\documentclass[a4paper,UKenglish,cleveref, autoref, thm-restate]{lipics-v2021}
\pdfoutput=1

\title{Structural Parameterization of Steiner Tree Packing}

\author{Niko Hastrich}{Saarland University and Max-Planck-Institute for Informatics, Saarbrücken, Germany}{niko.hastrich@cs.uni-saarland.de}{https://orcid.org/0000-0002-1825-0097}{This work is part of the project TIPEA that has received funding from the European Research Council (ERC) under the European Union’s Horizon 2020 research and innovation programme (grant agreement No 850979).}
\author{Kirill Simonov}{Department of Informatics, University of Bergen, Norway}{k.simonov@uib.no}{https://orcid.org/0000-0001-9436-7310}{}
\authorrunning{N. Hastrich and K. Simonov} 

\Copyright{Niko Hastrich and Kirill Simonov} 



\category{} 

\relatedversion{} 

\nolinenumbers 

\hideLIPIcs  

\EventEditors{John Q. Open and Joan R. Access}
\EventNoEds{2}
\EventLongTitle{42nd Conference on Very Important Topics (CVIT 2016)}
\EventShortTitle{CVIT 2016}
\EventAcronym{CVIT}
\EventYear{2016}
\EventDate{December 24--27, 2016}
\EventLocation{Little Whinging, United Kingdom}
\EventLogo{}
\SeriesVolume{42}
\ArticleNo{23}

\usepackage{mathtools}
\usepackage{booktabs}
\usepackage{bm}
\usepackage{stackengine}
\usepackage{scalerel}
\usepackage{todonotes}
\usepackage{array}

\theoremstyle{plain}
\newtheorem{reductionrule}[theorem]{Reduction Rule}
\crefname{reductionrule}{Reduction Rule}{Reduction Rules}

\newcommand\openbigstar[1][0.8]{%
  \scalerel*{%
    \stackinset{c}{-.125pt}{c}{}{\scalebox{#1}{\color{white}{$\bigstar$}}}{%
      $\bigstar$}%
  }{\bigstar}
}
\DeclareMathOperator{\ethOp}{\mathsf{ETH}}
\renewcommand{\eth}{\ethOp}
\DeclareMathOperator{\ilp}{\mathsf{ILP}}
\DeclareMathOperator{\xmod}{mod}
\DeclareMathOperator{\gstp}{\mathsf{GSTP}}
\DeclareMathOperator{\stp}{\mathsf{STP}}
\DeclareMathOperator{\edp}{\mathsf{EDP}}
\DeclareMathOperator{\fpt}{\mathsf{FPT}}

\DeclareMathOperator{\wone}{\mathsf{W}[1]}
\DeclareMathOperator{\wonehard}{\mathsf{W}[1]\text{-hard}}
\DeclareMathOperator{\p}{\mathsf{P}}
\DeclareMathOperator{\np}{\mathsf{NP}}
\DeclareMathOperator{\nphard}{\mathsf{NP}\text{-hard}}
\DeclareMathOperator{\pnphard}{\mathsf{paraNP}\text{-hard}}

\DeclareMathOperator{\pnph}{\mathsf{paraNP}\text{-h}}
\DeclareMathOperator{\w1h}{\mathsf{W}[1]\text{-h}}
\newcommand{\restrictFun}[2]{{\mathchoice{{#1}\big\vert_{#2}}{{#1}\big\vert_{#2}}{{#1}\big\vert_{#2}}{{#1}\big\vert_{#2}}}}
\newcommand{\natint}[1]{{\mathchoice{\left[#1\right]}{[#1]}{[#1]}{[#1]}}}
\newcommand{\natintZ}[1]{{\mathchoice{\left[#1\right]_0}{[#1]_0}{[#1]_0}{[#1]_0}}}
\newcommand{\abs}[1]{{\mathchoice{\left|#1\right|}{{|#1|}}{{|#1|}}{{|#1|}}}}
\renewcommand{\vec}[1]{{\bm{#1}}}
\ifdefined\N\else\DeclareMathOperator{\N}{\mathbb{N}}\fi
\ifdefined\R\else\DeclareMathOperator{\R}{\mathbb{R}}\fi

\DeclareMathOperator{\twOp}{tw}
\newcommand{\tw}[1]{{\mathchoice{\twOp\left(#1\right)}{\twOp(#1)}{\twOp(#1)}{\twOp(#1)}}}

\DeclareMathOperator{\fracNumOp}{fn}
\newcommand{\fracNum}[1]{{\mathchoice{\fracNumOp\left(#1\right)}{\fracNumOp(#1)}{\fracNumOp(#1)}{\fracNumOp(#1)}}}

\DeclareMathOperator{\fvsOp}{fvs}
\newcommand{\fvs}[1]{{\mathchoice{\fvsOp\left(#1\right)}{\fvsOp(#1)}{\fvsOp(#1)}{\fvsOp(#1)}}}
\DeclareMathOperator{\vcOp}{vc}
\newcommand{\vc}[1]{{\mathchoice{\vcOp\left(#1\right)}{\vcOp(#1)}{\vcOp(#1)}{\vcOp(#1)}}}
\DeclareMathOperator{\tcwOp}{tcw}
\newcommand{\tcw}[1]{{\mathchoice{\tcwOp\left(#1\right)}{\tcwOp(#1)}{\tcwOp(#1)}{\tcwOp(#1)}}}

\DeclareMathOperator{\stcwOp}{stcw}
\newcommand{\stcw}[1]{{\mathchoice{\stcwOp\left(#1\right)}{\stcwOp(#1)}{\stcwOp(#1)}{\stcwOp(#1)}}}
\DeclareMathOperator{\fenOp}{fen}
\newcommand{\fen}[1]{{\mathchoice{\fenOp\left(#1\right)}{\fenOp(#1)}{\fenOp(#1)}{\fenOp(#1)}}}
\DeclareMathOperator{\maxDegOp}{\max\!\text{-}\!\degOp}
\newcommand{\maxDeg}[1]{{\mathchoice{\maxDegOp\left(#1\right)}{\maxDegOp(#1)}{\maxDegOp(#1)}{\maxDegOp(#1)}}}
\newcommand{\inv}[1]{#1^{-1}}
\DeclareMathOperator{\degOp}{deg}
\renewcommand{\deg}[2][]{{\mathchoice{\degOp_{#1}\left(#2\right)}{\degOp_{#1}(#2)}{\degOp_{#1}(#2)}{\degOp_{#1}(#2)}}}
\DeclareMathOperator{\OOp}{\mathcal{O}}
\renewcommand{\O}[1]{{\mathchoice{\OOp\left(#1\right)}{\OOp(#1)}{\OOp(#1)}{\OOp(#1)}}}

\newcommand{\Ostar}[1]{{\mathchoice{\OOp^\ast\left(#1\right)}{\OOp^\ast(#1)}{\OOp^\ast(#1)}{\OOp^\ast(#1)}}}
\newcommand{\OstarLR}[1]{\OOp^\ast\left(#1\right)}
\DeclareMathOperator{\imageOp}{img}
\newcommand{\image}[1]{{\mathchoice{\imageOp\left(#1\right)}{\imageOp(#1)}{\imageOp(#1)}{\imageOp(#1)}}}

\DeclareMathOperator{\codomainOp}{cod}
\newcommand{\codomain}[1]{{\mathchoice{\codomainOp\left(#1\right)}{\codomainOp(#1)}{\codomainOp(#1)}{\codomainOp(#1)}}}
\DeclareMathOperator{\identityOp}{id}
\newcommand{\identity}[1]{{\mathchoice{\identityOp\left(#1\right)}{\identityOp(#1)}{\identityOp(#1)}{\identityOp(#1)}}}
\DeclareMathOperator{\compOp}{comp}
\newcommand{\comp}[1]{{\mathchoice{\compOp\left(#1\right)}{\compOp(#1)}{\compOp(#1)}{\compOp(#1)}}}
\DeclareMathOperator{\cutEdgesOp}{\delta}
\newcommand{\cutEdges}[2][]{{\mathchoice{\cutEdgesOp_{#1}\left(#2\right)}{\cutEdgesOp_{#1}(#2)}{\cutEdgesOp_{#1}(#2)}{\cutEdgesOp_{#1}(#2)}}}

\DeclareMathOperator{\adhesionOp}{adh}
\newcommand{\adhesion}[2][]{{\mathchoice{\adhesionOp_{#1}\left(#2\right)}{\adhesionOp_{#1}(#2)}{\adhesionOp_{#1}(#2)}{\adhesionOp_{#1}(#2)}}}
\DeclareMathOperator{\childrenOp}{chil}
\newcommand{\children}[2][]{{\mathchoice{\childrenOp_{#1}\left(#2\right)}{\childrenOp_{#1}(#2)}{\childrenOp_{#1}(#2)}{\childrenOp_{#1}(#2)}}}
\DeclareMathOperator{\thinChildrenOp}{t\text{-}chil}
\newcommand{\thinChildren}[2][]{{\mathchoice{\thinChildrenOp_{#1}\left(#2\right)}{\thinChildrenOp_{#1}(#2)}{\thinChildrenOp_{#1}(#2)}{\thinChildrenOp_{#1}(#2)}}}
\DeclareMathOperator{\boldChildrenOp}{b\text{-}chil}
\newcommand{\boldChildren}[2][]{{\mathchoice{\boldChildrenOp_{#1}\left(#2\right)}{\boldChildrenOp_{#1}(#2)}{\boldChildrenOp_{#1}(#2)}{\boldChildrenOp_{#1}(#2)}}}
\newcommand{\augment}[2]{{#1}^{#2}}
\newcommand{\augmentClique}[2]{#1^{\mathrm{cliq}(#2)}}
\newcommand{\augmentVert}[2]{#1^{\mathrm{vert}(#2)}}
\DeclareMathOperator{\augOp}{aug}
\newcommand{\aug}[2][{}]{{\mathchoice{\augOp_{#1}\left(#2\right)}{\augOp_{#1}(#2)}{\augOp_{#1}(#2)}{\augOp_{#1}(#2)}}}
\newcommand{\invAug}[2][{}]{{\mathchoice{\augOp_{#1}^{-1}\left(#2\right)}{\augOp_{#1}^{-1}(#2)}{\augOp_{#1}^{-1}(#2)}{\augOp_{#1}^{-1}(#2)}}}
\DeclareMathOperator{\demandOp}{dem}
\newcommand{\demand}[2][]{{\mathchoice{\demandOp_{#1}\left(#2\right)}{\demandOp_{#1}(#2)}{\demandOp_{#1}(#2)}{\demandOp_{#1}(#2)}}}
\DeclareMathOperator{\assignOp}{assign}
\newcommand{\assign}[4][]{{\mathchoice{\assignOp_{#1}\left(#2,#3,#4\right)}{\assignOp_{#1}(#2,#3,#4)}{\assignOp_{#1}(#2,#3,#4)}{\assignOp_{#1}(#2,#3,#4)}}}
\DeclareMathOperator{\assignSetsOp}{assignSets}
\newcommand{\assignSets}[4][]{{\mathchoice{\assignSetsOp_{#1}\left(#2,#3,#4\right)}{\assignSetsOp_{#1}(#2,#3,#4)}{\assignSetsOp_{#1}(#2,#3,#4)}{\assignSetsOp_{#1}(#2,#3,#4)}}}
\DeclareMathOperator{\supplyOp}{supl}
\newcommand{\supply}[2][]{{\mathchoice{\supplyOp_{#1}\left(#2\right)}{\supplyOp_{#1}(#2)}{\supplyOp_{#1}(#2)}{\supplyOp_{#1}(#2)}}}
\DeclareMathOperator{\confInstOp}{confInst}
\newcommand{\confInst}[4][]{{\mathchoice{\confInstOp_{#1}\left(#2,#3,#4\right)}{\confInstOp_{#1}(#2,#3,#4)}{\confInstOp_{#1}(#2,#3,#4)}{\confInstOp_{#1}(#2,#3#4)}}}
\DeclareMathOperator{\signatureOp}{sig}
\newcommand{\signature}[1]{{\mathchoice{\signatureOp\left(#1\right)}{\signatureOp(#1)}{\signatureOp(#1)}{\signatureOp(#1)}}}
\newcommand{\terminalsContainedS}{\mathcal{T}^\star}
\newcommand{\confInstSupplySetDefinition}{\{X \subseteq S\cap V(G)\mid \supply{X} > 0\}}

\DeclareMathOperator{\allViable}{\mathcal{V}}
\newcommand{\rhoLinRep}[1]{\bm{\rho\text{-}linRep}(#1)}
\newcommand{\assignLinRep}[1]{\bm{assign\text{-}linRep}(#1)}
\newcommand{\selectorLinRep}{\bm{selector\text{-}linRep}}
\DeclareMathOperator{\crossLinkOp}{cross}
\newcommand{\crossLink}[2][]{{\mathchoice{\crossLinkOp_{#1}\left(#2\right)}{\crossLinkOp_{#1}(#2)}{\crossLinkOp_{#1}(#2)}{\crossLinkOp_{#1}(#2)}}}
\DeclareMathOperator{\demandCrossLinkOp}{d-cross}
\newcommand{\demandCrossLink}[1]{{\mathchoice{\demandCrossLinkOp\left(#1\right)}{\demandCrossLinkOp(#1)}{\demandCrossLinkOp(#1)}{\demandCrossLinkOp(#1)}}}
\DeclareMathOperator{\enumerateCrossLinkBaseOp}{\eta}
\newcommand{\enumerateCrossLinkOp}[1]{\enumerateCrossLinkBaseOp_{#1}}
\newcommand{\enumerateCrossLink}[2]{{\mathchoice{\enumerateCrossLinkOp{#1}\left(#2\right)}{\enumerateCrossLinkOp{#1}(#2)}{\enumerateCrossLinkOp{#1}(#2)}{\enumerateCrossLinkOp{#1}(#2)}}}
\newcommand{\invEnumerateCrossLink}[2]{{\mathchoice{\inv{\enumerateCrossLinkOp{#1}}\left(#2\right)}{\inv{\enumerateCrossLinkOp{#1}}(#2)}{\inv{\enumerateCrossLinkOp{#1}}(#2)}{\inv{\enumerateCrossLinkOp{#1}}(#2)}}}
\DeclareMathOperator{\boundaryOp}{\partial}
\newcommand{\boundary}[1]{{\mathchoice{\boundaryOp_{#1}}{\boundaryOp_{#1}}{\boundaryOp_{#1}}{\boundaryOp_{#1}}}}
\DeclareMathOperator{\pastAssignBaseOp}{pastAssign}
\newcommand{\pastAssignOp}[1]{\pastAssignBaseOp_{#1}}
\newcommand{\pastAssign}[2]{{\mathchoice{\pastAssignOp{#1}\left(#2\right)}{\pastAssignOp{#1}(#2)}{\pastAssignOp{#1}(#2)}{\pastAssignOp{#1}(#2)}}}
\newcommand{\invPastAssign}[2]{{\mathchoice{\inv{\pastAssignOp{#1}}\left(#2\right)}{\inv{\pastAssignOp{#1}}(#2)}{\inv{\pastAssignOp{#1}}(#2)}{\inv{\pastAssignOp{#1}}(#2)}}}
\DeclareMathOperator{\pastPartBaseOp}{pastPart}
\newcommand{\pastPart}[1]{\pastPartBaseOp_{#1}}
\DeclareMathOperator{\futurePartBaseOp}{futPart}
\newcommand{\futurePart}[1]{\futurePartBaseOp_{#1}}
\DeclareMathOperator{\localSharedMapBaseOp}{\mu}
\newcommand{\localSharedMapOp}[1]{\localSharedMapBaseOp_{#1}}
\newcommand{\localSharedMap}[2]{{\mathchoice{\localSharedMapOp{#1}\left(#2\right)}{\localSharedMapOp{#1}(#2)}{\localSharedMapOp{#1}(#2)}{\localSharedMapOp{#1}(#2)}}}
\newcommand{\invLocalSharedMap}[2]{{\mathchoice{\inv{\localSharedMapOp{#1}}\left(#2\right)}{\inv{\localSharedMapOp{#1}}(#2)}{\inv{\localSharedMapOp{#1}}(#2)}{\inv{\localSharedMapOp{#1}}(#2)}}}
\newcommand{\terminalsContainedNode}[1]{\terminalsContainedS_{#1}}
\DeclareMathOperator{\enumerateSharedOp}{\xi}
\newcommand{\enumerateShared}[1]{{\mathchoice{\enumerateSharedOp\left(#1\right)}{\enumerateSharedOp(#1)}{\enumerateSharedOp(#1)}{\enumerateSharedOp(#1)}}}
\newcommand{\invEnumerateShared}[1]{{\mathchoice{\inv\enumerateSharedOp\left(#1\right)}{\inv\enumerateSharedOp(#1)}{\inv\enumerateSharedOp(#1)}{\inv\enumerateSharedOp(#1)}}}
\DeclareMathOperator{\unknownTerminalInOp}{\bot}
\newcommand{\unknownTerminalIn}[1]{\unknownTerminalInOp_{#1}}
\DeclareMathOperator{\blockedIndex}{\#}
\newcommand{\dpTransitionGraph}[1][s]{\widetilde{J_{#1}}}
\DeclareMathOperator{\subDivOp}{\mu}
\newcommand{\subDivA}[1]{\subDivOp_1(#1)}
\newcommand{\subDivB}[1]{\subDivOp_2(#1)}
\newcommand{\subDivE}[1]{\subDivOp_e(#1)}
\newcommand{\subDivP}[1]{\subDivOp_P(#1)}
\newcommand{\subDivV}[1]{\subDivOp_V(#1)}
\DeclareMathOperator{\enumerateTerminalOp}{\enumerateCrossLinkBaseOp}
\newcommand{\enumerateTerminal}[1]{{\mathchoice{\enumerateTerminalOp\left(#1\right)}{\enumerateTerminalOp(#1)}{\enumerateTerminalOp(#1)}{\enumerateTerminalOp(#1)}}}
\newcommand{\invEnumerateTerminal}[1]{{\mathchoice{\inv{\enumerateTerminalOp}\left(#1\right)}{\inv{\enumerateTerminalOp}(#1)}{\inv{\enumerateTerminalOp}(#1)}{\inv{\enumerateTerminalOp}(#1)}}}
\DeclareMathOperator{\extendOp}{ext}
\newcommand{\extend}[2]{{\mathchoice{\extendOp\left(#1,#2\right)}{\extendOp(#1,#2)}{\extendOp(#1,#2)}{\extendOp(#1,#2)}}}
\DeclareMathOperator{\mergeOp}{merge}
\newcommand{\merge}[2]{{\mathchoice{\mergeOp\left(#1,#2\right)}{\mergeOp(#1,#2)}{\mergeOp(#1,#2)}{\mergeOp(#1,#2)}}}
\DeclareMathOperator{\distributeOp}{distr}
\newcommand{\distribute}[2]{{\mathchoice{\distributeOp\left(#1,#2\right)}{\distributeOp(#1,#2)}{\distributeOp(#1,#2)}{\distributeOp(#1,#2)}}}

\begin{document}

\maketitle

\begin{abstract}
  \textsc{Steiner Tree Packing} \((\stp)\) is a notoriously hard problem in classical complexity theory, which is of practical relevance to VLSI circuit design.
  Previous research has approached this problem by providing heuristic or approximate algorithms.
  In this paper, we show the first \(\fpt\) algorithms for \(\stp\) parameterized by structural parameters of the input graph.
  In particular, we show that \(\stp\) is fixed-parameter tractable by the tree-cut width as well as the fracture number of the input graph.

  To achieve our results, we generalize techniques from \textsc{Edge-Disjoint Paths} (\(\edp\)) to \textsc{Generalized Steiner Tree Packing} (\(\gstp\)), which generalizes both \(\stp\) and \(\edp\).
  First, we derive the notion of the augmented graph for \(\gstp\) analogous to \(\edp\).
  We then show that \(\gstp\) is \(\fpt\) by
  \begin{itemize}
      \item the tree-cut width of the augmented graph,
      \item the fracture number  of the augmented graph,
      \item the slim tree-cut width of the input graph.
  \end{itemize}
  The latter two results were previously known for \(\edp\); our results generalize these to \(\gstp\) and improve the running time for the parameter fracture number.
  On the other hand, it was open whether \(\edp\) is \(\fpt\) parameterized by the tree-cut width of the augmented graph, despite extensive research on the structural complexity of the problem.
  We settle this question affirmatively.
\end{abstract}

\section{Introduction \label{sec:intro}}
\label{sec:org721bd70}
In \textsc{Steiner tree packing} (\(\stp\)), we are given a triple \((G, T, d)\), where \(G\) is a graph, \(T \subseteq V(G)\) is the \emph{terminal set}, and \(d \in \N^+\) is the \emph{demand}.
The goal is to decide whether there are \(d\) edge-disjoint trees \(F_1, F_2,\dots, F_d\) in \(G\) which all include the vertices \(T\).
We can see that this problem is polynomial time solvable if \(\abs{T} \leq 2\) or \(d = 1\).
However, if \(\abs{T} \geq 3\) or \(d \geq 2\), this problem becomes \(\nphard\)~\cite{Kaski04,AazamiCJ12}.
So, unless \(\p = \np\), there is no algorithm deciding each instance correctly in polynomial time.
Therefore, in order to give a polynomial time algorithm, we have to resort to heuristic solutions or abandon the goal of solving each instance.

\looseness=-1
Despite this complexity, \(\stp\)—and the related problems \(\edp\) and \(\gstp\), which we introduce shortly—has applications in practical fields like VLSI circuit design~\cite{Luk85,BursteinP83,CohoonH88,MartinW93,Pulleyblank95,GrotschelMW97}.
Here, the edges represent wires and the terminals represent endpoints that need to be connected.
Additionally, this problem finds applications in designing computer networks for multicasting~\cite{OfekY97,ChenGY98,GargKKP03}, which in turn has applications for video-conferencing~\cite{Zhao+14}.
Given this problem's significance, efficient algorithms are of high interest.
These are typically based on heuristics~\cite{Luk85,BursteinP83,CohoonH88,OfekY97,ChenGY98} or integer programming~\cite{GrotschelMW97,ChenGY98} as opposed to exploiting structural properties of specific instances.
These approaches typically provide no guarantees on the quality of their output compared to an optimal solution.

For \(\stp\), there are also algorithms known, which provide such guarantees.
If there are \(d\) edge-disjoint subgraphs connecting \(T\), we know that \(T\) is \(d\)-edge-connected (i.e.,~we need to remove at least \(d\) edges to disconnect \(T\)).
Kriesell~\cite{Kriesell03} proved that for each \(n \in \N\) there is a function \(f_n \colon \N \to \N\) such that if \(\abs{T} \leq n\) and \(T\) is \(f_n(d)\)-edge-connected, then there are \(d\) edge-disjoint subgraphs connecting \(T\) in \(G\).
In the same paper, Kriesell conjectured that these \(f_n\) are bounded by a single function.
Concretely, he conjectured that if \(T\) is \(2d\)-edge-connected, then \(G\) contains \(d\) edge-disjoint trees connecting \(T\).

This conjecture is known to be true for many special cases.
As an example, this conjecture has long been known to be true for \(T = V(G)\)~\cite{NashWilliams61,Tutte61}.
This means, if \(G\) is \(2d\)-edge-connected, it has at least \(d\) edge-disjoint spanning trees.
On the other end of the size of \(T\), we know that this conjecture is true for \(\abs{T} \leq 5\)~\cite{JainMS03}.
Additionally, the conjecture is known to be correct if \(G\) is Eulerian~\cite{Kriesell03}.

The first result proving, without additional assumptions on \(G, d\) and \(T\), that there is a constant \(c \in \R\) such that if \(T\) is \(cd\)-edge-connected, there are \(d\) edge-disjoint Steiner trees in \(G\), is due to Lau~\cite{Lau04}, who proved that \(c = 26\) is a valid choice.
This was later improved to \(c=6.5\)~\cite{WestW12}.
The state-of-the-art result is that there are \(d\) edge-disjoint Steiner trees, if \(T\) is \((5d+4)\)-edge-connected~\cite{DeVosMP16}.
These results are all obtained by providing polynomial time algorithms that given a \(cd\)-edge-connected graph, output \(d\) edge-disjoint Steiner trees.
So, the last result can be seen as a \(\frac{1}{5}-\O{1/\mathrm{OPT}}\) approximation algorithm.

Still, none of these algorithms provide exact solutions in polynomial time.
In this paper, we approach this problem from the other direction.
We significantly broaden the class of instances for which there is a known polynomial time algorithm.
We achieve this by exploiting inherent structure that exists in some instances.
More concretely, we study \(\stp\) parameterized by structural properties and develop algorithms running in \(\fpt\)-time.

In this paradigm of algorithm design, we consider instances where the underlying graphs exhibit some additional structure.
This structure is captured by an additional number \(k\in\N\)—the parameter—that is provided to algorithms as input.
Now, we try to find algorithms that are able to exploit this structure.
Ideally, there is a constant \(c \in \N\) such that the algorithm can decide each instance of size \(n\) in time \(\O{n^c}\).
Roughly speaking, if such an algorithm exists, we call a problem fixed-parameter tractable (\(\fpt\)) by the considered parameter.

An easy example of a structural parameter is the size of the smallest vertex cover.
Assume we know that there is a vertex cover \(S\) of size \(k \in \N\) in our graph.
An example of a problem, where we can exploit this structure is \textsc{Clique}.
In this problem, we are given a graph \(G\), an integer \(\ell \in \N\), and we need to decide whether there is a clique of size \(\ell\) in \(G\).
Now, consider any clique \(C\) in \(G\).
Notice that it contains at most one vertex of \(V(G) \setminus S\).
We now iterate through all \(X\subseteq S\) and search for the largest clique that contains \(X\) and at most one additional vertex of \(V(G) \setminus S\).
For each of the \(2^k\) different choices for \(X\), this can be done in \(\O{\abs{V(G)} + \abs{E(G)} + k^2}\leq \O{k\abs{V(G)}}\).
Thus, given \(S\), we can compute the largest clique in \(G\) in time \(\O{2^kk\abs{V(G)}}\).
Using the fact that we can compute \(S\) in time \(\O{2^kk\abs{V(G)}}\) given \(k\)~\cite{Mehlhorn84a}, we can solve \textsc{Clique} in linear time for every \(k\).
Thus, \textsc{Clique} is \(\fpt\) by the size of the smallest vertex cover of \(G\).

Before this paper, there was no known \(\fpt\) algorithm parameterized by structural parameters for \(\stp\).
In fact for one of the most widely used structural parameter—treewidth—\allowbreak{}this is unlikely.
This can be seen, as \(\stp\) has \textsc{Integer 2-Commodity Flow} as a special case~\cite{AazamiCJ12} and this problem has recently been shown to be \(\wonehard\) parameterized by treewidth~\cite{BodlaenderMOPL23}.
Therefore, \(\stp\) is \(\wonehard\) parameterized by treewidth, and it is believed~\cite[Chapter 13]{CyganFKLMPPS15} that there is no \(\fpt\) algorithm for a parameterized problem that is \(\wonehard\).
To this date, the only known \(\fpt\) algorithm for \(\stp\) is due to Robertson and Seymour~\cite{RobertsonS95b}, showing that \(\stp\) is fixed-parameter tractable parameterized by \(\abs{T} + d\).

To develop \(\fpt\) algorithms for \(\stp\) parameterized by structural parameters, we consider the closely related \textsc{Edge-Disjoint Paths} (\(\edp\)) problem.
An \(\edp\) instance is a tuple \((G, \mathcal{T})\), where \(G\) is a graph and \(\mathcal{T} \subseteq \binom{V(G)}{2}\) is a set of terminal pairs in \(V(G)\).
The task now is to decide if for each \(\{u,v\} \in \mathcal{T}\) there is an \(uv\)-path \(P_{uv}\) in \(G\) such that all \(\{P_{uv}\}_{\{u,v\} \in \mathcal{T}}\) are pairwise edge-disjoint.
This problem is famously \(\fpt\) with respect to \(\abs{\mathcal{T}}\)~\cite{RobertsonS95b}.
However, it is notoriously hard with respect to classical structural parameters like treewidth.
In fact, it is even \(\nphard\) on complete bipartite graphs where one bipartition contains 3 vertices~\cite{Fleszar2018Sep}.
These graphs have vertex cover number at most 3, which rules out an \(\fpt\) algorithm for \(\edp\) parameterized by many classical structural parameters—like treewidth, fracture number, size of the smallest feedback vertex set, size of the smallest vertex cover—unless \(\p = \np\).

However, there are some \(\fpt\) algorithms known for \(\edp\).
They fall roughly in two categories.
The first category are \(\fpt\) algorithms with respect to structural parameters that are based on edge-cuts like slim tree-cut width, rather than vertex-cuts like treewidth~\cite{GanianOR20,GanianK22,BrandCGHK22,GanianOR21}.
The second category of \(\fpt\) algorithms based on structural parameters, do not consider the parameter with respect to the host graph \(G\), but rather the augmented graph \(G + \mathcal{T}\), which is the graph \(G\) with an edge inserted among every terminal pair~\cite{GanianOR21}.
This helps to take into account where in the graph the terminal pairs are located and how much they ``influence'' each other.
The problem with transferring these results directly to \(\stp\) is that although the \(\edp\) and the \(\stp\) problem are closely related, there is no simple reduction known between instances of one problem to the other.

However, both are special cases of \textsc{Generalized Steiner Tree Packing} (\(\gstp\)), which is just \(\stp\) with multiple different terminal sets.
Formally, we are given a triple \((G, \mathcal{T}, d)\), where \(G\) is the underlying graph, \(\mathcal{T} \subseteq 2^{V(G)}\) is the set of terminal sets, and \(d \colon \mathcal{T} \to \N^+\) gives the demand for each terminal set.
Our task is to decide whether there is a set of pairwise edge-disjoint, connected subgraphs \(\mathcal{F}\) of \(G\) and an assignment function \(\pi \colon \mathcal{F} \to \mathcal{T}\).
This assignment function needs to satisfy that every solution subgraph \(F \in \mathcal{F}\) is assigned to a terminal set \(\pi(F)\), which is contained in \(F\) (i.e.,~\(\pi(F) \subseteq V(F)\)).
Additionally, for every terminal set \(T \in \mathcal{T}\), the assignment function needs to assign \(d(T)\) many solution subgraphs to \(T\).

In this paper, we generalize all known \(\fpt\) algorithms for \(\edp\) parameterized by structural parameters to \(\gstp\), which makes them much more widely applicable.
Firstly, this allows us to apply them to \(\stp\).
Secondly, in doing so, we also discover new \(\fpt\) algorithms for \(\edp\), which positively settles the open question, whether \(\edp\) is \(\fpt\) with respect to the tree-cut width of the augmented graph in the affirmative.
Finally, we also improve the running time of many known \(\fpt\) algorithms for \(\edp\).
See an overview of our results in comparison to known results in \Cref{tbl:overview-results}.

\begin{table}
\centering
\Crefname{theorem}{Thm.}{Thm.}
\Crefname{corollary}{Cor.}{Cor.}
\begin{tabular}{lp{1.75cm}<{\centering}p{1.75cm}<{\centering}p{1.75cm}<{\centering}p{1.75cm}<{\centering}p{1.75cm}<{\centering}p{1.75cm}<{\centering}}
\toprule
& \multicolumn{2}{c}{\(\bm{\stp}\)} & \multicolumn{2}{c}{\(\bm{\edp}\)} & \multicolumn{2}{c}{\(\bm{\gstp}\)}\\
 & Host & Aug. & Host & Aug. & Host & Aug. \\
\midrule

\(\twOp\) & \multicolumn{2}{c}{\(\w1h\)} & \(\pnph\) & \(\w1h\) & \(\pnph\) & \(\w1h\) \\
& \multicolumn{2}{c}{\cite{AazamiCJ12,BodlaenderMOPL23}} & \cite{Fleszar2018Sep} & \cite{GanianOR21} & \cite{Fleszar2018Sep} & \cite{GanianOR21} \\

\(\twOp + D\) & \multicolumn{2}{c}{\(\fpt\)~\(\bigstar\)} & \multicolumn{2}{c}{\(\fpt\)~\(\openbigstar\)} & \multicolumn{2}{c}{\(\fpt\)~\(\bigstar\)}\\
& \multicolumn{2}{c}{\Cref{stmt:gstp-fpt-tw-D}} & \multicolumn{2}{c}{\Cref{stmt:gstp-fpt-tw-D}} & \multicolumn{2}{c}{\Cref{stmt:gstp-fpt-tw-D}}\\

\(\fracNumOp\) & \multicolumn{2}{c}{\(\fpt\)~\(\bigstar\)} & \(\pnph\) & \(\fpt\)~\(\openbigstar\) & \(\pnph\) & \(\fpt\)~\(\bigstar\) \\
& \multicolumn{2}{c}{\Cref{stmt:augmented-frag-fpt}} & \cite{Fleszar2018Sep} & \Cref{stmt:augmented-frag-runtime} & \cite{Fleszar2018Sep} & \Cref{stmt:augmented-frag-fpt} \\

\(\tcwOp\) & \multicolumn{2}{c}{\(\fpt\)~\(\bigstar\)} & \(\w1h\) & \(\fpt\)~\(\bigstar\) & \(\w1h\) & \(\fpt\)~\(\bigstar\) \\
& \multicolumn{2}{c}{\Cref{stmt:tcw-stp-fpt}} & \cite{GanianOR20} & \Cref{stmt:gstp-fpt-tcw} & \cite{GanianOR20} & \Cref{stmt:gstp-fpt-tcw} \\

\(\stcwOp\) & \multicolumn{2}{c}{\(\fpt\)~\(\bigstar\)} & \multicolumn{2}{c}{\(\fpt\)} & \multicolumn{2}{c}{\(\fpt\)~\(\bigstar\)} \\
& \multicolumn{2}{c}{\Cref{stmt:gstp-fpt-stcw}} & \multicolumn{2}{c}{\cite{GanianK22,BrandCGHK22}} &\multicolumn{2}{c}{\Cref{stmt:gstp-fpt-stcw}} \\

\bottomrule
\end{tabular}
\Crefname{theorem}{Theorem}{Theorems}
\Crefname{corollary}{Corollary}{Corollaries}
\caption{\label{tbl:overview-results} The parameterized complexity of \(\stp\), \(\edp\), and \(\gstp\) with respect to structural parameters.
The structural parameters are taken with respect to the host graph or the augmented graph.
The value of \(D\) is the sum of demands.
New results are highlighted with the symbol \(\bigstar\), results where we improve the running time are highlighted with the symbol \(\openbigstar\).
For terminology of the parameters we refer to \Cref{sec:prelims}.
We abbreviate -hard with -h.
}
\end{table}

This paper is structured as follows.
We start in \Cref{sec:prelims} by introducing relevant definitions and notation.
In \Cref{sec:augmentation} the notion of augmentation from \(\edp\) to \(\gstp\).
To the best of our knowledge, there are no known results regarding augmentation of problems, where the terminals are arbitrary sets and not pairs like in \(\edp\).
For this, we compare two different approaches and finally settle on one of them which will be used for the remainder of this paper.

Building on this definition, we show in \Cref{sec:fn*} that \(\gstp\) is \(\fpt\) by the fracture number of the augmented graph; due to the way we define augmentation for \(\gstp\), this result directly applies to \(\stp\) as well.
This is in stark contrast to the fact, that \(\edp\), and therefore \(\gstp\), is \(\pnphard\) by the fracture number of the host graph~\cite{Fleszar2018Sep}.
The running time we obtain is doubly exponential in the parameter.
This improves upon the triply exponential running time obtained by Ganian et~al.~\cite{GanianOR21} for \(\edp\) parameterized by the fracture number of the augmented graph.

Following this result, we focus on results using tree-cut decompositions in \Cref{sec:gstp-by-tcw*-and-stcw}.
First, we define an additional property for tree-cut decompositions, which we call being \emph{simple}.
Then, we show how to decide an instance of \(\gstp\) given a simple tree-cut decomposition.
For tree-cut decompositions, we are interested in two measures of the decomposition, its width and slim width, where the former measure is bounded by a function of the latter.
We then show how the last result can be applied to obtain an \(\fpt\) algorithm parameterized by the slim tree-cut width of the host graph and the tree-cut width of the augmented graph.
This directly implies that \(\edp\) is \(\fpt\) by the tree-cut width of the augmented graph, which was not known before.
As \(\edp\) is \(\wonehard\) parameterized by the tree-cut width of the host graph, this result can not be extended to find an \(\fpt\) algorithm parameterized by the tree-cut width of the host graph, unless \(\fpt = \wone\), which is believed to be false~\cite[Chapter 13]{CyganFKLMPPS15}.

As augmentation—even of a single terminal set—can increase the tree-cut width arbitrarily, the previous result does not directly translate to an \(\fpt\) algorithm for \(\stp\), in contrast to the case with fracture number.
In \Cref{sec:stp-by-tcw}, we examine this further.
We show that we can avoid this hurdle and provide an \(\fpt\) algorithm for \(\stp\) parameterized by tree-cut width of the host graph.
This is the central result of this paper and combines most of the other results obtained.
To achieve this, our \(\fpt\) algorithm distinguishes two cases.
If the tree-cut width does not increase by too much due to augmentation, we apply the result obtained in \Cref{sec:gstp-by-tcw*-and-stcw}.
Otherwise, we prove that we can discard instances where the demand is not small.
As the tree-cut width can be used to put an upper bound on the treewidth and treewidth is a more general parameter, we then develop an \(\fpt\) algorithm for \(\gstp\) parameterized by the sum of treewidth and demand.
This \(\fpt\) algorithm can be used to solve the case remaining to solve \(\stp\) parameterized by tree-cut width.
Additionally, this algorithm directly translates to an \(\fpt\) algorithm for \(\edp\) parameterized by the sum of treewidth and number of terminal pairs.
This was known before, but nevertheless we improve upon the known running time~\cite{ZhouTN00}.

We conclude this paper in \Cref{sec:conclusion} by summarizing our work and giving open research questions.


\section{Preliminaries}
\label{sec:prelims}
We denote the natural numbers with \(\N\) and the positive natural numbers with \(\N^+\).
Analogously we define the real numbers as \(\R\) and the positive real numbers as \(\R^+\).
Let \(f,g \colon \N \to \R^+\) be functions.
We write \(f = \O{g}\) if \(\limsup_{n \in \N} \frac{f(n)}{g(n)} \neq \infty\) and \(f = \Omega(g)\) if \(\liminf_{n \in \N} \frac{f(n)}{g(n)} > 0\).
We also use this notation as part of terms—like \(f(n) = 2^\O{g(n)}\), this means that there is a \(g' = \O{g}\) such that \(f(n) = \O{2^{g'(n)}}\).

For all \(i \in \N\), we define \(\natintZ{i} \coloneqq \{j \in \N \mid j \leq i\}\) and \(\natint{i} \coloneqq \natintZ{i} \setminus \{0\}\).
Note that \(\natint{0} = \emptyset\).
Let \(A\) be a finite set of size \(n \in \N^+\).
For all \(k \in \N\), we denote the set of subsets of \(A\) of size \(k\) with \(\binom{A}{k} \coloneqq \{S \subseteq A \mid \abs{S} = k\}\).
Additionally, we denote the number of such subsets of \(A\) with \(\binom{n}{k} \leq n^\O{k}\), as well.
The number of subsets of \(A\) with size at most \(k\) is \(\sum_{i \leq k} \binom{n}{i} \leq \binom{n + k}{k}\).
For all \(k \in \N^+\), the set of functions \(f\colon A\to \N\) with \(\sum_{a \in A} f(a) = k\) has size \(\binom{n + k -1}{k - 1}\) and can be enumerated in time \(\O{n\binom{n + k - 1}{k-1}}\).
The set of functions \(f\colon A\to \N\) with \(\sum_{a \in A} f(a) \leq k\) has size \(\binom{n + k}{k}\) and can be enumerated in time \(\O{n\binom{n + k}{k}}\).
We denote the identify function on \(A\) with \(\identity{A}\) and the powerset of \(A\) with \(2^A \coloneqq \{S\}_{S \subseteq A}\).

Let \(A, B\) be sets and \(f \colon A \to B\) a function.
We abbreviate the codomain \(B\) with \(\codomain{B}\).
For every \(X \subseteq A\), we refer to \(f\) restricted to \(X\) with \(\restrictFun{f}{X}\), the image of \(X\) under \(f\) as \(f(X) \coloneqq \{f(x)\}_{x \in X}\) and the image of \(f\) as \(\image{f} \coloneqq f(A)\).
For \(Y \subseteq B\), we refer to the inverse of \(Y\) under \(f\) by \(\inv f (Y) \coloneqq \{a \in A \mid f(a) \in Y\}\).
For all \(b \in B\), we abbreviate \(\inv f(\{b\})\) with \(\inv f(b)\) and if \(\abs{\inv f(b)} = 1\), we use this as a function where the result is a single value of \(A\) and not a singleton set of \(A\).

Let \(G = (V, E)\) be a graph.
We refer to its vertices by \(V(G)\) and its edges by \(E(G)\).
Unless stated otherwise, we only consider simple undirected graphs, that is graphs without loops or multiple edges.
We define the size of \(G\) as \(\abs{G} \coloneqq \abs{V(G)} + \abs{E(G)}\).
All of these definitions apply to subgraphs as well.
We denote the set of connected components of \(G\) as \(\comp{G}\), where the elements of this set are the maximally connected subgraphs of \(G\).
Wherever possible, we refer to vertices or edges by lower case letters, (sub/hyper)graphs by upper case letters, and sets of vertices, edges, or graphs by calligraphic letters.

Let \(S\subseteq V(G)\) be a vertex set.
We denote the graph vertex induced on \(G\) by \(S\) with \(G[S] \coloneqq (S, \{uv \in E \mid u,v \in S\})\).
Now, let \(S' \subseteq E(G)\) be an edge set.
We denote the graph edge induced on \(G\) by \(S'\) with \(G[S'] \coloneqq (\bigcup_{uv \in S'}\{u,v\}, S')\) as all edges in \(S'\) with their incident vertices.
We also use this notation, if \(S\) or \(S'\) are not necessarily contained in \(V(G)\) or \(E(G)\), respectively.
In these cases, we implicitly refer to \(G[S \cap V(G)]\) and \(G[S' \cap E(G)]\), respectively.
Whether we refer to a vertex or edge set follows from context.
We set the difference graphs \(G-S \coloneqq G[V(G) \setminus S]\) and \(G - S' \coloneqq G[E(G) \setminus S']\).
Let \(H\) be another graph.
We denote the union of \(G\) and \(H\) with \(G \cup H \coloneqq (V(G) \cup V(H), E(G) \cup E(H))\).
For a set of vertices \(S\), we write \(G + S \coloneqq (V(G) \cup S, E(G))\) as the smallest graph that contains \(G\) and all vertices in \(S\); and for a set of edges \(S'\), we write \(G + S' \coloneqq \{V(G) \cup \bigcup_{uv \in S'}\{u,v\}, E(G) \cup S'\}\) as the smallest graph that contains \(G\) and all edges in \(S'\).

Let \(uv \in E(G)\).
Subdividing \(uv\) introduces a new vertex \(s_{uv}\) into \(G\), and connects it to \(u\) and \(v\).
Formally, we obtain the graph \(G - uv + \{s_{uv}u, s_{uv}v\}\).
Contracting the edge \(uv\), merges \(u\) and \(v\) into a single vertex, while keeping vertices adjacent to \(u\) or \(v\) adjacent to the combined vertex.
We denote the graph after contracting \(uv\) by \(G / uv \coloneqq G - v + \{ux\}_{vx \in E(G)}\).
Suppressing a vertex \(v \in V\) is the operation of directly connecting all neighbors of \(v\) and removing \(v\) from \(G\).
Formally, the graph after suppressing \(v\) is \(G - v + \{xy\mid \{x,y\} \in \binom{N(v)}{2}\}\).

For \(v \in V(G)\), we denote with the open neighborhood with \(N_G(v) \coloneqq \{u \mid uv \in E(G)\}\) and the closed neighborhood with \(N_G[v] \coloneqq N_G(v) \cup \{v\}\).
The degree of a vertex is \(\deg[G]{v} = \abs{N_G(V)}\).
For any \(S \subseteq V(G)\), we denote the edges not completely contained in \(S\) or \(V(G) \setminus S\) by \(\cutEdges[G]{S}\).
We leave off the index, if the graph is clear from context.

We call \(G\) edgeless, if \(E(G) = \emptyset\).
A walk \(W\) in \(G\) is a sequence of edges, such that neighboring edges in \(W\) share a common vertex.
A path \(P\) is a walk where each vertex appears at most once.
A cycle \(C\) is a walk \(W\coloneqq w_1w_2\dots w_\ell\), such that \(w_1 = w_\ell\).
We call a cycle simple, if \(\ell \geq 3\) and \(w_1w_2\dots w_{\ell - 1}\) is a path.
We call \(G\) cyclefree, if there is no simple cycle contained in \(G\).
There is a path between two distinct vertices \(u, v \in V(G)\), if and only if for each \(S \subseteq V(G)\) with \(u \in S\) and \(v \notin S\), we have \(\cutEdges{S} \neq \emptyset\).
Similarly, \(G\) is connected if and only if for all \(\emptyset \subset S \subset V(G)\), the set \(\cutEdges{S} \neq \emptyset\).

We call \(\phi \colon V(G) \to V(H)\) a graph homomorphism, is for all \(uv \in E(G)\), we have \(\phi(u)\phi(v) \in E(H)\).
To capture this fact, we also write \(\phi \colon G \to H\).
If \(\phi\) is a bijection, we call it an isomorphism.
We denote with \(\phi_e \colon E(G) \to E(H)\) the function \(uv \mapsto \phi(u)\phi(v)\).

Let \(\phi_v \colon V(G) \to V(H)\) be an injection and \(\phi_e\) be a function from edges \(xy \in E(G)\) to \(\phi_v(x)\phi_v(y)\)-paths in \(H\), such that all \(\image{\phi_e}\) are edge-disjoint, we call \((\phi_v, \phi_e)\) an immersion from \(G\) into \(H\).

Now, let \(G = (V, \mathcal{E})\) be a hypergraph.
We call \(G\) minimally-connected if for every \(E \in \mathcal{E}\), the hypergraph \(G - E\) is no longer connected.

An instance of \textsc{Generalized Steiner Tree Packing} is a triple \(\mathscr{P} \coloneqq (G, \mathcal{T}, d)\), where \(G\) is a graph, \(\mathcal{T}\subseteq 2^{V(G)}\) is the set of terminal sets, and \(d \colon \mathcal{T}\to \N^+\) is the demand for each terminal set.
This instance is positive, if there is a tuple \((\mathcal{F}, \pi)\), where \(\mathcal{F}\) is a set of edge-disjoint, connected subgraphs of \(G\) and \(\pi \colon \mathcal{F}\to \mathcal{T}\) is an assignment function such that for all \(F \in \mathcal{F}\), we have \(\pi(F) \subseteq V(F)\), and for all \(T \in \mathcal{T}\), we have \(\abs{\inv\pi(T)} = d(T)\).
Note that, if an instance is positive, we can always obtain a solution where every \(F \in \mathcal{F}\) is a tree.
We define the size of an instance \(\mathscr{P}\) as \(\abs{\mathscr{P}} \coloneqq \abs{G} + \sum_{T \in \mathcal{T}} \abs{T}\).

An instance of \textsc{Edge-Disjoint Paths} is a triple \((G, \mathcal{T})\), where \(G\) is a graph, \(\mathcal{T}\subseteq \binom{V(G)}{2}\) is the set of terminal pairs.
This instance is positive, if the instance \((G, \mathcal{T}, T \mapsto 1)\) of \textsc{Generalized Steiner Tree Packing} is positive.
An instance of \textsc{Steiner Tree Packing} is a triple \((G, T, d)\), where \(G\) is a graph, \(T \subseteq V(G)\) are the terminals, and \(d \in \N\) is the demand.
This instance is positive, if the instance \((G, \{T\}, T \mapsto d)\) of \textsc{Generalized Steiner Tree Packing} is positive.
Since the assignment function of any solution is trivial, we may also refer to a solution by simply \(\mathcal{F}\).

In this paper, we use the paradigm of integer linear programs.
Here, we consider a vector of variables \(\vec{v} \in \N^n\) where \(n\in\N^+\) is the length of the vector.
We want to maximize or minimize a linear function of \(\vec{v}\), where \(\vec{v}\) needs to satisfy some linear constraints.
Each constraint can enforce equality, less-than, and greater-than constraints between two affine functions of \(\vec{v}\).
To check, whether there is a feasible \(\vec{v}\), we maximize the linear function \(\vec{x} \mapsto 0\).
It is known that this problem can be solved in time \(n^{\O{n}}\)~\cite{Kannan87}.

\subsection{Parameterized Complexity}
\label{sec:orgd494e74}
Let \(\Sigma\) be a finite alphabet.
We call \(L \subseteq \Sigma^* \times \N\) a parameterized problem.
For an instance \((x, k) \in \Sigma^* \times \N\), the value \(k\) is called the parameter.
We call \(L\) \emph{fixed-parameter tractable} if there is an algorithm, a constant \(c \in \R^+\), and a computable function \(f\), such that for all \((x, k) \in \Sigma^* \times \N\) the algorithm decides \((x,k) \in L\) in time at most \(f(k)\abs{x}^c\)~\cite[Chapter 1]{CyganFKLMPPS15}.

Consider \textsc{Independent Set} parameterized by the solution size.
That is, for a graph \(G\) and \(k \in \N\), we decide whether there is a \(S \subseteq V(G)\) with \(\abs{S} = k\) such that \(G[S]\) is edge-less.
A parameterized problem \(L\) is called \(\wonehard\), if there is are computable functions \(f,g\), a constant \(c \in \R^+\), and an algorithm \(\mathcal{A}\) that outputs for every instance \((G, k)\) of independent set parameterized by the solution size in time \(f(k)\abs{G}^c\) an equivalent instance \(\mathcal{A}(G,k)\) of \(L\), such that the parameter of \(\mathcal{A}(G, k)\) is at most \(g(k)\).
It is commonly believer that no \(\wonehard\) problem is fixed-parameter tractable~\cite[Chapter 13]{CyganFKLMPPS15}.
We call a parameterized problem \(\pnphard\), if it is \(\nphard\) to the language restricted to instances where the parameter is bounded by some constant.

The \emph{exponential time hypothesis} is the claim that there is a \(\delta \in \R^+\) such that any algorithm that decides \textsc{3-Sat} takes on instances with \(n\) variables at least time \(\Omega(2^{\delta n})\)~\cite{ImpagliazzoP99}.

Consider a parameterized problem \((I, k)\).
Let \(r \colon \N \times \N \to \N^+\) be a runtime function.
For a \(f \colon \N \to \N\), we write \(r(\abs{I}, k) = \Ostar{f(k)}\), if there is a polynomial function \(p\) such that \(r(\abs{I}, k) = \O{f(k)p(\abs{I})}\).
That is, the \(\OOp^\ast\) notation hides a polynomial factor in the size of the instance.

\subsubsection{Structural Parameters}
\label{sec:org1632479}
\begin{figure}[tp]
\centering
\includegraphics[width=\textwidth]{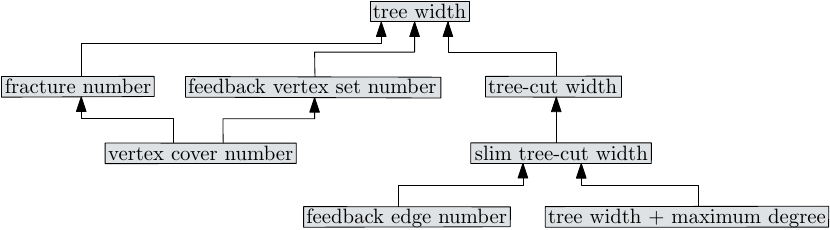}
\caption{\label{fig:relation-between-parameters} An overview of common structural parameters and their relation. We draw an edge from a parameter \(\alpha\) to a parameter \(\beta\), if given \(\alpha\) we can compute an upper bound on \(\beta\). A Hardness result for a parameter gives a similar results for all ancestors; an \(\fpt\)-algorithm for a parameter gives an \(\fpt\)-algorithm for all descendants.}
\end{figure}

In this paper, we are particularly interested in structural parameters.
We list the formal definition of all parameters considered in this paper here.
For an overview of their relation see \Cref{fig:relation-between-parameters}.
For all these definitions, let \(G\) be the considered graph.

Some of the parameters are based on a \emph{decomposition}, which contains a graph \(H\) on vertex-sets of \(G\).
We refer to the vertices of \(H\) as nodes or bags and to the edges of \(H\) as links.

\paragraph{Structural Parameters Based on Vertex Cuts}
The following parameters are based on vertex separators.
If they are bounded there is a small set of vertices such if we remove these vertices, the considered graph gets disconnected into considerably smaller parts.
What this means exactly varies from parameter to parameter.
However, all of them have in common, that the number of removed edges is not bounded by a function of the parameter.

\subparagraph{Vertex Cover Number}
Let \(S \subseteq V(G)\).
We call \(S\) a \emph{vertex-cover} if \(G - S\) is edgeless.
The size of the smallest vertex cover is called the vertex cover number \(\vc{G}\) of \(G\).
Finding the smallest vertex cover is \(\fpt\) by its size.

\subparagraph{Feedback Vertex Set Number}
Let \(S \subseteq V(G)\).
We call \(S\) a \emph{feedback vertex set} if \(G - S\) is cyclefree.
The size of the smallest feedback vertex set is called the feedback vertex set number \(\fvs{G}\) of \(G\).
Finding the smallest feedback vertex set is \(\fpt\) by its size.

\subparagraph{Fracture Number}
Let \(S \subseteq V(G)\).
We call \(S\) a \emph{fracture modulator} if every component in \(G - S\) contains at most \(\abs{S}\) vertices.
The size of the smallest fracture modulator is called the feedback vertex set number \(\fvs{G}\) of \(G\).
According to \Cref{stmt:fracture-modulator-algo}, we can find a smallest fracture modulator in \(\O{(2k - 1)^k\abs{G}}\).

\subparagraph{Treewidth}
Treewidth builds upon the notion of a tree decomposition of \(G\).

\begin{definition}
A \emph{tree decomposition} is a pair \((T, \mathcal{X})\), where \(T\) is a tree and \(\{X_t \subseteq V(G)\mid t \in V(T)\} \coloneqq \mathcal{X}\) is such that
\begin{itemize}
\item \(\bigcup_{t \in V(T)} X_t = V(G)\),
\item for every \(uv \in E(G)\), there is a \(t \in V(T)\) with \(\{u,v\} \subseteq X_t\),
\item for every \(u \in V(G)\), the subgraph induced on \(T\) by the bags containing \(u\) (i.e.,~\(T[\{t \in V(T) \mid u \in X_t\}]\)) is connected.\qedhere
\end{itemize}
\end{definition}

We define the width of a tree decomposition \((T, \mathcal{X})\) as \(\max_{t \in V(T)} \abs{X_t} - 1\) and the treewidth \(\tw{G}\) of \(G\) as the minimum width of any tree decomposition of \(G\).
There is an algorithm that for all \(w \in \N\) either outputs a tree decomposition of width \(5w+ 4\), or certifies that \(\tw{G} > w\).
This algorithm runs in time \(2^\O{w}\abs{V(G)}\)~\cite{BodlaenderDDFLP16}.

To better facilitate dynamic programming with tree decomposition, we introduce the notion of a \emph{nice} tree decomposition.
Contrary to normal tree decomposition, in a nice tree decomposition, the tree \(T\) is rooted at a vertex \(r\).
Additionally, for all leaves \(\ell \in V(T)\), we have \(X_\ell = X_r = \emptyset\).
Each inner node \(t \in V(T)\) has one of three types.
\begin{itemize}
\item \emph{Introduce Node} This node has exactly one child \(c \in V(T)\) and there is a \(v \in V(G) \setminus X_c\) such that \(X_t = X_c \cup \{v\}\).
\item \emph{Forget Node} This node has exactly one child \(c \in V(T)\) and there is a \(v \in X_c\) such that \(X_t = X_c \setminus \{v\}\).
\item \emph{Join Node} This node has exactly two children \(c, c' \in V(T)\) and \(X_c = X_{c'} = X_t\).
\end{itemize}

Given a tree decomposition of width \(w\), we can in time \(\O{w^2\max(\abs{V(G)},\abs{V(T)})}\) compute a nice tree decomposition of width \(w\) with at most \(w \abs{V(G)}\) nodes~\cite{Kloks94}.
As a shorthand notation, for each \(t \in V(T)\) denote with \(T_t\) the subgraph rooted at \(t\) and let \(Y_t \coloneqq \bigcup_{t' \in T_t} X_{t'}\) be all the vertices introduces at or below \(t\).

A \(\mathrm{MSO}_2\) formula for a graph is a logic formula, that has access to the set of vertices and edges and can check, whether they are incident to each other.
Additionally, quantification over vertices and edges and sets thereof are allowed.
There is an algorithm and a computable function \(f\colon \N \times \N \to \N\), such that all \(\mathrm{MSO}_2\) formulas of length \(\ell\) can be decided in time \(\O{f(\ell, \tw{G})\abs{V(G)}}\).
This fact is known as Courcelle's theorem~\cite{BoriePT92}.

\paragraph{Structural Parameters Based on Edge Cuts}
The following parameters are not only based on vertex cuts.
If any of the parameters presented here is small, we can remove a small number of edges to disconnect the graph into substantially smaller parts.
The exact meaning of this depends on the parameter.

\subparagraph{Tree-Cut Width}
The parameter \emph{tree-cut width} was introduced by Wollan~\cite{Wollan15}.
We orient our notation at the work of Ganian and Korchemna~\cite{GanianK22}.
Let \(X\) be a set and \(\mathcal{X} \coloneqq \{X_1, X_2, \dots, X_\ell\}\) be a family of subsets of \(X\) such that every two elements of \(\mathcal{X}\) are disjoint and \(X = \bigcup_{X_t \in \mathcal{X}} X_t\).
We call \(\mathcal{X}\) a \emph{near-partition} of \(X\).
Note that \(X_i = \emptyset\) is allowed.

\begin{definition}
A \emph{tree-cut decomposition} of a graph \(G\) is a tuple \(\mathcal{D} \coloneqq (T, \mathcal{X})\), where \(T\) is a rooted tree and \(\{X_t \subseteq V(G) \mid t \in V(T)\} \coloneqq \mathcal{X}\) a near-partition of \(V(G)\).
A set in \(\mathcal{X}\) is called a \emph{bag} of the tree-cut decomposition.
\end{definition}

For a node \(t \in V(T)\), we denote \(T_t^{\mathcal{D}}\) the sub-tree rooted at \(t\) and we set \(Y_t^{\mathcal{D}} \coloneqq \bigcup_{t \in T_t} X_t\) to be the vertices contained in the bags of \(T_t^\mathcal{D}\).
The \emph{adhesion} \(\adhesion[\mathcal{D}]{t}\) of \(t\) is defined as \(\abs{\cutEdges{Y_t^\mathcal{D}}}\) and the adhesion of the whole tree-cut decomposition is \(\max_{t \in V(T)}\adhesion[\mathcal{D}]{t}\).

The \emph{torso} of a tree-cut decomposition at node \(t\), denoted as \(H_t^\mathcal{D}\), is the graph obtained from \(G\) as follows.
Consider the connected components of \(T - t\) and for each \(C \in \comp{T - t}\), denote with \(Z_C \coloneqq \bigcup_{t \in C} X_t\).
The torso at \(t\) is obtained from \(G\), by contracting for all \(C \in \comp{T - t}\) the sets \(Z_C\) into a single vertex \(z_C\).
Note that this may create parallel edges and for \(\abs{V(T)} = 1\), the torso of the root is \(G\).
The vertices \(\{z_C\}_{C \in \comp{T - t}}\) are called \emph{peripheral} and the vertices in \(X_t\) are called \emph{core} vertices of \(H_t^\mathcal{D}\).
Consider the unique graph \(\tilde{H}_t^\mathcal{D}\), obtained from \(H_t^\mathcal{D}\), by repeatedly suppressing vertices of degree at most \(2\) from \(V(H_t^\mathcal{D}) \setminus X_t\) and removing loops.
This graph is called the \emph{3-center} of \(H_t^\mathcal{D}\) with respect to \(X_t\).

\begin{definition}
Let \(\alpha\) denote the adhesion of the tree-cut decomposition.
The width of the tree-cut decomposition is \[
\max\{\alpha\}\cup\left\{\abs{V(\tilde{H}_t^\mathcal{D})}\right\}_{t \in V(t)}
.\]
The tree-cut width of \(G\), written as \(\tcw{G}\), is the smallest width of a tree-cut decomposition.
\end{definition}

To better work with tree-cut decompositions, we need some additional notation and prior results.
For a node \(t \in V(T)\), we call \(t\) empty if \(X_t = \emptyset\) and we denote with \(\children{t}\) the set of children of \(t\) in \(T\).
Additionally, we call \(t \in V(T)\) \emph{thin} if \(\adhesion[\mathcal{D}]{t} \leq 2\) and \emph{bold} otherwise.
The set of thin children is denoted by \(\thinChildren[\mathcal{D}]{t}\) and the set of bold children as \(\boldChildren[\mathcal{D}]{s}\).
A tree-cut decomposition is \emph{nice}, if for all thin nodes \(t\) the sets \(N_G(Y_t)\) and \(\bigcup_{s \text{ is a sibling of } t} Y_{s}^\mathcal{D}\) are disjoint.
Any tree-cut decomposition can be transformed into a nice tree-cut decomposition in cubic time~\cite{GanianKS22}.
In all notation we leave off the tree-cut decomposition, if it is clear from context.

Let \(G'\) be such that there is an immersion from \(G'\) into \(G\), then \(\tcw{G'}\leq \tcw{G}\)~\cite{Wollan15}.
For all \(n, k \in \N^+\), consider the graph with a center vertex \(c\) and outer \(n\) vertices, each connected to \(c\) with \(k\) parallel edges.
We call this graph, \(S_{n,k}\).
Any graph with at most \(n\) vertices and maximum degree at most \(k\) has an immersion into \(S_{n,k}\)~\cite{Wollan15}.

The wall graph \(H_n\) is obtained from the \(n\times n\) grid graph by removing in each even row every even vertical edge and in each odd row, each odd edge vertical edge.
An illustration is given by Wollan~\cite{Wollan15} and they prove that for all \(k \in \N^+\), we have \(k \leq \tcw{H_{2k^2}}\).
In particular, \(H_n\) has \(2n^2\) vertices and maximum degree 3.

Kim et~al.~\cite{KimOPST18} provide an algorithm, that given a number \(w\in \N\) either provides a tree-cut decomposition of width at most \(2w\), or certifies that \(\tcw{G} > w\).
Its running time is \(2^\O{w^2\log 2}\abs{V(G)}^2\).

\subparagraph{Slim Tree-Cut Width}
The \emph{slim tree-cut width} is defined very similarly to tree-cut width.
Consider a tree-cut decomposition \(\mathcal{D} = (T, \mathcal{X})\) and a node \(t \in V(T)\).
Define the 2-center \(\tilde H^{2;\mathcal{D}}_t\) of \(t\) in \(\mathcal{D}\) to be the graph obtained from \(H^\mathcal{D}_t\) after exhaustively suppressing all vertices of degree at most \(1\).

\begin{definition}
Let \(\alpha\) be the adhesion of \(\mathcal{D}\).
The \emph{slim width} of \(\mathcal{D}\) is defined to be \[\max\{\alpha\}\cup \left\{\abs{\tilde V(H^{2,\mathcal{D}}_t)}\right\}_{t \in V(T)}.\qedhere\]
\end{definition}

The slim tree-cut width \(\stcw{G}\) is the smallest slim tree-cut width of any tree-cut width for \(G\).
Ganian et~al.~\cite{GanianK22} provide an algorithm, that given a number \(w\in \N\) either provides a nice tree-cut decomposition of slim width at most \(6(w+1)^3\), or certifies that \(\stcw{G} > w\).
Its running time is \(2^\O{w^2\log w}\abs{V(G)}^4\).

Let \(G'\) be such that \(G'\) has an immersion into \(G\), then \(\stcw{G'} \leq \stcw{G}\)~\cite{GanianK22}.
Let \(n \in \N^+\) and consider the graph of \(n\) cycles of length three that all contain the same vertex \(c\), but are otherwise vertex disjoint.
We call this graph \(W_n\) and the collection of all such graphs windmills.
For all \(r \in \N \setminus \natintZ{1}\), we have \(\stcw{W_{r^2}} \geq r\), while \(\tcw{W_r} = 2\).
An illustration of windmill graphs is given by Ganian et~al.~\cite{GanianK22}.

\subparagraph{Feedback Edge Set}
Let \(S \subseteq E(G)\).
If \(G - E\) is cyclefree, we call \(S\) a \emph{feedback edge set}.
The size of the smallest feedback edge set is called the \emph{feedback edge set number} \(\fen{G}\) of \(G\).
Note that \(\fen{G} = \abs{E(G)} - \abs{V(G)} + \abs{\comp{G}}\).

\subparagraph{Treewidth + Maximum Degree} This parameter is just \(\tw{G} + \max_{v \in V(G)} \deg{v}\) and gives a simple approach to turning \(\tw{G}\) into a parameter which ensures that there are small edge cuts that if removed disconnect the graph into substantially smaller parts.


\section{Augmentation for \texorpdfstring{$\gstp$}{GSTP}}%
\label{sec:augS}%
\label{sec:augmentation}
For the structural parameterization of \(\edp\) the tool of the augmented graph has found a wide range of applications.
For this consider an instance \((G, P)\) of \(\edp\), where \(G = (V,E)\) represents the underlying graph and \(P \subseteq \binom{V}{2}\) the pairs of vertices that ought to be connected.
The augmented graph of this instance is denoted by \(\augment{G}{P}\) and is defined to be \(G + P\), the graph obtained from \(G\) where we additionally connect all pairs of \(P\).
This concept was first introduced for \textsc{Multicut} by Gottlob and Lie~\cite{GottlobL07}.

For \(\edp\), structural parameterization by a parameter derived from the augmented graph, opens many doors.
As an example, consider the vertex cover number as a parameter.
For graphs with vertex cover number at least \(3\) this problem is \textsc{NP}-hard \cite{Fleszar2018Sep}.
If we restrict our view to graph with vertex cover number at most \(2\), this problem is solvable in polynomial time \cite{GanianOR21}.
This strict dichotomy is not necessarily desired and it does not help to fully understand which instances are actually hard.
However, taking the vertex cover number of the augmented graph as the parameter, \(\edp\) is fixed-parameter tractable.
This even holds for the weaker parameter of the fracture number of the augmented graph \cite{GanianOR21}.

We want to transfer this approach to \(\gstp\).
It is not directly clear how this should be done.
Consider an instance \((G, \mathcal{T}, d)\) of \(\gstp\).
There are at least two viable options on how to approach this.
First, for each \(T_{i} \in \mathcal{T}\) and distinct \(u,v \in T_{i}\) we could add an edge between \(u\) and \(v\), creating parallel edges if this edge exist already.
As this ensures that all \(T_{i} \in \mathcal{T}\) form a clique, we call this the clique-augmented graph \(\augmentClique{G}{\mathcal{T}}\).
There is a simple reduction for \(\edp\) that ensures that all terminals are non-adjacent and have degree one~\cite{Fleszar2018Sep}.
Assuming, each \(v \in V\) is part of at most \(\deg{v}\) terminal pairs, which can be ensured by applying \Cref{rr:degree-negative-instance}, all considered parameters, but vertex-cover number, only grow slightly as a function of the parameter.
Namely, fracture number grows at most quadratically, while all other considered parameters grow by at most a factor of two.
Therefore, the clique-augmented graph matches the augmented graph for \(\edp\) in the cases considered in previous research.

On the other hand, we can introduce a new vertex \(\aug{T_i}\) for each \(T_i\in \mathcal{T}\) and connect it to all \(v \in T_i\).
We call this the vertex-augmented graph \(\augmentVert{G}{\mathcal{T}}\).
This version differs from the augmented graph for \(\edp\) by subdividing the inserted edges.

We now study, how parameterized complexity results transfer between the two versions with respect to the structural parameters presented in \Cref{fig:relation-between-parameters,sec:prelims}.
To this end, we check for each such parameter \(\kappa\), whether or not there is a function \(f \colon \N \to \N\) such that for all instances \((G, \mathcal{T}, d)\) of \(\gstp\) we have \(\kappa(\augmentVert{G}{\mathcal{T}}) \leq f(\kappa(\augmentClique{G}{\mathcal{T}}))\) and vice versa.
In this case, we write \(\kappa(\augmentVert{G}{\mathcal{T}}) \leq_f \kappa(\augmentClique{G}{\mathcal{T}})\).
Note that \(d\) does not play any role in this definition.
If such a function exists, any result showing that \(\gstp\)—or \(\edp\), or \(\stp\)—is \(\fpt\) with respect to \(\kappa\) of the vertex-augmented graph, yields directly that \(\gstp\) is \(\fpt\) with respect to \(\kappa\) of the clique augmented graph.
Similarly, if \(\gstp\) is \(\wonehard\) with respect to \(\kappa\) of the clique-augmented graph, it is \(\wonehard\) with respect to \(\kappa\) of the vertex-augmented graph as well.
That is, algorithmic results would be stronger with respect to the vertex-augmented graph and hardness results would be stronger with respect to the clique augmented graph.

We now show, that for \(\kappa \in \{\twOp, \fvsOp, \fracNumOp, \twOp + \maxDegOp, \tcwOp, \stcwOp, \fenOp\}\) such a function in deed exists. For \(\kappa = \twOp + \maxDegOp\), the reverse holds as well.

\begin{lemma}
\label{stmt:vert-augment-bounded}
For all \(\kappa \in \{\twOp, \fvsOp, \fracNumOp, \tcwOp, \stcwOp, \fenOp, \twOp + \maxDegOp\}\), \(\kappa(\augmentVert{G}{\mathcal{T}}) \leq_f \kappa(\augmentClique{G}{\mathcal{T}})\).
If \(\kappa = \twOp + \maxDegOp\), we additionally have \(\kappa(\augmentClique{G}{\mathcal{T}}) \leq_f \kappa(\augmentVert{G}{\mathcal{T}})\).
\end{lemma}
\begin{proof}
\leavevmode
\begin{itemize}
\item[] \(\bm{\twOp\!.}\)\quad Let \((S, \mathcal{X})\) be a tree decomposition of \(\augmentClique{G}{\mathcal{T}}\).
We create a tree decomposition based on \((S, \mathcal{X})\) for \(\augmentVert{G}{\mathcal{T}}\), by choosing for all \(T \in \mathcal{T}\) a bag \(X_s \in \mathcal{X}\) with \(T \subseteq X_s\), which exists as \(T\) is a clique in \(\augmentClique{G}{\mathcal{T}}\)~\cite{BodlaenderM90}.
We add a node \(s'\) associated with the bag \(T \cup \{\aug{T}\}\) as a child of \(s\).
Note that this is a tree decomposition for \(\augmentVert{G}{\mathcal{T}}\) with width at most one larger than the treewidth of \((S, \mathcal{X})\).
\item[] \(\bm{\fvsOp\!.}\)\quad Let \(S\) be a feedback vertex set for \(\augmentClique{G}{\mathcal{T}}\).
We claim that \(S\) is also a feedback vertex set for \(\augmentVert{G}{\mathcal{T}}\).
Assume there is a cycle \(C\) in \(\augmentVert{G}{\mathcal{T}} - S\) and set \(C'\) to be \(C\) with all vertices in \(\aug{\mathcal{T}}\) removed.
As for every \(T \in \mathcal{T}\) the vertices in \(T\) are adjacent in \(\augmentClique{G}{\mathcal{T}}\), \(C'\) is a cycle in \(\augmentClique{G}{\mathcal{T}} - S\) as well, violating the fact that \(S\) is a feedback vertex set.
\item[] \(\bm{\fracNumOp\!.}\)\quad Let \(S\) be a fracture modulator of \(\augmentClique{G}{\mathcal{T}}\).
We now add arbitrary vertices to \(S\) until it is a fracture modulator for \(\augmentVert{G}{\mathcal{T}}\).
Note that the obtained set is a fracture modulator and its size is bounded by \(\abs{S} + \max_{C \in \comp{\augmentVert{G}{\mathcal{T}} - S}} \abs{V(C)}\).
Let \(C \in \comp{\augmentVert{G}{\mathcal{T}} - S}\).
All that remains is to bound \(\abs{V(C)}\) in terms of \(S\).
As no pair of augmented vertices is adjacent, if \(C\) contains only augmented vertices, we have \(\abs{C} = 1\).
Otherwise, there is a vertex \(v \in V(C)\) that is not augmented.
Let \(D \in \comp{\augmentClique{G}{\mathcal{T}} - S}\) be such that \(v \in V(D)\).
We now show, that \(V(C) \cap V(G) \subseteq V(D)\).
Let \(u \in V(C) \cap V(G)\).
There is a \(vu\)-path \(P\) in \(\augmentVert{G}{\mathcal{T}} - S\).
Note that \(P\) with all augmented vertices removed is a path in \(\augmentClique{G}{\mathcal{T}} - S\).
Thus, \(u \in V(D)\).

Now, consider any \(T \in \mathcal{T}\) with \(\aug{T} \in V(C)\).
As there is a path from \(v\) to \(\aug{T}\) in \(C\), we have \(T \cap V(C) \neq \emptyset\) and, in particular, \(T \subseteq S \cup V(C)\).
Thus, there are at most \(2^{\abs{S} + \abs{V(C)}} \leq 2^{2\abs{S}}\) many choices for \(T\) and at most that many augmented vertices contained in \(V(C)\); yielding that the obtained fracture modulator has size at most \(2\abs{S} + 2^{2\abs{S}}\).
  \item[] \(\bm{\tcwOp\!.}\)\quad Let \((S, \mathcal{X})\) be a tree-cut decomposition of \(\augmentClique{G}{\mathcal{T}}\) and let its tree cut width be \(w\).
  To obtain a tree-cut decomposition for \(\augmentVert{G}{\mathcal{T}}\), consider any \(T \in \mathcal{T}\) and let \(s_{T} \in V(S)\) be a node of \(S\) containing a vertex of \(T\) such that no vertex of \(T\) is contained in a descendant of \(s_{T}\) in \(S\).
  We add the vertex \(\aug{T}\) to the bag associated with the node \(s_{T}\).
  Doing this for all \(T\) does not change the adhesion of any node.
  So, we need to bound the increase in torso size.
  In the torso at \(s_{T}\), the vertex \(\aug{T}\) is only connected to vertices in \(X_{s_{T}}\) and the vertex \(z_{\mathrm{top}}\) to which \(V(G) \setminus Y_{s_{T}}\) was contracted.
  However, the degrees of all vertices in the torso at \(s_{T}\) with respect to \(\augmentVert{G}{\mathcal{T}}\) are bounded by their degrees with respect to \(\augmentClique{G}{\mathcal{T}}\) and the graph induced by all pendant vertices did not change.
  This means, that any contraction sequence of pendant vertices with respect to \(\augmentClique{G}{\mathcal{T}}\) is also viable with respect to \(\augmentVert{G}{\mathcal{T}}\).
  So, we only need to bound the number of new core vertices in \(s_{T}\), this means, bound the number of \(T' \in \mathcal{T}\) with \(s_{T} = s_{T'}\).
  All such \(T'\) have \(X_{s_{T}} \cap T' \neq \emptyset\).
  There are at most \(2^{\abs{X_{s_{T}}}}\) such \(T'\) with \(T' \subseteq X_{s_{T}}\). So, assume \(T' \not\subseteq X_{s_{T}}\).
  As all descendants of \(s_{T}\) in \(S\) are disjoint from \(T'\), we have \(T' \setminus Y_{s_{T}} \neq \emptyset\), meaning that \(T'\) contributes at least one edge to the adhesion of \(s_{T}\).
  Thus, there are at most \(w\) with \(T'\not\subseteq X_{s_{T}}\) and \(w + 2^{w}\) overall.
\item[] \(\bm{\stcwOp\!.}\)\quad The proof is analogous to the one for tree-cut width.
  \item[] \(\bm{\fenOp\!.}\)\quad Let \(S\subseteq E(\augmentClique{G}{\mathcal{T}})\) be a feedback edge set of \(\augmentVert{G}{\mathcal{T}}\).
  To obtain a feedback edge set for \(\augmentVert{G}{\mathcal{T}}\) let \(\mathcal{T}_{S} \subseteq \mathcal{T}\) be the terminal sets which added an edge to \(\augmentClique{G}{\mathcal{T}}\) that is contained in \(S\).
  Now, consider the edge set \(S'\) that contains \(S\cap E(G)\) and for each \(T\in \mathcal{T}_{S}\) all but one edge of \(\{\aug{T}t\}_{t \in T}\).
  We notice that \(\augmentVert{G}{\mathcal{T}} - S'\) is a subgraph of  \(G - S\) with some additional one degree vertices. Therefore, \(S'\) is a feedback edge set of \(\augmentClique{G}{\mathcal{T}}\).

        Now, we show that \(|S'| \leq 2|S|\).
        As all the edges contributed by each \(T \in \mathcal{T}\) are distinct, it is sufficient to show that each \(T \in \mathcal{T}\) contributes at most twice as many edge to \(S'\) compared to \(S\).
        If \(|T| \leq 2\), the amount of contributed edges is equal.
        So, let \(|T| \geq 3\).
        As the edges contributed by \(T\) to \(\augmentClique{G}{\mathcal{T}}\) form a clique, at least \(\binom{|T|}{2} - (\abs{T} - 1) = \frac{(\abs{T} - 1)(\abs{T} - 2)}{2}\) of them are contained in \(S\).
        On the other hand, at most \(\abs{T} - 1\) are contributed to \(S'\) by \(T\).
        Therefore, \(T\) contributes at most \(\frac{2(\abs{T} - 1)}{(\abs{T} - 1)(\abs{T} - 2)} \leq 2\) times as many edges to \(S'\) as it does to \(S\).
\item[] \(\bm{\twOp + \maxDegOp\!.}\)\quad We know that \(\tw{\augmentVert{G}{\mathcal{T}}} \leq 1 + \tw{\augmentClique{G}{\mathcal{T}}}\).
So, we only need to bound \(\maxDeg{\augmentVert{G}{\mathcal{T}}}\). To bound \(\maxDeg{\augmentVert{G}{\mathcal{T}}}\) in terms of \(\maxDeg{\augmentClique{G}{\mathcal{T}}}\), let \(v \in V(G)\) and set \(\mathcal{T}_v \coloneqq \{T \in \mathcal{T} \mid v \in T\}\) to be the terminal sets containing \(v\).
Then, \(v\) is adjacent in \(\augmentClique{G}{\mathcal{T}}\) to all its neighbors in \(G\) and all vertices for which there is a terminal set containing both, that is \(\deg[\augmentClique{G}{\mathcal{T}}]{v} = \abs{N_G(v) \cup ((\bigcup \mathcal{T}_v) \setminus \{v\})} \geq \max(\deg[G]{v}, \abs{\bigcup \mathcal{T}_v} - 1))\).
Additionally, in \(\augmentVert{G}{\mathcal{T}}\) the vertex \(v\) is adjacent to all its neighbors in \(G\) and the augmented vertices \(\aug{\mathcal{T}_s}\), that is \(\deg[\augmentVert{G}{\mathcal{T}}]{v} = \deg[G]{v} + \abs{\mathcal{T}_v}\).
For all \(T \in \mathcal{T}_s\), we have \(\{v\} \subseteq T \subseteq \bigcup \mathcal{T}_v\).
Thus, \(\abs{\mathcal{T}_v} \leq 2^{\abs{\bigcup \mathcal{T}_v} - 1}\) and we get \(\deg[\augmentVert{G}{\mathcal{T}}]{v} \leq \deg[\augmentClique{G}{\mathcal{T}}]{v} + 2^{\deg[\augmentClique{G}{\mathcal{T}}]{v}}\).
Now, consider an augmented vertex \(u \in V(\augmentVert{G}{\mathcal{T}})\).
Choose, \(v \in \invAug{u}\).
Then, \(\deg[\augmentVert{G}{\mathcal{T}}]{v} = \abs{T} \leq \abs{\bigcup \mathcal{T}_v} \leq \deg[\augmentClique{G}{\mathcal{T}}]{v} + 1\).
So, the parameter \(\twOp + \maxDegOp\) of \(\augmentVert{G}{\mathcal{T}}\) is bounded by a function of its value on \(\augmentClique{G}{\mathcal{T}}\).

Now, let \((S, \mathcal{X})\) be a tree decomposition of \(\augmentVert{G}{\mathcal{T}}\).
To create a tree decomposition for \(\augmentClique{G}{\mathcal{T}}\), replace for each \(T \in \mathcal{T}\) every occurrence of \(\aug{T}\) in \(\mathcal{X}\) with \(T\).
Notice that this is a tree decomposition.
Consider \(T \in \mathcal{T}\).
We have that \(\abs{T} = \deg[\augmentVert{G}{\mathcal{T}}]{\aug{T}} \leq \maxDeg{\augmentVert{G}{\mathcal{T}}}\).
Therefore, the treewidth of \(\augmentClique{G}{\mathcal{T}}\) is bounded by \((\tw{\augmentVert{G}{\mathcal{T}}} + 1)(\maxDeg{\augmentVert{G}{\mathcal{T}}} + 1)\).
Consider any \(v \in V(G)\).
We have \(\deg[\augmentClique{G}{\mathcal{T}}]{v} \leq \deg[G]{v} + \abs{\bigcup \mathcal{T}_v}\).
Observe that \(\abs{\bigcup \mathcal{T}_v} \leq \abs{\mathcal{T}_v}\max_{T \in \mathcal{T}_v} \abs{T}\) and that \(\abs{\mathcal{T}_v} \leq \deg[\augmentVert{G}{\mathcal{T}}]{v}\) and that \(\max_{T \in \mathcal{T}_v} \abs{T}\leq \maxDeg{\augmentVert{G}{\mathcal{T}}}\).
Thus, \(\abs{\bigcup \mathcal{T}_v} \leq ( \maxDeg{\augmentVert{G}{\mathcal{T}}})^2\), yielding that the maximum degree of \(\augmentClique{G}{\mathcal{T}}\) is also bounded by a function of the parameter with respect to \(\augmentVert{G}{\mathcal{T}}\).\qedhere
\end{itemize}
\end{proof}

Next, we show that for the only other considered parameter, namely vertex cover number, such a function bounding \(\vc{\augmentVert{G}{\mathcal{T}}}\) in terms of \(\vc{\augmentClique{G}{\mathcal{T}}}\) can not exist.

\begin{lemma}
\label{stmt:vert-augment-unbounded}
There exists a family of instances \((G_i, \mathcal{T}_i, d_i)_{i\in\N}\) of \(\gstp\) such that \(\{\vc{\augmentClique{G_i}{\mathcal{T}_i}}\}_{i\in \N}\) is bounded and \(\{\vc{\augmentVert{G_i}{\mathcal{T}_i}}\}_{i\in \N}\) is unbounded.
\end{lemma}
\begin{proof}
  Let \(i \in \N\) and consider the star graph \(S_{i}\) with \(i\) leaves and center vertex \(c\).
For each \(v \in V(S_{i}) \setminus \{c\}\), set \(T_v \coloneqq \{c, v\}\) to be a terminal set and let \(\mathcal{T}_i \coloneqq \{T_v\}_{v \in V(S_{i}) \setminus \{c\}}\).
As all edges that clique-augmentation adds are present in \(S_{i}\) and vertex covers are not affected by parallel edges, any vertex cover of \(S_{i}\) is a vertex cover of \(\augmentClique{S_{i}}{\mathcal{T}}\).
In particular, \(\{c\}\) is a vertex cover of \(\augmentClique{S_{i}}{\mathcal{T}}\).
On the other hand \(\augmentVert{S_{i}}{\mathcal{T}_i}\) is the windmill \(W_{i}\), which has a matching of size \(i\).
Thus, its vertex cover number is \(i\).
\end{proof}

Finally, we show that for all considered \(\kappa\), except for \(\twOp + \maxDegOp\), there also is no function bounding \(\kappa\) of the clique-augmented graph in terms of \(\kappa\) with respect to the vertex augmented graph.

\begin{lemma}
\label{stmt:clique-augment-unbounded}
For all \(\kappa \in \{\twOp, \fvsOp, \fracNumOp, \vcOp, \tcwOp, \stcwOp, \fenOp\}\), there exists a family of instances \((G_i, \mathcal{T}_i, d_i)_{i\in\N}\) of \(\gstp\) such that \(\{\kappa(\augmentVert{G_i}{\mathcal{T}_i})\}_{i\in \N}\) is bounded and \(\{\kappa(\augmentClique{G_i}{\mathcal{T}_i})\}_{i\in \N}\) is unbounded.
\end{lemma}
\begin{proof}
First, we prove this result for \(\kappa \in \{\twOp, \fvsOp, \fracNumOp, \vcOp\}\).
For this let \(i \in \N\) and consider the star graph \(S_i\) with \(i\) leaves around the center vertex \(c\).
Set the single terminal set \(T\) to be equal to all leaves.
The graph \(\augmentClique{S_i}{\mathcal{T}}\) has a clique with \(i\) vertices as a subgraph; so, its treewidth is at least \(i - 1\).
Whereas in the graph \(\augmentVert{S_i}{\mathcal{T}}\) the set \(\{c, \aug{T}\}\) is a vertex cover of size 2.
We have that \(\kappa(\augmentVert{S_i}{\mathcal{T}})\) is bounded by a function of the vertex cover number, yielding the result.

Now, let \(\kappa \in \{\tcwOp, \stcwOp\}\) and \(i \in \N\).
Consider the graph \(G_i\) which consists of a vertex \(c\) and \(4i^4\) isolated paths of length \(3\).
Denote with \(\{P_j\}_{j\in\natint{4i^4}}\) these isolated paths and set \(\mathcal{T} \coloneqq \{\{c\} \cup V(P_j)\}_{j \in \natint{4i^4}}\).
See \(G_i\), \(\augmentClique{G_i}{\mathcal{T}}\), and \(\augmentVert{G_i}{\mathcal{T}}\) illustrated in \Cref{fig:tcw-augmentation}.
Notice that the graph \(S_{4i^4, 3}\), which is the star graph on \(i^2\) vertices where each leave has \(3\) parallel edges to the center, immerses into \(\augmentClique{G_i}{\mathcal{T}}\).
As the wall \(H_{2i^2}\) has an immersion into \(S_{4i^4, 3}\)~\cite{Wollan15}, it also has an immersion into \(\augmentClique{G_i}{\mathcal{T}}\); so, \(\augmentClique{G_i}{\mathcal{T}}\) has tree-cut width at least \(i\)~\cite{Wollan15}.
For \(\augmentVert{G_i}{\mathcal{T}}\), we obtain a tree-cut decomposition as follows.
We set the root to be a bag containing only \(c\).
For each \(j \in \natint{4i^4}\), we create a bag as a child of the root node containing \(\{\aug{\{c\} \cup V(P_j)} \}\cup V(P_j)\).
This tree-cut decomposition is drawn in \Cref{fig:tcw-augmentation-vertex}.
We see that this tree-cut decomposition has width and slim width at most \(4\).

\begin{figure}[tp]
\centering
\begin{subfigure}[t]{0.32\textwidth}
\centering
\includegraphics[width=0.9\textwidth]{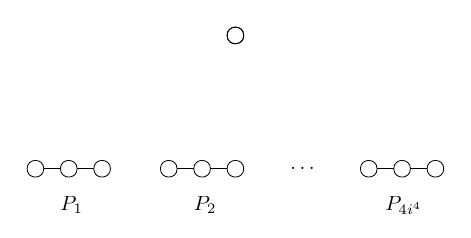}
\caption{\label{fig:tcw-augmentation-base} The graph \(G_i\)\vphantom{\(\augmentClique{G_i}{\mathcal{T}}\)}.}
\end{subfigure}
\begin{subfigure}[t]{0.32\textwidth}
\centering
\includegraphics[width=0.9\textwidth]{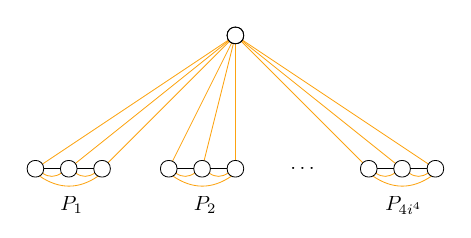}
\caption{\label{fig:tcw-augmentation-clique} The graph \(\augmentClique{G_i}{\mathcal{T}}\)\!.}
\end{subfigure}
\begin{subfigure}[t]{0.32\textwidth}
\centering
\includegraphics[width=0.9\textwidth]{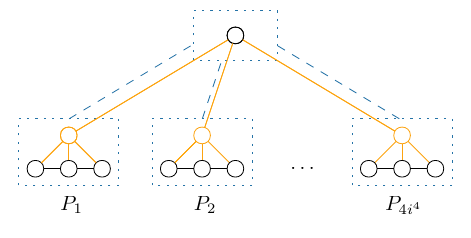}
\caption{\label{fig:tcw-augmentation-vertex} The graph \(\augmentVert{G_i}{\mathcal{T}}\)\!.}
\end{subfigure}
\caption{\label{fig:tcw-augmentation} The graphs \(G_i\), \(\augmentClique{G_i}{\mathcal{T}}\)\!, and \(\augmentVert{G_i}{\mathcal{T}}\)\!. We draw augmented vertices and edges in orange and the tree-cut decomposition in blue.}
\end{figure}

Finally, for \(\kappa = \fenOp\) and \(i \in \N\), we consider the graph \(G'_i\) on \(3i\) isolated vertices.
Partition \(V(G_i)\) into sets \(T_1, T_2, \dots, T_i\) of three vertices each and set \(\mathcal{T}\coloneqq \{T_j\}_{j \in \natint{i}}\).
We notice that \(\augmentClique{G_i'}{\mathcal{T}}\) is a union of \(i\) triangle graphs.
Thus, it has feedback edge number \(i\).
Additionally, we see that \(\augmentVert{G_i'}{\mathcal{T}}\) is a forest.
Consequently, it has feedback edge number \(0\).
\end{proof}

\begin{table}[tp]
\centering
\begin{tabular}{lcc}
\toprule
Parameter & \(\kappa(\augmentVert{G}{\mathcal{T}})\leq_{f}\kappa(\augmentClique{G}{\mathcal{T}})\) & \(\kappa(\augmentClique{G}{\mathcal{T}})\leq_{f}\kappa(\augmentVert{G}{\mathcal{T}})\) \\
\midrule
\(\twOp\) & \checkmark & \(\times\) \\
\(\fvsOp\) & \checkmark & \(\times\) \\
\(\fracNumOp\) & \checkmark & \(\times\) \\
\(\vcOp\) & \(\times\) & \(\times\) \\
\(\tcwOp\) & \(\checkmark\) & \(\times\) \\
\(\stcwOp\) & \(\checkmark\) & \(\times\) \\
\(\fenOp\) & \(\checkmark\) & \(\times\) \\
\(\twOp + \maxDegOp\) & \checkmark & \checkmark \\
\bottomrule
\end{tabular}
\caption{\label{tbl:summary-augmentation} Summary of the results obtained in \Cref{stmt:vert-augment-bounded,stmt:vert-augment-unbounded,stmt:clique-augment-unbounded}. For each parameter and statement, we write \checkmark if the statement holds with respect to this parameter and \(\times\) if it does not.}
\end{table}

The results obtained by \Cref{stmt:vert-augment-bounded,stmt:vert-augment-unbounded,stmt:clique-augment-unbounded} are summarized in \Cref{tbl:summary-augmentation}.
We see that the structural parameter of the vertex-augmented graph is strictly stronger with respect almost all considered parameters.
Only for vertex cover number they are incomparable and functionally equivalent for sum of tree width and maximum degree.
In light of these results and noticing that we mainly provide \(\fpt\)-algorithms, we decide to focus our attention on the vertex-augmented graph.
From now on, we simply call the vertex-augmented graph the augmented graph for \(\gstp\) and denote it with \(\augment{G}{\mathcal{T}}\).


\section{\texorpdfstring{$\gstp$}{GSTP} is \texorpdfstring{$\fpt$}{FPT} by the Augmented Fracture Number}
\label{sec:fn*}
\label{sec:gstp-fnS}
In this Chapter we show that \(\gstp\) is fixed-parameter tractable by the fracture number of the augmented graph.
This is an important stepping stone towards more intricate results, like showing that this problem is fixed-parameter tractable by the tree-cut width of the augmented graph.
Which in turn we can make use of to show that \(\stp\) is fixed-parameter tractable by the tree-cut width of the input graph.

To achieve the result of this Chapter, we represent the instance as an integer linear program with the number of variables bounded by a computable function of the parameter (i.e., the fracture number of the augmented graph).
Since checking whether an interger linear program has a feasible solution is \(\fpt\) in the number of variabel~\cite{Kannan87}, this yields an \(\fpt\)-algorithm, given that we can find a fracture modulator in \(\fpt\)-time as well.
This was claimed by Dvorák et~al.~\cite{DvorakEGKO21}, but their proof suffers from a minor inaccuracy.
We examine this and a corrected version of their algorithm in \Cref{sec:find-frac-mod-in-fpt}.

Our approach to represent the instance as an integer linear program uses similar ideas as Ganian et~al.~\cite{GanianOR21} use to show that \(\edp\) is \(\fpt\) by the fracture number of the augmented graph.
We first consider all the ways a component of the graph without the fracture modulator can interact with the remaining components.
Then, we group the components and we, finally, provide an integer linear program to represent the whole instance.
In doing so, we not only generalize there algorithm significantly, but we also reduce the running time from triply exponential in the parameter to doubly exponential in the parameter.
The exponential-time-hypothesis implies that there is no algorithm solving \(\ilp\) with \(n\) variables in \(2^{o(n)}\) time for all such instances.
Assuming \(\eth\), to decrease the running time to sub-doubly-exponential, either we need to use less than exponentially many variables to represent the instance, utilize some structure in the obtained \(\ilp\) instances, or switch to a new approach avoiding integer linear programming altogether.
All of which seem quite challenging.

Before we start into the three main parts of the proof, we provide common definitions and reduction rules.
Let \((G, \mathcal{T}, d)\) be an instance of \(\gstp\), let \(X\subseteq V(\augment{G}{\mathcal{T}})\) be a fracture modulator of \(\augment{G}{\mathcal{T}}\).
For all \(U \subseteq V(\augment{G}{\mathcal{T}})\), we call the set of all \(T \in \mathcal{T}\) with \(\aug{T}\in U\) by \(\mathcal{T}_U\).
Let \(C \in \comp{\augment{G}{\mathcal{T}} - X}\) and denote with \(C^+ \coloneqq \augment{G}{\mathcal{T}}[V(C) \cup X]\) the subgraph induced on the component and the fracture modulator.

For this Chapter, we need the following reduction rules.
They allow us to show, that in all relevant instances, most terminal sets only require few subgraphs to be assigned to them.

\begin{reductionrule}
\label{rr:sensible-terminal-sets}
If there is a \(T \in \mathcal{T}\) with \(|T| < 2\), remove \(T\) from \(\mathcal{T}\).
\end{reductionrule}
\begin{proof}
As we just remove constraints, whenever the original instance is a positive instance, so is the reduced instance.
Now assume that the reduced instance is a positive instance and let \(\mathcal{X}\) be a solution.
The set \(T\) can be interpreted as either an empty or a single-vertex tree.
In either case, it does not contain any edges and connects all vertices in \(T\).
Thus, we can add \(T\) to \(\mathcal{X}\) \(d(T)\)-times and obtain a solution for the original instance.
\end{proof}
\begin{reductionrule}
\label{rr:degree-negative-instance}
After applying \Cref{rr:sensible-terminal-sets} exhaustively, if there is a \(v \in V\) such that \(\sum_{T \in \mathcal{T} \colon v \in T}d(T) > \deg{v}\), output a trivial negative instance.
\end{reductionrule}
\begin{proof}
We show that no positive instance satisfies the premise.
Consider any solution, let \(v\in V\) and denote with \(\mathcal{F}_v\) the trees in the solution that contain \(v\).
Since for all \(T \in \mathcal{T}\) we have \(|T| \geq 2\), any tree in the solution contains an edge adjacent to every one of its vertices.
As all different trees are edge-disjoint, \(|\mathcal{F}_v| \leq \deg{v}\) and by definition of \(\mathcal{F}_v\) we have \(\sum_{T \in \mathcal{T} \colon v \in T}d(T)\leq |\mathcal{F}_v|\).
\end{proof}

To better assess, which terminal sets might need special attention, and to simplify our proofs, we only consider fracture modulators with a particular structure.

\begin{definition}
\label{def:nice-fracture-modulator}
Let \(X \subseteq V(\augment{G}{\mathcal{T}})\) be a fracture modulator of \(\augment{G}{\mathcal{T}}\).
We call \(X\) \emph{nice}, if
\begin{enumerate}
\item \label{item:nice-fracture-modulator-GS-edgeless}\(G[X]\) is edgeless,
\item \label{item:nice-fracture-modulator-TS} for all \(T \in \mathcal{T}_X\), we have that there are two distinct components \(C, C' \in \comp{\augment{G}{\mathcal{T}} - S}\) with \(V(C) \cap T \neq \emptyset\) and \(V(C') \cap T \neq \emptyset\).\qedhere
\end{enumerate}
\end{definition}

This is not a strong requirement.
In fact, we can always turn a fracture modulator into a nice fracture modulator of similar size.

\begin{lemma}
\label{stmt:nice-fracture-modulator-exists}
Let \(X\) be a fracture modulator of \(\augment{G}{\mathcal{T}}\).
There is an equivalent instance \(\mathscr{P}'=(G', \mathcal{T}, d)\) and a nice fracture modulator \(S\) of \(\augment{G'}{\mathcal{T}}\) with \(\abs{S} \leq 2\abs{X}\) and \(\abs{V(G')}\leq \abs{V(G)} + \binom{\abs{X}}{2} + 2\abs{X}\).
We can construct \(\mathscr{P}'\) and \(S\) in linear time.
\end{lemma}
\begin{proof}
First, we subdivide each edge of \(G[X]\), which adds at most \(\binom{\abs{X}}{2}\) vertices.
Call this graph \(H\).
Let \(Y\coloneqq \{\aug{T} \mid T \in \mathcal{T}_X; \exists C \in \comp{\augment{G}{\mathcal{T}} - S}\colon T \subseteq V(C^+)\}\) be the vertices of \(X\), that violate the condition of \Cref{item:nice-fracture-modulator-TS} and set \(Z \coloneqq X \setminus Y\).
The graph \(\augment{H}{\mathcal{T}} - Z\) can have connected components of size larger than \(Z\).
In order to obtain a fracture modulator, we new isolated vertices to \(H\) and \(Z\), obtaining \(G'\) and \(S\), until \(S\) is fracture modulator of \(\augment{G'}{\mathcal{T}}\).
We notice that \(S\) is a nice.
It remains to bound the size of \(S\).

Consider a \(C \in \comp{\augment{H}{\mathcal{T}} - Z}\) with \(\abs{C} > 2\).
As the vertices created by subdividing the edges of \(G[X]\) combined with \(Y\) form an independent set in \(\augment{H}{\mathcal{T}}\), there is a \(D \in \comp{\augment{G}{\mathcal{T}} - X}\) with \(V(C) \cap V(D) \neq \emptyset\) and, in particular, \(V(C) \supseteq V(D)\).
Let \(Y^* \coloneqq \{y \in Y \mid N(y) \cap V(D) \neq \emptyset\}\).
By definition of \(Y\), we have for all \(y \in Y^*\) that \(N_{\augment{G}{\mathcal{T}}}[y] \subseteq V(D) \cup Z \cup \{y\}\).
Additionally, for all \(d \in V(D)\), we have \(N_{\augment{G}{\mathcal{T}}}[d] \subseteq V(D) \cup Z \cup Y^*\).
So, \(N_{\augment{G}{\mathcal{T}} - Z}[V(D) \cup Y^*] = V(D) \cup Y^*\).
As \(V(C)\) contains all \(V(D) \cup Y^*\), there are no other vertices in \(V(C)\) since the vertices in \(V(D) \cup Y^*\) would be disconnected from them in \(\augment{G}{\mathcal{T}} - Z\) as well as \(\augment{H}{\mathcal{T}} - Z\).
Thus, \(V(C) = V(D) \cup Y^*\) and \(\abs{S} = \max(1, \abs{Z}, \max \{\abs{V(C)}\}_{C \in \comp{\augment{H}{\mathcal{T}} - S}}) \leq \abs{Z} + \max_{D \in \comp{\augment{G}{\mathcal{T}} - X}} \abs{V(D)} + \abs{Y} \leq 2\abs{X}\), meaning that we add at most \(2 \abs{X}\) isolated vertices to \(H\).
\end{proof}

From now on, we assume that \(S\) is a nice fracture modulator for \((G, \mathcal{T}, d)\) and that \Cref{rr:degree-negative-instance} has been applied.
Note that even though \(G[S]\) is edgeless \(\augment{G}{\mathcal{T}}[S]\) may contain edges.
Denote with \(\terminalsContainedS\) the set of all \(T \in \mathcal{T}\) satisfying \(T \subseteq S\).
We show, that in all relevant instances and for all \(T \in \mathcal{T}\setminus \terminalsContainedS\), the number of required subsets \(d(T)\) is bounded by \(2\abs{S}\).

\begin{lemma}
\label{stmt:terminal-sets-few-requirements}
After applying \Cref{rr:degree-negative-instance}, for all \(T \in \mathcal{T}\setminus\terminalsContainedS\) we have \(d(T) \leq 2\abs{S}\).
\end{lemma}
\begin{proof}
Consider a \(T \in \mathcal{T}\setminus \terminalsContainedS\).
Then, there is a \(v \in T \setminus S\) and let \(C \in \comp{\augment{G}{\mathcal{T}} - S}\) be such that \(v \in V(C)\).
Since \(N(v) \subseteq S \cup C\), we have \(\deg{v} \leq 2\abs{S}\) and since we applied \Cref{rr:degree-negative-instance}, this shows that \(d(T) \leq 2\abs{S}\).
\end{proof}

\subsection{Component Configurations}
\label{sec:org25bea82}
To fulfill the connectivity requirements, each component might need to use edges from other components, or other components might need to use edges from this component.
We now characterize how a component \(C \in \comp{\augment{G}{\mathcal{T}} - S}\) can interact with the remaining instance.
To this end, we introduce the concept of a \emph{component-configuration}.
The goal of this Section is to categorize positive and negative instances solely based on the configurations of the components.

As \(\abs{V(C^+)} \leq 2\abs{S}\), denote with \(u \coloneqq \binom{2\abs{S}}{2}\) the maximum number of edges in edges in \(C^+\) and define a configuration of \(C\) as a tuple \((\demandOp, \supplyOp, \assignOp)\) with \(\demandOp, \supplyOp \colon 2^{S\cap V(G)} \to \natintZ{u}\) and \(\assignOp\colon \mathcal{T}_S \times \natint{2\abs{S}} \times \natint{\abs{S}} \to 2^{S\cap V(G)}\).
For convenience, let \(\gamma\) be a component-configuration, then we may write \(\demandOp_\gamma\), \(\supplyOp_\gamma\), and \(\assignOp_\gamma\) for its first, second, and third component respectively.

The first part (i.e., \(\demandOp\)) signifies how often each subset of the fracture modulator gets connected by other components and used for the connection requirements of \(\mathcal{T}_C\)~--~what is the additional demand for each subset?
The second component (i.e., \(\supplyOp\)) signifies how often each subset of the fracture modulator gets connected inside this component, but these connections are \emph{not} used to satisfy connection requirements of terminal sets in \(\mathcal{T}_C\)~--~what is the additional supply for each subset?

Finally, consider a terminal set \(T \in \mathcal{T}_S\).
We have \(T \notin \terminalsContainedS\) and we have applied \Cref{rr:degree-negative-instance}; so, by \Cref{stmt:terminal-sets-few-requirements} we have \(d(T) \leq 2|S|\).
This allows us to explicitly store the contribution of \(C\) to every tree assigned to such a terminal set which is necessary since \(T\) can span an arbitrary number of components in \(\augment{G}{\mathcal{T}} - S\).
For each \(i \in \natint{d(T)}\), the information required for the \(i\)-th tree assigned to \(T\) is stored in \(\assign{T}{i}{\cdot}\).
Call this tree \(F^i\).
Let \(\sigma \colon \natint{\abs{S}}\to V(C)\) be a surjection.
Then, we can imagine for all \(j \in \natint{\abs{S}}\) that \(\assign{T}ij\) is a set of vertices in \(S\) which \(\sigma(j)\) can reach via \(F^i[V(C^+)]\).
Note that it might not be all such vertices.
So, we can think of the third argument like an index of a vertex in \(C\).
We do not use the actual vertices here to be able to easily say that two components are in the same configuration.

Not every component can be in any configuration in a valid solution.
For example, a component with \(k\) edges can not supply more than \(k\) additional connections to the other components.
To capture this concept, we say a component \emph{admits} a configuration if it can locally satisfy all the requirements of the configuration and of the terminal sets in \(\mathcal{T}_C\).

To make this rigorous, let \(\gamma=(\demandOp,\supplyOp,\assignOp)\) be a component-configuration.
We now provide an instances of \(\gstp\) to help us define when \(C\) admits a configuration \(\gamma\).
To do so, choose \(\sigma \colon \natint{\abs{S}} \to V(C)\) to be a surjection, which we use to encode \(\assignOp\) as explained above.
This instance is called \(\confInst{C}{\gamma}{\sigma}\).

First, we define the host-graph of the instance \(\confInst{C}{\gamma}{\sigma}\).
For this, we start with the graph \(C^+\) and for each \(Q \subseteq S \cap V(G)\), we add \(\demand{Q}\)-many vertices \(v\) to the graph with \(N(v) = Q\).
Denote this graph with \(H\).

Now, we define the terminal sets, that need to be connected.
Denote for all \(T \in \mathcal{T}_S\), \(i\in \natint{d(T)}\) and \(U\subseteq S\cap V(G)\) the set \(A(T, i, U, \sigma)\coloneqq \{\sigma(j)\mid j \in \natint{\abs{S}}; \assign{T}{i}{j} = U\}\) the set of vertices in \(C\) that is assigned to \(U\) to satisfy connections for the \(i\)-th tree of \(T\).
Now let \(\assignSets{T}{i}{\sigma}\coloneqq \bigcup_{\emptyset\neq U\subseteq S\cap V(G)\colon A(T, i, U,\sigma)\neq \emptyset}\{U \cup A(T, i, U, \sigma)\}\) be the subsets of \(S \cap V(G)\) that get assigned some vertex unioned with the assigned vertices.
Let
\begin{align*}
\mathcal{Q} &\coloneqq \{T \in \mathcal{T}_C \mid T \cap C\neq \emptyset\},\\
\mathcal{S} &\coloneqq \confInstSupplySetDefinition,\\
\mathcal{A} &\coloneqq
\bigcup_{T\in \mathcal{T}_S, i\in \natint{d(T)}}\assignSets{T}{i}{\sigma}.
\end{align*}
Note that if \(C\) consists of exactly one augmented vertex, \(\mathcal{Q}\) is empty; otherwise \(\mathcal{Q}=\mathcal{T}_C\).
Further, the set \(\mathcal{A}\) might intersect \(Q\), but one can verify that \(\mathcal{S}\) is disjoint from \(\mathcal{A}\) and \(Q\).
The complete set of terminal sets, for which we need connections, is \(\mathcal{U}\coloneqq Q\cup \mathcal{S} \cup \mathcal{A}\).

Finally, we need to specify the required number of connections \(d' \colon \mathcal{U} \to \N^+\).
For this extend \(d, \supplyOp\) and \(\assignOp\) canonically to yield 0 or \(\emptyset\) on arguments not in their original domain.
For all \(U \in \mathcal{U}\) let \[ d'(U) \coloneqq d(U) + \supply{U} + \sum_{T\in \mathcal{T}_S}\abs{\left\{i \in \natint{d(T)} \mid U \in \assignSets{T}{i}{\sigma} \right\}} ,\] and define \(\confInst{C}{\gamma}{\sigma} = (H, \mathcal{U}, d')\).

\begin{definition}
\label{def:admit-configuration}
We say a component \(C\in \comp{\augment{G}{\mathcal{T}} - S}\) admits a configuration \(\gamma\), if there is a surjection \(\sigma \colon \natint{\abs{S}} \to V(C)\) such that
\begin{enumerate}
\item \label{item:admit-terminal-set-assigned} for all \(T \in \mathcal{T_S}\), \(i \in \natint{d(T)}\), and \(j \in \sigma^{-1}(T \cap V(C))\), we have \(\assign{T}{i}{j} \neq \emptyset\),
\item there is a solution \((\mathcal{F}, \pi)\) to \(\confInst{C}{\gamma}{\sigma}\) such that
\begin{enumerate}
\item \label{item:admit-supply-assign-no-demand} for all \(F \in \pi^{-1}(\mathcal{S})\), we have that \(V(F) \subseteq V(C^+)\),
\item \label{item:admit-demand-used-sensibly} for all \(v \in V(H)\setminus V(C^+)\) where \(H\) is the host-graph of \(\confInst{C}{\gamma}{\sigma}\), there is exactly one \(F \in \mathcal{F}\) with \(v \in V(F)\) and for this \(F\), we have \(\deg[F]{v} \geq 2\),
\item \label{item:admit-every-F-uses-component-sensibly} for all \(F \in \mathcal{F}\), we have \(E(C^+) \cap E(F)\neq \emptyset\) and \(F\) is cycle-free.\qedhere
\end{enumerate}
\end{enumerate}
\end{definition}

We say that \((\sigma, \mathcal{F}, \pi)\) \emph{gives rise to} \(\gamma\) \emph{on} \(C\).
With \Cref{item:admit-terminal-set-assigned} we ensure for all \(T \in \mathcal{T}_{S}\) that every \(v\in T \cap V(C)\) is connected to a vertex in \(S \cap V(G)\).
\Cref{item:admit-supply-assign-no-demand} ensures that the supply claimed by this configuration is satisfied completely inside \(C^{+}\) while \Cref{item:admit-demand-used-sensibly} ensures that connections required from the outside are only used for a single solution tree.
\Cref{item:admit-every-F-uses-component-sensibly} is not necessary to ensure correctness, but to ensure that the number of admitted configurations stays singly exponential in \(\abs{S}\).
This is achieved by forbidding redundant configurations.

Later, we want to be able to only work with component-configurations and be able to ignore the underlying edge-disjoint trees that give rise to these configurations.
To this end, we need to know the set of configurations a component \(C\) admits.
We call this its signature \(\signature{C}\).

We now characterize, whether an instance is solvable solely based on the signatures of the available components.
For this, let \(\Gamma\) be a function that assigns each component of \(C \in \comp{\augment{G}{\mathcal{T}} - S}\) a configuration in \(\signature{C}\).
We call \(\Gamma\) a \emph{configuration selector}.
Not all configuration selectors are sensible, we only care about those that are \emph{valid}.

\begin{definition}
\label{def:valid-selector}
We call \(\Gamma\) \emph{valid}, if there is a function \(\rho \colon 2^{S \cap V(G)}\times \N \to 2^{2^{S\cap V(G)}}\) such that
\begin{enumerate}
\item \label{item:valid-selector-enough-supply} for all \(U \subseteq S\cap V(G)\), with \(r_U \coloneqq\abs{\{(T, i) \in \terminalsContainedS \times \N\mid U \in \rho(T,i)\}}\) we have \[r_U + \sum_{C\in \comp{\augment{G}{\mathcal{T}} - S}} \demand[\Gamma(C)]{U} \leq \sum_{C \in \comp{\augment{G}{\mathcal{T}} - S}} \supply[\Gamma(C)]{U},\]
\item \label{item:valid-selector-rho-connected}for all \(T \in \terminalsContainedS\) and \(i \in \natint{d(T)}\), the hypergraph \((T \cup \bigcup \rho(T, i), \rho(T,i))\) is minimally connected,
\item \label{item:valid-selector-assign-connected} for all \(T \in \mathcal{T}_S\) and \(i \in [d(T)]\), denote with \[\mathcal{H}\coloneqq\{\assign[\Gamma(C)]{T}{i}{j} \mid C \in \comp{\augment{G}{\mathcal{T}} - S},j\in \natint{\abs{S}}\}\] all subsets of \(S\) that are connected for this terminal set.
Then, the hypergraph \(\left((T\cap S)\cup\bigcup \mathcal{H}, \mathcal{H}\right)\) is connected.\qedhere
\end{enumerate}
\end{definition}

In the definition above, the function \(\rho\) is used to capture how the requirements of the terminal sets in \(\terminalsContainedS\) are fulfilled.
That is, for a \(T \in \terminalsContainedS\) and \(i \in \natint{d(T)}\), \(\rho(T, i)\) gives all connections that are needed to satisfy the \(i\)-th tree assigned to \(T\) and \Cref{item:valid-selector-rho-connected} ensures that these connections actually connect \(T\).
With \Cref{item:valid-selector-enough-supply} we ensure that there is enough supply to meet the demand.
Finally, \Cref{item:valid-selector-assign-connected} ensures that the explicitly stored solutions for the terminal sets \(\mathcal{T}_{S}\) are actually connected.

We require minimal connectivity in \Cref{item:valid-selector-rho-connected} to limit the number of needed variables in our ILP.
Note that we cannot require minimal connectivity in \Cref{item:valid-selector-assign-connected} as this would prevent for any \(T \in \mathcal{T}_S\), \(i \in d(T)\), and \(j \in \natint{\abs{S}}\) that \(\abs{\assign{T}{i}{j}} = 1\), which is required in some instances.

We now show, that a valid configuration selector exists for an instance if and only if it is a positive instance.

\begin{lemma}
\label{stmt:valid-selector-if-positive}
Let \(S\) be a nice fracture modulator of \(\augment{G}{\mathcal{T}}\).
Assume that \Cref{rr:degree-negative-instance} is applied.
Then, the instance is positive if and only if there exists a valid configuration selector with respect to \(S\).
\end{lemma}
\begin{proof}
Assume that the instance is positive and let \((\mathcal{F}, \pi)\) be a solution.
Assume without loss of generality that any \(F \in \mathcal{F}\) is a tree such that its leaves are contained in \(\pi(F)\).
We aim to define a valid configuration selector \(\Gamma\).
Consider a component \(C\in \comp{\augment{G}{\mathcal{T}} - S}\).
We first define \(\demandOp_{\Gamma(C)}\), then \(\supplyOp_{\Gamma(C)}\), and finally \(\assignOp_{\Gamma(C)}\).

Let \(\mathcal{F}_C\subseteq \mathcal{F}\) be the trees that use edges of \(C^+\).
We partition \(\mathcal{F}_C\) by which terminal set a tree is associated with.
Let
\begin{itemize}
\item \(\mathcal{X}_C \coloneqq \{F\in \mathcal{F}_C \mid \pi(F) \in \mathcal{T}_C\}\) be the trees that only have terminals in \(C^+\),
\item \(\mathcal{Y}_C \coloneqq \{F\in \mathcal{F}_C \mid \pi(F) \in \mathcal{T}_S\}\) be the trees that have terminals in multiple components of \(\augment{G}{\mathcal{T}} - S\),
\item \(\mathcal{Z}_C \coloneqq \mathcal{F}_C \setminus (\mathcal{X}_C \cup \mathcal{Y}_C)\) be the trees that do not have any terminal in \(C\),
\end{itemize}
where we omit the component index if the component is clear from context.

Consider the graph \(C^+\) and add for all \(F \in \mathcal{X}\) a copy of \(F[V(G) \setminus V(C^+)]\) connected to the same vertices of \(S\) as in \(F\) to \(C^+\).
Call the obtained graph \(G_C\).
Contract all edges in \(G_C\) that are not incident to \(C^+\) to obtain \(H_C\) (i.e., \(H_C\coloneqq G_C / E(G_C - V(C^+))\)).
This is exactly the graph we obtain when contracting all connected components in \(G_C - V(C^+)\) to a single vertex.
Since \(S\) is a fracture modulator, the vertices in \(V(H_C) \setminus V(C^+)\) are only incident to vertices in \(S\).
For all \(U \subseteq S \cap V(G)\), we set \[\demand[\Gamma(C)]{U} \coloneqq \abs{\{v \in V(H_C) \setminus V(C^+)\mid N_{H_C}(v) = U\}}.\]

Consider any \(F \in \mathcal{Z}\).
If we restrict our view to \(F[V(C)]\), we notice that this graph might be disconnected.
So, we supply the sets of \(S\) that are connected by each of these connected components to the other components, for which it is actually irrelevant these sets are connected via the same component.
Thus, for all \(U \subseteq S \cap V(G)\), we set \[\supply[\Gamma(C)]{U} \coloneqq \sum_{F \in \mathcal{Z}}\abs{\{K \in \comp{F[V(C)]}\mid N_F(V(K)) = U\}}.\]

To define \(\assignOp_{\Gamma(C)}\) and later \(\rho\), we need a global numbering on the trees.
So, let for all \(T \in \mathcal{T}\) the trees associated with this terminal let \(\{F^T_i\}_{i \in \natint{d(T)}} \coloneqq \pi^{-1}(T)\) be numbered globally.
Additionally, choose a surjection \(\sigma_C \colon \natint{\abs{S}} \to C\) globally.
Let \(T \in \mathcal{T}_S\), \(i \in \natint{2\abs{S}}\), and \(v\in C\) be given.
If \(v \notin V(F_i^T)\) or \(i > d(T)\), for all \(j \in \sigma^{-1}(v)\) we set \(\assign{T}{i}{j} \coloneqq \emptyset\).
Otherwise, let \(K\) be the connected component of \(v\) in \(F^T_i[V(C^+)]\).
For all \(j \in \sigma^{-1}(v)\) we set \[\assign[\Gamma(C)]{T}{i}{j} \coloneqq V(K) \cap S.\]

We now prove that \(\Gamma\) defined as above, is a valid configuration selector.
For this, we first show that \(\Gamma\) is a well-defined configuration selector, by showing that all components \(C \in \comp{\augment{G}{\mathcal{T}} - S}\) admit \(\Gamma(C)\).

Consider any \(T \in \mathcal{T}_S\), \(i \in \natint{d(T)}\), \(j \in \sigma_C^{-1}(T\cap V(C))\), and let \(v \coloneqq \sigma(j)\).
As \(v \in T\), we have \(v \in V(F^T_i)\).
Additionally, there is a \(C\neq C' \in \comp{\augment{G}{\mathcal{T}} - S}\) with \(V(C') \cap T \neq \emptyset\).
Let \(u \in T \cap V(C')\) and consider the simple \(vu\)-path \(P\) in \(F^T_i\).
As \(S\) is a fracture modulator, there is a first \(s \in S\) on \(P\).
Denote with \(P'\) the \(vs\)-subpath of \(P\).
We have \(V(P') \subseteq V(C^+)\).
So, the connected component of \(v\) in \(F^T_i[V(C^+)]\) contains \(s\) and by definition we have \(s \in \assign{T}{i}{j}\); so, \(\assign{T}{i}{j} \neq \emptyset\).

Next, we need to find a solution to \(\confInst{C}{\Gamma(C)}{\sigma}\).
We notice that the host graph of \(\confInst{C}{\Gamma(C)}{\sigma}\) is exactly \(H_C\).
If there is a \(T \in \mathcal{T}_C\) with \(T\cap V(C) = \emptyset\), the vertex \(\aug{T}\) is only connected to \(S\) in \(\augment{G}{\mathcal{T}}\) and so \(V(C) = \{\aug{T}\}\), which means that \(G[V(C)]\) is the empty graph as it does not contain any augmented vertices.
Thus, \(\mathcal{F}_C\) is empty, which means \(\demandOp_{\Gamma(C)}\) and \(\supplyOp_{\Gamma(C)}\) are the constant \(0\)-function.
Additionally all \(v \in V(C)\) satisfy \(v \notin V(G) \supseteq T\), so \(\assignOp_{\Gamma(C)}\) is the constant \(\emptyset\)-function.
Therefore, we can verify that in this case \(C\) admits \(\Gamma(C)\).
So, assume from now on, that such a \(T\) does not exist.

Consider any tree \(F \in \mathcal{X}\).
As \(\aug{\pi(F)} \in V(C)\), we have \(\pi(F) \subseteq V(C^+)\).
Denote with \(\phi \colon V(G_C) \to V(H_C)\) the mapping induced by the contractions of \(E(G_C - V(C^+))\) in \(G_C\).
Notice that for all \(\{u,v\} \in E(G_C)\), we have that either \(\phi(u) = \phi(v)\), or \(\phi(u)\phi(v) \in E(H_C)\).
So, \(\phi\) is almost a graph homomorphism.
Denote with \(\phi_E \colon E(G_C) \rightharpoonup E(H_C)\) the partial edge mapping function induced by \(\phi\).
Now let \(\psi \colon G_C \to G\) be the graph homomorphism that maps each \(v\in V(C^+)\) to \(v\) and all vertices in \(V(G_C) \setminus V(C^+)\) to their original vertices in \(G\).
Note that \(\psi\) restricted to edges is an injection, as edges in \(C^+\) are preserved and all other edges in \(G_C\) correspond to edges in edge-disjoint trees of \(\mathcal{X}\).
Thus, \(F' \coloneqq H_C[\phi_E(\psi^{-1}_E(F))]\) is a connected subgraph of \(H_C\).
As \(\aug{\pi(F)}\in V(C)\), we have \(\pi(F) \subseteq V(C^+)\).
Thus, \(\phi(\psi^{-1}(\pi(F))) = \pi(F)\) and \(F'\) connects \(\pi(F)\).

Consider any distinct \(F_1, F_2 \in \mathcal{X}\).
In creating \(G_C\), we added different copies of \(F_1[V(G)\setminus V(C^+)]\) and \(F_2[V(G)\setminus V(C^+)]\), which means that \(\phi_E(\psi_E^{-1}(F_1))\) and \(\phi_E(\psi^{-1}_E(F_2))\) are disjoint.
So, the set \(\mathcal{X}' \coloneqq \{H_C[\phi_E(\psi_E^{-1}(F))]\mid F \in \mathcal{X}\}\) contains edge-disjoint subgraphs of \(H_C\).
As all \(T \in \mathcal{T}_C\) satisfy \(T\cap V(C) \neq \emptyset\), every tree assigned to \(T\) must contain an edge of \(C^+\) and is therefore contained in \(\mathcal{F}_C\).
Thus, we can choose an assignment of \(\mathcal{X}'\) to \(T_C\) such that \(d(T)\)-many trees are assigned to each terminal set \(T\).

Consider \(\mathcal{Z}' \coloneqq \bigcup_{F \in \mathcal{Z}}\{F[N_F[V(K)]] \mid K \in \comp{F[V(C)]}\}\), that is for each component \(K \in \comp{F[V(C)]}\), the set \(\mathcal{Z}'\) contains the subgraph of \(F\) induced on the vertices of \(K\) and the vertices adjacent to \(K\) in \(F\).
Since every tree \(F \in \mathcal{Z}\) is cycle-free, the trees induced by \(F\) in \(\mathcal{Z}'\) do not share two or more vertices.
So, they also do not share any edges of \(F\) and as the trees in \(\mathcal{Z}\) are disjoint, this holds for \(\mathcal{Z}'\) as well.
Assigning an \(F' \in \mathcal{Z}'\) to \(S \cap V(F')\), we get \(\supply{U}\)-many trees assigned to \(U\) for each \(U \subseteq S\cap V(G)\).

Finally, let \(T \in \mathcal{T}_S\) and \(i\in \natint{d(T)}\).
We see that \(\mathcal{K}^T_{i} \coloneqq\comp{F^T_i[V(C^+)]}\) are vertex- and, thus, edge-disjoint trees.
By definition of \(\assignOp\), we have \(\{V(K)\}_{K \in K^T_{i}} = \assignSets{T}{i}{\sigma}\) and there is an assignment of \(\mathcal{K}_{i}^T\) to \(\assignSets{T}{i}{\sigma}\).
Let \(\mathcal{K} \coloneqq \bigcup_{T\in \mathcal{T}_S, i \in d(T)} \{K^T_{i}\}\) be a set of edge-disjoint trees.
For all \(T \in \mathcal{T}_S\), \(i \in \natint{d(T)}\), and \(U \in \assignSets{T}{i}{\sigma}\), we assign a different tree \(K \in \mathcal{K}\) with \(V(K) = U\) to \(U\).
So, to each \(U\subseteq V(C^+)\), we assign \(\sum_{T \in \mathcal{T}_S}\abs{\{i \in \natint{d(T)}\mid U \in \assignSets{T}{i}{\sigma}\}}\) distinct to trees \(U\).
As \(\bigcup_{K \in \mathcal{K}} E(K) \subseteq \bigcup_{F \in \mathcal{Y}} E(F)\), the trees in \(\mathcal{K}\) are edge-disjoint from the trees in \(\mathcal{X}'\) and \(\mathcal{Z}'\).
So, \(\mathcal{L}\coloneqq \mathcal{X}'\cup \mathcal{Z}' \cup \mathcal{K}\) is a set of edge-disjoint trees that with the assignment \(\pi\) described above solves \(\confInst{C}{\Gamma(C)}{\sigma_C}\).

Now, we verify that \((\mathcal{L}, \pi)\) not only solves \(\confInst{C}{\Gamma(C)}{\sigma_C}\), but \((\sigma_C, \mathcal{L}, \pi)\) actually gives rise to \(\Gamma(C)\) on \(C\).
\begin{enumerate}
\item Recall that \(\mathcal{S} = \confInstSupplySetDefinition\).
By choice of \(\pi\), \(\mathcal{Z'} = \pi^{-1}(\mathcal{S})\) and all \(F \in \mathcal{Z}'\) fulfill \(V(F) \subseteq V(C^+)\).
\item Consider some \(v \in V(H_C) \setminus V(C^+)\), denote with \(F\in \mathcal{X}\) the tree that created the vertices \(\phi^{-1}(v)\), and let \(F' \coloneqq H_C[\phi_E(\psi^{-1}_E(F))] \in \mathcal{X}'\) denote the corresponding tree in \(\mathcal{L}\).
Since all edges between \(\inv\phi(v)\) and \(N_{G_C}(\inv\phi(v))\) are contained in \(F\), all edges between \(v\) and \(N_{H_C}(v)\) are contained in \(F'\).
Thus, \(\deg[F']v = \deg[H_C]v\).
Since all trees in \(\mathcal{L}\) are edge-disjoint, this implies that no other \(F'' \in \mathcal{L}\) fulfills \(v\in V(F'')\).

We now show that \(\deg[H_C]v \geq 2\).
Consider any \(u \in \phi^{-1}(v) \subseteq V(F)\).
We know \(\aug{\pi(F)} \in V(C)\), so \(\pi(F) \subseteq V(C^+)\).
Since \(u \notin V(C^+)\), we have \(u \notin \pi(F)\), and, by assumption, \(u\) is not a leaf of \(F\).
Let \(x,y\in V(C^+)\) be leaves of \(F\) such that \(u\) is contained in the simple path \(P\) connecting \(x\) to \(y\) in \(F\).
Since \(x, y \in V(C^+)\), \(v=\phi(u)\) is neither \(x=\phi(x)\) nor \(y=\phi(y)\).
Let \(a\) denote the last vertex preceding any \(\phi^{-1}(v)\), and let \(b\) denote the first vertex following any \(\phi^{-1}(v)\) in \(P\).
Since \(S\) is a fracture modulator, \(a, b \in S\).
Therefore, \(v\) is connected to \(a\) as well as \(b\), and we have \(\deg[H_C]{v} \geq 2\).
\item Consider \(F \in \mathcal{L}\).
By the definitions of the terminal sets, we have \(\abs{\pi(F)} \geq 2\); so, all vertices in \(V(F)\) have an incident edge in \(F\).
Additionally, it is ensured that \(V(F) \cap V(C)\) is non-empty, showing that \(E(F) \cap E(C^+) \neq \emptyset\).
As the trees in \(\mathcal{Z}' \cup \mathcal{K}\) are subgraphs of trees in \(\mathcal{F}\), they are cycle-free.
Let \(F' \in \mathcal{X}'\) and let \(F \in \mathcal{X}\) be such that \(H_C[\phi_E(\inv\psi_E(F))]=F'\).
Note that \(F\) is a tree.
Since \(F\) is isomorphic to \(G_C[\inv\psi_E(F)]\), both are cycle-free.
Observe that \(F' = H_C[\phi_E(\inv\psi_E(F))]\) is equal to \(G_C[\inv\psi_E(F)]\) with some edges contracted; so, \(F'\) is cycle-free as well.
\end{enumerate}

We have proven that \(\Gamma\) is a configuration selector.
Next, we prove that \(\Gamma\) is valid.
For this we need to specify a function \(\rho \colon 2^{S\cap V(G)} \times \N \to 2^{2^{S\cap V(G)}}\) that specifies how, the terminal sets in \(\terminalsContainedS\) are fulfilled.
For any \(T \in \terminalsContainedS\) and \(i \in \natint{d(T)}\), consider \(F^T_i - S\).
This graph might be disconnected and each component connects some subset of \(S\).
So, we set \(\rho(T, i) = \{N_{F^T_i}(V(K))\mid K \in \comp{F^T_i - S}\}\) as each subset of \(S\) that gets connected by one of those components.
All other values of \(\rho\) we define to be \(\emptyset\).
Now, we verify the properties that are required for validity.
\begin{enumerate}
\item Consider any \(U \subseteq S \cap V(G)\).
First, we notice that for all \(C \in \comp{\augment{G}{\mathcal{T}} - S}\), we have
\begin{align*}
\demand[\Gamma(C)]{U}&=\abs{\{v \in V(H_C) \setminus V(C^+)\mid N(v) = U\}}\\
&=\abs{\{v \in V(H_C) \setminus V(C^+)\mid N_{G_C}(\inv\phi(v)) = U\}}\\
&=\sum_{F \in \mathcal{X_C}} \abs{\{K \in \comp{F - V(C^+)} \mid N_F(V(K)) = U\}}.
\end{align*}
Now, consider
\begin{align*}
&\sum_{C \in \comp{\augment{G}{\mathcal{T}} - S}}(\supply[\Gamma(C)]U - \demand[\Gamma(C)]U)\\
=&
 \begin{aligned}[t]
     \sum_{C \in \comp{\augment{G}{\mathcal{T}} - S}} \Bigg(&\sum_{F \in \mathcal{Z_C}}\abs{\{K \in \comp{F[V(C)]}\mid N_F(V(K)) = U\}}\\
     &- \sum_{F \in \mathcal{X_C}} \abs{\{K \in \comp{F - V(C^+)} \mid N_F(V(K)) = U\}}\Bigg)
 \end{aligned}\\
=&
 \begin{aligned}[t]
     \sum_{F \in \mathcal{F}}\Bigg(&\sum_{\substack{C \in \comp{\augment{G}{\mathcal{T}} - S};\\ F \in \mathcal{Z}_C}}\abs{\{K \in \comp{F[V(C)]}\mid N_F(V(K)) = U\}}\\
     &-\sum_{\substack{C \in \comp{\augment{G}{\mathcal{T}} - S};\\ F \in \mathcal{X}_C}}\abs{\{K \in \comp{F - V(C^+)} \mid N_F(V(K)) = U\}}\Bigg)
 \end{aligned}
\end{align*}
Denote with \(e(F)\) the inner term of the sum over \(F \in \mathcal{F}\) and denote with \(s(F)\) and \(d(F)\) the positive and negative sums in \(e(F)\), respectively.
We now prove that \[ e(F) =\begin{cases} 1,&\text{if } \pi(F) \in\terminalsContainedS\land \exists K \in \comp{F - S}: N_F(V(K)) = U\\ 0,&\text{otherwise}.\end{cases}\]

For this, we first show that for all \(F \in \mathcal{F}\) with \(\pi(F) \notin \terminalsContainedS\), it holds that \(e(F) = 0\).
We distinguish two cases.
First, let \(\pi(F) \in \mathcal{T}_S\).
For all \(C \in \comp{\augment{G}{\mathcal{T}} - S}\), we have \(F \in \mathcal{Y}_C\); so, \(s(F) = d(F) = e(F) = 0\).
Now let \(\pi(F) \in \mathcal{T} \setminus (\mathcal{T}_S\cup\terminalsContainedS)\).
Denote with \(C \in \comp{\augment{G}{\mathcal{T}} - S}\) the unique component with \(\pi(F) \in \mathcal{T}_C\).
Now assume there is a \(K \in \comp{F - V(C^+)}\) with \(N_F(V(K)) = U\).
As \(V(K)\) has at least 2 incident outgoing edges in \(F\) and since \(F\) is cycle-free, this \(K\) is unique and \(d(F) = 1\).
Since \(S\) is a fracture modulator, there is a \(C' \in \comp{\augment{G}{\mathcal{T}} - S}\) with \(V(K) \subseteq V(C')\).
This is also the unique component \(C'' \in \comp{\augment{G}{\mathcal{T}} - S}\) which has a \(K' \in \comp{F[V(C'')]}\) with \(N_F(V(K)) = U\).
Thus, \(s(F) = 1\) and \(e(F) = 0\).
If there is no \(K \in \comp{F - V(C^+)}\) with \(N_F(V(K)) = U\), we have that \(d(F) = s(F) = e(F) = 0\).

Now consider \(F \in \mathcal{F}\) with \(\pi(F) \in \terminalsContainedS\).
And let \(C\in \comp{\augment{G}{\mathcal{T}} - S}\) be the unique component with \(\pi(F) \in \mathcal{T}_C\).
Notice that \(\aug{F}\) is only connected to \(S\) and so \(C = \{\aug{F}\}\).
Thus, \(\mathcal{F}_C = \emptyset\) and, in particular, \(\mathcal{X}_C = \emptyset\).
Therefore, \(d(F) = 0\).
Now, assume there is a \(C \in \comp{\augment{G}{\mathcal{T}} - S}\) such that there is a \(K\in\comp{F[C]}\) with \(N_F(V(K)) = U\).
With the same argument as above, \(C\) is unique.
Thus, \(s(F) = 1\).
Finally, notice that the existence of the pair \((C, K)\) is equivalent to the existence of a \(K' \in \comp{F - S}\) with \(N_F(V(K')) = U\), which proves that \(e(F)\) has the desired alternative representation and so
\begin{align*}
&\sum_{C \in \comp{\augment{G}{\mathcal{T}} - S}}(\supply[\Gamma(C)]U - \demand[\Gamma(C)]U)\\
=&\abs{\left\{(T, i) \in 2^{S\cap V(G)} \times \N\;\middle|\; \exists K \in \comp{F^T_i- S}\colon N_{F^T_i}(V(K)) = U\right\}}.
\end{align*}

Now, we consider \(r_U\).
We observe that
\begin{align*}
 r_U=&\abs{\left\{(T, i) \in 2^{S\cap V(G)} \times \N\mid U \in \rho(T,i)\right\}}\\
 =&\abs{\left\{(T, i) \in 2^{S\cap V(G)} \times \N\mid U \in \{N_{F^T_i}(K)\mid K \in \comp{F^T_i - S}\}\right\}}\\
 =&\abs{\left\{(T, i) \in 2^{S\cap V(G)} \times \N\mid \exists K \in \comp{F^T_i - S} : N_{F^T_i}(V(K)) = U\right\}}\\
 =&\sum_{C \in \comp{\augment{G}{\mathcal{T}} - S}}(\supply[\Gamma(C)]U - \demand[\Gamma(C)]U),
\end{align*}
and so \(r_U + \sum_{C \in \comp{\augment{G}{\mathcal{T}} - S}}\demand[\Gamma(C)]U = \sum_{C \in \comp{\augment{G}{\mathcal{T}} - S}}\supply[\Gamma(C)]U.\)
\item \label{item:proof-equiv-selector-positive-valid-terminals-in-S} First we prove that the considered hypergraph is connected. Then, we show that it is minimally connected.

Let \(T \in \terminalsContainedS\) and \(i \in \natint{d(T)}\) be given.
Consider \(s,t \in T \cup \bigcup\rho(T,i)\).
We notice that \(u,v \in V(F^T_i)\), so there is a \(uv\)-path \(P\) in \(F^T_i\).
Denote with \(s_1, s_2, \dots, s_k\) the sequence of vertices in \(S\cap V(G)\) of \(P\).
As \(G[S]\) is edgeless, for all \(i \in \natint{k - 1}\) the vertices \(s_i\) and \(s_{i+1}\) are not adjacent in \(P\).
So, there is a the sub-path connecting \(s_i\) and \(s_{i+1}\), whose inner vertices are contained in one component of \(F^T_i - S\).
Therefore, there is a subset in \(\rho(T,i)\) containing both \(s_i\) and \(s_{i+1}\) and the hypergraph \((T\cup \bigcup \rho(T,i), \rho(T,i))\) is connected.

Now, assume there is an \(R \in \rho(T,i)\) such that the hypergraph \(H\coloneqq(T \cup \bigcup \rho(T,i),\allowbreak \rho(T,i)\setminus \{R\})\) is connected.
Let \(K \in \comp{F^T_i - S}\) with \(N_{F^T_i}(V(K)) = R\) be chosen.
Since \(H\) is connected, by a similar argument as above, \(F^T_i - K\) is connected as well.
We assumed that the leaves \(L\) of \(F^T_i\), fulfill \(L \subseteq \pi(F^T_i) = T\subseteq S\).
So, \(\abs{R} \geq 2\) and let \(s,t \in R\) be distinct.
Observe that there are \(st\)-paths in \(F^T_i - K\) and \(F^T_i[V(K^+)]\) which are edge-disjoint.
Thus, \(F^T_i\) is not cycle-free, which we assumed.
\item This follows analogously to the proof of connectivity in \Cref{item:proof-equiv-selector-positive-valid-terminals-in-S} and concludes the proof that \(\Gamma\) is a valid configuration selector.
\end{enumerate}

Now, assume that there is a valid configuration selector \(\Gamma\).
For every \(C\in\comp{\augment{G}{\mathcal{T}} - S}\), let \((\sigma_C, \mathcal{F}_C, \pi_C)\) give rise to \(\Gamma(C)\) on \(C\) and denote the host graph of \(\confInst{C}{\Gamma(C)}{\sigma_C}\) with \(H_C\).
We now prove that the instance is indeed solvable by constructing a solution.
For this, we split \(\mathcal{F}_C\) into parts according to the designated purpose each tree fulfills in our solution.

First, we set \(\mathcal{S}_C \coloneqq \inv\pi_C(\terminalsContainedS)\) to be all trees that \(\pi_C\) assigns to a terminal set contained in \(S\); so, it will be used to provide connections of subsets of \(S\) to other trees.
The remaining trees of \(\mathcal{F}_C\) will either directly contribute to some \(\mathcal{T}_C\), or will be used to fulfill some requirement of a terminal set in \(\mathcal{T}_S\).

Let \(U \subseteq V(C^+)\) with \(U \cap V(C) \neq \emptyset\).
For each \(T \in \mathcal{T}_S\) and \(i \in \natint{d(T)}\) with \(U \in \assignSets{T}{i}{\sigma_C}\) choose a distinct \(A^C_{T,i,U}\) from \(\inv\pi_C(U)\), which will be used to construct the \(i\)-th tree for the terminal set \(T\).
Finally, we set \(\mathcal{R}_{U} \coloneqq \inv\pi_C(U) \setminus \{A_{T,i,U}^C \mid T \in \mathcal{T}_S, i \in \natint{d(T)}\}\) to be the trees assigned to the terminal set \(U\) in the final solution.
As trees in \(\mathcal{F}_C \setminus \mathcal{S}_C\) might contain edges and vertices not present in \(G\), we cannot directly include them into our final solution.

To fix this problem, consider for each \(C \in \comp{\augment{G}{\mathcal{T}} - S}\) the set \(V(H_C) \setminus V(C^+)\), assume they are all disjoint and call their union \(K\).
Let \(\mathcal{S} \coloneqq \bigcup_{C \in \comp{\augment{G}{\mathcal{T}} - S}} \mathcal{S}_C\) and choose an injection \(\eta \colon K \to \mathcal{S}\) such that for all \(v \in K\) where \(C \in \comp{\augment{G}{\mathcal{T}} - S}\) is such that \(v \in V(H_C)\), we have that \(\pi_C(\eta(v)) = N(v)\).
The existence of \(\eta\) is guaranteed by \Cref{item:valid-selector-enough-supply} of \Cref{def:valid-selector}.
Note that \(\mathcal{S}\) is a set of edge-disjoint trees.

This function gives us a way to consider a vertex in \(K\) and find a unique replacement tree for it using edges and vertices in \(G\) that preserve connectivity.
Additionally, consider a \(C \in \comp{\augment{G}{\mathcal{T}} - S}\) and distinct \(F_1, F_2 \in \mathcal{F}_C\) if we replace the vertices \(v \in (V(F_1) \cup V(F_2)) \cap K\) by \(\eta(v)\) in the respective trees, they are still edge-disjoint since \(V(F_1) \cap K\) and \(V(F_2) \cap K\) are disjoint and so the vertices get replaced by different trees.

Let \(C \in \comp{\augment{G}{\mathcal{T}} - S}\) and \(T \in \mathcal{T}_C\) be such that \(T \cap V(C) \neq \emptyset\).
Consider any \(F \in \mathcal{R}_T\) and denote with \(F^*\) the graph obtained by replacing every vertex in \(v \in V(F) \cap K\) with \(\eta(v)\).
Note that \(F^*\) only uses vertices and edges of \(G\).
Let \(\mathcal{R}_T^*\coloneqq \{F^* \mid F \in \mathcal{R}_T\}\) and assign all those edge-disjoint subgraphs of \(G\) to the terminal set \(T\) in the final solution.
Noticing that \[\abs{\mathcal{R}_T^*} = d'(T) - \abs{\{U \in \mathcal{T}_S, i \in \natint{d(U)} \mid T \in \assignSets{U}{i}{\sigma}\}} = d(T),\] we can now turn our attention to the other terminal sets.

First, we consider \(T \in \mathcal{T}_S\) and let \(i \in \natint{d(T)}\).
Denote the union of all subgraphs that are selected for this terminal set to satisfy the \(i\)-th tree with \[F = \bigcup_{C \in \comp{\augment{G}{\mathcal{T}} - S}, U \in \assignSets{T}{i}{\sigma}} A_{T,i, U}^C.\]
We now prove that \(F\) is connected and that \(T \subseteq V(F)\).

Consider \(C \in \comp{\augment{G}{\mathcal{T}} - S}\), \(j \in \natint{\abs{S}}\), and \(s,t \in \assign[\Gamma(C)]{T}{i}{j}\subseteq S\).
We first prove that \(s\) can reach \(t\) via \(F\).
Any vertex in \(\inv\sigma(j)\) gives rise to an \(U \in \assignSets[\Gamma(C)]{T}{i}{\sigma_C}\subseteq V(C^+)\) with \(s,t \in U\).
The graph \(A^C_{T, i, U}\) is connected, a subgraph of \(F\), and \(s,t \in V(A^C_{T,i,U})\); so, \(s\) can reach \(t\) via \(A^C_{T,i,U}\) and \(F\) as well.

Consider \(s,t \in V(F) \cap S\).
We now show that \(s\) can reach \(t\) via \(F\).
We know by \Cref{item:valid-selector-assign-connected} of \Cref{def:valid-selector}, that there is a sequence of \((C_1, j_1), (C_2, j_2),\dots, (C_k, j_k)\) such that \(s \in \assign[\Gamma(C_1)]{T}{i}{j_1}\), \(t \in \assign[\Gamma(C_k)]{T}{i}{j_k}\), and for all \(\ell \in \natint{k - 1}\) the sets \(\assign[\Gamma(C_i)]{T}{i}{j_i}\) and \(\assign[\Gamma(C_{\ell+1})]{T}{i}{j_{\ell+1}}\) overlap.
As for all \(\ell \in \natint{k}\), the sets\allowbreak{} \(\assign[\Gamma(C_\ell)]{T}{i}{j_\ell}\) are connected via \(F\), there is a path from \(s\) to \(t\) via \(F\).

To show, that \(F\) is connected, it is now enough to show for every \(u \in V(F) \setminus S\), that there is a path to some vertex in \(V(F) \cap S\) using edges of \(F\).
As \(u \in V(F)\), there exists \(C\) and \(i\) such that there is an \(U\in\assignSets[\Gamma(C)]{T}{i}{\sigma_C}\) with \(u \in U\).
By definition, there is an \(s \in U \cap S\).
Since \(A^C_{T,i,U}\) is connected, a subgraph of \(F\), and contains \(u\) and \(s\), the vertex \(u\) can reach \(s\) in \(F\).

To show that \(T \subseteq V(F)\), first let \(t \in T \cap S\).
Consider \Cref{item:valid-selector-assign-connected} of \Cref{def:valid-selector}, and let \(H\) denote the considered hypergraph.
If \(\{t\} = V(H)\), let \(t' \in T\setminus S\)—which exists since \(T \not\subseteq S\)—and \(C\in \comp{\augment{G}{\mathcal{T}} - S}\) be the component containing \(t'\).
Then, for all \(j \in \inv\sigma_C(t')\), we have \(\assign[\Gamma(C)]Tij = \{t\}\).
Otherwise, since \(H\) is connected, there is a hyperedge incident to \(t\) in \(H\).
So, there are, by definition, \(C \in \comp{\augment{G}{\mathcal{T}} - S}\) and \(j \in \natint{\abs{S}}\) with \(t \in \assign[\Gamma(C)]Tij\).
In either case, there is a \(C \in \comp{\augment{G}{\mathcal{T}} - S}\) and \(j \in \natint{\abs{S}}\) with \(t \in \assign[\Gamma(C)]Tij\).
Now, since \(\sigma(j)\) is a assigned to some set containing \(t\), there is a \(U \in \assignSets[\Gamma(C)]Ti\sigma\) with \(t\in U \subseteq V(C^+)\).
So, \(t \in V(A^C_{T, i, U}) \subseteq V(F)\).
Now, consider \(t \in T \setminus S\).
Let \(C\in\comp{\augment{G}{\mathcal{T}} - S}\) be the component containing \(t\) and let \(j \in \inv\sigma_C(t)\).
By \Cref{item:admit-terminal-set-assigned} of \Cref{def:admit-configuration}, we know that \(\assign[\Gamma(C)]{T}{i}{j}\neq \emptyset\), so there is an \(U \in \assignSets[\Gamma(C)]{T}{i}{\sigma_C}\) with \(t \in U\).
Since \(U \subseteq V(A^C_{T,i,U}) \subseteq V(F)\), we have \(T \subseteq V(F)\).

As \(F\) is not a subgraph of \(G\), we cannot assign it to the terminal set \(T\) in our solution.
Now consider the graph \(F^*\) obtained from \(F\) by replacing all \(v \in V(F) \setminus V(G)\) by the subgraph \(\eta(v)\).
This is a connected subgraph of \(G\) and, as we only removed vertices not contained in \(V(G)\), we have \(T \subseteq V(F^*)\).
We assign \(F^*\) to \(T\).
Note that all \(F^*\) created in this manner are edge-disjoint from each other and the subgraphs already assigned to other terminal sets and that to each \(T\in \mathcal{T}_S\), we assign \(d(T)\)-many such subgraphs.

To construct the subgraphs assigned to the terminal sets \(\terminalsContainedS\), let \(\rho\) be a function witnessing validity of \(\Gamma\) on \(G\).
Let \(L\coloneqq\{(T, i, U) \mid T \in \mathcal{T}_S, i \in \natint{d(T)}, U \in \rho(T,i)\}\).
Choose an injection \(\mu \colon L \to \mathcal{S}\) such that \(\image{\mu}\) and \(\image{\eta}\) are disjoint, and for all \((T,i,U) \in L\) with \(C\in\comp{\augment{G}{\mathcal{T}} - S}\) denoting the component containing \(\mu(T,i,U)\), we have \(\pi_C(\mu(T,i,U)) = U\).
The existence of \(\mu\) is guaranteed by \Cref{item:valid-selector-enough-supply} of \Cref{def:valid-selector}.

Now let \(T \in \terminalsContainedS\) and \(i \in \natint{d(T)}\), consider \(F \coloneqq \bigcup_{U \in \rho(T, i)} \mu(T, i, U)\), and let the corresponding hypergraph considered in \Cref{item:valid-selector-rho-connected} of \Cref{def:valid-selector} be denoted by \(H\).
As we assume the instance to be reduced by \Cref{rr:degree-negative-instance}, which in turn applies \Cref{rr:sensible-terminal-sets} exhaustively, we have \(\abs{T} \geq 2\).
As \(T \subseteq V(H)\) and as \(H\) is connected, for every \(t \in T\), there is an \(U \in E(H)=\rho(T,i)\) with \(t \in U\) and so \(t \in V(F)\).
Additionally, as \(H\) is connected, so is \(F\).
We now assign \(F\) to \(T\) in our final solution.
Note that by the injectivity of \(\mu\) and the discontinues of \(\image{\mu}\) and \(\image{\eta}\), all these subgraphs are edge-disjoint from each other and all previously assigned subgraphs.
Additionally, we assign \(d(T)\) subgraphs to \(T\).
Thus, our final solution solves the instance.
\end{proof}

\subsection{Signatures and Equivalence Classes}
\label{sec:org824e70b}
When building our ILP, we want to treat components with the same signature equally.
More specifically, we want to represent all components that share the same signature by a common set of variables.
This Section has three goals.
First, we analyze the equivalence relation on \(\comp{\augment{G}{\mathcal{T}} - S}\) induced by whether or not two components have the same signature.
More specifically, we bound the number of non-empty equivalence classes.
Second, we bound the size of the signatures.
Third, we present a simple method to compute the signature of a component.

To analyze the number of non-empty equivalence classes, we provide a sufficient condition on when two components are equivalent.
We choose this condition in such a way, that it is almost mechanical to bound the number of non-empty equivalence classes for a given instance.

\begin{definition}
\label{def:indistinguishable-components}
Let \(C_1, C_2 \in \comp{\augment{G}{\mathcal{T}} - S}\), we call \(C_1\) and \(C_2\) \emph{indistinguishable}, if there exists a graph isomorphism \(\phi: C^+_1\to C^+_2\) such that
\begin{enumerate}
\item \label{item:modulator-identity} \(\phi\big|_S = \identity{S}\),
\item \label{item:same-augmented} for all \(v \in V(C^+)\), for all \(v \in V(C^+_{1})\), we have \(v \in V(G)\) if and only if \(\phi(v) \in V(G)\),
\item \label{item:same-demand} for all \(T \in \mathcal{T}_{V(C_{1})}\), we have \(d(T) = d(\invAug{\phi(\aug{T})})\).
\end{enumerate}
\end{definition}

We call the equivalence classes induced by whether or not components are indistinguishable by the name \emph{indistinguishablility classes}.
First, we show that if two components are indistinguishable, they are in fact equivalent as well.

\begin{lemma}
\label{stmt:indist-classes-preserve-signature}
Let \(C, C' \in \comp{\augment{G}{\mathcal{T}} - S}\) be indistinguishable components.
Then, \(C\) and \(C'\) are equivalent (i.e., \(\signature{C} = \signature{C'}\)).
\end{lemma}
\begin{proof}
Consider any \(\gamma\in \signature{C}\).
As indistinguishablility is a symmetric property, it is enough to show \(\gamma\in\signature{C'}\) to prove the claim.

Let \((\sigma, \mathcal{F}, \pi)\) give rise to \(\gamma\) on \(C\) and let \(\phi\) be an isomorphism between \(C\) and \(C'\) satisfying the additional requirements of \Cref{def:indistinguishable-components}.
Since for all \(v\in S\), we have \(v = \phi(v)\) and for all \(u \in V(C^+)\), we have \(u \in V(G)\) if and only if \(\phi(u) \in V(G)\), it is possible to extend \(\phi\big|_{V(C^+) \cap V(G)}\) to an isomorphism \(\psi\) between the host-graphs of \(\confInst{C}{\gamma}{\sigma}\) and \(\confInst{C'}{\gamma}{\sigma \circ \inv\phi}\) such that \(\phi\big|_{V(C^+) \cap V(G)}=\psi\big|_{V(C^+) \cap V(G)}\).
One can now verify that \[ \left(\sigma\circ\inv\phi, \left\{\psi_E(F)\mid F\in \mathcal{F}\right\}, \psi \circ \pi \circ \inv\psi\right) \] gives rise to \(\gamma\) on \(C'\).
\end{proof}

We now aim to bound the number of non-empty equivalence classes.
To this end, we bound the number of non-empty indistinguishablility classes, which is an upper bound on the number of equivalence classes.
We only consider instances, where \Cref{rr:degree-negative-instance} was applied.
Note that, if we did not apply this reduction rule, we could exploit \Cref{item:same-demand} of \Cref{def:indistinguishable-components} to create a family of instances, where the number of non-empty indistinguishablility classes is not bounded by any function of \(\abs{S}\).

\begin{lemma}
\label{stmt:frac-num-bounded-non-empty-indist-classes}
There are at most \(2^{\O{\abs{S}^2}}\) non-empty indistinguishablility classes on \(\comp{\augment{G}{\mathcal{T}} - S}\).
\end{lemma}
\begin{proof}
First, we bound the number of non-empty indistinguishablility classes whose components only contain augmented vertices.
As augmented vertices are not adjacent to each other, every component that only contains augmented vertices, contains exactly one such vertex and no non-augmented vertices.
For every terminal set we introduce exactly one augmented vertex.
So, for each pair of augmented vertices \(u,v\) with \(N(u) = N(v)\), we have that \(u = v\).
Thus, for each \(R \subseteq S\), there is at most one augmented vertex \(v\) with \(N(v) = R\) and at most one component \(C'\in\comp{\augment{G}{\mathcal{T}} - S}\) with \(N(V(C')) = R\).
Therefore, there are no more than \(2^{\abs{S}}\) non-empty indistinguishablility classes that only contain augmented vertices.

Now, we bound the number of non-empty indistinguishablility classes whose component contains at least one non-augmented vertex.
For each component \(C \in\comp{\augment{G}{\mathcal{T}} - S}\), the graph \(C^+\) has at most \(\O{|S|^2}\) potential edges.
So, there are at most \(2^\O{|S|^2}\) non-isomorphic graphs for each of the \(|S|\) different sizes of the components.
Overall, there are at most \(\abs{S}2^\O{\abs{S}^2}=2^\O{\abs{S}^2}\) non-isomorphic \(C^+\).
According to \Cref{item:same-augmented}, the isomorphism must preserve whether or not a vertex is augmented.
This increases the number of possible non-empty indistinguishablility classes by a factor of at most \(2^\abs{S}\).

Finally, we need to bound the impact of \Cref{item:same-demand}.
Consider any augmented vertex \(a \in V(C^+)\).
If \(a \in S\), \Cref{item:same-demand} is satisfied by every isomorphism \(\psi\) that satisfies \Cref{item:modulator-identity}.
So, assume \(a \in V(C)\).
As augmented vertices do not share adjacent edges, \(a\) needs to be adjacent to some non-augmented vertex \(v \in V(C)\).
So, \(\invAug{a} \notin \terminalsContainedS\) and by \Cref{stmt:terminal-sets-few-requirements}, we have that \(d(\invAug{a}) \leq 2\abs{S}\).
Therefore, every augmented vertex is in one of \(2\abs{S}\) states.
There are at most \(\abs{S}\) many such augmented vertices.
Thus, the number of non-empty indistinguishablility classes is increased by a factor of \({(2\abs{S})}^\abs{S}=2^\O{\abs{S}\log\abs{S}}\).
Overall the number of non-empty indistinguishablility classes is bounded by \(2^\abs{S} + 2^\O{\abs{S}^2}2^\abs{S}2^\O{\abs{S}\log\abs{S}}=2^\O{\abs{S}^2}\).
\end{proof}

Using \Cref{stmt:indist-classes-preserve-signature}, we conclude that this bound also applies to equivalence classes.

\begin{corollary}
\label{stmt:frac-num-bounded-non-empty-equiv-classes}
On \(\comp{\augment{G}{\mathcal{T}} - S}\) there are at most \(2^{\O{\abs{S}^2}}\) non-empty equivalence classes.
\end{corollary}

Now, we continue this Section by bounding the size of signatures that are admitted by any component.
For this, we consider a necessary condition for a component configuration to be admitted by any component.

\begin{definition}
\label{def:viable-configuration}
Let \(\gamma\) be a component-configuration.
We call \(\gamma\) viable, if \[\sum_{U\subseteq S\cap V(G)}\demand{U}\leq u\abs{S}\text{ and }\sum_{U\subseteq S\cap V(G)} \supply{U} \leq u.\]
The set of all viable configurations is denoted by \(\allViable\).
\end{definition}

We now show, that any configuration, that is admitted by a component is indeed viable.

\begin{lemma}
\label{stmt:admitted-configuration-is-viable}
For all \(C \in \comp{\augment{G}{\mathcal{T}} - S}\), we have \(\signature{C} \subseteq \allViable\).
\end{lemma}
\begin{proof}
Let \(\gamma\in\signature{C}\), let \((\sigma, \mathcal{F}, \pi)\) give rise to \(\gamma\) on \(C\), and let \(H\) be the host-graph of \(\confInst{C}{\gamma}{\sigma}\).
For any \(F\in \mathcal{F}\), every vertex \(x \in V(F) \setminus V(G)\) has \(N(x) \subseteq S\) and \(\deg[F]x \geq 2\).
Since \(F\) is cycle-free, the number of such vertices is bounded by \(|S|\).
Additionally, every vertex in \(V(H) \setminus V(G)\) is contained in some tree of \(\mathcal{F}\).
Therefore \(\abs{V(H) \setminus V(G)}\leq \abs{S}\cdot\abs{\mathcal{F}}\).
Since every tree of \(\mathcal{F}\) contains an edge from \(E(G[C^+])\) of which there are at most \(u\), and since each such edge is contained in at most one tree of \(\mathcal{F}\), we have \(\abs{\mathcal{F}}\leq u\).
Thus, \(\abs{V(H) \setminus V(G)} = \sum_{U\subseteq S\cap V(G)} \demand{U} \leq u \abs{S}\).

Furthermore, \(\sum_{U\subseteq S\cap V(G)}\supply{U} = \sum_{U\subseteq S\cap V(G)} d'(U) \leq \abs{\mathcal{F}} \leq u\).
\end{proof}

Now, we count the number of viable configurations. This gives an upper bound on the size of any signature.

\begin{lemma}
\label{stmt:size-signature}
For all \(C \in \comp{\augment{G}{\mathcal{T}} - S}\), we have \(\abs{\signature{C}}\leq \abs{\allViable}\leq 2^\O{\abs{S}^4}\).
The set \(\allViable\) can be enumerated in running time \(2^\O{\abs{S}^4}\).
\end{lemma}
\begin{proof}
For all \(C \in \comp{\augment{G}{\mathcal{T}} - S}\), from \Cref{stmt:admitted-configuration-is-viable} it follows that \(\abs{\signature{C}} \leq \abs{\!\allViable\!}\).
The number of different \(\demandOp\) and \(\supplyOp\) functions in viable configurations is \(\binom{2^\abs{S} + u\abs{S}}{u\abs{S}}\leq\left(2^\abs{S} + u\abs{S}\right)^{u\abs{S}} = 2^{\O{\abs{S}^4}}\) and \(\binom{2^\abs{S} + u}{u}\leq\left(2^\abs{S} + u\right)^u = 2^{\O{\abs{S}^3}}\) respectively.
They can be enumerated in time \(\O{2^\abs{S}2^\O{\abs{S}^4}} = 2^\O{\abs{S}^4}\) and \(\O{2^\abs{S}2^\O{\abs{S}^3}} = 2^\O{\abs{S}^3}\) as well.
The number of different \(\assignOp\) functions in any configuration can be bounded by \(\left(2^\abs{S}\right)^{\abs{\mathcal{T}_S}\cdot 2\abs{S}\cdot\abs{S}}\leq\left(2^\abs{S}\right)^{\abs{S}\cdot 2\abs{S}\cdot \abs{S}} = 2^{\O{\abs{S}^4}}\) and enumerated in time \(\O{2^\abs{S}2^\O{\abs{S}^4}} = 2^\O{\abs{S}^4}\) as well.
Therefore, the number of different viable configurations and by extension \(\abs{\signature{C}}\) is bounded by \(2^{\O{\abs{S}^4}}\) and can be enumerated in this running time as well.
\end{proof}

Finally, we exploit this fact, to show, that the signature of a component can be computed reasonably fast.

\begin{lemma}
\label{stmt:compute-signature}
Let \(C \in \comp{\augment{G}{\mathcal{T}} - S}\).
We can compute \(\signature{C}\) with running time \(2^{\O{\abs{S}^4\log\abs{S}}}\).
\end{lemma}
\begin{proof}
First note, that all viable configurations—a super-set of \(\signature{C}\)—can be enumerated in running time \(2^\O{\abs{S}^4}\).
Consider any viable configuration \(\gamma\).
Given some \((\sigma, \mathcal{F}, \pi)\) it is possible to check in linear time, whether \((\sigma, \mathcal{F}, \pi)\) gives rise to \(\gamma\) on \(C\).
Denote with \(\mathcal{E} \coloneqq E(\mathcal{F})\) the edge sets of the disjoint trees of the possible solution.
As all \(F\in \mathcal{F}\) intersect with \(E(C^+)\) and since \(\abs{\mathcal{F}} \leq u\), we have that \(\mathcal{E}\cup \left\{E(\confInst{C}{\gamma}{\sigma})\setminus \bigcup \mathcal{E})\right\}\) partitions \(E(\confInst{C}{\gamma}{\sigma})\) into at most \(u + 1\) parts.
So, it is easy to enumerate a superset \(\mathscr{X}\) of all \((\sigma, \mathcal{F},\pi)\) that solve \(\confInst{C}{\gamma}{\sigma}\) and therefore gives rise to \(\gamma\) on \(C\).
We set \(\mathscr{X}\) to be the set of all surjections \(\sigma\colon \natint{\abs{S}} \to C\), all partitions of \(E(\confInst{C}{\gamma}{\sigma})\) into at most \(u + 1\) parts and functions from the current \(\mathcal{F}\) to the set of all terminal sets \(\mathcal{U}\).

The number of surjections is bounded by \(\abs{S}^\abs{S} = 2^\O{\abs{S}\log\abs{S}}\).
Since \[ \abs{E(\confInst{C}{\gamma}{\sigma})}\leq u + \sum_{U \subseteq S} \abs{U}\demand{U} \leq u + u\abs{S}^2= \O{\abs{S}^4} ,\] the number of partitions of \(E(\confInst{C}{\gamma}{\sigma})\) into at most \(u + 1\) parts is bounded by \(\sum_{i=1}^{u+1}i^{\O{\abs{S}^4}}\leq (u+1)\left(u+1\right)^{\O{\abs{S}^4}} = 2^\O{\abs{S}^4\log\abs{S}}\).
Additionally, we have \(\abs{\mathcal{U}}\leq\abs{V(C)} + u + \abs{S} = \O{\abs{S}^2}\).
Therefore, the number of different possible \(\pi\) is bounded by \(\left(\abs{S}^2\right)^\O{\abs{S}^2} = 2^{\O{\abs{S}^2\log\abs{S}}}\) and so \(\abs{\mathscr{X}} \leq 2^\O{\abs{S}\log\abs{S}} 2^{\O{\abs{S}^4\log\abs{S}}} 2^{\O{\abs{S}^2\log\abs{S}}} = 2^{\O{\abs{S}^4\log\abs{S}}}\).
Thus, the running time to check whether \(C\) admits \(\gamma\) is bounded by \(2^\O{\abs{S}^4\log\abs{S}}\O{\abs{S}^4} = 2^\O{\abs{S}^4\log\abs{S}}\).
Since we need to check \(2^\O{\abs{S}^4}\) different configurations, the running time to compute \(\signature{C}\) is in \(2^\O{\abs{S}^4\log\abs{S}}\) as well.
\end{proof}

\subsection{Integer Linear Program Representation}
\label{sec:orgf3aac70}
In this Section, we finally provide a linear programming representation of the instance of \(\gstp\).
We first create separate integer-linear-program representations for \Cref{item:valid-selector-rho-connected,item:valid-selector-assign-connected} in \Cref{def:valid-selector}.
Finally, we create a single ILP combining those such that this ILP has a feasible assignment if and only if the considered instance is positive.

To construct our ILPs, assume that \(\Gamma\) is a configuration selector.
First we construct an ILP to check \Cref{item:valid-selector-rho-connected} of \Cref{def:valid-selector}.
For this, we assume that for all \(U \subseteq S\cap V(G)\), \(\vec{s}_U\) corresponds to \(r_U\) of \Cref{def:valid-selector}.
\begin{definition}
\label{def:rho-lin-rep}
For all \(U \subseteq S\cap V(G)\), denote with \(\mathcal{M}_U\) all minimally connected hypergraphs on the vertex set \(U\).
Let \(\vec{s}\) be indexed by \(U \subseteq S\cap V(G)\).
Then, \(\rhoLinRep{\vec{s}}\) is the ILP
\begin{subequations}
\begin{alignat}{2}
\vec{s}_{U}&\in \N&\quad&\forall U \subseteq S \cap V(G),\notag\\
\vec{p}_{T, H} &\in \N&&\forall T \in \terminalsContainedS, T \subseteq U \subseteq S, H \in \mathcal{M}_U,\notag\\
\vec{q}_{T, H, E} &\in \N&&\forall T \in \terminalsContainedS, T \subseteq U \subseteq S, H \in \mathcal{M}_U, R \in E(H),\notag\\
\sum_{\substack{T \subseteq U \subseteq S \cap V(G),\\H \in \mathcal{M}_U}}\vec{p}_{T, H} &= d(T)&&\forall T \in \terminalsContainedS,\label{eq:rho-lin-rep-exact}\\
\vec{q}_{T, H, R}&\geq \vec{p}_{T, H}&&\forall T \in \terminalsContainedS,T\subseteq V \subseteq S\cap V(G), H \in \mathcal{M}_U, R \in E(H),\label{eq:rho-lin-rep-enough-edges}\\
\sum_{\substack{T \in \terminalsContainedS,\\T \subseteq U \subseteq S\cap V(G),\\H\in \mathcal{M}_U\!\colon\\R \in E(H)}}\vec{q}_{T, H, R}&\leq \vec{s}_R&&\forall R \subseteq S \cap V(G).\label{eq:rho-lin-rep-enough-supply}
\end{alignat}
\end{subequations}
\end{definition}

For all \(T \in \terminalsContainedS\), \(T\subseteq U \subseteq S\cap V(G)\), and \(H \in \mathcal{M}_U\), the variable \(\vec{p}_{T,H}\) denotes how often the hypergraph \(H\) is used to fulfill requirements of \(T\).
For all \(R \in E(H)\), the variable \(\vec{q}_{T,H,R}\) denotes how often the hyperedge \(R\) gets used to construct the hypergraph \(H\) that gets assigned to \(T\).
We first check, that enough hypergraphs are assigned to \(T\).
Then, we check that each assigned hypergraph has enough edges available.
Finally, we check that each edge is not used more often globally than allowed.

Now we show, that indeed \(\rhoLinRep{\vec{s}}\) fully captures \Cref{item:valid-selector-rho-connected} of \Cref{def:valid-selector} with few variables.
\begin{lemma}
\label{stmt:rho-lin-rep-sensible}
Assume that \Cref{rr:sensible-terminal-sets} is applied exhaustively and let \(\vec{s}\) be a vector indexed by \(U \subseteq S\cap V(G)\).
Then, \(\rhoLinRep{\vec{s}}\) has
\begin{enumerate}
\item \(2^\O{\abs{S}^2}\) variables,
\item a feasible assignment if and only if there is a function \({\rho \colon \terminalsContainedS \times \N \to 2^{2^{S\cap V(G)}}}\) that satisfies \Cref{item:valid-selector-rho-connected} of \Cref{def:valid-selector} and for all \(U \subseteq S\cap V(G)\), it holds that \(r_U \leq \vec{s}_U\).\qedhere
\end{enumerate}
\end{lemma}
\begin{proof}
First, we bound the number of variables needed.
The number of \(\vec{s}\) variables is bounded by \(2^{\abs{S\cap V(G)}}\).
Consider \(T \in \terminalsContainedS\), \(T\subseteq U \subseteq S \cap V(G)\) and \(H \in \mathcal{M}_U\).
As \(H\) is minimally-connected, \(\abs{E(H)} \leq \abs{U} - 1\).
Any \(F \in E(H)\) is chosen from \(2^{\abs{S\cap V(G)}}\) many possibilities.
As hypergraphs in \(\mathcal{M}_U\) only differ by their hyperedges, we have \(\abs{\mathcal{M}_U} \leq \binom{2^{\abs{S\cap V(G)}} + \abs{S\cap V(G)}}{\abs{S\cap V(G)}} = 2^\O{\abs{S\cap V(G)}^2}\).
Since \(\abs{\terminalsContainedS}\leq 2^{\abs{S}}\), the number of \(\vec{p}\) variables is bounded by \(\abs{\terminalsContainedS}2^\abs{S\cap V(G)}2^\O{\abs{S\cap V(G)}^2} = 2^\O{\abs{S}^2}\).
For each \(\vec{p}\) variable, the number of \(\vec{q}\) variables is bounded by \(\abs{S} - 1\).
So, the number of \(\vec{q}\) variables is bounded by \(2^\O{\abs{S}^2}\) as well.

Let \(\vec{s}\) be given and let \(\vec{p}\) and \(\vec{q}\) be chosen such that the assignment is feasible.
By \Cref{eq:rho-lin-rep-exact}, it is possible to choose \(\rho\) such that for all \(T \in \terminalsContainedS\), \(T\subseteq U \subseteq S\cap V(G)\), and \(H \in \mathcal{M}_U\) the number of \(i\in \natint{d(T)}\) with \(\rho(T,i) = E(H)\) is exactly \(\vec{p}_H\).
For all \(T \in \terminalsContainedS\), we have \(\abs{T} \geq 2\).
So, for all \(i \in \natint{d(T)}\), we have \(T \subseteq \bigcup\rho(T, i)\).
As the hypergraph \((T \cup \bigcup\rho(T,i), \rho(T,i)) = (\bigcup\rho(T,i), \rho(T,i))\) is contained in \(\mathcal{M}_{\bigcup \rho(T,i)}\), it is minimally connected and \(\rho\) satisfies \Cref{item:valid-selector-rho-connected} of \Cref{def:valid-selector}.
Additionally, let any \(R \subseteq S \cap V(G)\) be given.
By \Cref{eq:rho-lin-rep-enough-edges,eq:rho-lin-rep-enough-supply}, we have that
\[
r_R = \abs{\{(T, i) \in \terminalsContainedS \times \N\mid R \in \rho(T,i)\}}
= \sum_{\substack{T \in \terminalsContainedS,\\T \subseteq U \subseteq S\cap V(G),\\H\in \mathcal{M}_U\!\colon\\R \in E(H)}}\vec{p}_{T, H}
\leq \sum_{\substack{T \in \terminalsContainedS,\\T \subseteq U \subseteq S\cap V(G),\\H\in \mathcal{M}_U\!\colon\\R \in E(H)}}\vec{q}_{T, H, U}
 \leq \vec{s}_R.
\]

Now let \(\rho\) satisfy \Cref{item:valid-selector-rho-connected} of \Cref{def:valid-selector} and for all \(R \subseteq S\cap V(G)\) assume that \(r_R\leq \vec{s}_R\).
For \(T \in \terminalsContainedS\), \(T \subseteq U \subseteq S\cap V(G)\) and \(H \in \mathcal{M}_U\), we set \(\vec{p}_{T, H} = \abs{\{i \in \natint{d(T)} \mid E(H) = \rho(T,i)\}}\) and for all \(R \in E(H)\), we set \(\vec{q}_{T, H, R} = \vec{p}_{T,H}\).
By definition, \Cref{eq:rho-lin-rep-enough-edges} holds.
Since \(\abs{T} \geq 2\), for all \(i \in \natint{d(T)}\), we have \(T \subseteq \bigcup\rho(T,i)\).
So, the hypergraph \((\bigcup \rho(T,i), \rho(T,i))\) is well-defined and minimally connected.
Therefore, \(\sum_{T \subseteq U \subseteq S\cap V(G),H\in \mathcal{M}_U} \vec{p}_{T, H} = d(T)\), satisfying \Cref{eq:rho-lin-rep-exact}.
Finally consider any \(R \subseteq S \cap V(G)\).
By choice of \(\vec{p}\) and \(\vec{q}\), we have that
\[
\sum_{\substack{T \in \terminalsContainedS,\\T \subseteq U \subseteq S\cap V(G),\\H\in \mathcal{M}_U\!\colon\\R \in E(H)}}\vec{q}_{T, H, U}
= \sum_{\substack{T \in \terminalsContainedS,\\T \subseteq U \subseteq S\cap V(G),\\H\in \mathcal{M}_U\!\colon\\R \in E(H)}}\vec{p}_{T, H}
=\abs{\{(T, i) \in \terminalsContainedS \times \N\mid R \in \rho(T,i)\}}
= r_R
\leq \vec{s}_R,
\]
showing that \Cref{eq:rho-lin-rep-enough-supply} is satisfied as well.
\end{proof}

Now, we construct an ILP to check \Cref{item:valid-selector-assign-connected} of \Cref{def:valid-selector}.
For this, we assume that for all \(T \in \mathcal{T}_S\), \(i\in \natint{d(T)}\), and \(U \subseteq S\cap V(G)\), the value \(\vec{a}_{T,i,U}\) corresponds to the indicator variable whether there is a \(C \in \comp{\augment{G}{\mathcal{T}} - S}\), and \(j \in \natint{\abs{S}}\) with \(U = \assign[\Gamma(C)]{T}{i}{j}\).

\begin{definition}
\label{def:assign-lin-rep}
Let \(\vec{a}\) be indexed by \(T\in \mathcal{T}_S\), \(i \in \natint{d(T)}\), and \(U \subseteq S\cap V(G)\).
Then, \(\assignLinRep{\vec{a}}\) is the ILP
\begin{subequations}
\begin{alignat}{2}
\vec{a}_{T,i,Y}&\in \{0,1\}&\quad&\forall T\in \mathcal{T}_S, i \in \natint{d(T)}, Y \subseteq S\cap V(G),\notag\\
\vec{b}_{T,i,U}&\in \{0,1\}&&\forall T\in \mathcal{T}_S, i \in \natint{d(T)}, T\cap S \subseteq U \subseteq S\cap V(G),\notag\\
\sum_{T\cap S \subseteq U \subseteq S\cap V(G)}\vec{b}_{T,i,U}&= 1&&\forall T\in \mathcal{T}_S, i \in \natint{d(T)},\label{eq:assign-lin-rep-exact}\\
\sum_{\emptyset \neq Y \subseteq U}\vec{a}_{T,i,Y}&\geq \vec{b}_{T,i,U}&&\forall T\in \mathcal{T}_S, i \in \natint{d(T)}, T\cap S\subseteq U\subseteq S\cap V(G),\label{eq:assign-lin-rep-at-least-one-edge}\\
\sum_{Y \subseteq U \subseteq S\cap V(G)} \vec{b}_{T,i,U}&\geq \vec{a}_{T,i,Y} &&\forall T\in \mathcal{T}_S, i \in \natint{d(T)}, Y \subseteq S\cap V(G).\label{eq:assign-lin-rep-well-formed}\\
\sum_{\substack{Y \subseteq U\!\colon\\Y \cap X \neq \emptyset,Y \setminus X \neq \emptyset}}\vec{a}_{T,i,Y}&\geq \vec{b}_{T,i,U}&&\forall T\in \mathcal{T}_S, i \in \natint{d(T)}, T\cap S\subseteq U\subseteq S\cap V(G),\emptyset \subset X \subset U,\label{eq:assign-lin-rep-cut}
\end{alignat}
\end{subequations}
\end{definition}

\looseness=-1
In the above definition, for all \(T \in \mathcal{T}_S\) and \(i \in \natint{d(T)}\), we essentially choose one \(T\subseteq U \subseteq S\cap V(G)\), where \(U\) is the vertex set of the hypergraph considered in \Cref{item:valid-selector-assign-connected} of \Cref{def:valid-selector}.
This is the unique \(U\) with \(\vec{b}_{T,i,U} = 1\).
First, we encode that at least one hyperedge is present in the considered hypergraph.
Then, we ensure that the hypergraph is well formed.
To do so, we check that all sets \(Y\) that get assigned to (i.e., \(\vec{a}_{T,i,Y} = 1\)) actually only use vertices of \(U\).
Finally, we ensure that each cut of \(U\) is crossed by at least one edge.

Now we show, that indeed \(\assignLinRep{\vec{a}}\) fully captures \Cref{item:valid-selector-assign-connected} of \Cref{def:valid-selector} with few variables.

\begin{lemma}
\label{stmt:assign-lin-rep-sensible}
Let \(\vec{a}\) be a vector indexed by \(T \in \mathcal{T}_S\), \(i \in \natint{d(T)}\), and \(U \subseteq S\cap V(G)\). Then,
\begin{enumerate}
\item \(\assignLinRep{\vec{a}}\) has \(2^\O{\abs{S}}\) variables,
\item assuming that for all \(T \in \mathcal{T}_S\), \(i\in\natint{d(T)}\), \(U\subseteq S\cap V(G)\) we have \[\vec{a}_{T,i,U}=\begin{cases} 1, & \text{if }\exists C \in \comp{\augment{G}{\mathcal{T}} - S}, j \in \natint{\abs{S}}\colon \assign[\Gamma(C)]{T}{i}{j} = U\\0,&\text{otherwise}.\end{cases}\]
Then, \(\assignLinRep{\vec{a}}\) has a feasible assignment if and only if \(\Gamma\) satisfies \Cref{item:valid-selector-assign-connected} of \Cref{def:valid-selector}.\qedhere
\end{enumerate}
\end{lemma}
\begin{proof}
Since \(S\) is nice, we have \(\mathcal{T}_S \cap \terminalsContainedS = \emptyset\).
Therefore, by \Cref{stmt:terminal-sets-few-requirements}, we need at most \(\O{\abs{S}\abs{S}2^{\abs{S\cap V(G)}}} = 2^\O{\abs{S}}\) variables.

Let \(T \in \mathcal{T}_S\) and \(i \in \natint{d(T)}\).
For all \(Y \subseteq S\cap V(G)\), if there is a \(C \in \comp{\augment{G}{\mathcal{T}} - S}\) and \(j \in \natint{\abs{S}}\) with \(\assign[\Gamma(C)]Tij = Y\), we assume that \(\vec{a}_{T,i,Y}= 1\) and otherwise that \(\vec{a}_{T,i,Y} = 0\).
Additionally, denote with \(H\) the hypergraph considered in \Cref{item:valid-selector-assign-connected} of \Cref{def:valid-selector}

First, assume that \(\assignLinRep{\vec{a}}\) is feasible and let \(\vec{b}\) be chosen accordingly.
Let \(T\cap S \subseteq Z\subseteq S \cap V(G)\) be the unique set with \(\vec{b}_{T,i,Z} = 1\).
By \Cref{eq:assign-lin-rep-well-formed}, we know that for all \(Y\subseteq S\cap V(G)\) with \(\vec{a}_{T,i,Y} = 1\), we have \(Y \subseteq Z\).
So, the hypergraph \(J \coloneqq (Z, \{F\subseteq S\cap V(G)\mid \vec{a}_{T,i,F} = 1\})\) is well-defined.
Notice that \(E(J) = E(H)\).
By \Cref{eq:assign-lin-rep-cut}, \(J\) is connected.
We now show \(Z = V(H)\), which shows that \(J = H\) and that \(H\) is connected.
If \(\abs{Z} \geq 2\), we have that \(Z \subseteq \bigcup_{F\in E(J)} F = \bigcup_{F\in E(H)} F\subseteq V(H)\).
Now, assume \(\abs{Z} \leq 1\).
Since \(S\) is nice, there is a \(C \in \comp{\augment{G}{\mathcal{T}} - S}\) and \(j \in \natint{\abs{S}}\) with \(\assign[\Gamma(C)]Tij \neq \emptyset\).
Let \(U \coloneqq \assign[\Gamma(C)]Tij\in E(J)\).
Thus, \(U \subseteq Z\), combined with \(\abs{Z} \leq 1\), we have \(U = Z\).
As \(U \in E(H)\), we have \(U = Z \subseteq V(H)\).
Now, let \(v \in V(H)\).
If \(v \in T \cap S\), we have by definition of \(Z\) that \(v \in Z\).
Otherwise, if \(\{v\} \neq Z\), by \Cref{eq:assign-lin-rep-cut}, there is a \(Y \subseteq S\cap V(G)\) with \(v \in Y\) and \(\vec{a}_{T, i, Y} = 1\).
So, there is a \(C \in \comp{\augment{G}{\mathcal{T}} - S}\) and a \(j \in \natint{\abs{S}}\) with \(\assign[\Gamma(C)]Tij = Y\) and \(Y \in E(J)\).
Thus, \(v \in \bigcup_{Y \in E(J)} Y \subseteq Z\).
If \(\{v\} = Z\), by \Cref{eq:assign-lin-rep-at-least-one-edge}, we have \(\vec{a}_{T,i,Z} \geq \vec{b}_{T,i,Z} = 1\) and by the same argument as above \(v \in \bigcup_{Y \in E(J)} Y \subseteq Z\) as well.

Second, assume that \(\Gamma\) satisfies \Cref{item:valid-selector-assign-connected} of \Cref{def:valid-selector}.
Denote with \(Z \coloneqq (T\cap S) \cup\allowbreak\bigcup_{Y \subseteq S\cap V(G)\colon \vec{a}_{T,i,Y} = 1} Y\).
Set \(\vec{b}_{T,i,Z} \coloneqq 1\) and for all \(Z \neq W \subseteq S \cap V(G)\) we set \(\vec{b}_{T,i,W} \coloneqq 0\).
We notice immediately that \Cref{eq:assign-lin-rep-exact} is satisfied.
For all \(T \cap S \subseteq U \subseteq S\cap V(G)\) with \(U \neq Z\), \Cref{eq:assign-lin-rep-at-least-one-edge,eq:assign-lin-rep-cut} is satisfied.
Since \(S\) is nice fracture modulator, there is a \(C \in \comp{\augment{G}{\mathcal{T}} - S}\) with \(V(C) \cap T \neq \emptyset\).
Consider \((\sigma, \mathcal{F}, \pi)\) that gives rise to \(\Gamma(C)\) on \(C\).
Let \(j \in \inv\sigma(V(C) \cap T)\), then \(\assign[\Gamma]Tij \neq \emptyset\).
So, we have \(\sum_{\emptyset \neq Y \subseteq Z} \vec{a}_{T,i,Y} \geq \vec{a}_{T,i,\assign[\Gamma(C)]Tij} = 1=\vec{b}_{T,i,Z}\), meaning that \Cref{eq:assign-lin-rep-at-least-one-edge} is satisfied for \(U = Z\).
Consider any \(\emptyset \subset X \subset U\).
As \(Z = V(H)\), there is a hyperedge \(F \in E(H)\) with \(X \cap F \neq \emptyset\) and \(X \setminus F \neq \emptyset\).
Thus, \(\sum_{Y \subseteq U\colon Y \cap X \neq \emptyset,Y \setminus X \neq \emptyset}\vec{a}_{T,i,Y} \geq \vec{a}_{T,i,F} = 1 \geq \vec{b}_{T, i, U}\) and \Cref{eq:assign-lin-rep-cut} is satisfied for \(X = Z\) as well.
Finally, consider any \(Y \subseteq S \cap V(G)\).
If \(\vec{a}_{T,i,Y} = 0\), \Cref{eq:assign-lin-rep-well-formed} is trivially satisfied.
Otherwise, \(\vec{a}_{T,i,Y} = 1\) and \(Y \subseteq Z\), which means that \(\sum_{Y \subseteq U \subseteq S\cap V(G)} \vec{b}_{T,i,U} \geq \vec{b}_{T,i,Z} = 1 \geq \vec{a}_{T,i,Y}\).
Thus, \Cref{eq:assign-lin-rep-well-formed} is satisfied as well.
\end{proof}

\looseness=-1
Now, we combine all the pieces to obtain an ILP that fully captures whether \(\Gamma\) is valid.
For this, denote with \(\mathfrak{C}\) the set of non-empty equivalence classes of \(\comp{\augment{G}{\mathcal{T}} - S}\) and extend \(\signatureOp\) to \(\mathfrak{C}\) to be defined as the signature of any component in the corresponding equivalence class.
For all \(\mathcal{X}\in \mathfrak{C}\), we denote with \(n_\mathcal{X}\) the number of components in \(\mathcal{X}\) (i.e., \(\abs{\mathcal{X}}\)).

\begin{definition}
\label{def:selector-lin-rep}
Let \(N\coloneqq \sum_{\mathcal{X}\in \mathfrak{C}}n_\mathcal{X}\). We denote with \(\selectorLinRep\) the ILP
\begin{subequations}
\begin{alignat}{2}
\vec{d}_{\mathcal{X}, \gamma} &\in \N &&\forall \mathcal{X} \in \mathfrak{C}, \gamma \in \signature{\mathcal{X}},\notag\\
\vec{s}_{U}&\in \N&&\forall U \subseteq S \cap V(G),\notag\\
\vec{a}_{T,i,U}&\in \{0,1\}&&\forall T \in \mathcal{T}_S, i \in \natint{d(T)}, U \subseteq S\cap V(G),\notag\\
\sum_{\gamma \in \signature{\mathcal{X}}} \vec{d}_{\mathcal{X}, \gamma} &= n_\mathcal{X} &&\forall \mathcal{X} \in \mathfrak{C},\label{eq:selector-lin-rep-exact}\\
\vec{s}_U + \sum_{\substack{\mathcal{X} \in \mathfrak{C},\\\gamma \in \signature{\mathcal{X}}}}\demand[\gamma]U\vec{d}_{\mathcal{X}, \gamma} &\leq \sum_{\substack{\mathcal{X} \in \mathfrak{C},\\\gamma \in \signature{\mathcal{X}}}}\supply[\gamma]U\vec{d}_{\mathcal{X}, \gamma}&\quad&\forall U \subseteq S\cap V(G),\label{eq:selector-lin-rep-dem-sup}\\
\sum_{\substack{\mathcal{X}\in \mathfrak{C}, \gamma \in \signature{\mathcal{X}}\!\colon\\\exists j \in \natint{\abs{S}}\colon U = \assign[\gamma]Tij}}\vec{d}_{\mathcal{X}, \gamma}&\geq \vec{a}_{T,i,U}&&\forall T \in \mathcal{T}_S,i\in \natint{d(T)}, U\subseteq S\cap V(G),\label{eq:selector-lin-rep-a-lb}\\
\sum_{\substack{\mathcal{X}\in \mathfrak{C}, \gamma \in \signature{\mathcal{X}}\!\colon\\\exists j \in \natint{\abs{S}}\colon U = \assign[\gamma]Tij}}\vec{d}_{\mathcal{X}, \gamma}&\leq N\vec{a}_{T,i,U}&&\forall T \in \mathcal{T}_S,i\in \natint{d(T)}, U\subseteq S\cap V(G),\label{eq:selector-lin-rep-a-ub}\\
\rhoLinRep{\vec{s}},\notag\\
\assignLinRep{\vec{a}}.&&&\tag*{\qedhere}
\end{alignat}
\end{subequations}
\end{definition}

In this definition, we basically represent a configuration selector.
For all \(\mathcal{X}\in \mathfrak{C}\) and \(\gamma \in \signature{\mathcal{X}}\), the variable \(\vec{d}_{\mathcal{X}, \gamma}\) denotes how many components in the equivalence class \(\mathcal{X}\) take the configuration \(\gamma\) in our solution.
We then use auxiliary variables \(\vec{s}\) and \(\vec{a}\) together with the previously analyzed ILPs \(\rhoLinRep{\vec{s}}\) and \(\assignLinRep{\vec{a}}\) to ensure that this solution corresponds to a valid configuration selector.

\begin{lemma}
\label{stmt:selector-lin-rep-sensible}
The ILP \(\selectorLinRep\) has
\begin{enumerate}
\item \(2^\O{\abs{S}^4}\) many variables,
\item a feasible assignment if and only if there is a valid configuration selector.\qedhere
\end{enumerate}
\end{lemma}
\begin{proof}
The ILP \(\selectorLinRep\) itself uses \(2^\O{\abs{S}} + \O{\sum_{\mathcal{X} \in \mathfrak{C}}\abs{\signature{\mathcal{X}}}}\) variables and the sub-ILPs use \(2^\O{\abs{S}^2}\) additional variables.
Using \Cref{stmt:frac-num-bounded-non-empty-equiv-classes} and \Cref{stmt:size-signature}, we observe that \(\sum_{\mathcal{X} \in \mathfrak{C}}\abs{\signature{\mathcal{X}}} = 2^\O{\abs{S}^4}\); so, in total the ILP \(\selectorLinRep\) uses \(2^\O{\abs{S}^4}\) variables.

Now consider any feasible assignment to the variables \(\vec{d}, \vec{s}\) and \(\vec{a}\) and consider any configuration selector \(\Gamma\) that assigns for all \(\mathcal{X}\) and \(\gamma \in \signature{\mathcal{X}}\), exactly \(\vec{d}_{\mathcal{X}, \gamma}\) many components in this equivalence class the configuration \(\gamma\).
This is possible by \Cref{eq:selector-lin-rep-exact}.
By \Cref{stmt:rho-lin-rep-sensible}, there is a function \(\rho\colon \terminalsContainedS\times\N \to 2^{2^{S\cap V(G)}}\) that satisfies \Cref{item:valid-selector-rho-connected} of \Cref{def:valid-selector} and for all \(U \subseteq S\cap V(G)\), it holds that \(r_U \leq \vec{s}_U\).
Notice that \Cref{eq:selector-lin-rep-dem-sup} ensures that \Cref{item:valid-selector-enough-supply} of \Cref{def:valid-selector} is satisfied as well.

To conclude that \(\Gamma\) is valid, we aim to apply \Cref{stmt:assign-lin-rep-sensible}.
Consider any \(T\in \mathcal{T}_S, i \in \natint{d(T)}\), and \(U \subseteq S\cap V(G)\).
Assume there is a \(C \in \comp{\augment{G}{\mathcal{T}} - S}\) and \(j \in \natint{\abs{S}}\) such that \(\assign[\Gamma(C)]Tij = U\) and denote with \(\mathcal{X}\) the equivalence class of \(C\).
By \Cref{eq:selector-lin-rep-a-ub}, we have \(1 \leq \vec{d}_{\mathcal{X}, \Gamma(C)} \leq N \vec{a}_{T,i,U}\); so, \(\vec{a}_{T,i,U} = 1\).
Now, assume that \(\vec{a}_{T,i,U} = 1\).
By \Cref{eq:selector-lin-rep-a-lb} and since the maximum of a set of numbers is bounded below by their mean, there is a \(\mathcal{X} \in \mathfrak{C}\) and \(\gamma \in \signature{\mathcal{X}}\) such that there is a \(j \in \natint{S}\) with \(\assign[\gamma]Tij = U\) and \(\vec{d}_{\mathcal{X}, \gamma} \geq \frac{\vec{a}_{T,i,U}}{\sum_{\mathcal{X}\in \mathfrak{C}} \abs{\{\gamma \in \signature{\mathcal{X}} \mid \exists j \in \natint{\abs{S}} \colon \assign[\gamma]Tij = U\}}} > 0\).
So, \(\vec{d}_{\mathcal{X},\gamma} \geq 1\) and there is a \(C \in \mathcal{X}\) with \(\Gamma(C) = \gamma\).
Thus, \(\assign[\Gamma(C)]Tij = U\).
Now we know from \Cref{stmt:assign-lin-rep-sensible} that the configuration selector \(\Gamma\) satisfies \Cref{item:valid-selector-assign-connected} of \Cref{def:valid-selector} and so \(\Gamma\) is valid.

Let \(\Gamma\) be a valid configuration selector.
We aim to provide a feasible solution to \(\selectorLinRep\).
First let \(\mathcal{X}\in \mathfrak{C}\) and \(\gamma \in \signature{\mathcal{X}}\).
We set \(\vec{d}_{\mathcal{X},\gamma} \coloneqq \abs{\{C \in \mathcal{X}\mid \Gamma(C) = \gamma\}}\) as the number of components in \(\mathcal{X}\) that get assigned the configuration \(\gamma\), satisfying \Cref{eq:selector-lin-rep-exact}.
Let \(\rho\) be a function witnessing validity of \(\Gamma\).
We set for all \(U \subseteq S \cap V(G)\), \(\vec{s}_U = r_U\).
Since \(\Gamma\) is valid, \Cref{eq:selector-lin-rep-dem-sup} is satisfied.
By \Cref{stmt:rho-lin-rep-sensible}, the sub-ILP \(\rhoLinRep{\vec{s}}\) is feasible.

Finally, consider any \(T \in \mathcal{T}_S\), \(i \in \natint{d(T)}\), and \(U \subseteq S \cap V(G)\).
If there is a \(C \in \comp{\augment{G}{\mathcal{T}} - S}\) and \(j \in \natint{\abs{S}}\) with \(\assign[\Gamma(C)]Tij = U\), we set \(\vec{a}_{T,i,U} \coloneqq 1\).
Otherwise, we set \(\vec{a}_{T,i,U} \coloneqq 0\).
By \Cref{stmt:assign-lin-rep-sensible}, the sub-ILP \(\assignLinRep{\vec{a}}\) is feasible.
If \(\vec{a}_{T,i,U} = 0\), \Cref{eq:selector-lin-rep-a-lb} is trivially satisfied and \Cref{eq:selector-lin-rep-a-ub} is satisfied by choice of \(\vec{a}\).
If \(\vec{a}_{T,i,U} = 1\), \Cref{eq:selector-lin-rep-a-lb} is satisfied by choice of \(\vec{a}\).
Since \(\sum_{\mathcal{X}\in \mathfrak{C}, \gamma \in \signature{\mathcal{X}}} d_{\mathcal{X}, \gamma} = N\vec{a}_{T,i,U}\), \Cref{eq:selector-lin-rep-a-ub} is satisfied as well.
\end{proof}

We know that deciding whether a feasible assignment for an ILP exists, is \(\fpt\) by the number of variables~\cite{Kannan87}.
Additionally, we know that deciding whether a fracture modulator of size \(k\) exists, and possibly finding it, is \(\fpt\) by \(k\).
So, we can construct and evaluate \(\selectorLinRep\) reasonably fast, yielding the following theorem.

\begin{theorem}
\label{stmt:augmented-frag-runtime}
Let \(\mathscr{P}=(G, \mathcal{T}, d)\) be an instance of \(\gstp\).
Denote the fracture number of \(\augment{G}{\mathcal{T}}\) with \(k\).
We can decide whether \(\mathscr{P}\) is a positive instance in running time \(2^{2^\O{k^4}}\abs{G} + \O{\abs{\mathscr{P}}}\).
\end{theorem}
\begin{proof}
First, we apply \Cref{rr:degree-negative-instance} in running time \(\O{\abs{\mathscr{P}}}\).
After applying this reduction rule, for all \(v \in V(G)\) the number of incident augmented edges is bounded by the number of non-augmented edges, yielding \(\abs{E(\augment{G}{\mathcal{T}})} \leq 2\abs{E(G)}\).
Now, we find a fracture modulator \(X\) of size \(k\) in \(\augment{G}{\mathcal{T}}\).
By \Cref{stmt:fracture-modulator-algo}, this can be done in time \(\O{(2k)^k\abs{\augment{G}{\mathcal{T}}}} = \O{(2k)^k\abs{G}}\).

By \Cref{stmt:nice-fracture-modulator-exists}, we can find an equivalent instance \(\mathscr{P}'=(G', \mathcal{T}, d)\) and a nice fracture modulator \(S\) of \(\augment{G'}{\mathcal{T}}\) with \(\abs{S} = \O{k}\) and \(\abs{V(\augment{G'}{\mathcal{T}})} = \O{\abs{V(\augment{G}{\mathcal{T}})}}\) in linear time.
Now, we again apply \Cref{rr:degree-negative-instance} to \(\mathscr{P}'\) in running time \(\O{\abs{G}}\).
If the reduction rule supplied a negative instance, we abort here and output that \(\mathscr{P}\) is a negative instance.
Otherwise, for each of the components in \(\augment{G'}{\mathcal{T}} - S\) we compute the signature in time \(2^\O{k^4\log k}\), according to \Cref{stmt:compute-signature}.
Overall, this takes at most \(2^\O{k^4\log k}\abs{V(G)}\) time.
Next, we compute the size of the non-empty equivalence classes of the components.
This can be achieved in time \(\O{\abs{G}\log\abs{\!\allViable\!}} = \O{k^4\abs{G}}\).

Now, it is straightforward to construct \(\selectorLinRep\) in time \(2^\O{k^4}\).
As all scalars in the ILP are bounded by \(\O{\abs{G}}\), we can check for feasibility in time \(2^{2^\O{k^4}}\log\abs{G}\)~\cite{Kannan87}.
We output that \(\mathscr{P}\) is a positive instance if and only if \(\selectorLinRep\) is feasible.
By \Cref{stmt:selector-lin-rep-sensible,stmt:valid-selector-if-positive}, this output is correct.
\end{proof}

\begin{corollary}
\label{stmt:augmented-frag-fpt}
\(\gstp\) is \(\fpt\) by the fracture number of the augmented graph.
\end{corollary}


\section{\texorpdfstring{$\gstp$}{GSTP} is \texorpdfstring{$\fpt$}{FPT} by the Augmented/Slim Tree-Cut Width}
\label{sec:gstp-by-tcw*-and-stcw}
This Chapter is devoted to showing, that \(\gstp\) is fixed-parameter tractable by both the tree-cut width of the augmented graph and the slim tree-cut width of the original graph.
We later use the former fact to show, that \(\stp\) itself is fixed-parameter tractable by the tree-cut width of the input graph.
This result does not follow immediately as augmentation—even of a single terminal set—might increase the tree-cut width arbitrarily.

Let \(\mathscr{P}=(G, \mathcal{T}, d)\) be an instance of \(\gstp\).
The central part of deciding whether \(\mathscr{P}\) is positive, is a dynamic program.
This dynamic program works on a tree-cut decomposition for \(G\) with some additional assumptions and heavily uses the fact that the number of bold children of any node in the tree-cut decomposition is bounded by a function of its width.

Ganian et~al.~\cite{GanianKS22} claimed that in a nice tree-cut decomposition the number of bold children of any node is bounded by \(w+1\).
However, we provide counter examples to this, showing that in a nice tree-cut decomposition, the number of bold children is actually not bounded by a function its width.
Then, we show how to transform a nice tree-cut decomposition of an arbitrary graph into a \emph{friendly} tree-cut decomposition in \(\fpt\)-linear time.
In fact, the running time is even at most polynomial in the size of the graph.

\begin{definition}
\label{def:friendly-tcw}
We call a tree-cut decomposition \((S,\mathcal{X})\) of width \(w\) \emph{friendly}, if it is nice and for all \(s \in S\), we have \(\abs{\boldChildren{s}} + \abs{X_s} \leq w+2\).
\end{definition}

Based on this, we give the definition of a \emph{simple} tree-cut decomposition.
Given a simple tree-cut decomposition for \(G\), we provide a dynamic program deciding whether \(\mathscr{P}\) is a positive instance.

\begin{definition}
\label{def:simple-tcw}
Let \((G, \mathcal{T}, d)\) be an instance of \(\gstp\).
Consider a tree-cut decomposition \(\mathcal{D} \coloneqq (S, \mathcal{X})\) of \(G\) and let \(s \in V(S)\).
Denote with \(\crossLink{s}\coloneqq \{T \in \mathcal{T}\mid T \cap Y_s \neq \emptyset \land T \setminus Y_s \neq \emptyset\}\) the set of terminal sets crossing the link between \(s\) and its parent.
We call \(s\) \emph{simple} if it is thin, \(\abs{Y_s} = 1\), \(\adhesion{s} = 2\), and \(\crossLink{s} = \emptyset\).
We call \(\mathcal{D}\) \emph{simple}, if it is friendly and all thin nodes are simple.
\end{definition}

This Chapter is structured as follows.
First, we provide reduction rules that are necessary for the whole Chapter and in particular for reducing the instance from augmented tree-cut width and slim tree-cut width to a simple tree-cut decomposition.
Then, we provide a dynamic program to solve \(\gstp\) given a simple tree-cut decomposition parameterized by its width.
Building on this result, we show how to construct a simple tree-cut decomposition of width \(w\) from a tree-cut decomposition with slim width \(w\).
Finally, we show how to use the fact that we can decide instances with a simple tree-cut decomposition in \(\fpt\)-time by the width of this decomposition to decide an instance given a tree-cut decomposition of the augmented graph.
As an important ingredient, we show how to obtain a friendly tree-cut decomposition from a nice tree-cut decomposition in polynomial time, without considerably increasing its width.

\subsection{General Techniques for \texorpdfstring{$\gstp$}{GSTP} Tree-Cut Decompositions \label{sec:tcw-red-rules}}
\label{sec:org14ea6f3}
In this Section, we present the techniques, which we need across this chapter.
We start with the general reduction rules.
Then, we show that we can obtain a simple tree-cut decomposition from a tree-cut decomposition without increasing the width by much, if the given decomposition is almost friendly and for each node the number of non-simple, thin children is small.

Let \(\mathcal{D} = (S, \mathcal{X})\) be a tree-cut decomposition of width \(w\) for \(G\).
Let \(Q \subseteq V(G)\) be given.
If there is no \(T \in \mathcal{T}\) with \(T \cap Q \neq \emptyset\) and \(T \setminus Q \neq \emptyset\) (i.e., \(T\) crosses \(Q\)), we denote with \(\mathcal{T}'\coloneqq \{T \in \mathcal{T}\mid T \subseteq Q\}\) the terminals contained in \(Q\) and with \(\mathscr{P}[Q]\) the instance \((G[Q], \mathcal{T}', d\vert_{\mathcal{T}'})\) restricted to \(Q\).
For every \(C \in \comp{G}\), we have that \(\mathscr{P}\) is positive if and only if \(\mathscr{P}[V(C)]\) and \(\mathscr{P}[V(G)\setminus V(C)]\) are both positive, providing a way to consider connected components of an instance separately.
This immediately yields our first reduction rule.

\begin{reductionrule}
\label{rr:connected-components}
If there is a \(T \in \mathcal{T}\), such that there is no \(C \in \comp{G}\), with \(T \subseteq V(C)\), output a trivial negative instance.
Otherwise, output that the instance is positive if and only if all \(\{\mathscr{P}[V(C)]\}_{C \in \comp{G}}\) are positive instances.
\end{reductionrule}

Denote with \(r\) the root of \(S\).
After applying \Cref{rr:connected-components}, all \(s \in V(S) \setminus \{r\}\) with \(\adhesion{s} = 0\) are empty leaves, which we can remove from \(\mathcal{D}\).
Thus, we assume from now on that \(\adhesion{s} \geq 1\).

We now provide a reduction rule to limit the number of connections that need to be made across links.
In any positive instance the accumulated demands of all \(T \in \crossLink{s}\) can not be very large as each such demand needs to be fulfilled by a tree crossing the link between \(S_s\) and its parent, of which there can not be that many.
Denote this value by \(\demandCrossLink{s}\coloneqq\sum_{T \in \crossLink{s}} d(T)\).
If \(\mathcal{D}\) would be a tree-cut decomposition for \(\augment{G}{\mathcal{T}}\), denote with \(\mathcal{U}\) the terminal sets inducing the augmented edges in \(\cutEdges[\augment{G}{\mathcal{T}}]{Y_s}\).
Initially, one might even think that \(\crossLink{s} = \mathcal{U}\) holds, but this is not necessarily the case.
We note that \(\crossLink{s} \subseteq \mathcal{U}\).

\begin{reductionrule}
\label{rr:cross-link-demand-large}
If there is a node \(s \in V(S)\) with \(\demandCrossLink{s} > \adhesion{s}\), we output a trivial negative instance.
\end{reductionrule}
\begin{proof}
Consider any positive instance.
We show that for all \(s \in V(S)\) we have \(\demandCrossLink{s} \leq \adhesion{s}\).
Let \((\mathcal{F}, \pi)\) be a solution and set \(\{F_1, F_2, \dots, F_{\demandCrossLink{s}}\}\coloneqq \inv\pi(\crossLink{s})\).
As all terminal sets in \(\crossLink{s}\) contain a terminal in \(Y_s\) and \(V(G) \setminus Y_s\), for each \(i \in \natint{\demandCrossLink{s}}\) there is a distinct \(e_i \in E(F_i) \cap \cutEdges[G]{Y_s}\).
As \(\abs{\cutEdges[G]{Y_s}} = \adhesion{s}\), we have that \(\demandCrossLink{s} \leq \adhesion{s}\).
\end{proof}

Now, we consider thin nodes \(s \in V(S) \setminus \{r\}\), with \(\adhesion{s} = 1\).

\begin{reductionrule}
\label{rr:adh-1}
Assume \Cref{rr:cross-link-demand-large} has been applied exhaustively.
Let \(s \in V(S)\) with \(\adhesion{s} = 1\) and consider \(\{uv\} \coloneqq \cutEdges{Y_s}\) with \(u \in Y_s\) and \(v \notin Y_s\).
Then, remove \(uv\) from \(G\) and if there is a \(T \in \crossLink{s}\), increase the demand of \((T\cap Y_s) \cup \{u\}\) and \((T \setminus Y_s) \cup \{v\}\) by 1 while removing \(T\) from \(\mathcal{T}\) (if necessary, add \((T\cap Y_s) \cup \{u\}\) and \((T \setminus Y_s) \cup \{v\}\) to \(\mathcal{T}\)).
\end{reductionrule}
\begin{proof}
As \Cref{rr:cross-link-demand-large} is applied, we have \(\demandCrossLink{s} \in \{0, 1\}\).
First, assume that \(\demandCrossLink{s} = 0\).
Then, for all \(T \in \mathcal{T}\), we either have \(T \subseteq Y_s\) or \(T \subseteq V(G) \setminus Y_s\).
As restricting connected subgraphs to either \(Y_s\) or \(V(G) \setminus Y_s\) keeps them connected, the bridge \(e\) in \(G\) is not needed to connect terminal sets.

Now, assume \(\demandCrossLink{s} = 1\) and let \(\{T\} \coloneqq \crossLink{s}\).
Let \((\mathcal{F}, \pi)\) be a solution to the original instance.
Consider \(\{F\} \coloneqq \inv\pi(T)\).
Then, \(uv \in E(F)\).
As \(uv\) is a bridge in \(G\), \(F - uv\) has two connected components.
Denote the component contained in \(Y_s\) with \(P\) and the other with \(Q\).
Then, \((T \cap Y_s) \cup \{u\} \subseteq V(P)\) and \((T \setminus Y_s) \cup \{v\} \subseteq V(Q)\).
So, in the reduced instance we assign all of \(\mathcal{F}\setminus \{F\}\) to the same terminal sets and we additionally assign \(P\) to \((T \cap Y_s) \cup \{u\}\) and \(Q\) to \((T \setminus Y_s) \cup \{v\}\), solving the reduced instance.

Let \((\mathcal{F}, \pi)\) now be a solution to the reduced instance and let \(P \in \inv\pi((T \cap Y_s) \cup \{u\})\) and \(Q \in \inv\pi((T\setminus Y_s)\cup \{v\})\) and set \(F \coloneqq (P \cup Q) + uv\).
Note that \(F\) is connected, edge-disjoint from all \(\mathcal{F} \setminus \{P, Q\}\) and \(V(F) \supseteq (T \cap Y_s) \cup (T \setminus Y_s)= T\).
To get a solution for the original instance, we assign all of \(\mathcal{F}\setminus \{P,Q\}\) to the same terminal sets and \(F\) to \(T\), solving the instance.
\end{proof}
After applying \Cref{rr:connected-components,rr:cross-link-demand-large,rr:adh-1} exhaustively, we remove all empty leaves from \(S\), which achieves that for all thin nodes \(s \in V(S)\setminus \{r\}\) we have \(\adhesion{s} = 2\).
This already hints at how we solve GSTP parameterized by slim tree-cut width.
After applying these reduction rules, we show that in a nice tree-cut decomposition the number of children is bounded by a function of its slim width.
As a friendly tree-cut decomposition without thin nodes is simple, we just treat all children like bold children, yielding an \(\fpt\) algorithm.

It is known that EDP—and therefore GSTP—is \(\wonehard\) parameterized by tree-cut width~\cite{GanianOR20}.
We show, that GSTP can be solved in \(\fpt\)-time parameterized by the tree-cut width of a simple tree-cut decomposition.
Additionally, we show that we can transform any tree-cut decomposition into a friendly tree-cut decomposition in polynomial time.
Assume there exist general reduction rules, that make all thin nodes simple and only increases the tree-cut width by a computable function.
Then, we could transform the tree-cut decomposition to be friendly and without thin links that are not simple in polynomial time.
Thus, \(\mathrm{GSTP}\) would be \(\fpt\) by the tree-cut width, which is not possible unless \(\fpt = \wone\).
Therefore, the existence of such a reduction rule is unlikely.

Finally, we show that if \(\mathcal{D}\) is close to being simple, we can actually make \(\mathcal{D}\) simple without increasing its width by much.
To formalize this idea, let \(s \in V(S)\) and denote with \(N_s \subseteq \thinChildren{s}\) the set of thin and non-simple children of \(s\).
If for all \(s\) both \(\abs{N_s}\) and \(\abs{\boldChildren{s}} + \abs{X_s} - w\) are small, we obtain a simple tree-cut decomposition of width that is not increased by much.

\begin{lemma}\label{stmt:almost-simple-tcd-made-simple}
Given a nice tree-cut decomposition \(\mathcal{D}\) of \(G\).
We can compute in linear time an equivalent instance \((G', \mathcal{T}, d)\) and a simple tree-cut decomposition \(\mathcal{C}\) of \(G'\).
Let \(\Delta_s \coloneqq \abs{N_s} + \abs{\boldChildren{s}} + \abs{X_s} - w - 1\), then the tree-cut width of \(\mathcal{C}\) is \[ w^* \coloneqq w+ 4 + \max(0, \max_{s\in V(S)}\Delta_s).\qedhere\]
\end{lemma}
\begin{proof}
For each \(s\in V(S)\) add vertices \(t_s\) and \(t'_s\) to \(G\) and \(X_s\), and for each \(r \in N_s\) the edges \(\{t_s, t'_s\} \times \{t_r, t'_r\}\).
Now add new isolated vertices to the root vertex of the tree-cut decomposition until its width is at least \(w^*\) and remove all empty leaves.
These operation can be applied in linear time.
We call the obtained decomposition \(\mathcal{C} \coloneqq (S', \mathcal{X'})\).
Observe that as \(\mathcal{D}\) is nice, so is \(\mathcal{C}\).
Since all added vertices and edges are disconnected from \(G\), this instance is solvable if and only if the original instance is solvable.
So, consider any \(s \in V(S)\).
We have \[ \abs{\boldChildren[\mathcal{C}]{s}} + \abs{X'_s} =\abs{N_s} + \abs{\boldChildren[\mathcal{D}]{s}} + \abs{X_s} + 2 \leq \Delta_s + w + 3 \leq w^* - 1.
\] Thus, \(\mathcal{C}\) is friendly and we see that all remaining thin nodes are simple; so, \(\mathcal{C}\) is simple.

Finally, we compute the width of \(\mathcal{C}\).
By construction, we have \(\adhesion[\mathcal{C}]s \leq \adhesion[\mathcal{D}]s + 4 \leq w + 4\leq w^*\).
Consider the torso at \(s\) with respect to \(\mathcal{C}\).
If \(s\) is not the root, the size of the 3-center of the torso at \(s\) in \(\mathcal{C}\) is bounded by \(1+\abs{\boldChildren[\mathcal{C}]{s}} + \abs{X'_s}\leq w^*\).
The same holds for the root, before we start adding isolated vertices.
So, the width of \(\mathcal{C}\) is exactly \(w^*\).
\end{proof}

\subsection{\texorpdfstring{$\gstp$}{GSTP} is \texorpdfstring{$\fpt$}{FPT} by the Width of a Simple Tree-Cut Decomposition \label{sec:tcw-simple}}
\label{sec:org71807c5}
In this Section, we assume that the given tree-cut decomposition \((S, \mathcal{X})\) is simple.
How to use this to solve the general problem parameterized by slim width or tree-cut width of the augmented graph, is left for the following Sections.
A simple tree-cut decomposition heavily restricts how vertices in thin bags can be connected with the rest of the graph.
This can be exploited to construct a dynamic program for this problem.

From now on, we assume that \Cref{rr:cross-link-demand-large,rr:sensible-terminal-sets} are applied.
For each \(s \in V(S)\) choose a function \(\enumerateCrossLinkOp{s} \colon \natint{\demandCrossLink{s}} \to \crossLink{s}\) such that for all \(T \in \crossLink{s}\) we have \(\abs{\invEnumerateCrossLink{s}{T}}=d(T)\).
This function gives us the ability to identify the different trees crossing the link between this node and its parent by a number in the set \(\natint{\demandCrossLink{s}} \subseteq \natint{w}\).

To define the dynamic program, we often refer to the \emph{boundary} \(\boundary{s}\) of a node, which is defined as the graph edge-induced by \(\cutEdges{Y_s}\) on \(G\), and the graph \(G_s \coloneqq G[Y_s] \cup \boundary{s}\) which is the graph edge-induced by all edges with at least one endpoint in \(Y_s\).
Note that there can be edges in \(G[V(\boundary{s})]\) that are not included in \(\boundary{s}\); in particular, \(\boundary{s}\) is bipartite with the bipartitions \(V(\boundary{s}) \cap Y_s\) and \(V(\boundary{s}) \setminus Y_s\).
We have \(\abs{E(\boundary{s})} \leq w\); so, \(\abs{V(\boundary{s})} \leq 2w\).
The boundary is a small separator between the vertices \(Y_s \setminus V(\boundary{s})\) and \(V(G) \setminus (Y_s \cup V(\boundary{s}))\).
We use this fact to define our dynamic program.

The data-table \(D(s)\) at \(s\), is a set of tuples \(\tau = (\pastPart\tau, \pastAssignOp{\tau}, \futurePart{\tau})\) with \(\pastAssignOp\tau \colon \pastPart\tau \to \natint{w}\) where \(\pastPart{\tau}\) and \(\futurePart{\tau}\) are partitions of a (not necessarily proper) subset of \(E(\boundary{s})\).
If \(\bigcup\pastPart\tau\) and \(\bigcup\futurePart\tau\) are disjoint, we call \(\tau\) a syntactically valid tuple at \(s\).
Note that we do not assign the partitions in \(\futurePart\tau\) to indices, since at this point we do not care, whether they are eventually used for the same terminal set.

Intuitively, a tuple \(\tau\) gives almost complete information about the part of the final solution that crosses \(\boundary{s}\).
Consider a solution \((\mathcal{F}, \pi)\) to the whole instance and denote with \(\mathcal{F}'\subseteq \mathcal{F}\) the trees containing an edge of \(E(\boundary{s})\).
Set \(\mathcal{F}^\downarrow \coloneqq \{F \in \mathcal{F}'\mid \pi(F) \cap Y_s \neq \emptyset\}\) and \(\mathcal{F}^\uparrow \coloneqq \mathcal{F}'\setminus \mathcal{F}^\downarrow\) as the subgraphs crossing the link between \(s\) and its parent which are assigned to terminal set starting below this link and starting above this link respectively.
This solution corresponds to a tuple \(\tau\in D(s)\) with \(\pastPart\tau \coloneqq \{\{E(K) \cap E(\boundary{s}) \mid K \in \comp{F[E(G_s)]}\}\}_{F \in \mathcal{F}^\downarrow}\) and \(\futurePart\tau \coloneqq \{\{E(K) \cap E(\boundary{s}) \mid K \in \comp{F[E(G_s))]}\}\}_{F \in \mathcal{F}^\uparrow}\) where we consider the edges in \(E(\boundary{s})\) per connected component of the trees in \(\mathcal{F}'\) restricted to the edges with at least one endpoint in \(Y_s\).
For each \(F \in \mathcal{F}^\downarrow\) we choose a distinct \(\lambda_F \in \natint{w}\) such that if \(\pi(F) \in \crossLink{s}\), we have \(\lambda_F \in \natint{\demandCrossLink{s}}\) and \(\enumerateCrossLink{s}{\lambda_F} = \pi(F)\) and if \(\pi(F) \notin \crossLink{s}\), we have \(\lambda_F \in \natint{w}\setminus\natint{\demandCrossLink{s}}\).
Now, we set \(\pastAssignOp{\tau}\) to \(\lambda_F\) for each set induced by \(F\).

However, when computing \(D(s)\) we do not know enough about the whole instance to ensure that the remaining instance can actually be solved by a solution that is valid up to this node.
So, we cannot require that exactly the tuples induced by complete solutions are the members of \(D(s)\).
Rather, we only require that the tuples in \(D(s)\) correspond to solutions that are valid up to the considered node that could be extended to complete solutions assuming that partitions of \(\pastPart{\tau}\) assigned to the same tree are connected in the whole solution and that the connections provided by \(\futurePart\tau\) are enough to fulfill the requirements not intersecting \(Y_s\).

To formally define what this means, consider a syntactically valid tuple \(\tau\) and let \(\mathcal{U}_s \coloneqq \{T \in \mathcal{T}\mid T \subseteq Y_s\}\).
Note that \(\mathcal{U}_s\) and \(\crossLink{s}\) are disjoint and their union is the set of all terminals \(T \in \mathcal{T}\) that are not disjoint from \(Y_s\).
We add \(w\) vertices \(\{q_{s,i}\}_{i\in\natint{w}}\) to \(G_s\) such that each has the neighborhood \(V(\boundary{s}) \setminus Y_s\).
Call this graph \(G_{s}^*\).
These additional vertices can be used to simulate that a subgraph gets connected outside \(G_s\).
Let \(i\in\natint{\demandCrossLink{s}}\) and denote with \(Q_{s,i} \coloneqq (\enumerateCrossLink{s}{i} \cap Y_s) \cup \{q_{s,i}\}\) the vertices in \(Y_s\) of the terminal set assigned to the \(i\)-th subgraph crossing \(\boundary{s}\), which we called \(\enumerateCrossLink{s}{i}\), combined with \(q_{s,i}\).
Additionally, define for each \(P \in \futurePart\tau\) the set \(R_{s,P} \coloneqq V(G[P])\setminus Y_s\) to be all vertices of edges contained in \(P\) that are not in \(Y_s\).
Finally, define the instance \(\mathscr{D}_{s,\tau}\coloneqq(G_{s}^*, \mathcal{U}_s\cup \{Q_{s,i}\}_{i \in\natint{\demandCrossLink{s}}} \cup \{R_{s,P}\}_{P \in \futurePart\tau}, d')\), where for all \(T \in \mathcal{U}_s\) we have \(d'(T) = d(T)\), for all \(i \in \natint{\demandCrossLink{s}}\), we have \(d'(Q_{s,i}) = 1\), and for all \(P \in \futurePart\tau\), we have \(d'(R_{s,P}) = \abs{\{P' \in \futurePart\tau\mid R_{s, P'} = R_{s,P}\}}\).

\begin{definition}
\label{def:simple-tcw-dp}
For each \(s \in V(S)\) the data-table \(D(s)\) is exactly the set of syntactically valid tuples \(\tau\) at \(s\) where the instance \(\mathscr{D}_{s,\tau}\) has a solution \((\mathcal{F}, \pi)\) such that
\begin{enumerate}
\item \label{item:simple-tcd-edge-partitions} we have \(E(\boundary{s}) \cap \bigcup E(\mathcal{F})\subseteq\bigcup\pastPart\tau \cup \bigcup\futurePart\tau\),
\item \label{item:simple-tcd-future-sensible} for all \(P \in \futurePart\tau\) and \(F \in \inv\pi(R_{s,P})\), the set \(V(F)\) is disjoint from \(\{q_{s,i}\}_{i \in \natint{w}}\),
\item \label{item:simple-tcd-future-exact} for all \(P \in \futurePart\tau\), there is a \(F \in \inv\pi(R_{s,P})\) with \(E(F) \cap E(\boundary{S}) = P\),
\item \label{item:simple-tcd-past} for all \(i \in \natint{w}\), let \(\mathcal{P}_i \coloneqq \invPastAssign{\tau}{i}\), then
\begin{itemize}
\item if \(\mathcal{P}_i = \emptyset\), we have \(q_{s,i} \notin \bigcup V(\mathcal{F})\),
\item otherwise, there is exactly one \(F \in \mathcal{F}\) with \(q_{s,i} \in V(F)\) and this \(F\) additionally satisfies that \(E(F - q_{s,i})\) can be partitioned into \(\{E_P\}_{P \in \mathcal{P}_i}\) such that for all \(P \in \mathcal{P}_i\), the graph \(F[E_P]\) is connected and \(E_P \cap E(\boundary{s}) = P\).\qedhere
\end{itemize}
\end{enumerate}
\end{definition}

With \Cref{item:simple-tcd-edge-partitions} we ensure that the edges used in the solution are accounted for in \(\pastPart\tau\) and \(\futurePart\tau\).
We want to be able to assume that for every \(P \in \futurePart\tau\) there is a subgraph completely contained in \(G[Y_s] \cup \boundary{s}\) connecting all edges of \(P\).
So, in \Cref{item:simple-tcd-future-sensible} we ensure that the vertices \(\{q_{s,i}\}_{i \in \natint{w}}\)—that are used to simulate that a subgraph gets connected outside of this graph—are not included in the subgraphs connecting the edges in \(P\).
These subgraphs should also use exactly the edge-set \(P\) in \(E(\boundary{s})\), which we ensure with \Cref{item:simple-tcd-future-exact}.
Finally, consider \Cref{item:simple-tcd-past} and \(i \in \natint{w}\).
If \(\mathcal{P}_i = \emptyset\), that is, no edges are assigned to the \(i\)-th subgraph, the vertex \(q_{s,i}\), which is used to mark the \(i\)-th subgraph crossing \(\boundary{s}\), is not used in any subgraph.
Otherwise, we again ensure that for each \(P \in \mathcal{P}_i\) there is a subgraph in this solution that connects the edges of \(P\) inside \(G_s\).
Notice that we can combine \Cref{item:simple-tcd-edge-partitions,item:simple-tcd-future-exact,item:simple-tcd-past} to show, that for any \(F \in \mathcal{F}\) with \(E(F) \cap E(\boundary{s}) \neq \emptyset\) there is a \(P \in \pastPart{\tau} \cup \futurePart{\tau}\) with \(P\subseteq E(F)\).

Immediately, we observe that this dynamic program can indeed be used to determine whether \(\mathscr{P}\) is a positive instance, if we assume that it is calculated correctly.

\begin{lemma}
\label{stmt:simple-tcw-dp-root-correct}
Let \(r\) be the root of \(S\).
Then, \(\mathscr{P}\) is a positive instance if and only if \(D(r) \neq \emptyset\).
\end{lemma}
\begin{proof}
We note that \(\cutEdges{r} = \emptyset\).
Therefore, \(\boundary{r}\) is the empty graph, which means that there is exactly one syntactically valid tuple \(\tau\) at \(r\).
\(\mathscr{D}_{r,\tau}\) is the same instance as \(\mathscr{P}\) except that there are \(w\) nodes added with neighborhood \(V(\boundary{r}) \setminus Y_r = \emptyset\).
So, every solution to \(\mathscr{P}\) can be transformed into a valid solution of \(\mathscr{D}_{r,\tau}\) and vice versa.
Additionally, every solution to \(\mathscr{D}_{r,\tau}\) also satisfies the additional requirements of \Cref{def:simple-tcw-dp}, proving the statement.
\end{proof}

To compute this dynamic program, we consider that the size of the data-table is bounded by a function of the parameter.
Denote with \(\mathcal{V}_s\) the set of all syntactically valid tuples at \(s\).
As \(\abs{\crossLink{s}} \leq w\) and as \(\abs{E(\boundary{s})} \leq w\), we have \(\abs{D(s)} \leq \abs{\mathcal{V}_s} \leq 2^\O{w\log w}\).
Now assume that for all bold children \(b\in\boldChildren{s}\) of \(s\), we have already computed the data-table \(D(b)\).
We now show how to compute \(D(s)\) in \(\fpt\)-time.
To achieve this, it is enough to decide in \(\fpt\)-time for a given \(\tau \in \mathcal{V}_s\) whether \(\tau \in D(s)\).
For this we iterate over all simultaneous choices of \(\tau_b \in D(b)\) for \(b \in \boldChildren{s}\) and check whether the solutions witnessing \(\tau_b \in D(b)\) can be extended to solutions witnessing \(\tau \in D(s)\).
Then, we show how to create instances of \(\gstp\) with bounded fracture number of its augmented graph such that one of these instances is positive if and only if the aforementioned condition is met.
Combined with the results from \Cref{sec:gstp-fnS} this yields the result of this Section.

Assume for all \(b \in \boldChildren{s}\) a \(\tau_b \in D(b)\) is fixed.
To combine the subgraphs of the sub-solutions witnessing \(\tau_b \in D(b)\) to a solution witnessing \(\tau \in D(s)\), we need to be able to translate the local numbering of the solution subgraphs into a numbering shared across all solution subgraphs not fully contained in one connected component of \(S - s\).
When dealing with a shared mapping, we want to avoid that we map solution subgraphs assigned to different terminal sets to the same index.
For convenience denote with \(A \coloneqq \{s\} \cup \boldChildren{s}\) the set of \(s\) and all its bold children, with \(\terminalsContainedNode{s} \coloneqq \{T \in \mathcal{T}\mid T \subseteq X_s\}\) the set of terminal sets completely contained in \(X_s\), and with \(\mathcal{X}\coloneqq \terminalsContainedNode{s}\cup \bigcup_{a \in A}\crossLink{a}\) the set of all terminal sets, which are not completely contained in one sub-tree of \(S-s\).
As \(s\) is simple it has at most \(w + 2\) bold children.
Since we applied \Cref{rr:cross-link-demand-large}, \(w(w + 3)\) is an upper bound on the cumulated demand of all terminal sets crossing the links between \(s\) and its parent or any of its bold children.
It is also an upper bound on the number of edges crossing bold links adjacent to \(s\).
So, we need at most \(w(w+3)\) shared indices for the crossing subgraphs.
We set \(u \coloneqq w(w + 3) +w\) slightly larger than this upper bound to have spare indices.

\begin{definition}
\label{def:local-shared-map}
For all \(a \in A\) let \(\localSharedMapOp{a} \colon \natint{w}\to \natint{u}\) be an injective function and if \(s \neq a\) let \(\unknownTerminalIn{a}\) be a unique symbol.
Additionally, let \(\blockedIndex\) be a another unique symbol.
We call \((\localSharedMapOp{a})_{a \in A}\) a \emph{local to shared mapping} if there is a function \(\enumerateSharedOp \colon \natint{u} \to \mathcal{X} \cup \{\unknownTerminalIn{b}\}_{b \in \boldChildren{s}}\cup\{\blockedIndex\}\), called a \emph{shared mapping enumerator}, such that
\begin{enumerate}
\item \label{item:local-shared-def-preserves-crossLink} for all \(a \in A\) and \(T \in \crossLink{a}\), we have \(\invEnumerateShared{T}= \localSharedMap{a}{\invEnumerateCrossLink{a}{T}}\),
\item \label{item:local-shared-def-number-containedInNode} for all \(T \in \terminalsContainedNode{s}\), we have \(\abs{\invEnumerateShared{T}}\leq d(T)\),
\item \label{item:local-shared-def-unnamed-sensible} for all \(b \in \boldChildren{s}\), we have \(\inv{\enumerateSharedOp}(\unknownTerminalIn{b}) \subseteq \localSharedMap{b}{\natint{w}\setminus \natint{\demandCrossLink{b}}}\),
\item \label{item:local-shared-def-unknown-is-unnamend-or-blocked} for all \(b \in \boldChildren{s}\) and \(i \in \natint{w}\setminus \natint{\demandCrossLink{b}}\), we have \[\enumerateShared{\localSharedMap{b}{i}}= \begin{cases}\unknownTerminalIn{b},&\text{if }\invPastAssign{\tau_b}{i}\neq\emptyset\\\blockedIndex,&\text{otherwise.}\end{cases}\]
\end{enumerate}
\end{definition}

In the definition above, we map all link-local indices to shared indices across all considered bold links using the functions \((\localSharedMapOp{a})_{a \in A}\).
Then, we identify each shared index \(i\in\natint{u}\) with a terminal set \(\enumerateShared{i}\).
If the concrete terminal set is known at this point (i.e.,~it is contained in \(\mathcal{X}\)), we assign it to exactly this terminal set.
We ensure this with \Cref{item:local-shared-def-preserves-crossLink,item:local-shared-def-number-containedInNode}.
In particular, \Cref{item:local-shared-def-preserves-crossLink} ensures that for all \(a \in A\) and \(T \in \crossLink{s}\) there is a bijection between the local and shared indices for subgraphs designated for \(T\).
For terminal sets \(T \in \mathcal{T}\) completely contained in \(X_s\) (i.e.,~\(T \in \terminalsContainedNode{s}\)), we ensure with \Cref{item:local-shared-def-number-containedInNode} that the number of designated subgraphs does not exceed the number of required subgraphs.

However, if the terminal set is completely contained in a \(\{Y_b\}_{b \in \boldChildren{s}}\), we do not know the exact terminal set to which this subgraph gets assigned in the solution.
This limitation stems from the fact that if we would store the concrete terminal sets, the size of the dynamic program might no longer be bounded by a function of the parameter.
Luckily, we can treat them interchangeably at this point.
So, we introduce the symbols \(\{\unknownTerminalIn{b}\}_{b\in\boldChildren{s}}\), where for any \(b \in \boldChildren{s}\) the symbol \(\unknownTerminalIn{b}\) is used to denote an arbitrary terminal set completely contained in \(Y_b\).
We introduce the symbol \(\blockedIndex\) to symbolize that an index does not correspond to a subgraph of the solution.
This is done, in part, for \(\enumerateSharedOp\) to be a total function.
Let \(b \in \boldChildren{s}\).
Using \Cref{item:local-shared-def-unnamed-sensible}, we ensure that each subgraph that is designated to be assigned to a terminal set contained in \(Y_b\) crosses \(\boundary{b}\).
The purpose of \Cref{item:local-shared-def-unknown-is-unnamend-or-blocked} is two-fold.
First and in combination with \Cref{item:local-shared-def-unnamed-sensible}, we require that exactly the indices \(j \in \natint{u}\) for which there is a \(i \in \natint{w}\setminus \natint{\demandCrossLink{s}}\) with \(\localSharedMap{b}{i} = j\) map to \(\unknownTerminalIn{b}\).
That is \(\invEnumerateShared{\unknownTerminalIn{b}} = \{\localSharedMap{b}{i}\mid i\in\natint{w}\setminus\natint{\demandCrossLink{s}}; \invPastAssign{\tau_b}{i} \neq \emptyset\}\).
Additionally, \Cref{item:local-shared-def-unknown-is-unnamend-or-blocked} ensures that the remaining indices of \(\invLocalSharedMap{b}{\natint{w}\setminus \natint{\demandCrossLink{s}}}\) are not used.

Consider the graph \(J_s\coloneqq G^*_s\left[X_s \cup V(\boundary{s}) \cup \{q_{s,i}\}_{i \in \natint{w}}\right] \cup \bigcup_{c \in \children{s}} \boundary{c}\), which is graph obtained from \(G^*_s\) after removing all edges completely contained in a \(\{Y_b\}_{b\in\boldChildren{s}}\) and all vertices isolated by this operation.
This graph is the part of \(G_s\) that is most relevant for propagating the dynamic program from the bold children to their parent \(s\).
For all \(b \in \boldChildren{s}\) and \(P\in \futurePart{\tau_b} \cup \pastPart{\tau_b}\) we add a vertex \(\tilde{p}_{b, P}\) connected to all \(V(G[P])\).
The vertices \(\tilde{p}_{b,P}\) can be used to simulate that the edges in \(P\) are connected using a subgraph contained in \(Y_b\).
For each \(i \in \natint{u} \setminus \invEnumerateShared{\blockedIndex}\) we add a vertex \(\tilde{m}_i\) to the graph connected to all \(V(J_s) \setminus \bigcup_{t \in \thinChildren{s}} X_t\).
These vertices are used to mark the subgraphs for the particular indices.
Call the obtained graph \(\dpTransitionGraph\).

First, we consider the subgraphs that are assigned a shared index \(i \in \natint{u}\setminus \invEnumerateShared\blockedIndex\).
If \(\enumerateShared{i} \in \mathcal{X}\), we set \(\widetilde{C_i} \coloneqq \enumerateShared{i}\cap X_s\) and otherwise \(\widetilde{C_i} \coloneqq \emptyset\), as the vertices of \(X_s\) which need to be contained in this subgraph by the assigned terminal set.
Now, if \(i \in \localSharedMap{s}{\natint{\demandCrossLink{s}}}\), we set \(\widetilde{M_i} \coloneqq \widetilde{C_i} \cup \{q_{s,\invLocalSharedMap{i}{s}}, \tilde{m}_i\}\) and otherwise \(\widetilde{M_i} \coloneqq \widetilde{C_i} \cup \{\tilde{m}_i\}\).
This ensures that all vertices of the final terminal set \(\enumerateShared{i}\) contained in \(Y_s\) are included, that the subgraph is marked using \(\tilde{m}_i\), and that we can simulate, if necessary, that the subgraph is connected outside of \(G_s\).
Next, we consider the subgraphs that supply connections to the remaining graph.
These are the subgraphs crossing a bold link adjacent to \(s\) that are not assigned an index.
More concretely, we consider the subgraphs witnessing that \(P \in \futurePart{\tau}\) can be connected in \(G_s\).
Recall, that \(R_{s,P} \coloneqq V(G[P]) \setminus Y_s\).
Finally, we consider the subgraphs that do not cross into a \(G[Y_b]\) for any \(b \in \boldChildren{s}\).
These are the remaining subgraphs that get assigned to a \(T \in \terminalsContainedNode{s}\).
For these we do not need to alter the set of terminals and set \(\widetilde{\terminalsContainedNode{s}} \coloneqq \{T \in \terminalsContainedNode{s}\mid \abs{\invEnumerateShared{T}} < d(T)\}\).
The terminal sets for the instance we define as \(\widetilde{\mathcal{T}_s}\coloneqq \{\widetilde{M_i}\}_{i \in \natint{u}\setminus\invEnumerateShared\blockedIndex}\cup \{R_{s,P}\}_{P \in \futurePart{\tau}} \cup \widetilde{\terminalsContainedNode{s}}\).

Note that the sets \(\{\widetilde{M_i}\}_{i \in \natint{u}\setminus\invEnumerateShared\blockedIndex}\), \(\{R_{s,P}\}_{P \in \futurePart{\tau}}\), and \(\widetilde{\terminalsContainedNode{s}}\) are pairwise disjoint.
For all \(T \in \{\widetilde{M_i}\}_{i \in \natint{u}\setminus\invEnumerateShared\blockedIndex}\), we set the demand to 1.
For all \(T \in \{R_{s,P}\}_{P \in \futurePart{\tau}}\), we set the demand equal to the number of \(P\) giving rise to \(T\).
For all \(T \in \widetilde{\terminalsContainedNode{s}}\), we set the demand to \(d(T) - \abs{\invEnumerateShared{T}}\).
Call the obtained instance \(\mathscr{C}(s,\tau, (\tau_b)_{b \in \boldChildren{s}},\enumerateSharedOp, (\localSharedMapOp{a})_{a\in A})\).

\begin{definition}
\label{def:simple-tcd-witness}
Choose for all \(b \in \boldChildren{s}\) a \(\tau_b\in D(b)\).
We say that \((\tau_b)_{b\in \boldChildren{s}}\) witnesses \(\tau \in D(s)\), if there is a local to shared mapping \((\localSharedMapOp{a})_{a\in A}\) with a shared terminal enumerator \(\enumerateSharedOp\) such that there is a solution \((\widetilde{\mathcal{F}}, \tilde\pi)\) to the instance \(\mathscr{C}(s,\tau, (\tau_b)_{b \in \boldChildren{s}},\enumerateSharedOp, (\localSharedMapOp{a})_{a\in A})\) such that
\begin{enumerate}
\item\label{item:simple-tcd-witness-edge-partitions} we have \(E(\boundary{s}) \cap \bigcup E(\widetilde{\mathcal{F}})\subseteq\bigcup\pastPart\tau \cup \bigcup\futurePart\tau\),
\item\label{item:simple-tcd-witness-future-sensible} for all \(P \in \futurePart\tau\) and \(\widetilde{F} \in \inv{\tilde\pi}(R_{s,P})\), the set \(V(\widetilde{F})\) is disjoint from \(\{q_{s,i}\}_{i \in \natint{w}}\),
\item\label{item:simple-tcd-witness-future-exact} for all \(P \in \futurePart\tau\), there is a \(\widetilde{F} \in \inv{\tilde\pi}(R_{s,P})\) with \(E(\widetilde{F}) \cap E(\boundary{S}) = P\).
\item\label{item:simple-tcd-witness-past} for all \(i \in \natint{w}\), let \(\mathcal{P}_i \coloneqq \invPastAssign{\tau}{i}\), then
\begin{itemize}
\item if \(\mathcal{P}_i = \emptyset\), we have \(q_{s,i} \notin \bigcup V(\mathcal{\widetilde{F}})\),
\item otherwise, there is exactly one \(\widetilde{F} \in \widetilde{\mathcal{F}}\) with \(q_{s,i} \in V(\widetilde{F})\) and this \(\widetilde{F}\) additionally satisfies that \(E(\widetilde{F} - q_{s,i})\) can be partitioned into \(\{\widetilde{E}_{P}\}_{P\in\mathcal{P}_i}\) such that for all \(P \in \mathcal{P}_i\), the graph \(\widetilde{F}[\widetilde{E}_P]\) is connected and \(\widetilde{E}_P \cap E(\boundary{s}) = P\), and for all \(b \in \boldChildren{s}\) and \(P' \in \pastPart{\tau_b} \cup \futurePart{\tau_b}\), there is at most one \(P \in \mathcal{P}\) with \(\tilde{p}_{b, P'} \in V(\widetilde{F}[\widetilde{E}_P])\), additionally this \(P\) satisfies \(P' \subseteq E_P\).
\end{itemize}
\item\label{item:simple-tcd-witness-marked} for all \(i \in \natint{u}\setminus\invEnumerateShared\blockedIndex\), let \(\{\widetilde{F}\} \coloneqq \inv{\tilde\pi}(\widetilde{M_i})\), then
\begin{itemize}
\item there is at most one edge adjacent to \(m_i\) in \(\bigcup E(\mathcal{\widetilde{F}})\),
\item for all \(a \in A\) with \(\invLocalSharedMap{a}{i} \neq \emptyset\) and \(P \in \invPastAssign{\tau_b}{\invLocalSharedMap{b}{i}}\), we have \(P \subseteq E(\widetilde{F})\) and if \(a \in \boldChildren{s}\) even \(\tilde{p}_{a, P} \in V(\widetilde{F})\).
\end{itemize}
\item \label{item:simple-tcd-witness-sub-connected} for all \(b \in \boldChildren{s}\) and \(P \in \futurePart{\tau_b} \cup \pastPart{\tau_b}\), there is at most one \(\widetilde{F} \in \widetilde{\mathcal{F}}\) with \(\tilde{p}_{b, P} \in V(\widetilde{F})\) and this particular \(\widetilde{F}\) satisfies \(P \subseteq E(\widetilde{F})\).
\end{enumerate}
\end{definition}

To understand the definition above, consider a solution \((\mathcal{F}, \pi)\) to the instance \(\mathscr{D}_{s, \tau}\).
This solution induces for all \(b \in \boldChildren{s}\) a solution \((\mathcal{F}_b, \pi_b)\) to sub-instances \(\mathscr{D}_{b, \tau_b}\), where \(\tau_b\) is an element of \(D(b)\).
Additionally, \((\mathcal{F}, \pi)\) induces a solution \((\widetilde{F}, \tilde{\pi})\) to \(\mathscr{C}(s,\tau, (\tau_b)_{b \in \boldChildren{s}},\enumerateSharedOp, (\localSharedMapOp{a})_{a\in A})\).
We see that \Cref{item:simple-tcd-witness-edge-partitions,item:simple-tcd-witness-future-sensible,item:simple-tcd-witness-future-exact} of \Cref{def:simple-tcd-witness} are direct adaptations of \Cref{item:simple-tcd-edge-partitions,item:simple-tcd-future-sensible,item:simple-tcd-future-exact} of \Cref{def:simple-tcw-dp} transferred from \((\mathcal{F}, \pi)\) to \((\widetilde{F}, \tilde{\pi})\).
\Cref{item:simple-tcd-witness-sub-connected} of \Cref{def:simple-tcd-witness} ensures that the connections simulated by the vertices \(\tilde{p}_{b, P}\) are only used once.
This is also reflected in the change of \Cref{item:simple-tcd-witness-past} in \Cref{def:simple-tcd-witness}.
We adapt it to make sure that the connections simulated by the vertices \(\tilde{p}_{b,P}\) are only used for one subgraph giving the connections for a single \(P' \in \pastPart{\tau}\).
With \Cref{item:simple-tcd-witness-marked} of \Cref{def:simple-tcd-witness}, we ensure that the vertices marking the subgraphs are not used to introduce additional connections, and that the edges and connections inside a \(\{G[Y_b]\}_{b\in\boldChildren{s}}\) which are assigned to this index are all used in the marked subgraph.

We now show in \Cref{stmt:simple-tcd-dp-implies-witness,stmt:simple-tcd-dp-implies-witness}, that \Cref{def:simple-tcd-witness} fully captures whether \(\tau \in D(s)\).

\begin{lemma}
\label{stmt:simple-tcd-dp-implies-witness}
Let \(\tau \in D(s)\).
Then, for all \(b \in \boldChildren{s}\) there is a \(\tau_b \in D(b)\) such that \((\tau_b)_{b \in \boldChildren{s}}\) witnesses \(\tau \in D(s)\).
\end{lemma}
\begin{proof}
As \(\tau \in D(s)\), there is a solution \((\mathcal{F}, \pi)\) to \(\mathscr{D}_{s, \tau}\) satisfying the additional requirements of \Cref{def:simple-tcw-dp}.
First, we define the shared terminal enumerator \(\enumerateSharedOp\) and a shared labelling of the solution subgraphs containing an edge of \(\bigcup_{b \in \boldChildren{s}}E(\boundary{b})\).
Then, we define the \((\tau_b)_{b\in\boldChildren{s}}\) using local labellings of the solution subgraphs using an edge of \(E(\boundary{b})\).
Based on the last two results, we then define the local to shared mappings \((\localSharedMapOp{a})_{a \in A}\) and prove that the shared terminal enumerator defined previously actually fulfills the requirements of \Cref{def:local-shared-map}.
Finally, we show that the \((\tau_b)_{b\in\boldChildren{s}}\) satisfy the requirements of \Cref{def:simple-tcd-witness}.

Let \(\mathcal{F}' \coloneqq \{F \in \mathcal{F} \mid E(F) \cap E(J_s) \neq \emptyset\}\) be the solution subgraphs containing an edge of \(J_s\).
These are the subgraphs relevant for propagating the dynamic program.
Denote with \(\mathcal{F}^\star \coloneqq \{F \in \mathcal{F}' \mid V(F) \subseteq X_s \cup \bigcup_{t\in \thinChildren{s}} X_t\}\) the solution subgraphs that are completely contained in \(X_s\) and the vertices of the thin children of \(s\).
These subgraphs are very local and not of high importance for the propagation of the dynamic program.
We can disregard them for most of the following arguments.
Let \(\mathcal{F}^{\boundaryOp} \coloneqq \mathcal{F}' \setminus \mathcal{F}^\star\) be the remaining subgraphs, which cross the boundary of \(s\) or a bold child.
Partition \(\mathcal{F}^{\boundaryOp}\) by whether or not the subgraph is designated for a terminal set starting above or below \(s\) in the final solution, that is \(F^{\uparrow} \coloneqq \bigcup_{P\in\futurePart{\tau}}\inv\pi(R_{s,P})\subseteq \mathcal{F}^{\boundaryOp}\) and \(\mathcal{F}^\downarrow \coloneqq \mathcal{F}^{\boundaryOp} \setminus \mathcal{F}^\uparrow\).
The tree-cut decomposition is simple and therefore friendly; so, \(s\) has at most \(w + 2\) bold children and \(\ell \coloneqq \abs{\mathcal{F}^\downarrow} \leq \abs{\mathcal{F}^{\boundaryOp}} \leq w (w + 3) = u - w\).
So, let \(\sigma \colon \mathcal{F}^\downarrow \to \natint{\ell}\) be a bijection, which gives a shared numbering on the subgraphs crossing bold links.

Now, we define the shared mapping enumerator.
For all \(i \in \natint{u} \setminus \natint{\ell}\), we set \(\enumerateShared{i} \coloneqq \blockedIndex\).
Let \(i \in \natint{\ell}\) and \(F \coloneqq \inv\sigma(i)\).
If there is a \(j \in \natint{\demandCrossLink{s}}\) with \(\pi(F) = Q_{s,j}\), then set \(\enumerateShared{i} \coloneqq \enumerateCrossLink{s}{j}\).
If there is a \(b \in \boldChildren{s}\) with \(\pi(F) \subseteq Y_b\), then set \(\enumerateShared{i} \coloneqq \unknownTerminalIn{b}\).
Otherwise, set \(\enumerateShared{i} \coloneqq \pi(F)\).
At this point it suffices to show that
\begin{itemize}
\item \(\enumerateSharedOp\) is actually well defined, that is that all values are contained in \(\codomain{\enumerateSharedOp} = \mathcal{X} \cup \{\unknownTerminalIn{b}\}_{b\in\boldChildren{s}}\),
\item for all \(T \in \mathcal{X} \setminus \terminalsContainedNode{s}\), we have \(\abs{\invEnumerateShared{T}} = d(T)\),
\item for all \(F \in \mathcal{F}^\downarrow\), with \(\enumerateShared{\sigma(F)} \in \mathcal{X}\), we have \(\pi(F) \cap Y_s = \enumerateShared{\sigma(F)} \cap Y_s\).
\end{itemize}

For all \(i \in \natint{\demandCrossLink{s}}\), we have \(\enumerateCrossLink{s}{i} \in \crossLink{s} \subseteq \mathcal{X}\).
Now, consider an \(i \in \natint{\ell}\) with \(F \coloneqq \inv\sigma(i)\) such that there is no \(j \in \natint{\demandCrossLink{s}}\) with \(Q_{s,j} = \pi(F)\) and no \(b \in \boldChildren{s}\) with \(\pi(F) \subseteq Y_b\), that is the last case of the definition of \(\enumerateSharedOp\) applies.
As \(F \in \mathcal{F}^\downarrow\), there is no \(P \in \futurePart{\tau}\) with \(F \in \inv\pi(R_{s,P})\).
So, \(\pi(F) \in \mathcal{U}_s\).
If \(\pi(F) \subseteq X_s\), then \(\pi(F) \in \terminalsContainedNode{s}\subseteq \mathcal{X}\).
Otherwise, there is a \(b \in \boldChildren{s}\) with \(\pi(F) \cap Y_b \neq \emptyset\).
As \(\pi(F) \not\subseteq Y_b\), we have \(\pi(F) \in \crossLink{b} \subseteq \mathcal{X}\).
Therefore, \(\enumerateSharedOp\) is well defined.

Now, let \(T \in \crossLink{s}\).
For all \(i \in \invEnumerateCrossLink{s}{T}\), there is a subgraph \(F_{s,i} \in \mathcal{F}^\downarrow\) with \(q_{s,i} \in V(F_{s,i})\).
So, \(\enumerateShared{\inv\sigma(F_{s,i})} = T\), showing \(\abs{\invEnumerateShared{T}} \geq d(T)\).
By construction, no other solution subgraphs \(F'\) exhibit \(\enumerateShared{\inv\sigma(F')} = T\).
So, \(\abs{\invEnumerateShared{T}} = d(T)\).
Now, let \(T \in \mathcal{X} \setminus (\crossLink{s} \cup \terminalsContainedNode{s})\).
Then, \(T \in \mathcal{U}_s\) and we see that \(\inv\pi(T) \subseteq \mathcal{F}^\downarrow\) and so \(\invEnumerateShared{T} = \inv\pi(T)\).
Thus, \(\abs{\invEnumerateShared{T}} = d(T)\).

Consider \(F \in \mathcal{F}^\downarrow\) with \(\enumerateShared{\sigma(F)} \in \mathcal{X}\).
If there is a \(j \in \natint{\demandCrossLink{s}}\) with \(Q_{s,j} = \pi(F)\), then \(\pi(F) \cap Y_s = \enumerateCrossLink{s}{j} \cap Y_s = \enumerateShared{\sigma(F)} \cap Y_s\).
Otherwise, we have \(\pi(F) \in \mathcal{U}_s\); so, \(\pi(F) = \enumerateShared{\sigma(F)}\), which is sufficient to prove the claim.

Let \(b \in \boldChildren{s}\), we now define \(\tau_b\).
Let \(\mathcal{F}_b \coloneq \{F \in \mathcal{F} \mid E(F) \cap E(G_b) \neq \emptyset\}\) be the trees that cross into \(G_b = G[Y_b] \cup \boundary{b}\) and let \(\mathcal{F}_b^{\boundaryOp}\coloneqq \{F \in \mathcal{F}_b \mid E(F) \cap E(\boundary{b}) \neq \emptyset\}\) be the trees of \(\mathcal{F}_b\) that are not completely contained in \(G[Y_b]\).
Let \(\mathcal{F}^\downarrow_b \coloneqq \{{F \in \mathcal{F}^{\boundaryOp}_b} \mid \pi(F) \cap Y_b \neq \emptyset\}\) be the trees of \(\mathcal{F}^{\boundaryOp}_b\) that are assigned to terminal sets starting below the link between \(b\) and its parent and let \(\mathcal{F}^\uparrow_b \coloneqq \mathcal{F}^{\boundaryOp}_b \setminus \mathcal{F}^\downarrow_b\) be the remaining trees crossing \(\boundary{b}\).
Recall that the terminal sets of \(\mathscr{D}_{s,\tau}\) can be partitioned into \(\mathcal{U}_s\), \(\{R_{s, P}\}_{P \in \futurePart{\tau}}\), and \(\{Q_{s,i}\}_{i \in \natint{\demandCrossLink{b}}}\).
We note that for all \(P \in \futurePart{\tau}\), we have by definition that \(R_{s,P}\) is disjoint from \(Y_s\supseteq Y_b\) and so \(\inv{\pi}(R_{s,P})\) is disjoint from \(F^\downarrow_b\).
Thus, \(F^\downarrow_b\) are the subgraphs that, by our interpretation of the dynamic program, get assigned to terminal sets starting at or below \(b\) in our final solution.

We set \(\futurePart{\tau_b} \coloneqq \bigcup_{F \in \mathcal{F}^{\uparrow}_b}\{E(K) \cap E(\boundary{b}) \mid K \in \comp{F[E(G_b)]}\}\) to be the edge sets of the components of the solution subgraphs of \(\mathcal{F}^{\uparrow}_b\) after removing all edges outside \(G_s\) and all vertices isolated by this operation.

To define \(\pastPart{\tau_b}\) and \(\pastAssignOp{\tau_b}\), we define a local numbering \(\lambda_b \colon \mathcal{F}^\downarrow_b \to \natint{w}\) where \(\lambda_b\) is an injective function.
Let \(T \in \crossLink{b}\), and set \(\mathcal{F}_{T} \coloneqq \inv\sigma(\invEnumerateShared{T})\).
Notice that \(\mathcal{F}_{T}\subseteq \mathcal{F}^\downarrow_b\) and that \(\abs{\mathcal{F}_{T}} = \abs{\invEnumerateShared{T}} = d(T)\).
This enables us to choose for each \(F \in \mathcal{F}_{T}\) a unique \(i \in \invEnumerateCrossLink{b}{T}\) and set \(\lambda_b(F) \coloneqq i\).
Note that \(\inv\sigma(\invEnumerateShared{T}) = \inv\lambda_b(\invEnumerateCrossLink{b}{T})\), that is we designate exactly the same solution subgraphs to \(T\) in the solution \(\mathscr{D}_{b, \tau_b}\) as in the solution to \(\mathscr{D}_{s,\tau}\).
Then, we choose for each \(F \in \mathcal{F}^\downarrow_b \setminus \bigcup_{T \in \crossLink{b}} \mathcal{F}_{T}\) a unique \(i \in \natint{w} \setminus \natint{\demandCrossLink{b}}\) and set \(\lambda_b(F) \coloneqq i\), concluding the description of \(\lambda_b\).

Now, consider \(F \in \mathcal{F}\).
If there is a \(i \in \natint{w}\) such that \(q_{s,i} \in V(F)\), let \(\mathcal{P} \coloneqq \invPastAssign{s}{i}\).
By \Cref{item:simple-tcd-past} of \Cref{def:simple-tcw-dp}, we can partition \(E(F - q_{s,i})\) into \(\mathcal{E}_F\coloneqq\{E_P\}_{P \in \mathcal{P}}\) such that for all \(P \in \mathcal{P}\), the graph \(F[E_P]\) is connected and \(E_P \cap E(\boundary{s}) = P\).
Otherwise, set \(\mathcal{E}_F \coloneqq \{E(F)\}\).
Note that for all \(F \in \mathcal{F}\), we have \(\bigcup \mathcal{E}_F = E(F - \{q_{s,i}\}_{i \in \natint{w}})\).
We now set \(\pastPart{\tau_b} \coloneqq \bigcup_{F \in \mathcal{F}^\downarrow_b, D \in \mathcal{E}_F}\{E(K) \cap E(\boundary{b})\mid K \in \comp{F[D \cap E(G_b)]}\}\), that is we consider for each \(F \in \mathcal{F}^\downarrow_b\) the partitioning \(\mathcal{E}_F\) of the edge set of the graph \(F - \{q_{s,i}\}_{i \in \natint{w}}\) given by \Cref{item:simple-tcd-past} and for each \(D \in \mathcal{E}_F\) and each connected component \(K\) of \(F[D]\) with all edges outside \(G_b\) removed, we create a partition in \(\pastPart{\tau_b}\) containing the edges of \(K\) in \(\boundary{b}\).
Finally, for all \(F\in \mathcal{F}^\downarrow_b, D \in \mathcal{E}_F\) and \(K \in \comp{F[D \cap E(G_b)]}\) we set \(\pastAssign{\tau_b}{E(K)\cap E(\boundary{b})} \coloneqq \lambda_b(F)\).
Note that we have \(\futurePart{\tau_b} = \bigcup_{F \in \mathcal{F}^\uparrow, D \in \mathcal{E}_F}\{E(K) \cap E(\boundary{b})\mid K \in \comp{F[D \cap E(G_b)]}\}\).

We now show that \(\tau_b \in D(b)\).
For this, we need to provide a solution to \(\mathscr{D}_{b, \tau_b}\) satisfying the additional requirements of \Cref{def:simple-tcw-dp}.
First, consider a \(F \in \mathcal{F}^\uparrow_b\) and let \(K \in \comp{F[E(G_b)]}\).
In our solution, we assign \(K\) to the terminal set \(V(K) \setminus Y_b = R_{b, E(K)\cap E(\boundary{b})}\).
Note that this already satisfies \Cref{item:simple-tcd-future-exact,item:simple-tcd-future-sensible} of \Cref{def:simple-tcw-dp}.
Additionally, it is easy to see that the demand of all \(\{R_{b,P}\}_{P \in \futurePart{\tau_b}}\) is satisfied by this assignment.

Now, consider \(F\in \mathcal{F}^\downarrow_b\subseteq \mathcal{F}^{\boundaryOp}_b\).
Let \(F' \coloneqq F[E(G_s)]\) which might be disconnected.
Note that any \(K \in \comp{F'}\) contains a \(v_K \in V(\boundary{b})\setminus Y_b\).
Let \(F^*_{\lambda_b(F)}\) be \(F'\) combined with \(q_{b, \lambda_b(F)}\) and all incident edges.
As \(q_{b,\lambda_b(F)}\) is connected to a vertex in each component of \(F'\), the graph \(F^*_{\lambda_b(F)}\) is connected and as the underlying \(F \in \mathcal{F}^\downarrow_b\) are edge-disjoint and as \(\lambda_b\) is injective, all such \(F^*_{\lambda_b(F)}\) are edge-disjoint.
Notice that \(F^*_{\lambda_b(F)}\) is contained in \(G^*_b\).

Let \(T \coloneqq \enumerateShared{\sigma(F)}\).
As \(F \in \mathcal{F}^\downarrow_b\), we have \(\pi(F) \cap Y_b \neq \emptyset\).
If \(T \in \crossLink{b}\), we have \(\lambda_b(F) \in \invEnumerateCrossLink{b}{T}\); so, \(V(F^*_{\lambda_b(F)}) \supseteq (\enumerateCrossLink{b}{\lambda_b(F)} \cap Y_b)\cup\{q_{b, \lambda_b(F)}\} = Q_{b,\lambda_b(F)}\).
Consequently, we assign \(F^*_{\lambda_b(F)}\) to \(Q_{b, \lambda_b(F)}\).
If \(T \in \mathcal{X}\setminus \crossLink{b}\), we have \(T \cap Y_b = \pi(F) \cap Y_b \neq \emptyset\), but since \(T \notin \crossLink{b}\), we have \(T \setminus Y_b = \emptyset\).
Thus, \(T \notin \mathcal{U}_b\), violating \(T \in \mathcal{X}\).
Finally, assume there is a \(b' \in \boldChildren{s}\) with \(T = \unknownTerminalIn{b'}\).
As \(\pi(F) \cap Y_b \neq \emptyset\), we have \(b' = b\) and \(\pi(F) \subseteq Y_b\).
So, we assign \(F^*_{\lambda_b(F)}\) to \(\pi(F) \in \mathcal{U}_b\) in our solution to \(\mathscr{D}_{b, \tau_b}\).

Assigning the \(\{F^*_i\}_{i\in\image{\lambda_b}}\) as described satisfies the demand of all \(\{Q_{b,i}\}_{i\in \natint{\demandCrossLink{b}}}\) and for all \(U \in \mathcal{U}_b\), this assigns \(\abs{\{F \in \inv\pi(U) \mid E(F)\cap E(\boundary{b}) \neq \emptyset\}}\) subgraphs to \(U\).
Let \(U \in \mathcal{U}_b\), and consider a \(F \in \inv\pi(U)\) with \(E(F) \cap E(\boundary{s}) = \emptyset\).
As \(U \subseteq Y_b\) and \(F \notin \mathcal{F}^{\boundaryOp}_b\), the graph \(F\) is completely contained in \(G[Y_b]\); so, we assign \(F\) to \(U\) is our solution.
This satisfies the requirements of all \(\mathcal{U}\).
Which completes the description of our solution to \(\mathscr{D}_{b,\tau_b}\).

To complete the prove that \(\tau_b \in D(b)\), we need to show that our solution satisfies the additional requirements posed in \Cref{def:simple-tcw-dp}.
Note that by construction \Cref{item:simple-tcd-future-sensible,item:simple-tcd-future-exact} are already satisfied.
We see that the edges of \(\boundary{b}\) that are used in our solution are exactly \(E(\boundary{b})\cap \bigcup E(\mathcal{F}^{\boundaryOp}_b)\).
By definition of \(\pastPart{\tau_b}\) and \(\futurePart{\tau_b}\), we have \(\bigcup \pastPart{\tau_b} = E(\boundary{b}) \cap \bigcup E(\mathcal{F}^\downarrow_b)\) and \(\bigcup \futurePart{\tau_b} = E(\boundary{b}) \cap \bigcup E(\mathcal{F}^\uparrow_b)\).
So, \(\bigcup \pastPart{\tau_b}\cup\bigcup \futurePart{\tau_b} = E(\boundary{b}) \cap \bigcup E(\mathcal{F}^{\boundaryOp}_b)\) showing that \Cref{item:simple-tcd-edge-partitions} of \Cref{def:simple-tcw-dp} is satisfied.

Now, consider \(i \in \natint{w}\) and let \(\mathcal{P}_i \coloneqq \invPastAssign{\tau_b}{i}\).
If \(P_i = \emptyset\), there is no \(F \in \mathcal{F}^\downarrow_b\) with \(\lambda_b(F) = i\) and by construction there is no solution subgraph \(F\) in our solution with \(q_{b,i} \in V(F)\).
Otherwise, we have that \(q_{b,i} \in V(F^*_i)\).
The subgraph \(F^*_i\) is part of our solution and the unique subgraph of our solution containing \(q_i\).
Additionally, \(F^*_i - q_{b,i} = F[E(G_s)]\) and \(\mathcal{P}_i = \bigcup_{D \in \mathcal{E}_F}\{E(K) \cap E(\boundary{b})\mid K \in \comp{F[D \cap E(G_b)]}\}\).
Let \(P \in \mathcal{P}_i\).
Then, there is a \(D \in \mathcal{E}_F\) and \(K \in \comp{F[D \cap E(G_b)]}\) with \(P = E(K) \cap E(\boundary{b})\).
Set \(E_P^F \coloneqq E(K)\).
Notice that for all \(P' \in \mathcal{P}_i\), the set \(E^F_P\) is disjoint from \(E^F_{P'}\).
Additionally, \(\bigcup_{P \in \mathcal{P}_i} E^F_P = \bigcup_{D \in \mathcal{E}_F} D \cap E(G_b)= E(G_b)\cap\bigcup_{D \in \mathcal{E}_F} D = E(G_b) \cap E(F - \{q_{s,i}\}_{i\in\natint{w}}) = E(F^*_i - q_{b,i})\).
So, \(\{E^F_P\}_{P \in \mathcal{P}_i}\) satisfies the requirements of \Cref{item:simple-tcd-past} in \Cref{def:simple-tcw-dp} and \(\tau_b \in D(b)\).

We now show, that \((\tau_b)_{b \in \boldChildren{s}}\) witness \(\tau \in D(s)\).
For this, we first define a local to shared mapping \((\localSharedMapOp{a})_{a \in A}\).
We aim to use the local mappings \((\lambda_b)_{b\in B}\) and the shared numbering \(\sigma\) of the trees in \(F^\downarrow\).
To use the same approach for \(s\) as for the bold children, we now define a local numbering for \(s\).
Let \(\mathcal{F}^\downarrow_s \coloneqq \{F \in \mathcal{F}^\downarrow \mid E(F) \cap E(\boundary{s}) \neq \emptyset\}\) be the trees assigned to a terminal set starting at or below \(s\) crossing \(\boundary{s}\), analogously to \(\{\mathcal{F}\}_{b \in \boldChildren{s}}\).
We now choose an injection \(\lambda_s \colon \mathcal{F}^\downarrow_s \to \natint{w}\).
Let \(F \in \mathcal{F}^\downarrow_s\).
By \Cref{item:simple-tcd-edge-partitions,item:simple-tcd-future-exact,item:simple-tcd-past} of \Cref{def:simple-tcw-dp}, there is a unique \(i \in \natint{w}\) such that \(q_{s,i} \in V(F)\).
Set \(\lambda_s(F) \coloneqq i\).
Note that analogously to \(\{\lambda_b\}_{b\in\boldChildren{s}}\), we have for all \(T \in \crossLink{s}\) that \(\inv\sigma(\invEnumerateShared{T}) = \inv\lambda_s(\invEnumerateCrossLink{s}{T})\).
To define the local to shared mappings, let \(a \in A\).
For all \(i \in \image{\lambda_a}\), we set \(\localSharedMap{a}{i} \coloneqq \sigma(\inv\lambda_a(i))\).
For each \(i \in \natint{w} \setminus \image{\lambda_a}\), we choose a unique \(j \in \invEnumerateShared{\blockedIndex} = \natint{u}\setminus \natint{\ell}\), which works since \(u - l \geq w\), and set \(\localSharedMap{a}{i} \coloneqq j\).

We now show, that \((\localSharedMapOp{a})_{a \in A}\) combined with \(\enumerateSharedOp\) are actually a local to shared mapping and a shared mapping enumerator according to \Cref{def:local-shared-map}.
Let \(a \in A\) and \(T \in \crossLink{a}\).
We know \(\inv\sigma(\invEnumerateShared{T}) = \inv\lambda(\invEnumerateCrossLink{a}{T})\).
Thus, \(\invEnumerateShared{T} = \sigma(\inv\lambda(\invEnumerateCrossLink{a}{T})) = \localSharedMap{a}{\invEnumerateCrossLink{a}{T}}\).
Consider \(T \in \terminalsContainedNode{s}\subseteq \mathcal{X}\).
We have \(\inv\sigma(\invEnumerateShared{T}) \subseteq \inv\pi(T)\).
Thus, \(\abs{\invEnumerateShared{T}} \leq d(T)\).
Let \(b \in \boldChildren{s}\).
Consider \(i \in \invEnumerateShared{\unknownTerminalIn{b}}\) and let \(F \coloneqq \inv\sigma(i)\).
For all \(j \in \natint{\crossLink{s}}\), we have that \(\pi(F) \neq Q_{s,j}\).
Thus, \(\pi(F) \in \mathcal{U}_s\) and, in particular, \(\pi(F) \in \mathcal{U}_b\), implying \(F \in \mathcal{F}^\downarrow_b\).
Now, assume \(\lambda_b(F) \in \natint{\demandCrossLink{b}}\).
Then, \(\enumerateShared{i} \in \crossLink{b}\), violating \(\pi(F) \in \mathcal{U}_b\).
Thus, \(\lambda_b(F) \in \natint{w} \setminus\natint{\demandCrossLink{b}}\) and \(i = \sigma(F) = \sigma(\inv\lambda_b(\lambda_b(F)))) = \localSharedMap{b}{\lambda_b(F)}\in \localSharedMap{b}{\natint{w} \setminus \natint{\demandCrossLink{b}}}\).
Finally, let \(i \in \natint{w} \setminus \natint{\demandCrossLink{b}}\).
If \(i \in \image{\lambda_b}\), let \(F \coloneqq \inv\lambda_b(i) \in \mathcal{F}^\downarrow_b\), \(D \in \mathcal{E}_F\) and \(K \in \comp{F[D \cap E(G_b)]}\) be a component of the solution subgraph \(F\) after removing all edges outside \(D \cap E(G_b)\) and all vertices that got isolated that way.
By definition \(\pastAssign{\tau_b}{E(K) \cap E(\boundary{s})} = i\).
Additionally, \(\enumerateShared{\sigma(F)} \notin \crossLink{b}\) and as \(\pi(F) \cap Y_b \neq \emptyset\), we have \(\enumerateShared{\sigma(F)} = \unknownTerminalIn{b}\) and, in particular, \(\enumerateShared{\sigma(\inv\lambda_b(\lambda_b(F)))} = \enumerateShared{\localSharedMap{b}{\lambda_b(F)}} = \enumerateShared{\localSharedMap{b}{i}} = \unknownTerminalIn{b}\).
If \(i \notin \image{\lambda_b}\), then, by construction, there is no \(P \in \pastPart{\tau_b}\) with \(\pastAssign{\tau_b}{P} = i\) and \(\enumerateShared{\localSharedMap{b}{i}} = \blockedIndex\),completing the proof that \((\localSharedMapOp{a})_{a\in A}\) is a local to shared mapping and \(\enumerateSharedOp\) a shared mapping enumerator.

To prove that \((\tau_b)_{b \in \boldChildren{s}}\) witness \(\tau \in D(s)\), we provide a solution to the instance \(\mathscr{C}(s,\tau, (\tau_b)_{b \in \boldChildren{s}}, \enumerateSharedOp, (\localSharedMapOp{b})_{b \in \boldChildren{s}})\).
For this, consider any \(F \in \mathcal{F}' = \{F \in \mathcal{F}\mid E(F)\cap E(J_s) \neq \emptyset\}\).
We now show how to transform \(F\) into a connected subgraph \(\widetilde{F}\) of \(\widetilde{J_s}\) such that the edges and vertices in \(J_s\) do not change that much.
Let \(F' \coloneqq F[E(J_s)]\) be \(F\) with all edges outside \(J_s\) removed and all vertices which get isolated by this operation removed.
This graph might be disconnected.
For all \(b \in \boldChildren{s}\) and \(P \in \pastPart{\tau_b} \cup \futurePart{\tau_b}\) with \(E(F') \cap P \neq \emptyset\), add \(\tilde{p}_{b, P}\) and all incident edges to \(F'\).
If there is a \(q_{s, i} \in F'\), add the edges \(\{vq_{s,i}\mid v\in N(q_{s,i})\}\).
Call the obtained graph \(\widetilde{F}\).
Notice that \(V(\widetilde{F}) \cap (X_s \cup V(\boundary{s}) \cup \{q_{s,i}\}_{i\in\natint{w}}) = V(F)\cap (X_s \cup V(\boundary{s}) \cup \{q_{s,i}\}_{i\in\natint{w}})\) as well as \(E(\widetilde{F} - \{q_{s,i}\}_{i \in \natint{w}}) \cap E(J_s) = E(F - \{q_{s,i}\}_{i \in \natint{w}}) \cap E(J_s)\).
Additionally, note that \(P \subseteq E(F') \subseteq E(\widetilde{F})\), which shows that \Cref{item:simple-tcd-witness-sub-connected} of \Cref{def:simple-tcd-witness} is fulfilled if we do not use any \(\{\tilde{p}_{b, P}\}_{b\in\boldChildren{s}, P \in \pastPart{\tau_b}\cup \futurePart{\tau_b}}\) in another solution subgraph.
By \Cref{item:simple-tcd-future-exact,item:simple-tcd-witness-past}, we also know that all such \(\widetilde{F}\) are edge-disjoint.

To show that \(\widetilde{F}\) is connected, we first show that each connected component of \(\widetilde{F}\) contains an edge of \(E(J_s)\).
Then, we show—a slightly stronger statement as—that for all \(D \in \mathcal{E}_F\), all vertices of \(\widetilde{F}[D]\) are contained in one connected component of \(\widetilde{F}\).
Finally, we show that there is a single vertex contained in every connected component, which shows that there is at most one connected component.

First, assume there is a component of \(K \in \comp{\widetilde{F}}\) with \(E(K) \cap E(J_s) = \emptyset\).
If there is a \(b \in \boldChildren{s}\) and \(P \in \pastPart{\tau_b}\cup\futurePart{\tau_b}\) with \(\tilde{p}_{b, P} \in V(K)\), then \(N[\tilde{p}_{b,P}] \subseteq V(K)\).
As \(P\subseteq E(\widetilde{F}[N[p_{b,P}]])\), we have \(\emptyset \neq P \subseteq E(K)\) and since \(P \subseteq E(J_s)\), this is not possible.
As \(V(\widetilde{F}) \subseteq V(J_s) \cup \{\tilde{p}_{b, P}\}_{b\in\boldChildren{s}, P \in \bigcup\pastPart{\tau_b}\cup\bigcup\futurePart{\tau_b}}\), the graph \(K\) is equal to an isolated node of \(J_s\).
As \(F\) has no isolated vertices and since \(F'\) is an edge-induced graph on \(F\), the graph \(F'\) and hence \(\widetilde{F}\) does not contain isolated vertices, making this impossible as well.
Thus, \(E(K)\) is not disjoint from \(J_s\).

Let \(D \in \mathcal{E}_F\).
Let \(\widetilde{Z} \coloneqq \{\tilde{p}_{b,P}\mid b \in \boldChildren{s}, P \in \pastPart{\tau_b} \cup \futurePart{\tau_b}\}\) to be all vertices that simulate connectivity of partitions of the sub-solutions \(\{\tau_b\}_{b\in\boldChildren{s}}\).
Let \(\widetilde{Z}_D \coloneqq \{\tilde{p}_{b,P} \in Z\mid P \cap D \neq \emptyset\}\) be the vertices simulating connectivity inside a component for a partition induced by \(D\).
Set \(\widetilde{Z}^E_D \coloneqq \{uz\mid z \in Z_D, u \in N(z)\}\) to be the edges adjacent to the vertices of \(Z_D\).
Note that \(\widetilde{Z}^E_D \subseteq E(\widetilde{F})\).
We now show, that \(\widetilde{F}[D \cup \widetilde{Z}^E_D]\) is connected.
First, we note that any connected component of \(\widetilde{F}[D \cup \widetilde{Z}^E_D]\) contains a vertex of \(\widetilde{F}[D]\).
Consider as \(u, v \in V(\widetilde{F}[D])\) vertices in different connected components of \(\widetilde{F}[D \cup \widetilde{Z}^E_D]\) such that their distance in \(F[D]\) is minimized.
Let \(P\) be a shortest path connecting \(u\) and \(v\) in \(F[D]\).
By minimality and since \(V(F[D]) \cap \{q_{s,i}\}_{i \in \natint{w}} = \emptyset\), we have \(V(P) \cap V(J_s) = \{u,v\}\).
Note that the sets \(\{Y_b \setminus V(J_s)\}_{b\in \boldChildren{s}}\) are disconnected from each other in \(G^*_s\) and \(V(G^*_s) \setminus V(J_s) = \bigcup_{b\in \boldChildren{s}}(Y_b \setminus V(J_s))\).
So, there is a \(b \in \boldChildren{s}\) with \(V(P) \subseteq Y_b \cup \{u,v\}\).
By minimality, we even have \(u,v \in Y_b\).
Now, consider \(K \in \comp{F[D \cap E(G_b)]}\) with \(u \in V(K)\).
Then, \(V(P) \subseteq V(K)\) and in particular \(v \in V(K)\).
As \(u\) and \(v\) are not isolated in \(F[D]\), there are \(e_u\in E(F'[D])\) and \(e_v\in E(F'[D])\) adjacent to \(u\) and \(v\) respectively.
As all edges completely contained in \(Y_b\) are removed from \(F'\), we have \(e_u, e_v \in E(\boundary{b})\) and, by definition of \(\pastPart{\tau_b}\) and \(\futurePart{\tau_b}\), there is a \(P\in\pastPart{\tau_b}\cup\futurePart{\tau_b}\) with \(e_u, e_v \in P\).
Since \(\tilde{p}_{b, P} \in \widetilde{Z}_D\) and all adjacent edges are contained in \(\widetilde{Z}^E_D\), \(u\) and \(v\) are actually in the same component of \(\widetilde{F}[D \cup \widetilde{Z}^E_D]\).
Thus, \(\widetilde{F}[D\cup \widetilde{Z}^E_D]\) is connected.

If \(\abs{\mathcal{E}_F} = 1\), the last argument shows that \(\widetilde{F}\) is connected.
So, assume \(\abs{\mathcal{E}_F} > 1\).
This is only the case if there is an \(i \in \natint{w}\) such that \(q_{s,i} \in V(F)\).
Let \(K \in \comp{\widetilde{F}}\).
Since \(E(K) \cap E(J_s) \neq \emptyset\), there is a \(D \in \mathcal{E}_F\) with \(E(K) \cap D \neq \emptyset\).
By definition of \(D\), there is an edge \(uv \in E(\boundary{s}) \cap D\) with \(v\notin Y_s\).
Note that, by construction, \(vq_{s,i} \in E(\widetilde{F})\).
Since all vertices of \(V(\widetilde{F}[D])\) are in the same component of \(\widetilde{F}\), we have \(v \in V(\widetilde{K})\) and by extension \(q_{s,i} \in V(K)\), showing that \(\widetilde{F}\) has one connected component.

To describe our solution consider a \(F \in \mathcal{F}^\downarrow\) and choose a \(v \in V(\widetilde{F}) \cap V(J_s)\).
Set \(\widetilde{F_{\sigma(F)}} \coloneqq \widetilde{F} \cup \{vm_{\sigma(F)}\}\).
If \(\enumerateShared{\sigma(F)} \notin \mathcal{X}\), then \(\widetilde{C_{\sigma(F)}} = \emptyset \subseteq V(\widetilde{F_{\sigma(F)}})\).
If \(\enumerateShared{\sigma(F)} \in \mathcal{X}\), then \(\widetilde{C_{\sigma(F)}} = \enumerateShared{\sigma(F)} \cap X_s\).
We have \(\pi(F) \cap Y_s = \enumerateShared{\sigma(F)} \cap Y_s\).
Therefore, \(\enumerateShared{\sigma(F)} \cap X_s \subseteq V(F)\) and, in particular, \(\enumerateShared{\sigma(F)} \cap X_s \subseteq V(\widetilde{F_{\sigma(F)}})\).
Since \(m_{\sigma(F)} \in V(\widetilde{F_{\sigma(F)}})\), we have \(\widetilde{M_{\sigma(F)}}\subseteq V(\widetilde{F_{\sigma(F)}})\); so, we assign \(\widetilde{F_{\sigma(F)}}\) to \(\widetilde{M_{\sigma(F)}}\) satisfying the demand of all \(\{\widetilde{M_i}\}_{i \in \natint{u} \setminus \invEnumerateShared{\blockedIndex}}\).

Consider \(F \in \mathcal{F}^\uparrow = \bigcup_{P\in\futurePart{\tau}} \inv\pi(R_{s,P})\).
We note that \(\widetilde{F}\) is edge-disjoint from all already assigned subgraphs.
In our solution we assign \(\widetilde{F}\) to \(\pi(F)\) as well.
This satisfies the demand of all \(\{R_{s,P}\}_{P \in \futurePart{\tau}}\).

Finally, consider \(F \in F^\star = F' \setminus (\mathcal{F}^\uparrow \cup \mathcal{F}^\downarrow)\).
As \(V(F) \subseteq X_s \cup \bigcup_{t\in\thinChildren{s}} X_t\), we have \(F = \widetilde{F}\) and notice that \(\widetilde{F}\) is edge-disjoint from all already assigned subgraphs.
So, we assign \(\widetilde{F}\) to \(\pi(F)\) as well.
We claim that this satisfies the demand of all \(\terminalsContainedNode{s}\).
Let \(T \in \terminalsContainedNode{s}\) and consider \(\inv\pi(T)\).
Set \(\mathcal{F}^{\boundaryOp}_T \coloneqq \inv\sigma(\invEnumerateShared{T})\) and notice that \(\mathcal{F}^{\boundaryOp}_T \subseteq \inv\pi(T)\).
Let \(F \in \inv\pi(T) \setminus \mathcal{F}^{\boundaryOp}_T\).
As \(\pi(F) \subseteq X_s\), we have \(V(F) \cap X_s \neq \emptyset\) and since \(F \notin \mathcal{F}^{\boundaryOp}\), we have that \(E(F)\) is disjoint from \(\bigcup_{a \in A} E(\boundary{a})\).
Therefore, \(V(F) \subseteq X_s \cup \bigcup_{t\in\thinChildren{s}} X_t\) and \(F \in \terminalsContainedNode{s}\).
So, we assign \(\abs{\inv\pi(T)} - \abs{\mathcal{F}^{\boundaryOp}_T} = \abs{\inv\pi(T)} - \abs{\invEnumerateShared{T}}\) subgraphs to \(T\) satisfying its demand.
Completing our solution of \(\mathscr{C}(s,\tau, (\tau_b)_{b \in \boldChildren{s}}, \enumerateSharedOp, (\localSharedMapOp{b})_{b \in \boldChildren{s}})\).

To complete the prove that \((\tau_b)_{b\in\boldChildren{s}}\) witness \(\tau \in D(s)\), we need to verify the additional requirements to our solution \((\widetilde{\mathcal{F}}, \widetilde\pi)\) posed in \Cref{def:simple-tcd-witness}.
Notice that \Cref{item:simple-tcd-witness-edge-partitions,item:simple-tcd-witness-future-sensible,item:simple-tcd-witness-future-exact} of \Cref{def:simple-tcd-witness} directly follow from \(\tau \in D(s)\) and we already argued \Cref{item:simple-tcd-witness-sub-connected} of \Cref{def:simple-tcd-witness}.

Consider \(i \in \natint{w}\) and \(\mathcal{P}_i = \pastAssign{\tau}{i}\).
If \(\mathcal{P}_i = \emptyset\), there is no \(F \in \mathcal{F}\) with \(q_{s,i} \in V(F)\).
As for all \(\widetilde{F} \in \widetilde{\mathcal{F}}\), we have \(q_{s,i} \in \widetilde{F}\) if and only if \(q_{s,i} \in F\), there also is no \(\widetilde{F}' \in \widetilde{\mathcal{F}}\) with \(q_{s,i} \in V(\widetilde{F}')\).
Otherwise, there is exactly one \(F \in \mathcal{F}\) with \(q_{s,i} \in V(F)\).
For \(P \in \widetilde{\mathcal{P}_i}\) consider \(D \in \mathcal{E}_F\) with \(P = D \cap E(\boundary{s})\) and choose as the partition for \(P\) the set \(\widetilde{E}_P \coloneqq (D \cap E(J_s))\cup \widetilde{Z}^E_D\).
We know that \(\widetilde{F}[D \cup \widetilde{Z}^E_D]=\widetilde{F}[\widetilde{E}_P]\) is connected.
Consider a \(\tilde{p}_{b, P'} \in \widetilde{Z}_D\).
By definition of \(\{\tau_b\}_{b \in \boldChildren{s}}\), we have \(\emptyset \neq P' \subseteq D \cap E(J_s) \subseteq \widetilde{E}_P\).
So, there is no \(D' \in \mathcal{E}_F \setminus \{D\}\) with \(\tilde{p}_{b,P'} \in \widetilde{Z}_{D'}\) implying that the partitions are actually disjoint and that any vertex of \(\widetilde{Z}_D = V(\widetilde{F}[\widetilde{E}_P]) \cap \widetilde{Z}\) is only a member of \(V(\widetilde{F}[\widetilde{E}_P])\).
Additionally, we can verify that \(\bigcup_{P \in \mathcal{P}_i} \widetilde{E}_P = E(\widetilde{F} - \{q_{s,i}\})\), showing that \Cref{item:simple-tcd-witness-past} of \Cref{def:simple-tcd-witness} is satisfied.

Let \(i \in \natint{u}\setminus \invEnumerateShared{\blockedIndex}\) and \(\{\widetilde{F}\} \coloneqq \inv{\widetilde\pi}(\widetilde{M_i})\).
By construction \(\widetilde{F}\) is the unique \(\widetilde{F}' \in \widetilde{\mathcal{F}}\) with \(m_i \in V(\widetilde{F}')\) and only one edge adjacent to \(m_i\) is contained in \(E(\widetilde{F})\).
Now, let \(a \in A\) and \(P \in \pastPart{\tau_a}\) be such that \(\localSharedMap{a}{\pastAssign{\tau_a}{P}} = i\).
Consider \(F' \coloneqq \inv\sigma(\localSharedMap{a}{\pastAssign{\tau_a}{P}}) = \inv\sigma(i)\); so, by construction, we have \(\widetilde{F'} = \widetilde{F}\).
If \(a = s\) by \Cref{item:simple-tcd-past} of \Cref{def:simple-tcw-dp} and otherwise by construction, we have \(P \subseteq E(F') \cap E(\boundary{a}) = E(\widetilde{F}) \cap E(\boundary{a})\).
If \(a \in \boldChildren{s}\), then we add \(\tilde{p}_{a, P}\) to \(\widetilde{F}\) during construction showing that \Cref{item:simple-tcd-witness-marked} of \Cref{def:simple-tcd-witness}, and by extension \Cref{def:simple-tcd-witness} itself, is satisfied.
\end{proof}

\begin{lemma}
\label{stmt:simple-tcd-witness-implies-dp}
If for all \(b \in \boldChildren{s}\) there is a \(\tau_b \in D(b)\) such that \((\tau_b)_{b \in \boldChildren{s}}\) witnesses \(\tau \in D(s)\), then \(\tau \in D(s)\) indeed holds.
\end{lemma}
\begin{proof}
Let \((\localSharedMapOp{b})_{b\in \boldChildren{s}}\) and \(\enumerateSharedOp\) be local to shared mappings and a shared mapping enumerator such that they fulfill the requirement of \Cref{def:simple-tcd-witness}.
Let \((\widetilde{\mathcal{F}}, \tilde{\pi})\) be a solution to \(\mathscr{C}(s,\tau, (\tau_b)_{b \in \boldChildren{s}},\enumerateSharedOp, (\localSharedMapOp{a})_{a\in A})\) satisfying the additional requirements of \Cref{def:simple-tcd-witness}. 
For all \(b \in \boldChildren{s}\) let \((\mathcal{F}_b, \pi_b)\) be a solution to \(\mathscr{D}_{b, \tau_b}\) satisfying the additional requirements of \Cref{def:simple-tcw-dp}.
Based on this, we now provide a solution to \(\mathscr{D}_{s, \tau}\) satisfying the additional requirements of \Cref{def:simple-tcw-dp}.

For every \(b\in \boldChildren{s}\), set \(\mathcal{F}^{\boundaryOp}_b \coloneqq \{F \in \mathcal{F}_b \mid E(F) \cap E(\boundary{b}) \neq \emptyset\}\).
For all \(F \in \mathcal{F}_b \setminus \mathcal{F}^{\boundaryOp}_b\), we know \(\pi_b(F) \in \mathcal{U}_b \subseteq \mathcal{U}_s\) and that \(F\) is contained in \(G[Y_b]\), which in turn is a subgraph of \(G^*_s\).
So, we assign \(F\) to \(\pi_b(F)\) in our solution of \(\mathscr{D}_{s, \tau}\).
For all \(U \in \mathcal{U}_b\subseteq \mathcal{U}_s\), this assigns \(d(U) - \abs{\inv\pi_b(U) \cap \mathcal{F}^{\boundaryOp}_b}\) many subgraphs to \(U\).

Now, consider any \(F \in \widetilde{\mathcal{F}}\).
Let \(\widetilde{\mathcal{F}}^{\star}\coloneqq \{F \in \widetilde{\mathcal{F}} \mid V(F) \subseteq X_s \cup \bigcup_{t \in \thinChildren{s}} X_t\}\) and set \(\widetilde{\mathcal{F}}^{\boundaryOp} \coloneq \widetilde{\mathcal{F}} \setminus \widetilde{\mathcal{F}}^\star\).
All \(F \in \widetilde{\mathcal{F}}^\star\) are contained in \(G^*_s\) and \(\pi(F) \in \terminalsContainedNode{s}\).
As \(F\) is disjoint from all previously assigned subgraphs, we assign \(F\) to \(\pi(F)\) in our solution.
For all \(T \in \terminalsContainedNode{s}\), we assign \(d'(T) - \abs{\inv{\tilde{\pi}}(T) \cap \widetilde{\mathcal{F}}^{\boundaryOp}} = d(T) - \abs{\invEnumerateShared{T}} - \abs{\inv{\tilde{\pi}}(T) \cap \widetilde{\mathcal{F}}^{\boundaryOp}}\) subgraphs to \(T\).

Finally, consider \(\widetilde{F} \in \widetilde{\mathcal{F}}^{\boundaryOp}\).
The graph \(\widetilde{F}\) might not be a subgraphs of \(G^*_s\).
In particular, the vertices \(\widetilde{Z} \coloneqq \{\tilde{p}_{b, P}\}_{b \in \boldChildren{s}, P \in \pastPart{\tau_b} \cup \futurePart{\tau_b}}\) and \(\widetilde{W}\coloneqq \{\tilde{m}_i\}_{i\in\natint{u}\setminus\invEnumerateShared{\blockedIndex}}\) are not contained in \(G^*_s\).
We aim to create a connected subgraph \(\psi(\widetilde{F})\) similar to \(\widetilde{F}\) contained in \(G^*_s\).

By \Cref{item:simple-tcd-witness-marked} of \Cref{def:simple-tcd-witness}, we can remove \(\widetilde{W}\) from \(\widetilde{F}\) without destroying connectedness.
For each \(\tilde{p}_{b, P} \in \widetilde{Z}\), we now find edge-disjoint connected subgraphs \(H_{b, P}\) completely contained in \(G_b\) with \(E(H_{b, P}) \cap E(\boundary{b}) = P\).
If \(P \in \futurePart{\tau_b}\), choose a unique \(F \in \inv\pi_b(R_{b,P})\) and set \(H_{b, P} \coloneqq F\).
This works by choice of the demand of \(R_{b,P}\) and \Cref{item:simple-tcd-future-sensible,item:simple-tcd-future-exact} of \Cref{def:simple-tcw-dp}.
Otherwise \(P \in \pastPart{\tau_b}\).
Set \(i \coloneqq {\pastAssign{\tau_b}{P}}\) and let \(F_b^i \in \mathcal{F}_b\) be the unique subgraph with \(q_{b,i} \in V(F^i_b)\).
By \Cref{item:simple-tcd-past}, we can partition \(E(F^i_b - q_{b,i})\) into \(\{E_{P'}\}_{P' \in \invPastAssign{\tau_b}{i}}\) such that for all \(P'\in\invPastAssign{\tau_b}{i}\), the graph \(F[E_{P'}]\) is connected and \(E_{P'} \cap E(\boundary{s}) = P'\).
Then, choose \(H_{b, P} \coloneqq F^i_b[E_P]\) and note that \(\bigcup_{P' \in \pastAssign{\tau_b}{i}} F^i_b[E_{P'}] = F^i_b - q_{b,i}\).
Now, set \(\psi(\widetilde{F})\) equal to \(\widetilde{F}\) where we replace all \(\tilde{p}_{b, P} \in V(\widetilde{F})\cap \widetilde{Z}\) by \(H_{b,P}\), that is \(\psi(\widetilde{F}) \coloneqq (\widetilde{F} - \widetilde{W} - \widetilde{Z}) \cup \bigcup_{\tilde{p}_{b,P} \in V(\widetilde{F}) \cap \widetilde{Z}} H_{b,P}\).
By \Cref{item:simple-tcd-witness-sub-connected} of \Cref{def:simple-tcd-witness}, all such graphs are edge-disjoint from each other and from all already assigned subgraphs.
To see that \(\psi(\widetilde{F})\) is actually connected, notice that, by \Cref{item:simple-tcd-witness-sub-connected} of \Cref{def:simple-tcd-witness}, \(N(\tilde{p}_{b,P}) = V(F[P]) \subseteq V(H_{b,P})\) and that \(H_{b,P}\) is connected. So, any path in \(F\) going over \(\tilde{p}_{b,P}\) can be replaced by a path going through \(H_{b,P}\).

Now, we aim to assign \(\psi(\widetilde{F})\) to a terminal set.
If \(\tilde{\pi}(\widetilde{F}) \in \{R_{s, P}\}_{P \in \pastPart{\tau}\cup\futurePart{\tau}}\cup\terminalsContainedNode{s}\), we assign \(\psi(\widetilde{F})\) to \(\tilde{\pi}(\widetilde{F})\subseteq V(\widetilde{F} - \widetilde{W} - \widetilde{Z})\subseteq V(\psi(\widetilde{F}))\).
For \(T \in \{R_{s, P}\}_{P \in \pastPart{\tau}\cup\futurePart{\tau}}\), we have \(\inv{\tilde\pi}(T) \subseteq \widetilde{\mathcal{F}}^{\boundaryOp}\).
Thus, this satisfies the demand of \(T\) and \Cref{item:simple-tcd-future-sensible,item:simple-tcd-future-exact} of \Cref{def:simple-tcw-dp}.
For \(T \in \terminalsContainedNode{s}\), this newly assigns \(\abs{\inv{\tilde\pi}(T) \cap \widetilde{\mathcal{F}}^{\boundaryOp}}\) subgraphs to \(T\).
Otherwise, there is a unique \(i \in \natint{u} \setminus \invEnumerateShared{\blockedIndex}\) with \(\tilde{\pi}(\widetilde{F}) = \widetilde{M_i}\).
Consider \(b \in \boldChildren{s}\) with \(\invLocalSharedMap{b}{i} \neq \emptyset\).
Set \(\ell \coloneqq \invLocalSharedMap{b}{i}\).
By \Cref{item:simple-tcd-witness-marked} of \Cref{def:simple-tcd-witness}, for all \(P \in \invPastAssign{\tau_b}{\ell}\) we have \(\tilde{p}_{b, P} \in V(\widetilde{F})\).
Thus, \(H_{b,P}\) is contained in \(\psi(\widetilde{F})\) and in particular \(V(F^{\ell}_b - q_{b,\ell}) \subseteq V(\psi(\widetilde{F}))\).
If \(\ell \in \natint{w} \setminus \natint{\demandCrossLink{s}}\), by \Cref{item:local-shared-def-unknown-is-unnamend-or-blocked} of \Cref{def:local-shared-map} we have \(\enumerateShared{i} = \unknownTerminalIn{b}\).
In particular \(\pi_b(F^\ell_b) \in \mathcal{U}_b\).
Since \(\pi_b(F^\ell_b) \subseteq Y_b\), we have \(\pi_b(F^\ell_b) \subseteq V(F^\ell_b - q_{b,\ell}) \subseteq V(\psi(\widetilde{F}))\) and we assign \(\psi(\widetilde{F})\) to \(\pi_b(F^\ell_b)\).
Otherwise, we have by \Cref{item:local-shared-def-unnamed-sensible} of \Cref{def:local-shared-map} that \(\enumerateShared{i} \in \mathcal{X}\).
As \(\enumerateShared{i} \cap X_s \subseteq \widetilde{M_i} = \tilde{\pi}(\widetilde{F})\), we have \(\enumerateShared{i} \cap X_s \subseteq V(\psi(\widetilde{F}))\).
Now, let \(b \in \boldChildren{s}\) be such that \(\enumerateShared{i} \in \crossLink{b}\) and again set \(\ell \coloneqq \invLocalSharedMap{b}{i}\).
Then, \(\ell \in \natint{\demandCrossLink{s}}\) and \(\pi_b(F^\ell_b) \setminus\{q_{b, \ell}\} = Q_{b, \ell} \setminus \{q_{b,\ell}\} = \enumerateCrossLink{b}{\ell} \cap X_b = \enumerateShared{i} \cap Y_b\).
Additionally, \(\pi_b(F^\ell_b) \setminus \{q_{b,\ell}\} = \enumerateShared{i} \cap Y_b \subseteq V(F^\ell_b - q_{b, \ell}) \subseteq V(\psi(\widetilde{F}))\).
For \(b \in \boldChildren{s}\) with \(\enumerateShared{i} \notin \crossLink{b}\), we have \(\enumerateShared{i} \cap Y_b = \emptyset \subseteq V(\psi(\widetilde{F}))\).
Thus, \(\enumerateShared{i} \cap Y_s = \enumerateShared{i} \cap (X_s \cup \bigcup_{b\in \boldChildren{s}} Y_b) \subseteq V(\psi(\widetilde{F}))\) and if \(\enumerateShared{i} \subseteq Y_s\), we assign \(\psi(\widetilde{F})\) to \(\enumerateShared{i}\).
Otherwise, \(\enumerateShared{i} \in \crossLink{s}\) and by \Cref{item:local-shared-def-preserves-crossLink} of \Cref{def:local-shared-map}  we have \(\invLocalSharedMap{s}{i} \in \natint{\demandCrossLink{s}}\).
Set \(\ell \coloneqq \invLocalSharedMap{s}{i}\).
Thus, \(q_{s, \ell} \in \pi(\widetilde{F}) = \widetilde{M_i} \subseteq V(\psi(\widetilde{F}))\).
So, \(Q_{s, \ell} = (\enumerateCrossLink{s}{\ell} \cap Y_s) \cup \{q_{s, \ell}\} = (\enumerateShared{i} \cap Y_s) \cup \{q_{s, \ell}\} \subseteq V(\psi(\widetilde{F}))\) and we assign \(\psi(\widetilde{F})\) to \(Q_{s, \ell}\).

We claim that now all demands are satisfied.
Recall, that we already proved this for all terminal sets in \(\{R_{s, P}\}_{P \in \futurePart{\tau}}\).
For \(T \in \terminalsContainedNode{s}\), before considering \(F \in \inv{\tilde\pi}(\{\widetilde{M_i}\}_{i \in \natint{u}\setminus\invEnumerateShared{\blockedIndex}})\) we already assigned \(d(T) - \abs{\invEnumerateShared{T}}\) subgraphs to \(T\).
As \(T \subseteq Y_s\), for all \(i \in \invEnumerateShared{T}\) the subgraph \(\psi(\inv{\tilde\pi}(M_i))\) gets assigned to \(T\) satisfying its demand.
Let \(T \in \mathcal{X} \setminus (\terminalsContainedNode{s} \cup \{\crossLink{s}\})\).
By \Cref{item:local-shared-def-preserves-crossLink} of \Cref{def:local-shared-map}, \(\abs{\invEnumerateShared{T}} = d(T)\) and by construction for all \(i \in \invEnumerateShared{T}\), the subgraph \(\psi(\inv{\tilde\pi}(\widetilde{M_i}))\) gets assigned to \(T\) satisfying its demand.
Let \(b \in \boldChildren{s}\) and \(T \in \mathcal{U}_b\).
Out of \(\mathcal{F}_b\setminus \mathcal{F}^{\boundaryOp}_b\) we assign \(d(T) - \abs{\inv{\pi_b}(T) \cap \mathcal{F}^{\boundaryOp}_b}\) subgraphs to \(T\).
Now, consider any \(F \in \inv{\pi_b}(T) \cap \mathcal{F}^{\boundaryOp}_b\).
According to \Cref{item:simple-tcd-edge-partitions,item:simple-tcd-future-exact,item:simple-tcd-past} of \Cref{def:simple-tcw-dp}, there is a \(P \in \pastPart{\tau_b}\) with \(P \subseteq E(F)\).
So, there is a \(i \in \natint{w}\setminus \natint{\demandCrossLink{b}}\) with \(q_{b, i} \in V(F)\).
Consider \(\widetilde{F} \in \inv{\tilde\pi}(\widetilde{M_i})\), and notice that \(\psi(\widetilde{F})\) gets assigned to \(\pi_b(F) = T\).
As for each \(F \in \inv{\pi_b}(T) \cap \mathcal{F}^{\boundaryOp}_b\) the corresponding \(\psi(\widetilde{F})\) is different, we additionally assign \(\abs{\inv{\pi_b}(T) \cap \mathcal{F}^{\boundaryOp}_b}\) subgraphs to \(T\), satisfying its demand.
Therefore the demand of all \(\widetilde{\terminalsContainedNode{s}} \cup (\mathcal{X}\setminus \crossLink{s}) \cup \bigcup_{b\in\boldChildren{s}} \mathcal{U}_b = \mathcal{U}_s\) is satisfied.
Finally, consider \(i \in \natint{\demandCrossLink{s}}\) and let \(\widetilde{F} \coloneqq \inv{\tilde\pi}(\widetilde{M_{\localSharedMap{s}{i}}})\).
We assign \(\psi(\widetilde{F})\) to \(Q_{s,i}\) satisfying its demand.
Showing that we have a solution \((\mathcal{F}, \pi)\) for \(\mathscr{D}_{s,\tau}\).

Now, we need to show that this solution satisfies the additional requirements of \Cref{def:simple-tcw-dp}.
Recall that we already argued \Cref{item:simple-tcd-future-exact,item:simple-tcd-future-sensible}.
Note that \(E(\boundary{s}) \cap \bigcup E(\mathcal{F}) = E(\boundary{s}) \cap \bigcup E(\widetilde{\mathcal{F}})\).
Thus, by \Cref{item:simple-tcd-witness-edge-partitions} of \Cref{def:simple-tcd-witness} we have \Cref{item:simple-tcd-edge-partitions} of \Cref{def:simple-tcw-dp}.

Finally, consider \(i \in \natint{w}\) and let \(\mathcal{P}_i \coloneqq \pastAssign{\tau}{i}\).
If \(\mathcal{P}_i = \emptyset\), by \Cref{item:simple-tcd-witness-past} of \Cref{def:simple-tcd-witness} there is no \(\widetilde{F} \in \widetilde{\mathcal{F}}\) with \(q_{s,i} \in V(\widetilde{F})\).
As we do not add \(q_{s,i}\) to any solution subgraph, this shows that there is no \(F \in \mathcal{F}\) with \(q_{s,i} \in V(F)\).
Otherwise, let \(\widetilde{F} \coloneqq \inv{\tilde\pi}(\widetilde{M_{\localSharedMap{s}{i}}})\).
We have \(q_{s,i} \in \widetilde{F}\) and \(\{\psi(\widetilde{F})\} = \inv\pi(Q_{s,i})\) and note that \(\psi(\widetilde{F})\) is the unique \(F' \in \mathcal{F}\) with \(q_{s,i} \in V(F')\).

According to \Cref{item:simple-tcd-witness-past} of \Cref{def:simple-tcd-witness}, there is a partitioning of \(E(\widetilde{F} - q_{s,i})\) into \(\{\widetilde{E}_P\}_{P \in \mathcal{P}_i}\) such that for every \(P \in \mathcal{P}_i\) the graph \(\widetilde{F}[\widetilde{E}_P]\) is connected and \(\widetilde{E}_P \cap E(\boundary{s}) = P\) for all \(\tilde{p} \in \widetilde{Z} \cap V(\widetilde{F})\) there is exactly one \(P \in \mathcal{P}_i\) with \(\tilde{p} \in V(\widetilde{F}[\widetilde{E}_P])\).
Now, let \(P \in \mathcal{P}_i\) and set the partition required by \Cref{item:simple-tcd-past} of \Cref{def:simple-tcw-dp} to be \(E_P \coloneqq E(\psi(\widetilde{F}[\widetilde{E}_P]))\).
That is \(E_P = (\widetilde{E}_P \cap E(J_s)) \cup\bigcup_{b \in \boldChildren{s}, P' \in \pastPart{\tau_b}\cup\futurePart{\tau_b}\colon P' \subseteq \widetilde{E}_P} E(H_{b, P'}) = E(\psi(\widetilde{F})[E_P])\).
As \(\widetilde{E}_P \cap E(\boundary{s}) = P\) and as \(\widetilde{F}[\widetilde{E}_P]\) is connected, so is \(E_P \cap E(\boundary{s}) = P\) and \(\psi(\widetilde{F}[\widetilde{E}_P]) = \psi(\widetilde{F})[E_P]\) is connected as well, as desired.
All \(\{\widetilde{E}_P\}_{P \in \mathcal{P}_i}\) are pairwise disjoint; so, are \(\{E_P\}_{P \in \mathcal{P}_i}\) as well.
For each \(\tilde{p}_{b,P'} \in \widetilde{Z} \cap V(\widetilde{F})\) there is exactly one \(P \in \mathcal{P}_i\) with \(P' \subseteq P\).
So, \(\bigcup_{P \in \mathcal{P}_i}E_P = E(\psi(\widetilde{F} - q_{s,i})) = E(\psi(\widetilde{F}) - q_{s,i})\) and \(\{E_P\}_{P \in \mathcal{P}_i}\) is a partitioning of the edges in \(\psi(\widetilde{F}) - q_{s,i}\) and \Cref{item:simple-tcd-past} as well as \Cref{def:simple-tcw-dp} itself is satisfied.
\end{proof}

We have shown, that it is enough to check for each syntactically valid tuple \(\tau\) whether there are \((\tau_b)_{b \in \boldChildren{s}}\) such that for all \(b \in \boldChildren{s}\) we have \(\tau_b \in D(b)\) and \((\tau_b)_{b \in \boldChildren{s}}\) witnesses \(\tau \in D(s)\).
Now, we consider how to check for given \((\tau_b)_{b\in\boldChildren{s}}\) in \(\fpt\)-time whether \((\tau_b)_{b\in\boldChildren{s}}\) witness \(\tau \in D(s)\) with respect to some local to shared mappings \((\localSharedMapOp{a})_{a\in A}\) and shared mapping enumerator \(\enumerateSharedOp\).

For this, let \(C \coloneqq V(\dpTransitionGraph) \setminus \bigcup_{t \in \thinChildren{s}} X_t\).
For all \(t \in \thinChildren{s}\), we have that \(\abs{X_t} = 1\) and \(N(X_t) \subseteq X_s \cup V(\boundary{s}) \subseteq C\).
Thus, \(C\) is a vertex cover of \(\dpTransitionGraph\).
As \(C \neq \emptyset\) whenever \(V(\dpTransitionGraph) \setminus C \neq \emptyset\), it is a fracture modulator as well.
We see that \[C = V(X_s) \cup V(\boundary{s}) \cup \{q_{s,i}\}_{i \in \natint{w}} \cup \bigcup_{b \in \boldChildren{s}} \left(V(\boundary{b}) \cup \bigcup_{P \in \pastPart{\tau_b} \cup \futurePart{\tau_b}} V(\tilde{p}_{b, P})\right).\] Thus, \(\abs{C} \leq w + 2w + w + (w + 2)(2w + w) = \O{w^2}\).
As the tree-cut decomposition is simple, for all \(T \in \widetilde{\mathcal{T}}_s\) we have \(T \subseteq C\) and \(C\) is a fracture modulator for \(\augment{\widetilde{J}_s}{\widetilde{\mathcal{T}}_s}\) (i.e.,~the augmented graph of \(\mathscr{C}(s, \tau, (\tau_b)_{b \in \boldChildren{s}}, \enumerateSharedOp, (\localSharedMapOp{a})_{a\in A})\)) as well.
So, \Cref{stmt:augmented-frag-fpt} allows us to find a solution to the instance in \(\fpt\) time.
This motivates us to construct gadgets to simulate the additional restrictions given by \Cref{def:simple-tcd-witness}.

We can group the requirements of \Cref{def:simple-tcd-witness} into five categories.
\begin{itemize}
\item \Cref{item:simple-tcd-witness-edge-partitions,item:simple-tcd-witness-past} of \Cref{def:simple-tcd-witness} require some predefined edges and vertices not be used in the solution.
\item \Cref{item:simple-tcd-witness-future-sensible} of \Cref{def:simple-tcd-witness} requires some vertices not to be part of the solution subgraph assigned to specific terminal sets.
\item \Cref{item:simple-tcd-witness-future-exact,item:simple-tcd-witness-past,item:simple-tcd-witness-marked,item:simple-tcd-witness-sub-connected} of \Cref{def:simple-tcd-witness} require some vertices and edges to be part of specific solution subgraphs.
\item \Cref{item:simple-tcd-witness-past} of \Cref{def:simple-tcd-witness} requires that the solution subgraphs assigned to a specific terminal set can be partitioned into edge-disjoint connected subgraphs.
\item \Cref{item:simple-tcd-witness-marked} of \Cref{def:simple-tcd-witness} requires that adjacent to a vertex at most one edge is chosen in any solution.
\end{itemize}

Ensuring that some vertices and edges are not used in any solution subgraph is quite easy.
We can just remove them.
Ensuring that some vertices are not used in a specific solution subgraph is more difficult and we leave this for later.

For the remaining gadgets it is very useful to sub-divide the edges contained in \(\dpTransitionGraph\) with two vertices.
That is for an edge \(e \in E(\dpTransitionGraph{})\) we replace \(e\) by a \(4\)-path \(\subDivP{e}\).
Call the newly created vertices \(\subDivA{e}\) and \(\subDivB{e}\) respectively, set \(\subDivV{e} \coloneqq \{\subDivA{e}, \subDivB{e}\}\) to be the inner vertices of \(\subDivP{e}\), set \(\subDivE{e} \coloneqq \subDivA{a}\subDivB{e}\) to be the edge between \(\subDivA{e}\) and \(\subDivB{e}\), and call the modified graph \(\dpTransitionGraph{}'\).
Note that if \(\abs{C} \geq 5\), the set \(C\) stays a fracture modulator of the augmented graph.
It is easy to transfer any solution of \(\dpTransitionGraph\) to \(\dpTransitionGraph{}'\).
Now consider a solution \((\mathcal{F}', \pi')\) on \(\dpTransitionGraph{}'\), for all \(F \in \mathcal{F}'\) let \(L_{F} \coloneqq \{e \in E(\dpTransitionGraph{}[C]) \mid \subDivE{e} \in E(F)\}\) be the edges \(e \in E(\dpTransitionGraph{})\) where \(F'\) contains the middle edge of the 4-path that replaced \(e\).
Note that \(\phi(F) \coloneqq \dpTransitionGraph{}[L_{F}]\) is connected, \(V(\phi(F)) = V(F) \cap V(\widetilde{J}_{s})\), and for every edge \(e \in E(\dpTransitionGraph)\) there is at most one \(F' \in \mathcal{F}'\) with \(e \in E(\phi(F')) = L_{F'}\).
Thus, we can assign \(\phi(F)\) to \(\pi'(F')\) to obtain a solution for the instance on the host graph \(\dpTransitionGraph\).

Now, consider a set of edges \(D \subseteq E(\dpTransitionGraph{})\) and a terminal set \(T \in \widetilde{\mathcal{T}}_s\).
To ensure that there is a \(F \in \inv{\tilde\pi}(T)\) with \(D \subseteq E(F)\), we reduce the demand of \(T\) by one and increase set the demand of \(T \cup \bigcup\subDivV{D}\) to 1.
Call the modified terminal set \(\widetilde{\mathcal{T}}_s'\) and the modified demand function \(d'\).
Note that \(C \cup \aug{T \cup \bigcup\subDivV{D}}\) is a fracture modulator of the augmented graph of the modified instance.

\begin{lemma}
\label{stmt:gadget-edges-contained-in-specific-solution}
There is a solution \((\widetilde{\mathcal{F}}, \tilde\pi)\) to the instance \((\dpTransitionGraph, \widetilde{\mathcal{T}}_s, d)\) such that there is a \(F \in \inv{\tilde\pi}(T)\) with \(D \subseteq E(F)\) if and only if there is a solution \((\widetilde{\mathcal{F}}', \tilde\pi')\) to the instance \((\dpTransitionGraph{}', \widetilde{\mathcal{T}}_s', d')\).
The fracture number of \(\augment{\dpTransitionGraph{}'}{\widetilde{\mathcal{T}}_s'}\) is at most \(\abs{C} + 1\).
\end{lemma}
\begin{proof}
The solution \((\widetilde{\mathcal{F}}, \tilde\pi)\) can be transferred canonically to the instance \((\dpTransitionGraph{}', \widetilde{\mathcal{T}}_s', d')\).
Now, consider a solution \((\widetilde{\mathcal{F}}', \tilde\pi')\) to the instance \((\dpTransitionGraph{}', \widetilde{\mathcal{T}}_s', d')\).
Let \(\{F\} \coloneqq \inv{\tilde\pi}(T \cup \bigcup \subDivV{D})\), \(e \in D\), and let \(F' \in \widetilde{\mathcal{F}}' \setminus \{F\}\).
Then, \(F\) contains at-least two edges of \(E(\subDivP{e})\).
Therefore \(F'\), contains at most 1 edge of \(E(\subDivP{e})\) and as \(F'\) is connected and contains a vertex apart from \(\subDivV{e}\), the graph \(F'\) does not contain the edge \(\subDivE{e}\).
So, we assign \(\phi(F \cup \bigcup_{e\in D} \subDivE{e})\) to \(T\) and \(\phi(F')\) to \(\tilde\pi'(F')\), obtaining a solution to the original instance with the desired additional property.
\end{proof}

A similar idea also works, to ensure that for a terminal set \(T \in \widetilde{\mathcal{T}}_s\) and a set \(U\subseteq V(\dpTransitionGraph{})\), there is a \(F \in \inv{\tilde\pi}(T)\) with \(U \subseteq V(F)\) and \(D \subseteq E(F)\).
We reduce the demand of \(T\) by 1 and increase the demand of \(T \cup U\cup\bigcup_{e \in D}\subDivV{e}\) by 1.
Call the modified set of terminal sets \(\widetilde{\mathcal{T}}''\) and the demand function \(d''\).

\begin{corollary}
\label{stmt:gadget-edges-and-vertices-contained-in-specific-solution}
There is a solution \((\widetilde{\mathcal{F}}, \tilde\pi)\) to the instance \((\dpTransitionGraph, \widetilde{\mathcal{T}}_s, d)\) such that there is a \(F \in \inv{\tilde\pi}(T)\) with \(U \subseteq V(F)\) and \(D \subseteq E(F)\) if and only if there is a solution \((\widetilde{\mathcal{F}}'', \tilde\pi'')\) to the instance \((\dpTransitionGraph{}', \widetilde{\mathcal{T}}_s'', d'')\).
The fracture number of \(\augment{\dpTransitionGraph{}'}{\widetilde{\mathcal{T}}_s''}\) is at most \(\abs{C} + 1\).
\end{corollary}

To model \Cref{item:simple-tcd-witness-future-exact} of \Cref{def:simple-tcd-witness}, we first notice that according to \Cref{item:simple-tcd-witness-past} of \Cref{def:simple-tcd-witness} all edges in \(\bigcup\pastPart\tau\) are used by solution subgraphs containing a vertex of \(\{q_{s,i}\}_{i \in \natint{w}}\).
Now, consider a solution \((\widetilde{\mathcal{F}}, \widetilde{\pi})\) to the instance \((\widetilde{J}_{s},\widetilde{\mathcal{T}}_{s}, d)\) such that for all \(P \in \futurePart\tau\), there is a distinct \(F_{P}\in \inv{\widetilde{\pi}}(R_{s,P})\) with \(P\subseteq E(F_{P})\) and \(E(F_{P}) \cap E(\boundary{s}) \subseteq \bigcup\futurePart\tau\).
Then, for each \(P \in \futurePart\tau\) we notice that there is no edge in \(\bigcup\futurePart\tau\setminus P\) that is not contained in a \(\{F_{P'}\}_{P' \in \in \futurePart\tau \setminus \{P\}}\).
So, we even have \(E(F_{P}) \cap E(\boundary{s}) = P\) and \Cref{stmt:gadget-edges-and-vertices-contained-in-specific-solution} is enough to model \Cref{item:simple-tcd-witness-future-exact} of \Cref{def:simple-tcd-witness} completely.

Also \Cref{item:simple-tcd-witness-sub-connected} of \Cref{def:simple-tcd-witness} can almost be modeled using \Cref{stmt:gadget-edges-and-vertices-contained-in-specific-solution}.
But in this case we do not know to which terminal set the solution subgraph \(F \in \widetilde{\mathcal{F}}'\) with \(\tilde{p}_{b,P} \in V(F)\) and \(P \subseteq E(F)\) gets assigned.
This is needed to apply \Cref{stmt:gadget-edges-and-vertices-contained-in-specific-solution}.
To remedy this situation, we check each possible simultaneous assignment of \(\tilde{p}_{b,P}\) to a terminal set \(\pi(F)\).
There are \(\abs{\natint{u} \setminus \invEnumerateShared\blockedIndex} + \abs{\futurePart{\tau}} + \abs{\widetilde{\terminalsContainedNode{s}}} \leq u + w + 2^{w} = 2^{\O{w}}\) different terminal sets and at most \(w(w+2)\) vertices \(\tilde{p}_{b,P}\).
So, there are at most \(\left(2^{\O{w}}\right)^{w(w+2)} = 2^{\O{w^3}}\) different simultaneous assignments, each of which we can check using the gadget from \Cref{stmt:gadget-edges-and-vertices-contained-in-specific-solution}.
The fracture number of the augmented graph increases by at most \(w(w+2)\); so, it still bounded by \(\O{w^2}\) and we can decide the resulting instances in \(\fpt\)-time using \Cref{stmt:augmented-frag-fpt}.

Now, consider an \(i \in \natint{w}\) with \(\mathcal{P}_i \coloneqq \invPastAssign{\tau}{i} \neq \emptyset\).
First, we assume that \(i \in \natint{\demandCrossLink{s}}\).
This ensures that we know, which solution subgraph uses \(q_{s,i}\) in \((\widetilde{\mathcal{F}}, \tilde\pi)\).
Let \(j\coloneqq \localSharedMap{s}{i}\) and set \(\{F\}\coloneqq \inv{\tilde\pi}(\widetilde{M_{j}})\).
Notice that \(\widetilde{M_{j}} = \{q_{s,i}, m_j\} \cup (\enumerateCrossLink{s}{i} \cap X_s)\).
If \((\widetilde{\mathcal{F}}, \tilde{\pi})\) satisfies \Cref{item:simple-tcd-witness-past} of \Cref{def:simple-tcd-witness}, we can partition \(E(F - q_{s,i})\) into \(\{\widetilde{E}_P\}_{P \in \mathcal{P}_i}\) such that for all \(P \in \mathcal{P}_i\) the graph \(F[\widetilde{E}_P] = \dpTransitionGraph{}[\widetilde{E}_P]\) is connected, \(\widetilde{E}_P \cap E(\boundary{s}) = P\).

To model this, we remove \(q_{s,i}\) from \(\dpTransitionGraph{}'\) and remove the terminal set \(\widetilde{M_j}\).
Then, for all \(P \in \mathcal{P}_i\), we increase the demand of \(R_{s, P}\) by 1 and require the there is a \(F_P \in \inv{\tilde\pi'}(R_{s,P})\) with \(E(F_P) \cap E(\boundary{s}) = P\).
Additionally, we require for all \(x \in \widetilde{M_j} \setminus \{q_{s,i}\} =\{m_j\}\cup( \enumerateCrossLink{s}{i}\cap X_s)\) there is a \(P \in \mathcal{P}_i\) with \(x \in V(\dpTransitionGraph{}'[F_P])\) which we can ensure using the same approach as for \Cref{item:simple-tcd-witness-sub-connected} of \Cref{def:simple-tcd-witness}.
This approach also allows us to model \Cref{item:simple-tcd-witness-past} of \Cref{def:simple-tcd-witness} completely for \(i \in \natint{\demandCrossLink{s}}\).
For \(i \in \natint{w} \setminus \natint{\demandCrossLink{s}}\), we enumerate all choices of solution subgraphs that might contain a solution subgraph using \(q_{s,i}\) and then check each simultaneous choice using the same approach.
We have to watch out to not use any solution subgraph assigned to a \(\{R_{s,P}\}_{P \in \futurePart{\tau}}\), as this would violate \Cref{item:simple-tcd-witness-future-sensible} of \Cref{def:simple-tcd-witness}.
Given a solution to one of these this modified instance on \(\dpTransitionGraph{}'\), we can construct a solution to the original instance on \(\dpTransitionGraph\) such that \Cref{item:simple-tcd-witness-past} of \Cref{def:simple-tcd-witness} is satisfied.
Each such instance has a fracture modulator of its augmented graph of size \(\O{w^2}\) and we need to check at most \(\left(2^{\O{w}}\right)^{w - \demandCrossLink{s}} = 2^{\O{w^2}}\) such instances, which we can do in \(\fpt\)-time using \Cref{stmt:augmented-frag-fpt}.

\begin{remark}
We can model \Cref{item:simple-tcd-witness-edge-partitions,item:simple-tcd-witness-future-sensible,item:simple-tcd-witness-future-exact,item:simple-tcd-witness-past,item:simple-tcd-witness-sub-connected} completely using \(2^{\O{w^3}}\) instances of the \(\gstp\)-problem each with fracture number of the augmented graph at most \(\O{w^2}\).
\end{remark}

It remains to design a gadget that at most one edge adjacent to any \(\{m_i\}_{i \in \natint{u}\setminus\invEnumerateShared\blockedIndex}\) is used in a solution subgraph.
The easiest solution to this is to choose for each \(i \in \natint{u}\setminus \invEnumerateShared\blockedIndex\) at most one edge adjacent to \(m_i\) that should be used in the solution subgraph \(\inv{\tilde\pi}(\widetilde{M_i})\) and then remove all other edges adjacent to \(m_i\) from the graph.
Whether there is a simultaneous choice for all \(i\) can be checked using \Cref{stmt:gadget-edges-contained-in-specific-solution}.

But there is a more elegant solution that allows to simulate this restriction within in single instance.
For this, let \(i\in \natint{u}\setminus \invEnumerateShared{\blockedIndex}\) be given and order the neighbors \(N(m_i)\) arbitrarily and call them \(u_0, u_1, \dots, u_{k-1}\) where \(k + 1 \coloneqq \deg{m_i}\).
For each \(j \in \natintZ{k-1}\), let the vertex \(t_{j,0} \in \subDivV{m_iu_j}\) be adjacent to \(m_i\) and let \(t_{j,3} \in \subDivV{m_iu_j}\) be adjacent to \(u_j\).
Connect \(t_{j,0}\) to \(t_{(j+1) \xmod k,0}\) with a 4-path and call the two newly created vertices \(t_{j,1}\) and \(t_{j,2}\) respectively where \(t_{j,1}\) is adjacent to \(t_{j,0}\).
Analogously, connect \(t_{j,3}\) to \(t_{(j+1)\xmod k,3}\) calling the new vertices \(t_{j,4}\) and \(t_{j,5}\) where \(t_{j,4}\) is adjacent to \(t_{j,3}\).
Finally, connect \(t_{k-1,2}\) and \(t_{k-1,5}\) with an additional edge.
This is illustrated in \Cref{fig:gadget-at-most-one-adjacent}.
Additionally, we add for all \(j \in \natintZ{k-1}\) the terminal set \(T^i_j \coloneqq \{t_{j, \ell}\}_{\ell \in \natintZ{5}}\) with demand one.

\begin{figure}[tp]
\centering
\includegraphics[width=0.6\textwidth]{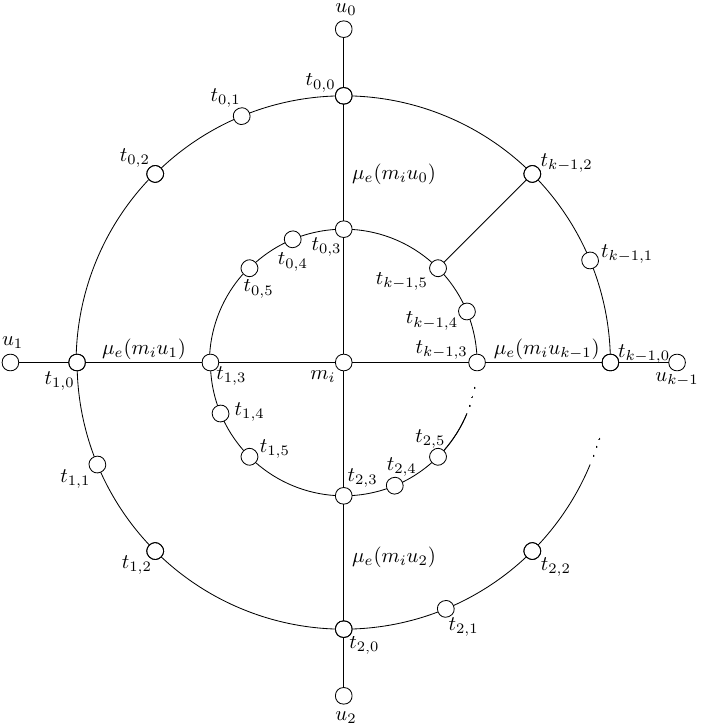}
\caption{\label{fig:gadget-at-most-one-adjacent} The gadget ensuring that at most one edge adjacent to \(m_i\) is used in the solution.}
\end{figure}

\begin{lemma}
\label{stmt:gadget-at-most-one-adjacent-correct}
The gadget depicted in \Cref{fig:gadget-at-most-one-adjacent} ensures that at most one edge adjacent to \(m_i\) is used in a solution subgraph.
The fracture number of the augmented graph containing this gadget for all \(i \in \natint{u}\setminus \invEnumerateShared{\blockedIndex}\) is at most \(\O{w^2}\).
\end{lemma}
\begin{proof}
First, note that any solution to the original instance where at most one edge adjacent to \(m_i\) is used can be transferred to a solution to the instance with the gadget.
Now, consider a solution \((\widetilde{\mathcal{F}}',\tilde\pi')\) to the instance with the gadget.
Let \(i \in \natint{u} \setminus \invEnumerateShared\blockedIndex\) and denote with \(D \coloneqq \bigcup_{j \in \natintZ{k-1}} E(\inv{\tilde\pi'}(T^i_j))\) all edges that are used to connect the new terminal sets for the gadget of this \(i\).
Set \(L \coloneqq \{\subDivE{m_iu_\ell}\}_{\ell\in\natintZ{k-1}} \cup \{t_{k-1,2}t_{k-1,5}\}\).
This set disconnects each terminal set in \(\{T^i_j\}_{j \in \natintZ{k-1}}\) and \(\abs{L} = k+1\).
So, there is at most one edge \(e^* \in L\setminus D\).
Denote with \(\widetilde{\mathcal{F}}''\) the solution subgraphs in \(\widetilde{\mathcal{F}}'\) that are not assigned to any \(\{T^i_j\}_{i \in \natint{u}\setminus\invEnumerateShared\blockedIndex, j \in \natintZ{k-1}}\).
If there is a \(\ell^* \in \natintZ{k-1}\) with \(e^* = \subDivE{m_iu_{\ell^*}}\), we claim that \(\phi(\widetilde{\mathcal{F}}'')\) with the canonical mapping is a solution to the original instance where at most one edge adjacent to \(m_i\) is used.
Otherwise, we have \(e^* = t_{k-1,2}t_{k-1,5}\).
Then, we replace the edge \(t_{k-1,2}t_{k-1,5}\) by \(\subDivE{m_iu_0}\) in the solution subgraphs in \(\widetilde{\mathcal{F}}'\).
We now again claim that after replacement \(\phi(\widetilde{\mathcal{F}}'')\) with the canonical mapping is a solution to the original instance where at most one edge adjacent to \(m_i\) is used.

First, we show that after removing the edges in \(D\) from the gadget, the vertices \(\{u_j\}_{j \in \natintZ{k-1}}\) are disconnected inside the gadget.
Showing that all \(\phi(\widetilde{\mathcal{F}}'')\setminus \phi(\inv{\tilde\pi'}(\widetilde{M_i}))\) are connected.
Consider any path \(P\) between \(u_{\ell}\) and \(u_{\ell'}\) where \(\ell \neq \ell'\) that only uses edges introduced by the gadget and \(\bigcup_{j \in \natintZ{k-1}} E(\subDivP{m_iu_j})\) such that \(E(P)\) is disjoint from \(D\).
Then, \(P\) either contains both edges adjacent to a \(p_{x, 1}\) for some \(x \in \natintZ{k-1}\) or two edge of \(L\).
The first case is not possible as at least one edge adjacent to \(p_{x,1}\) is contained in \(D\); the second case is not possible as \(\abs{L \setminus D} \leq 1\).
Additionally, all \(\phi(\widetilde{\mathcal{F}}'')\setminus \phi(\inv{\tilde\pi'}(\widetilde{M_i}))\) do not use an edge adjacent to \(m_i\).

Now, consider \(\{F\} \coloneqq \inv{\tilde\pi'}(\widetilde{M_{i}})\).
The edge \(e^*\) is reachable from all terminals in \(\widetilde{M_i}\) using only edges from \(F\).
If \(e^* = t_{k-1,2}t_{k-1,5}\), we also have that \(t_{0,0}\) or \(t_{0,3}\) as not both edges adjacent to \(t_{k-1,1}\) or \(t_{k-1,4}\) can be past of \(F\).
So, when we replace \(e^*\) by \(t_{0,0}t_{0,3}\) in \(F\), it stays connected.
Thus, \(\phi(F)\) is connected in any case.
This shows that the gadget works as expected.

Finally, consider \(C\) and notice that \(m_i\) and all \(\{u_j\}_{j \in \natintZ{k-1}}\) are contained in \(C\).
So, the connected component after removing \(C\) created by the gadget has at most \(6(k+1)=\O{\deg{m_i}}\) vertices.
We have that \(N(m_i) \subseteq C\). So, \(6(k+1) = \O{w^2}\) and there is a fracture modulator of size \(\O{w^2}\).
\end{proof}

So, given \((\tau_b)_{b\in\boldChildren{s}}\), \((\localSharedMapOp{a})_{a\in A}\), and \(\enumerateSharedOp\), we can decide by solving \(2^{\O{w^3}}\) \(\gstp\)-instances each with fracture number of the augmented graph at most \(\O{w^2}\), whether \((\tau_b)_{b\in\boldChildren{s}}\) witness \(\tau \in D(s)\) with respect to \((\localSharedMapOp{a})_{a\in A}\) and \(\enumerateSharedOp\).
As there are at most \(u^{w(w+3)} = 2^{\O{w^{2}\log w}}\) local shared mappings and \(\left(2^{w} + (w+3)w +w + 2+ 1\right)^{u} = 2^{\O{w^{3}}}\) shared mapping enumerators, we can check whether overall \((\tau_{b})_{b \in \boldChildren{s}}\) witness \(\tau \in D(s)\) by solving \(2^{\O{w^3}}\) \(\gstp\)-instances each with fracture number of the augmented graph at most \(\O{w^2}\)
According to \Cref{stmt:augmented-frag-runtime} this takes time at most \(2^{2^{\O{w^8}}}\abs{\widetilde{J}_s}\).
Combined with \Cref{stmt:simple-tcd-witness-implies-dp,stmt:simple-tcd-dp-implies-witness} this shows that the \(D(s)\) can be computed in \(\fpt\)-time given that for all bold children the dynamic program is calculated correctly.

\begin{corollary}
\label{stmt:simple-dp-propagation-time}
If for all \(b \in \boldChildren{s}\), we have computed \(D(b)\), we can compute \(D(s)\) in time \(2^{2^\O{w^8}}\abs{\widetilde{J}_s}\).
\end{corollary}

Combined with \Cref{stmt:simple-tcw-dp-root-correct}, we get that \(\gstp\) can be solved in \(\fpt\)-time given a simple tree-cut decomposition of appropriate width.

\begin{theorem}
\label{stmt:fpt-simple-tcw}
Given a simple tree-cut \((S, \mathcal{X})\) decomposition of width \(w\), we can decide whether the instance is positive in time \(2^{2^{\O{w^8}}}\abs{V(G)}\).
\end{theorem}
\begin{proof}
Given the simple tree-cut composition, we compute the dynamic program for all bold nodes and output whether the dynamic program of the root is empty.
By \Cref{stmt:simple-tcw-dp-root-correct}, this is correct.
It remains to calculate the running time.
According to \Cref{stmt:simple-dp-propagation-time}, we only need to compute upper bounds on \(\sum_{s \in S\colon s \text{ is bold}} \abs{\widetilde{J}_s}\).

Let \(s \in V(S)\).
Then,
\begin{align*}
  \abs{\widetilde{J}_s}
  \leq& \abs{G^*_s\left[X_s \cup V(\boundary{s}) \cup \{q_{s,i}\}_{i \in \natint{w}}\right]} + \sum_{c \in \children{s}} \abs{\boundary{c}}\\
  &+ \sum_{a \in A} \sum_{P \in \pastPart\tau\cup\futurePart\tau} (w + 1) + u\abs{V(J_{s}) \setminus \bigcup_{t \in \thinChildren{s}} X_{t}}\\
  \leq& (4w)^2 + 3 \abs{\thinChildren{s}} + \sum_{b \in \boldChildren{s}} 3 \adhesion{s} + (w+3)(w + 1) + u\O{w^{2}}\\
  =& \O{\abs{\thinChildren{s}} + w^4}.
\end{align*}
As every node is a child of at most one other node, we have \(\sum_{s \in S} \abs{\thinChildren{s}} \leq \abs{V(S)} \leq 2 \abs{V(G)}\); so, \(\sum_{s \in S\colon s \text{ is bold}} \abs{\widetilde{J}_s} \leq \O{w^4\abs{G}}\).
Therefore, the overall running time is as claimed.
\end{proof}

\subsection{\texorpdfstring{$\gstp$}{GSTP} is \texorpdfstring{$\fpt$}{FPT} by the Slim Tree-Cut Width}
\label{sec:orga3f51cc}
In this Section, we show that \(\gstp\) is \(\fpt\) by the slim tree-cut width of host graph.
For this, let \(\mathscr{P} \coloneqq (G, \mathcal{T}, d)\) be an instance of \(\gstp\) and let \(\mathcal{D}=(S, \mathcal{X})\) be a nice tree-cut decomposition of \(G\) of slim width \(\overline{w}\) and width \(w \leq \overline{w}\).
To achieve our result, we transform \(\mathcal{D}\) in such a way, that it becomes simple, while not increasing its tree-cut width beyond \(\overline{w} + 4\).
We do this, by adding vertices and edges such that all links are bold.
Then, we apply \Cref{stmt:fpt-simple-tcw} to decide whether the instance is positive.

Now, assume that the reduction rules presented in \Cref{sec:tcw-red-rules} are applied exhaustively to this instance.
In particular, \Cref{rr:connected-components,rr:adh-1} together ensure that there are no links with adhesion one in the tree-cut decomposition.
Which allows to bound the number of children of each node in \(S\) by a function of \(\overline{w}\).

\begin{lemma}
\label{stmt:slim-tcw-number-children}
Assume \Cref{rr:connected-components,rr:adh-1} are applied exhaustively and that all empty leaves are removed.
Let \(s \in V(S)\).
Then, \(\abs{\children{s}} + \abs{X_s} \leq \overline{w}\).
\end{lemma}
\begin{proof}
Consider the torso \(H_s\) at \(s\).
The 2-center of \(H_s\) with respect to \(X_s\) is obtained by repeatedly suppressing vertices in \(H_s\) of degree at most \(1\).
For all \(c \in \children{s}\), denote with \(z_c\in V(H_s)\) the vertex in \(H_s\) obtained by contracting the subgraph \(G[Y_c]\).
As \Cref{rr:connected-components,rr:adh-1} are applied exhaustively and all empty leaves are removed, we have \(\deg[H_s]{z_c}\geq 2\).
Thus, for any of them to be suppressed a vertex in \(V(H_s) \setminus (X_s \cup \{z_c\}_{c \in \children{s}}) = \{z_{top}\}\) has to be suppressed first.
As \(\deg[H_s]{z_{top}} = \adhesion{s} \geq 2\), this will not happen.
Therefore, \(\abs{\children{s}} + \abs{X_s} \leq \abs{\tilde H^2_s} \leq \overline{w}\).
\end{proof}

This allows us to apply \Cref{stmt:almost-simple-tcd-made-simple} to \(\mathcal{D}\) to obtain an equivalent instance with a simple tree-cut decomposition \(\mathcal{C}\).
Let \(s \in V(S)\) and consider \(\Delta_s = \abs{N_s} + \abs{\boldChildren{s}} + \abs{X_s} - w - 1\).
We have \(N_s \subseteq \thinChildren{s}\).
Thus, by \Cref{stmt:slim-tcw-number-children}, \(\Delta_s \leq \abs{\children{s}} + \abs{X_s} - w - 1 \leq \overline{w} - w - 1\) and we obtain a simple \(\mathcal{C}\) tree-cut decomposition of width \(\overline{w} + 4\).

Let \(k\) be the slim tree-cut width of \(G\).
There is an algorithm that computes a nice tree-cut decomposition of slim width \(6(k+1)^3\) in time \(2^\O{k^2\log k}\abs{V(G)}^4\)~\cite{GanianK22}.
So, we obtain a simple tree-cut decomposition of tree-cut width \(6(k+1)^3 + 4\).
Applying \Cref{stmt:fpt-simple-tcw}, we know that we can decide whether \(\mathscr{P}\) is positive in time \(2^{2^\O{k^{24}}}\abs{V(G)}\).

\begin{corollary}
\label{stmt:gstp-fpt-stcw}
Let \(k\) be the slim tree-cut width of \(G\).
We can decide whether \(\mathscr{P}\) is positive in time \(\OstarLR{2^{2^\O{k^{24}}}}\).
So, \(\gstp\) is \(\fpt\) by the slim tree-cut width of \(G\).
\end{corollary}

\subsection{\texorpdfstring{$\gstp$}{GSTP} is \texorpdfstring{$\fpt$}{FPT} by the Augmented Tree-Cut Width}
\label{sec:org3dbd85a}
In this Section, we show that GSTP is \(\fpt\) by the tree-cut width of the augmented graph.
In contrast to the last Section, we actually need to treat thin and bold links differently.
In particular, we can not assume that the number of thin children is bounded by a function of the parameter.
Therefore, we first show how to compute a friendly tree-cut decomposition from a nice tree-cut decomposition in polynomial time.
Then, we use this to show how to reduce the problem of deciding GSTP by the tree-cut width of the augmented graph to solving multiple instances of GSTP with respect to simple tree-cut decompositions.
\subsubsection{Friendly Tree-Cut Decompositions}
\label{sec:orgf3e3bf6}
Let \(G\) be a graph and \((S, \mathcal{X})\) be a nice tree-cut decomposition of width \(w\).
Ganian et~al.~\cite{GanianKS22} claimed that for all \(s \in V(S)\), we now have \(\abs{\boldChildren{s}} \leq w + 1\).
To be able to use dynamic programming on the tree-cut decomposition, this is a crucial fact.
However, this is not true and we can find counter examples showing that in fact the number of bold-children is not bounded by a function of \(w\).
So, we introduce the notion of a friendly tree-cut decomposition in \Cref{def:friendly-tcw}, where we add this as an additional requirement.

The goal of this Section is to show, that we can compute a friendly tree-cut decomposition from a nice tree-cut decomposition in FPT-linear and quartic time while not increasing its width beyond a constant.
Then, we extend this to show that we can always find a tree-cut decomposition of the same width that is friendly in polynomial time.
We achieve this by introducing the operation of \emph{blowing up} a node \(s \in S\) that reduces the number of bold children of \(s\) to at most \(w + 2 - \abs{X_s}\) while increasing the width of the tree-cut decomposition to at most 4 and keeping the tree-cut decomposition nice.
Afterwards we extend this result to not increase the tree-cut width at all.

\begin{figure}
\begin{subfigure}{0.49\textwidth}
\includegraphics[width=\textwidth]{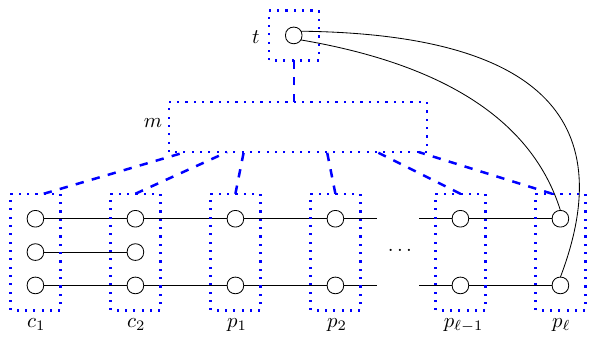}
\caption{\label{fig:unlimited-bold-children-tcw-graph} A graph family where the bags and links of the tree-cut decomposition are indicated in blue.}
\end{subfigure}
\begin{subfigure}{0.49\textwidth}
\includegraphics[width=\textwidth]{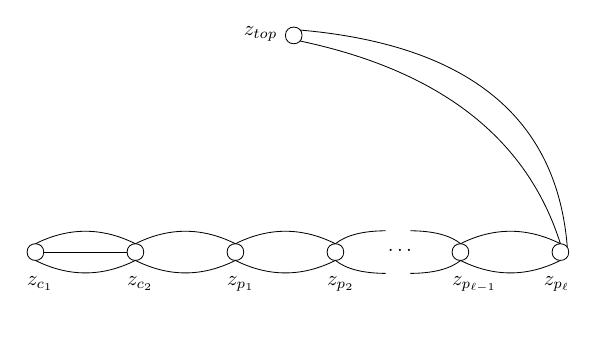}
\caption{\label{fig:unlimited-bold-children-tcw-torso} The torso at \(m\). The 3-center is the graph induced by \(z_{c_1}\) and \(z_{c_2}\).}
\end{subfigure}
\caption{\label{fig:unlimited-bold-children-tcw} A family of graphs with tree-cut width at most 5. In the depicted nice tree-cut decomposition the node \(m\) has \(\ell + 2\) bold children, where \(\ell\) can be chosen freely.}
\end{figure}

First, consider the family of graphs depicted in \Cref{fig:unlimited-bold-children-tcw}.
In \Cref{fig:unlimited-bold-children-tcw-graph}, we see a nice tree-cut decomposition for the graphs in this family.
Note that \(X_m = \emptyset\) and that \(m\) has \(\ell + 2\) bold children.
The adhesion of each node is bounded by \(5\) and the size of the 3-center of the torso of any node apart from \(m\) is bounded by 4.
Now consider the torso \(H_m\) at \(m\), which is depicted in \Cref{fig:unlimited-bold-children-tcw-torso}.
For each \(c \in \children{m}\), we denote with \(z_c\) the vertex to which the sub-tree rooted at \(c\) is contracted and with \(z_{top}\) the vertex to which \(G - Y_m\) was contracted.
As \(\deg{z_{top}} = 2\), the vertex \(z_{top}\) gets suppressed, which introduces a loop at \(z_{p_\ell}\).
Loops get removed, so now \(\deg{z_{p_{\ell}}} = 2\) and \(z_{p_{\ell}}\) gets suppressed, reducing the degree of \(z_{p_{\ell - 1}}\) to 2.
This continues until only \(z_{c_1}\) and \(z_{c_2}\) remain.
Therefore, the 3-center of the torso at \(m\) has 2 vertices and the width of this tree-cut decomposition is 5.

\begin{remark}
\label{stmt:nice-tree-cut-unlimited-bold}
For each \(\ell \in \N\), there is a graph \(G\) and nice tree-cut decomposition \((S, \mathcal{X})\) of width at most 5 such that there is a node \(m \in S\) with \(\abs{\boldChildren{m}} \geq \ell\).
\end{remark}

To remove these additional bold children, we aim to understand their structure better.
Let \(b \in \boldChildren{s}\) such that \(z_b\) is not part of the 3-center of the torso at \(s\).
We call such a \(b\) \emph{fake}.
We now show, that these fake nodes more or less have to look like the fake vertices in \Cref{fig:unlimited-bold-children-tcw-torso}.
That is, they form a narrow path from \(z_{top}\) to the 3-center.

\begin{lemma}
\label{stmt:tcw-fake-nodes-form-path}
Let \(\ell\) be the number of fake nodes in \(\boldChildren{s}\).
If \(\ell \geq 1\), there is an induced path \(P = z_{top}p_1p_2\dots p_{\ell}\) in \(H_s\) such that
\begin{itemize}
\item \(\{p_i\}_{i \in \natint{\ell}} = \{z_f\}_{f \in \boldChildren{f}; f \text{ is fake}}\),
\item the multiplicity of each edge in \(P\) is at most 2,
\item there is a suppression sequence that first suppresses all thin nodes, then \(z_{top}\), and finally all \(p_1,p_2, \dots,p_\ell\) in order,
\item for all \(i \in \natint{\ell - 1}\), we have \(N(p_i) \subseteq V(P)\), and when \(p_i\) gets suppressed its only neighbor is \(p_{i+1}\).\qedhere
\end{itemize}
\end{lemma}
\begin{proof}
Assume \(\ell \geq 1\) and consider the graph \(H^*\) obtained from \(H_s\) after suppressing all vertices in \(\{z_t\}_{t\in \thinChildren{s}}\).
Note that for all \(b \in \boldChildren{s}\), the vertex \(z_b\) is not adjacent to any \(\{z_t\}_{t\in \thinChildren{s}}\) and we have \(\deg[H^*]{z_b} = \deg[H_s]{z_b}\).
Fix the fake nodes \(f_1, f_2, \dots, f_\ell \in \boldChildren{s}\) such that they get suppressed in this order and set \(Q = q_0q_1q_2\dots q_\ell \coloneqq z_{top}z_{f_1}z_{f_2}\dots z_{f_\ell}\).
To show that \(Q\) is a valid assignment for \(P\), we inductively prove that for all \(i \in \natintZ{\ell}\)
\begin{itemize}
\item let \(H_{(i)}\) be the graph obtained from \(H^*_s\) after suppressing all vertices of \(\bigcup_{j \in \natintZ{i - 1}}q_j\), then we have \(H_{(i)} = H^* - (\bigcup_{j \in \natintZ{i - 1}}q_j)\),
\item \(\deg[H_{(i)}]{q_i} \leq 2\),
\item if \(i < \ell\), we have \(N_{H_{(i)}}(q_i) = N_{H^*}(\{q_j\}_{j \in \natintZ{i}}) = \{q_{i+1}\}\).
\end{itemize}
For all \(i \in \natint\ell\) the edges incident to \(z_{f_i}\) are equal in \(H^*\) and \(H_s\); so, the claim follows.

As noted above, for all \(b \in \boldChildren{s}\), we have \(\deg[H^*]{z_{b}}=\deg[H_s]{z_{b}} \geq 3\).
In particular, another vertex needs to be suppressed before \(z_{f_1}\) can be suppressed.
As no vertex of \(X_s\) or \(\{z_b\}_{b \in \boldChildren{s}}\) can be suppressed in \(H^*\), the only choice for this is \(z_{top} = q_0\).
If \(\deg[H^*]{z_{top}}\geq 3\), it will not be suppressed.
Since \(z_{f_1}\) gets suppressed, we have \(\deg[H^*]{z_{top}} = \deg[H_{(0)}]{q_0}\leq 2\).
Note that only degrees of adjacent vertices may change upon suppression.
Thus, \(\abs{N_{H^*}(z_{top})} \geq 1\).
If \(\abs{N_{H^*}(z_{top})} \geq 2\), suppressing \(z_{top}\) does not decrease the degree of the adjacent vertices as direct edges get inserted between them.
So, \(\abs{N_{H^*}(z_{top})} = 1\).
More concretely, we have \(N_{H^*}(z_{top}) = \{z_{f_1}\} = \{q_1\}\) as otherwise the degree of \(z_{f_1}\) after suppressing \(z_{top}\) is still at least 3, so it can not be suppressed next, violating our assumption.

Now, let \(i \in \natint{\ell}\) and assume that the claim holds for \(i-1\).
So, we can suppress \(q_{i-1}\) in \(H_{(i-1)}\).
Since \(\abs{N_{H_{(i-1)}}(q_{i-1})} = 1\), this equates to deleting \(q_{i-1}\); so, we have \(H_{(i)} = H_{(i-1)} - q_{i-1} = H^* - (\bigcup_{j \in \natintZ{i - 1}}q_j)\).
By assumption, there is a suppression sequence suppressing \(z_{f_i} = q_i\) next.
Thus, \(\deg[H_{(i)}]{q_i} \leq 2\).
Now assume \(i < \ell\).
By the induction hypothesis, the vertices \(\{q_j\}_{j\in\natintZ{i-1}}\) are not adjacent to any \(\{q_{j'}\}_{j' \in \natintZ{\ell} \setminus \natintZ{i}}\) in \(H^*\).
Thus, for all \(j' \in \natintZ{\ell} \setminus \natintZ{i}\), we have \(\deg[H_{(i)}]{q_{j'}} = \deg[H^*]{q_{j'}}\geq 3\) and in particular \(\deg[H_{(i)}]{q_{i+1}}\geq 3\).
For \(q_{i+1}\) to be suppressible after suppressing \(q_i\), its degree must decrease, implying \(N_{H_{(i)}}(q_i) = \{q_{i+1}\}\).
Combined with \(N_{H^*}(\{q_j\}_{j \in \natintZ{i-1}}) = \{q_{i}\}\), we obtain \(N_{H^*}(\{q_j\}_{j \in \natintZ{i}}) = \{q_{i + 1}\}\), concluding the proof.
\end{proof}

We now describe the operation of \emph{expanding} a node \(s \in V(S)\).
Recall that \(H_s\) denotes the torso at \(s\).
Let \(a,b \in \boldChildren{s}\) be fake and denote with \(m\) the multiplicity of the edge \(z_az_b\) in \(H_s\), that is the number of edges between \(Y_a\) and \(Y_b\).
If \(m \geq 1\) and \(3 \leq \adhesion{a} + \adhesion{b} - 2m\), we can expand the node \(s\) with respect to \(a\) and \(b\) by introducing a new node \(c\) associated with an empty bag as a child of \(s\) and moving \(a\) and \(b\) to be children of \(c\).
Expanding \(s\) until this operation is not applicable any more is called \emph{blowing up} \(s\).
We now show, that these operations keep \(S\) nice while not increasing the width of the decomposition beyond 4.

\begin{lemma}
\label{stmt:tcw-expanding-upper-nice-and-width}
Let \((S', \mathcal{X}')\) be the tree-cut decomposition after expanding \(s \in S\) with respect to \(a\) and \(b\).
Then, \((S',\mathcal{X}')\) is nice and its width is at most \(\max(4,w)\).
\end{lemma}
\begin{proof}
For all nodes \(x \in V(S) \setminus\{s\}\), the children did not change and so the niceness property holds.
To show the niceness property for \(s\), we show that \(c\) is a bold child, where \(c\) is the newly introduced node.
Notice that \(3 \leq \adhesion{a} + \adhesion{b} - 2m = \abs{\cutEdges{Y_a}} + \abs{\cutEdges{Y_b}} - 2m\).
As \(m\) is equal to the number of edges between \(Y_a\) and \(Y_b\) and since \(Y_c = Y_a \cup Y_b\), we have \(\abs{\cutEdges{Y_a}} + \abs{\cutEdges{Y_b}} - 2m = \abs{\cutEdges{Y_a \cup Y_b}} = \adhesion{c}\).
Therefore, \(c\) is a bold child of \(s\) and the niceness property at \(s\) is satisfied as well.
As \(a\) and \(b\) are bold, the niceness at \(c\) is also satisfied.

Now, we show that the width does not increase beyond \(\max(4,w)\) by considering adhesion and the 3-center of the torsos separately.
For each \(x \in V(S)\) the set \(Y_x\) does not change and so \(\adhesion{x}\) does not change.
Notice that \(\deg[H_s]{z_a} = \adhesion{a}\) and \(\deg[H_s]{z_b} = \adhesion{b}\).
So, we know from \Cref{stmt:tcw-fake-nodes-form-path} that \(\adhesion{a},\adhesion{b} \in \{3,4\}\).
If \(\adhesion{a} =\adhesion{b} = 3\), as \(\adhesion{c} = \adhesion{a} + \adhesion{b} - 2m\), we have \(\adhesion{c}\leq 4\).
Otherwise at least one of \(a\) or \(b\) has adhesion \(4\), implying \(m = 2\).
In this case, \(\adhesion{c} \leq 8 - 4 \leq 4\) as well.

For all \(x \in V(S) \setminus \{s\}\) the torso does not change and the torso at \(c\) contains at most \(3\leq 4\) vertices.
Finally, consider the torso at \(s\) in \(S\) and \(S'\), which we denote with \(H\) and \(H'\) respectively.
Observe that \(H' = H / z_az_b\), where \(z_c\) is the vertex in \(H'\) to which \(z_a\) and \(z_b\) get contracted.
Let \(P \coloneqq z_{top}p_1p_2\dots p_\ell\) be the path of vertices associated with the fake nodes in \(H\) obtained from \Cref{stmt:tcw-fake-nodes-form-path} and let \(i \in \natint\ell\) be such that, without loss of generality, \(p_i = z_a\) and \(p_{i+1} = z_b\).
Observe that in \(H'\) we can also suppress the vertices associated with thin nodes, then \(z_{top}\), and then \(p_1, p_2, \dots, p_{i-2}\) in this order as these nodes are not adjacent to \(z_a\) or \(z_b\) in \(H\).
Call this graph \(H_{(i-1)}'\) and the graph after suppressing the same vertices in \(H\) by \(H_{(i-1)}\).
As nodes adjacent to either \(z_a\) or \(z_b\) are not suppressed yet, \(H_{(i-1)}' = H_{(i-1)}/z_az_b\) and as \(N_{H_{(i-1)}}(p_{i-1}) = \{p_i\} = \{z_a\}\), we have \(N_{H_{(i-1)}'}(p_{i-1}) = \{z_c\}\).
Additionally, the multiplicity of \(p_{i-1}z_c\) is at most 2 in \(H_{(i-1)}'\) as well.
So, \(p_{i-1}\) is suppressible in \(H_{(i-1)}'\), which equates to deleting \(p_{i-1}\).
The same is true for \(p_{i-1}\) in \(H_{(i-1)}\), where we can then suppress \(p_i = z_a\), which also equates to deleting it.
Thus, the graph obtained from \(H_{(i-1)}'\) after suppressing \(p_{i-1}\) is \(H_{(i-1)}' - p_{i-1}=H_{(i-1)}/z_az_b - p_{i-1}\).
The graph obtained from \(H_{(i-1)}\) after suppressing \(p_{i-1}\) and \(z_a\) is \(H_{(i-1)} - p_{i-1} - z_a\).
As \(N_{H_{(i-1)}}(z_a) = \{p_{i-1},z_b\}\), these two graphs are equal.
Thus, there are suppression sequences for \(H\) and \(H'\) arriving at the same graph; so, their 3-centers are equal.
\end{proof}

Now, we show that after blowing up \(s\), the number of its bold children is bounded by \(w+2 - \abs{X_s}\).

\begin{lemma}
\label{stmt:tcw-blowing-up-number-children}
Denote with \((S', \mathcal{X'})\) the tree-cut decomposition after blowing up \(s\).
Then, \((S', \mathcal{X}')\) is nice, has width at most \(\max(4,w)\), and \(\abs{\boldChildren[S']{s}} + \abs{X_s} \leq w+2\).
\end{lemma}
\begin{proof}
Combined with \Cref{stmt:tcw-expanding-upper-nice-and-width}, we only need to show, that whenever \(\abs{\boldChildren[S]{s}} +\abs{X_s} \geq w+3\), there are two \(a,b \in \boldChildren[S]{s}\) with respect to which we can expand \(s\).
Let \(P = z_{top}z_{c_1}z_{c_2}\dots z_{c_\ell}\) be the path of vertices associated with the fake nodes in \(H\) obtained from \Cref{stmt:tcw-fake-nodes-form-path}.
As \(\abs{\boldChildren[S]{s}} + \abs{X_s} \geq w+3\), we have \(\ell \geq 3\).
For \(i \in \natint{\ell-1}_0\), denote with \(m_i\) the multiplicity of the edge \(z_{c_i}z_{c_{i+1}}\).
If \(m_1 = 1\), we claim that \(s\) can be expanded with respect to the nodes \(c_1\) and \(c_2\); otherwise, with respect to the nodes \(c_2\) and \(c_3\).
Assume, that \(m_1 = 1\).
Then, \(\adhesion{c_1} = \adhesion{c_2} = 3\) and \(\adhesion{c_1} + \adhesion{c_2} - 2m_1 = 4 \geq 3\).
If \(m_1 \neq 1\), we have \(m_1 = 2\) and \(\adhesion{c_2} + \adhesion{c_3} - 2m_2 = m_1 + \adhesion{c_3} - m_2 \geq 2 + 1 = 3\), meaning that \(s\) can be expanded with respect to the nodes associated with \(p_2\) and \(p_3\).
\end{proof}

Finally, we note that the operation of blowing up \(s\) can be implemented in linear time in size of the respective torso.
Additionally, the newly introduced nodes do not need to be altered again to make the tree-cut decomposition friendly.
These facts allow us to design an algorithm running in \(\fpt\)-linear time.
In fact, this algorithm always runs in at most quartic time in \(\abs{V(G)}\).

\begin{theorem}
\label{stmt:tcw-blowing-up-running-time}
Given a nice tree-cut decomposition \((S, \mathcal{D})\) of width \(w\), we can compute a friendly tree-cut decomposition of width at most \(\max(4,w)\) in time \(\O{w\abs{G} + \abs{S}}\).
\end{theorem}
\begin{proof}
In time \(\O{\abs{S}}\), we can reduce the number of nodes in \(\abs{S}\) to \(2\abs{V(G)}\) while not increasing its width~\cite{GanianKS22}.
We assume, without loss of generality, that this operation has already been applied.

We can compute all torsos in time \(\O{\abs{G} + \sum_{s\in V(S)} \adhesion{s}} = \O{w \abs{G}}\).
For every \(s \in V(S)\), we can now compute the 3-center and hence the fake bold children in time \(\O{\abs{H_s}}\).
As expanding \(s\) with respect to any two nodes only result in these nodes being contracted and by \Cref{stmt:tcw-fake-nodes-form-path} these nodes have degree at most 4, we can blowup \(s\) in time \(\O{\ell_s} = \O{\abs{H_s}}\) as well, where \(\ell_s\) denotes the number of fake bold children at \(s\).
The newly introduced nodes all have exactly 2 bold children and so do not need to be blown up to make the tree-cut decomposition friendly.
Therefore, the running time to make the tree-cut decomposition friendly is \(\O{\sum_{s \in V(S)} \abs{H_s}} = \O{\sum_{s \in V(S)} \abs{V(H_s)} + \abs{E(H_s)}}= \O{\abs{V(S)} + \sum_{s \in V(S)} \abs{E(H_s)}}\).

Consider \(s \in V(S)\) and denote with \(r_s \coloneqq \abs{\boldChildren{s}} - \ell_s\) the number of non-fake bold children of \(s\).
As all thin children have degree at most 2 and all vertices associated with fake node children have degree at most 4.
Thus, \(\abs{E(H_s)} \leq 4\abs{\children{s}} + (1 + \abs{X_s} + r_s)^2\) and \(\sum_{s \in V(S)} \abs{E(H_s)} \leq 4\abs{V(S)} + \sum_{s \in V(S)}(1+\abs{X_s} + r_s)^2\).
As \((1+\abs{X_s} + r_s)^2\) is convex and since for all \(s \in V(S)\), we have \(\abs{X_s} + r_s \leq w\), we conclude that for fixed \(\sum_{s \in V(S)} \abs{X_s} + r_s\) the value of \(\sum_{s \in V(S)}(1+\abs{X_s} + r_s)^2\) is maximized, when the number of \(s\) with \(\abs{X_s} + r_s = w\) is maximized.
As \(\sum_{s \in V(S)} \abs{X_s} + r_s = \O{\abs{V(G)}}\), we conclude that \(\sum_{s \in V(S)}(1+\abs{X_s} + r_s)^2 \leq \O{\frac{\abs{V(G)}}{w} w^2 + \abs{V(G)}} = \O{w\abs{V(G)}}\).
\end{proof}

When creating \(\fpt\) algorithms parameterized by tree-cut width, we can from now on assume that the number of bold children is bounded by \(w+2\), or more concretely, that \(\abs{X_s} + \abs{\boldChildren{s}} \leq w+2\).
If, however, the algorithms only work for constant tree-cut width or more concretely up to tree-cut width at most 3, we can not make this simplifying assumption based on the arguments presented above.

We now sketch how to close this gap for tree-cut decompositions of width exactly 3.
For these, we expand a node \(s\) with respect to fake nodes \(a,b\in \boldChildren{s}\) if the multiplicity of the edge \(z_az_b\) is exactly 2.
One can verify, analogously to \Cref{stmt:tcw-expanding-upper-nice-and-width,stmt:tcw-blowing-up-number-children}, that such a pair of fake nodes always exists if \(\abs{\boldChildren{s}} + \abs{X_s} \geq w+3\), but that contrary to the result of \Cref{stmt:tcw-expanding-upper-nice-and-width} the width of the obtained tree-cut decomposition is bounded by 3, but the newly introduced node turns out to be a thin child of \(s\) and so the tree-cut decomposition is not necessarily nice anymore.
Therefore, we can not use this operation as a post-processing step after having obtained a nice tree-cut decomposition.

To fix this, we take a closer look at the algorithm provided by Ganian et~al.~\cite{GanianKS22} to obtain a nice tree-cut decomposition from an arbitrary tree-cut decomposition.
Their algorithm works by repeatedly considering the thin node \(p\) with parent \(s\) of minimum depth that is bad (i.e.,~there is a \(q \in \children{s}\) with \(N(p) \cap Y_q \neq \emptyset\)).
Then, \(p\) gets moved to be included in the sub-tree rooted at \(q\).
The position to which \(p\) gets moved is chosen in such a way as to ensure that the width of the tree-cut decomposition does not increase and that after at most \(2\abs{V(S)}\) many moving operations, the tree-cut decomposition is nice.
Notice that whenever at depth \(t\) there is no bad node anymore, this fact is maintained throughout the algorithm.

The operation of expanding a node in a tree-cut decomposition of width 3 only introduces a bad node at a larger depth than at the expanded node.
Therefore, we consider the tree-cut decomposition depth layer by depth layer.
If at a depth layer, there is a bad node \(p\), we apply the moving procedure described by Ganian et~al.~\cite{GanianKS22}.
If at a depth layer, there is no bad node, but a node \(s\) with \(\abs{X_s} + \abs{\boldChildren{s}} \geq w+3\), we blow up \(s\) using the modified expansion operation.
This ensures that after a layer is processed, there are no bad nodes on this or a lower layer and all such nodes \(s\) have \(\abs{\boldChildren{s}} + \abs{X_s} \leq w+2\).
Additionally, per depth layer we apply at most \(3\abs{V(S')}\) many operations, where \(S'\) denotes the tree-cut decomposition when beginning to process the respective depth layer.
As all internal empty nodes have at least two children and we remove all empty leaves, the number of nodes in \(S\) is bounded by \(2\abs{V(G)}\)~\cite{GanianKS22}.
Thus, we do at most \(6\abs{V(G)}^2\) operations.

\begin{corollary}
\label{stmt:tcw-blowing-up-general}
Given a nice tree-cut decomposition of width \(w\), we can compute a friendly tree-cut decomposition of width at most \(w\) in polynomial time.
\end{corollary}

\subsubsection{Reducing \texorpdfstring{$\gstp$}{GSTP} by Augmented Tree-Cut Width to \texorpdfstring{$\gstp$}{GSTP} with a Simple Tree-Cut Decomposition}
\label{sec:orgd9164ac}
In this Section, we present some reduction rules that, taken together, are enough to remove all nodes in the tree-cut decomposition that violate the condition for a simple tree-cut decomposition.
Some of the reduction rules have as a precondition that a specific sub-instance is positive.
So, they are not reduction rules in the classical sense, as they can not necessarily be applied in polynomial but rather only in \(\fpt\)-time.

For this Section, assume that \(\mathcal{D} \coloneqq (S, \mathcal{X})\) is a friendly tree-cut decomposition of width \(w\) for \(\augment{G}{\mathcal{T}}\).
Set \(X_s' \coloneqq X_s \cap V(G)\) and let \(\mathcal{X}' \coloneqq \{X'_s\}_{s\in V(S)}\) for all \(s \in S\).
Then, \(\mathcal{D}' \coloneqq (S, \mathcal{X}')\) is a tree-cut decomposition of \(G\) with width at most \(w\).
When we refer to a reduction rule presented in \Cref{sec:tcw-red-rules}, namely \Cref{rr:connected-components,rr:cross-link-demand-large,rr:adh-1}, we mean that they are applied with respect to \(\mathcal{D}'\) while keeping \(\mathcal{D}\) in sync.
These reduction rules already bring us quite close to \(\mathcal{D}'\) being simple with respect to the nodes that are thin in \(\mathcal{D}\).

\begin{lemma}\label{stmt:general-tcw-rr-almost-simple}
After exhaustively applying \Cref{rr:sensible-terminal-sets,rr:connected-components,rr:adh-1} with respect to \(\mathcal{D}'\) and removing nodes that are empty in \(\mathcal{D}\) and \(\mathcal{D}'\), we have for all \(s \in V(S) \setminus \{r\}\) that are thin in \(\mathcal{D}\) that \(\adhesion[\mathcal{D}]{s} = 2\) and \(\crossLink[\mathcal{D'}]{s} = \emptyset\).
In particular, \(\cutEdges[\augment{G}{\mathcal{T}}]{Y^\mathcal{D}_s}\) does not contain an augmented and a non-augmented edge.
Additionally, the tree-cut decompositions can be maintained efficiently while not increasing the width of \(\mathcal{D}\) and keeping \(\mathcal{D}\) friendly if it was friendly before.
\end{lemma}
\begin{proof}
Let \(s \in V(S) \setminus \{r\}\) be thin in \(\mathcal{D}\).
Assume \(\adhesion[\mathcal{D}]{s} < 2\).
If \(\adhesion[\mathcal{D}]{s} = 0\), either \(Y^\mathcal{D}_s\supseteq Y^{\mathcal{D}'}_s\) is empty, which means that this node was removed, or \(G\) is disconnected, which means that \Cref{rr:connected-components} would split this instance.
New let \(\adhesion[\mathcal{D}]s = 1\) and set \(\{uv\} = \cutEdges{Y^\mathcal{D}_s}\).
If \(uv \in E(G)\), then \(\adhesion[\mathcal{D}']{s} = 1\) and \Cref{rr:adh-1} would remove \(uv\).
So, \(uv \notin E(G)\) and \(uv\) is an augmented edge.
Let \(T\) be the terminal set inducing \(uv\).
As \(\cutEdges[G]{Y_s^\mathcal{D'}}\) is empty, either \(Y^{\mathcal{D}'}_s\) or \(V(G)\setminus Y^\mathcal{D'}_s\) is empty or \(G\) would be disconnected, which would violate that \Cref{rr:connected-components} has been applied exhaustively.
Thus, \(\abs{T} \leq 1\) and \(T\) should have been removed by \Cref{rr:sensible-terminal-sets} and we have \(\adhesion[\mathcal{D}]{s} = 2\).

Now, assume that there is a \(T \in \crossLink[\mathcal{D}']{s}\).
This means that at least one edge of \(\cutEdges[\augment{G}{\mathcal{T}}]{Y^\mathcal{D}_s}\) is augmented.
If both are augmented, either \(Y^\mathcal{D'}_s\) or \(V(G) \setminus Y^\mathcal{D'}_s\) is empty.
As \(T \subseteq V(G)\), this means that either \(T \cap Y_s^{\mathcal{D}} \neq \emptyset\) or \(T \setminus Y_s^{\mathcal{D}} \neq \emptyset\) is empty, violating our assumption.
Thus, exactly one edge of \(\cutEdges[\augment{G}{\mathcal{T}}]{Y^\mathcal{D}_s}\) is augmented and so \(\adhesion[\mathcal{D'}]{s} = 1\), which means that \Cref{rr:adh-1} was applicable.

The only reduction rule, that might be hard to efficiently maintain and might increase the width of \(\mathcal{D}\) is \Cref{rr:adh-1} in the case where the terminal sets are modified.
This case only applies when \(\adhesion[\mathcal{D}']{s} = 1\) and there is a \(T \in \mathcal{T}\) with \(T \cap Y_s^{\mathcal{D}} \neq \emptyset\) and \(T \setminus Y_s^{\mathcal{D}} \neq \emptyset\).
Thus, \(\adhesion[\mathcal{D}]{s} = 2\).
Let \(\{uv\} \coloneqq \cutEdges[G]{Y^\mathcal{D'}_s}\) and \(\{uv, xy\} \coloneqq \cutEdges[\augment{G}{\mathcal{T}}]{Y_s^\mathcal{D}}\) and, without loss of generality, assume \(u \in Y^\mathcal{D'}_s\), \(v \notin Y^\mathcal{D'}_s\), and \(x = \aug{T}\).
Additionally, assume \(x \in Y^\mathcal{D}_s\), the case \(x \notin Y^\mathcal{D}_s\) follows analogously.
The augmented graph after removing \(uv\) and \(T\) is \(G^* \coloneqq \augment{G}{\mathcal{T}} - uv - xy\).

Let \(A\) be the component of \(G^*\) containing \(u\).
If before the reduction rule applied, we already had \((T \cap Y^\mathcal{D'}_s) \cup \{u\}\) as a terminal set, \(A = \augment{G}{\mathcal{T}}[Y^\mathcal{D}_s]\); so, a tree-cut decomposition of appropriate width and friendliness for this decomposition is \((S_s, \{X_{s'} \setminus \{u\} \mid s' \in S_s\})\).
Now, assume that the terminal set \((T \cap Y^\mathcal{D'}_s) \cup \{u\}\) was newly introduced.
Then, \(A = \augment{G}{\mathcal{T}}[Y^\mathcal{D}_s] + xu\), where we identify \(u\) with the augmented vertex of \((T \cap Y^\mathcal{D'}_s) \cup \{u\}\).
Consider the tree-cut decomposition \(\mathcal{D}_A \coloneqq (S_s, \{X_{s'}\}_{s' \in S_s})\).
Note that \(\mathcal{D}_A\) is friendly if \(\mathcal{D}\) is friendly.
Let \(s \in V(S_s) \setminus \{s\}\).
We have \(xy \in\cutEdges[\augment{G}{\mathcal{T}}]{Y^\mathcal{D}_s}\) if and only if \(xu \in\cutEdges[A]{Y^{\mathcal{D}_A}_s}\).
So, \(\adhesion[\mathcal{D}]{s} = \adhesion[\mathcal{D}_A]{s}\).
Additionally, the torso at \(s\) does not change between \(\mathcal{D}\) and \(\mathcal{D}_A\).
Lastly, denote with \(H_s\) and \(H'_s\) the torsos at \(s\) with respect to \(\mathcal{D}\) and \(\mathcal{D}_A\), respectively.
If \(x \in X_s\), \(H'_s = H_s - z_{top} + xu\), so the 3-center does not increase as we only introduce an edge between center vertices.
If \(x \notin X_s\), denote with \(z\) the vertex of \(H_s\) whose associated sub-tree contains \(x\).
Then, \(H'_s = H_s - z_{top} + uz\).
As \(\deg[H_s]{z_{top}} = \adhesion[\mathcal{D}]{s} = 2\), the vertex \(z_{top}\) is suppressible.
We have \(N_{H_s}(z_{top}) = \{u, z\}\).
Thus, \(H'_s\) is equal to \(H_s\) after suppressing \(z_{top}\) and their 3-centers are equal.

Let \(B\) be the component of \(G^*\) containing \(v\).
If before the reduction rule applied, we already had \((T \setminus Y^\mathcal{D'}_s) \cup \{v\}\) as a terminal set, \(B = \augment{G}{\mathcal{T}}[V(\augment{G}{\mathcal{T}}) \setminus Y^\mathcal{D}_s]\); so, a tree-cut decomposition of appropriate width and friendliness for this decomposition is \((S - S_s, \{X_{s'}\}_{s' \in V(S) - V(S_s)})\).
Now, assume that the terminal set \((T \setminus Y^\mathcal{D'}_s) \cup \{v\}\) was newly introduced.
Then, \(B = \augment{G}{\mathcal{T}} / Y^\mathcal{D}_s\) and let \(t\) be the augmented vertex of \((T \setminus Y^\mathcal{D'}_s) \cup \{v\}\).
Denote with \(p\) the parent of \(s\) in \(S\) and consider the tree-cut decomposition \(\mathcal{D}_B\) induced by \(S - S_s\), where we add a node \(s'\) with \(X_{s'} \coloneqq \{t\}\) as a child of \(p\).
Note, that \(\mathcal{D}_B\) is friendly if \(\mathcal{D}\) was friendly and that the 3-center of the torso at \(s'\) is empty.
As above, the adhesion does not increase and the remaining torsos also do not change between \(\mathcal{D}\) and \(\mathcal{D}_B\).
To obtain a tree-cut decomposition for the whole graph, set \(\mathcal{D}_A\) as a child of the root of \(\mathcal{D}_B\).
\end{proof}

There are two task remaining for this Section.
First, we need to ensure for all thin nodes \(s \in V(S) \setminus \{r\}\) that \(\abs{Y^\mathcal{D'}_s} \leq 1\).
We call nodes \(s \in V(S) \setminus \{r\}\) that are thin in \(\mathcal{D}\), but have \(\abs{Y^\mathcal{D'}_s} \geq 2\) \emph{cluttered}.
Second, we need to take care of all nodes that are bold in \(\mathcal{D}\), but thin in \(\mathcal{D}'\).
The second task can be taken care of rather quickly by applying \Cref{stmt:almost-simple-tcd-made-simple}.
As \(\mathcal{D}\) is friendly, for each \(s\in V(S)\) the number of such children is bounded by the number of bold children, or more concretely, \(w + 2 - \abs{X_s}\).
Thus, we can just treat them like bold children.

\begin{lemma}\label{stmt:no-cluttered}
Assume \(\mathcal{D}\) is friendly, has no cluttered nodes, and that \Cref{rr:sensible-terminal-sets,rr:connected-components,rr:adh-1} were applied exhaustively.
We can compute in linear time an equivalent instance \((G', \mathcal{T}, d)\) and a simple tree-cut decomposition \(\mathcal{C}\) of \(G'\) of width at most \(w+5\).
\end{lemma}
\begin{proof}
To obtain \(G'\) and \(\mathcal{C}\), we apply \Cref{stmt:almost-simple-tcd-made-simple} to \(G\) and \(\mathcal{D}'\), where we remove all empty leaves.
It remains to bound the tree-cut width of \(\mathcal{C}\).
Let \(s \in V(S)\), we now need to bound \(\Delta_s = \abs{N_s} + \abs{\boldChildren[\mathcal{D}']{s}} + \abs{X_s} - w - 1\).
For this, we show that any \(c \in \thinChildren[\mathcal{D}]{s}\) is simple in \(\mathcal{D}'\).
As \(c\) is not cluttered, we have \(\abs{Y^{\mathcal{D}'}} \leq 1\) and since we remove empty leaves, we even have \(\abs{Y^{\mathcal{D}'}_s} = 1\).
By \Cref{stmt:general-tcw-rr-almost-simple}, we have that \(\adhesion{c} =2\) and \(\crossLink{c} = \emptyset\), yielding that \(c\) is simple.
Thus, \(N_s \subseteq \boldChildren[\mathcal{D}]{s}\) and, in particular, \(N_s \cup \boldChildren[\mathcal{D}']{s} \subseteq \boldChildren[\mathcal{D}]{s}\).
Combined with the fact that \(\mathcal{D}\) is friendly, we get \(\Delta_s \leq 1\) and the width of \(\mathcal{C}\) is bounded by \(w+5\).
\end{proof}

To tackle the cluttered nodes, we solve sub-instances of \(\gstp\).
The reduction rules we present now, are no reduction rules in the classical sense (i.e.,~they do not necessarily run in polynomial time), but rather recursion rules.
We later show, how to apply these rules in a way, that we only solve simple sub-instances and we mostly preserves the running time obtained in \Cref{sec:tcw-simple}.

The crux of why this problem is \(\fpt\) by the tree-cut width of the augmented graph, but \(\wonehard\) by the tree-cut width of the host-graph~\cite{GanianOR20} lies in the fact, that for a cluttered node \(s\in V(S)\), we have for all \(T \in \mathcal{T}\) that either \(T \cap Y^\mathcal{D'}_s\) or \(T \setminus Y^\mathcal{D'}_s\) is empty.
This means that no terminal set crosses \(Y^\mathcal{D'}_s\).
Therefore, we can mostly disregard how the instance looks on \(V(G) \setminus Y^\mathcal{D'}_s\) for deciding how terminal set contained in \(Y^\mathcal{D'}_s\) are solved in a solution of the whole instance.
Let \(u,x \in Y_s\) and \(v,y \in V(G) \setminus Y_s\) be such, that \(\{uv,xy\} = \cutEdges[\augment{G}{\mathcal{T}}]{Y^\mathcal{D}_s}\), and let \(\mathcal{U}\coloneqq \{T \in \mathcal{T}\mid T \subseteq Y_s\}\) be those terminal sets contained in \(Y^\mathcal{D'}_s\).
Note that \(uv, xy \in E(G)\).
First, we consider the case, where we can satisfy the requirements of \(\mathcal{U}\), while supplying an additional connection for \(vy\) to \(V(G)\setminus Y^\mathcal{D'}_s\), while solving the requirements of \(\mathcal{U}\).

\begin{reductionrule}
\label{rr:tcw-thin-supply}
Consider the instance \(\mathscr{X}_s \coloneqq (G[Y^\mathcal{D'}_s], \mathcal{U} \cup \{u,x\}, d')\) where \(d'\) is \(d\vert_{\mathcal{U}\cup \{u,x\}}\) increased by one for the argument \(\{u,x\}\).
If \(\mathscr{X}_s\) is positive, remove all \(\mathcal{U}\) from \(\mathcal{T}\) and contract \(Y^\mathcal{D'}_s\) in the original instance \(\mathscr{P}\).
\end{reductionrule}
\begin{proof}
Let \((\mathcal{F}, \pi)\) be a solution for the original instance and denote with \(\mathcal{Z} \coloneqq \inv\pi(\mathcal{T}\setminus \mathcal{U})\) all trees that are assigned to a terminal set whose vertices are not contained in \(Y^\mathcal{D'}_s\cap V(G)\).
Denote with \(h\) the vertex to which \(Y^\mathcal{D'}_s\) gets contracted.
Consider any \(Z \in \mathcal{Z}\), if \(uv \in E(Z)\), replace \(Z\) by \((Z - uv) + hv\) and if \(xy \in E(Z)\), replace \(Z\) by \((Z - xy) + hy\).
Denote with \(\mathcal{Z}'\) the set of all \(\mathcal{Z}\) after applying the transformation.
Notice that all of them are connected, pairwise edge-disjoint, and contained in the reduced instance.
Assigning \(Z \in \mathcal{Z}'\) to the same terminal set as the original subgraphs yields a solution for the reduced instance.

Now, let \((\mathcal{F}, \pi)\) be a solution for the reduced instance and \((\mathcal{H}, \rho)\) be a solution to \(\mathscr{X}_s\).
Without loss of generality, we assume that for all \(T \in \mathcal{T}\setminus \mathcal{U}\), we have that all leaves of all trees in \(\inv\pi(T)\) are contained in \(T\).
Denote with \(h\) the vertex to which \(Y^\mathcal{D'}_s\) gets contracted.
As \(h\) is not contained in any terminal set and \(\deg{h} = 2\), there is at most one \(F \in \mathcal{F}\) with \(h \in V(F)\).
Let \(P \in \inv\rho(\{u,x\})\) and set \(F^* \coloneqq ((F - h) + uv + xy) \cup P\).
Notice that \(V(F^*) \supseteq \pi(F)\), that \(F^*\) is connected, and that edge-disjoint from all \(\mathcal{H}\setminus P\) and \(\mathcal{F}\setminus F\).
A solution for the original instance can be obtained by assigning all \(F' \in \mathcal{F}\setminus F\) to \(\pi(F)\), all \(H \in \mathcal{H}\setminus P\) to \(\rho(H)\), and \(F^*\) to \(\pi(F)\).
\end{proof}

If \Cref{rr:tcw-thin-supply} is not applicable, we know that we cannot use \(uv\) and \(xy\) in a tree for terminals contained in \(V(G)\setminus Y^\mathcal{D'}_s\).
So, we check, whether the terminal sets \(\mathcal{U}\) can be solved only using edges of \(G[Y^\mathcal{D'}_s]\).

\begin{reductionrule}\label{rr:tcw-thin-independent}
Assume that \Cref{rr:tcw-thin-supply} is not applicable to \(s\).
Consider the instance \(\mathscr{Y}_s\coloneqq (G[Y^\mathcal{D'}_s], \mathcal{U}, d\vert_\mathcal{U})\).
If \(\mathscr{Y}_s\) is positive, remove \(\mathcal{U}\) from \(\mathcal{T}\) and \(Y_s^{\mathcal{D}'}\) from \(G\) in the original instance \(\mathscr{P}\).
\end{reductionrule}
\begin{proof}
Let \((\mathcal{F}, \pi)\) be a solution for the original instance and denote with \(\mathcal{Z} \coloneqq \inv\pi(\mathcal{T}\setminus \mathcal{U})\) all trees that are assigned to a terminal set whose vertices are not contained in \(Y^\mathcal{D'}_s\cap V(G)\).
Since \(\mathscr{X}_s\) is negative, for all \(Z \in \mathcal{Z}\) the set of vertices \(V(Z)\) is disjoint from \(Y^\mathcal{D'}_s\).
So, \((\mathcal{Z}, \pi\vert_{\mathcal{Z}})\) is a solution for the reduced instance.

Now, let \((\mathcal{F}, \pi)\) be a solution for the reduced instance and \((\mathcal{H}, \rho)\) be a solution to \(\mathscr{Y}_s\).
Let \(\mathcal{J} \coloneqq \mathcal{F} \cup \mathcal{H}\) be a set of edge-disjoint connected subgraphs of \(G\) and denote with \(\tau \colon \mathcal{J} \to \mathcal{T}\) the function satisfying \(\tau\vert_{\mathcal{F}} = \pi\) and \(\tau\vert_{\mathcal{H}} = \rho\).
Then, \((\mathcal{J}, \tau)\) is a solution for the original instance.
\end{proof}

Finally, we need to take care of the case, where the terminal sets \(\mathcal{U}\) cannot be solved using only edges of \(G[Y^\mathcal{D'}_s]\).

\begin{reductionrule}
\label{rr:tcw-thin-demand}
Assume both \Cref{rr:tcw-thin-supply,rr:tcw-thin-independent} are not applicable to \(s\).
Consider the instance \(\mathscr{Z}_s\coloneqq (G / (V(G) \setminus Y^\mathcal{D'}_s), \mathcal{U}, d\vert_\mathcal{U})\).
If \(\mathscr{Z}_s\) is positive, we remove \(\mathcal{U}\) from \(\mathcal{T}\), \(Y^\mathcal{D'}_s\) from \(G\), and add 1 demand to the terminal set \(\{v,y\}\) in the remaining instance.
Otherwise, output a trivial negative instance.
\end{reductionrule}
\begin{proof}
Let \((\mathcal{F}, \pi)\) be a solution for the original instance and denote with \(\mathcal{H} \coloneqq \inv\pi(\mathcal{U})\) all trees that are assigned to a terminal set whose vertices are contained in \(Y^\mathcal{D'}_s\cap V(G)\).
Denote with \(z\) the vertex of \(G / (V(G) \setminus Y^\mathcal{D'}_s)\) to which the vertices \(V(G) \setminus Y^\mathcal{D'}_s\) are contracted.
For all subgraphs \(H\in \mathcal{H}\), if \(uv \in E(H)\) replace \(H\) by \((H - uv) + uz\), and if \(xy \in E(Y)\) replace \(H\) by \((H - xy) + xz\).
Denote with \(\mathcal{H'}\) the set of all \(\mathcal{H}\) after applying the modification.
Notice that they are connected and pairwise edge-disjoint, and contained in \(G / (V(G) \setminus Y^\mathcal{D'}_s)\).
Assigning each subgraph of \(\mathcal{H}'\) to the original terminal pair, we obtain a solution for \(\mathscr{Z}_s\).
So, rejecting the instance if \(\mathscr{Z}_s\) is negative is correct.

Additionally, since \(\mathscr{Y}_s\) is negative, we know that there is a \(H \in \mathcal{H}\) with \(V(H) \setminus Y^\mathcal{D'}_s \neq \emptyset\).
In particular, \(\{uv, xy\} \subseteq E(H)\).
Let \(P \coloneqq H - Y^\mathcal{D'}_s\) be the part of \(H\) outside \(Y_s\) and let \(\mathcal{J}\coloneqq \inv\pi(\mathcal{T}\setminus \mathcal{U})\).
Note that \(P\) is connected, edge-disjoint from all \(\mathcal{J}\), contained in the reduced instance and \(V(P) \supseteq \{v,y\}\).
To obtain a solution for the reduced graph, we assign all subgraphs in \(\mathcal{J}\) to their respective terminal set and \(P\) to \(\{v,y\}\).

Now, let \((\mathcal{F}, \pi)\) be a solution for the reduced instance and \((\mathcal{H}, \rho)\) be a solution to \(\mathscr{Z}_s\).
Since \(\mathscr{Y}_s\) is negative, there is a unique \(H \in \mathcal{H}\) such that \(V(H) \setminus Y^\mathcal{D'}_s \neq \emptyset\).
Let \(z\) denote the graph to which \(V(G) \setminus Y^\mathcal{D'}_s\) are contracted in the host-graph of \(\mathscr{Z}_s\).
Then, \(\{uz, xz\}\subseteq E(H)\).
Let \(P \in \inv\pi(\{v,y\})\) and denote with \(H^* \coloneqq ((H - z) + uv + xy) \cup P\).
Notice that \(V(H^*) \supseteq \rho(H)\), that \(H^*\) is connected, and that \(H^*\) is edge-disjoint from all subgraphs in \(\mathcal{H}\setminus H\) and \(\mathcal{F}\setminus P\).
A solution for the original instance can be obtained by assigning all \(F \in \mathcal{F}\setminus P\) to \(\pi(F)\), all \(H' \in \mathcal{H}\setminus H\) to \(\rho(H')\), and \(H^*\) to \(\rho(H)\).
\end{proof}

Additionally, the tree-cut width of the augmented graphs of \(\mathscr{X}_s, \mathscr{Y}_s\), and \(\mathscr{Z}_s\) is bounded by \(\tcw{\augment{G}{\mathcal{T}}}\) and applying any of the previous three reduction rules does not increase the tree-cut width of the augmented graph.

\begin{lemma}
\label{stmt:rr-tcw-thin-decomposition-sensible}
We can obtain a tree-cut decomposition for the augmented graphs of \(\mathscr{X}_s\), \(\mathscr{Y}_s\), and \(\mathscr{Z}_s\) of at most the same width as the tree-cut decomposition of \(\augment{G}{\mathcal{T}}\) in linear time.
Additionally, we can obtain tree-cut decomposition of the augmented graphs after applying any of \Cref{rr:tcw-thin-supply,rr:tcw-thin-independent,rr:tcw-thin-demand} in linear time as well, while not increasing their width.
\end{lemma}
\begin{proof}
Let \(H \coloneqq \augment{G}{\mathcal{T}} / (V(G) \setminus Y^\mathcal{D'}_s)\) be the augmented graph of the instance with all but the vertices in \(Y^\mathcal{D'}_s\) contracted.
Notice that \(H\) is the augmented graph of \(\mathscr{Z}_s\), and if \(\{u,x\} \notin \mathcal{U}\) it is also the augmented graph of \(\mathscr{X}_s\).
Denote with \(h\) the vertex to which \(V(\augment{G}{\mathcal{T}}) \setminus Y^\mathcal{D'}_s\) were contracted.
As the tree-cut decomposition for this graph we choose the tree-cut decomposition induced by \(S_s\) on \(\mathcal{D}\) where we add a node \(p\) containing only \(h\) as the parent of \(s\).
Denote this tree-cut decomposition with \(\mathcal{D}_H\).
We have \(\adhesion[\mathcal{D}_H]{s} = 2 = \adhesion[\mathcal{D}]s\), and the 3-center of the torso at \(p\) in \(\mathcal{D}_H\) is the empty graph, while the torso at \(s\) is the same with respect to \(\mathcal{D}_H\) and \(\mathcal{D}\).
Thus, the width of the \(\mathcal{D}_H\) is bounded by the width of \(\mathcal{D}\).

Let \(H' \coloneqq \augment{G}{\mathcal{T}}[Y^\mathcal{D'}_s]\) be the augmented graph of the instance with all but the vertices in \(Y^\mathcal{D'}_s\) removed.
Notice that \(H'\) is the augmented graph of \(\mathscr{Y}_s\) and if \(\{u,x\} \in \mathcal{U}\) it is also the augmented graph of \(\mathscr{X}_s\).
As tree-cut decomposition for \(H'\) we choose the tree-cut decomposition induced by \(S_s\) on \(\mathcal{D}_H\).

The augmented graph of the reduced instance is either \(I \coloneqq \augment{G}{\mathcal{T}} - Y_s\) or \(J \coloneqq\augment{G}{\mathcal{T}} / Y_s\).
We obtain a tree-cut decomposition for \(I\) by removing \(S_s\) from \(\mathcal{D}\).
Denote with \(z\) the vertex in \(J\) to which the vertices in \(Y_s\) are contracted.
For \(J\), we obtain a tree-cut decomposition by replacing \(S_s\) by a node associated with a bag containing only \(z\).
\end{proof}

Together \Cref{rr:tcw-thin-supply,rr:tcw-thin-independent,rr:tcw-thin-demand}, can be used to remove a cluttered node—or at least make it non-cluttered.
We now show how to apply these rules to solve \(\gstp\) parameterized by tree-cut width of the augmented graph.
To do so, we solve multiple sub-instances of \(\gstp\) with respect to simple tree-cut decompositions.

\begin{theorem}
Assume \(\gstp\) can be solved in time \(r(g, k)\), given a graph of size at most \(g\) and a simple tree-cut decomposition of width at most \(k\).
Then, \(\gstp\) for the instance \(\mathscr{P}\coloneqq(G, \mathcal{T}, d)\) given a tree-cut decomposition of width \(w\) for \(\augment{G}{\mathcal{T}}\) can be solved in time \(\O{\abs{\mathscr{P}} + \abs{V(G)}^7 + \abs{V(G)}^2r(3\abs{G},w + 12)} = \Ostar{r(\abs{G},w + 12)}\).
\end{theorem}
\begin{proof}
First, we exhaustively apply \Cref{rr:sensible-terminal-sets,rr:degree-negative-instance} in time \(\O{\abs{\mathscr{P}}}\).
After this \(\abs{T} \leq 2\abs{E(G)}\), implying \(\abs{V(\augment{G}{\mathcal{T}})} \leq 3\abs{G}\).

Now, we exhaustively apply \Cref{rr:connected-components,rr:adh-1}.
This can be done in time \(\O{w\abs{G}}\).
Now, we make the tree-cut decomposition nice in cubic time~\cite{GanianKS22}.
Then, we make it friendly using \Cref{stmt:tcw-blowing-up-running-time} in time \(\O{w\abs{G}}\), which might increase the tree-cut width to \(\max(4,w)\leq w + 4\).
We continue with this until none of \Cref{rr:connected-components,rr:adh-1} is applicable anymore.
Overall, this might take \(\O{\abs{E(G)}(w\abs{G} + \abs{V(G)}^3})\).
Additionally, these algorithms ensure that each leave is not empty and all empty nodes have at least two children.

If there is no cluttered node, we solve the instance in time \(r(\abs{G}, w + 9)\) using \Cref{stmt:no-cluttered}.
So, let \(s\) be the cluttered node with largest depth, resolving ties arbitrarily.
We want to apply \Cref{rr:tcw-thin-supply,rr:tcw-thin-independent,rr:tcw-thin-demand}.
To not increase the running-time too much, when applying one of these rules, we first ensure \(2 \leq \abs{Y^\mathcal{D}_s} < \frac{\abs{\augment{G}{\mathcal{T}}}}{2}\).

We know, that \(2 \leq \abs{Y^\mathcal{D}_s} \leq \abs{\augment{G}{\mathcal{T}}} - 1\).
Let \(p\) denote the parent of \(t\).
If \(\abs{Y^\mathcal{D}_s} = \abs{\augment{G}{\mathcal{T}}} - 1\), we notice that by choice of \(s\), it is the only cluttered node in \(T\).
We now merge \(X_p\) into \(X_s\) and remove—the now empty node—\(p\) from the tree-cut decomposition.
Call this tree-cut decomposition \(\mathcal{D}_s\).
Since \(\mathcal{D}\) is friendly, so is \(\mathcal{D}_s\) and the width of \(\mathcal{D}_s\) is bounded by \((w+4)+3\).
Additionally, since \(s\) was the only cluttered node in \(\mathcal{D}\), there is no cluttered node in \(\mathcal{D}_s\).
Using this, we solve the instance in running time \(r(\abs{G}, w + 12)\) using \Cref{stmt:no-cluttered}.
Note that in this case we are done now; so, the increase in width is negligible.

If \(\frac{\abs{\augment{G}{\mathcal{T}}}}{2} \leq \abs{Y_s} < \abs{\augment{G}{\mathcal{T}}} - 1\), we re-root \(S\) to \(s\), this might make \(S\) not-nice, but we do not need that at this stage anymore.
Notice that \(p\) is cluttered in the modified tree-cut decomposition and that \(\abs{Y_s} + \abs{Y_p} = \abs{\augment{G}{\mathcal{T}}}\).
Thus, in the modified tree-cut decomposition \(2 \leq \abs{Y_p} < \frac{\abs{\augment{G}{\mathcal{T}}}}{2}\).
So, we assume, without loss of generality, that we choose a cluttered node \(s'\) with \(2 \leq \abs{Y_{s'}} < \frac{\abs{\augment{G}{\mathcal{T}}}}{2}\).

We now construct the tree-cut decompositions for \(\mathscr{X}_{s'}\), \(\mathscr{Y}_{s'}\), and \(\mathscr{Z}_{s'}\) in linear time using \Cref{stmt:rr-tcw-thin-decomposition-sensible} and recursively check which of \(\mathscr{X}_{s'}\), \(\mathscr{Y}_{s'}\), and \(\mathscr{Z}_{s'}\) are positive instances.
Based on this information, we apply the appropriate rule out of \Cref{rr:tcw-thin-supply,rr:tcw-thin-independent,rr:tcw-thin-demand}, which takes at most \(\O{\abs{\augment{G}{\mathcal{T}}}} = \O{\abs{G}}\) time and recursively solve the remaining instance.

To bound the running time, we calculate how many uncluttered instances will be solved.
Let \(i(n)\) denote the maximum number of uncluttered instances solved, if the augmented graph contains \(n\) vertices.
Assume that we consider a cluttered node \(s\).
The augmented graphs of the sub-instances have at most \(\abs{Y_s} + 1 \leq \frac{n}{2}\) vertices.
After applying the reduction rules, the vertices in \(Y_s\) get contracted to a single vertex.
Thus, \(\abs{Y_s} - 1\) vertices are removed from the instance.
So, for large enough \(n\), the function \(i\) satisfies \(i(n) \leq \max_{2 \leq k < \frac{n}{2}} 3i(k + 1) + i(n - k + 1)\).
Using this, one can prove inductively that there is a \(c \in \R\) such that for all \(n \in \N\) we have \(i(n) \leq c (n - 3)^2  + c = \O{n^2}\).
Thus, we solve at most \(\O{\abs{V(G)}^2}\) uncluttered instances, each of which takes at most time \(r(\abs{G}, w + 12)\).

Now, we only need to analyze the additional cost.
In each recursion step, we do at most do at most \(\O{\abs{E(G)}(w\abs{G} + \abs{V(G)}^3)}\) additional work.
Therefore, the overall amount of additional work is \(\O{\abs{E(G)}(w\abs{G} + \abs{V(G)}^3)i(\abs{V(\augment{G}{\mathcal{T}})})} = \O{\abs{V(G)}^7}\).
\end{proof}

Kim et~al.~\cite{KimOPST18} proved that for all \(k\in \N\) in time \(2^\O{k^2\log w} \abs{V(G)}^2\) we can either compute a tree-cut decomposition of width \(2k\), or we can certify that no tree-cut decomposition of width \(k\) exists.
Combined with \Cref{stmt:fpt-simple-tcw}, we know that \(\gstp\) is FPT by the tree-cut width of the augmented graph.

\begin{corollary}
\label{stmt:gstp-fpt-tcw}
Let \((G, \mathcal{T}, d)\) be an instance of \(\gstp\) and set \(k \coloneqq \tcw{\augment{G}{\mathcal{T}}}\).
We can decide whether this instance is positive in time \(\OstarLR{2^{2^\O{k^8}}}\); so, \(\gstp\) is \(\fpt\) by the tree-cut width of the augmented graph.
\end{corollary}


\section{\texorpdfstring{$\stp$}{STP} is \texorpdfstring{$\fpt$}{FPT} by the Tree-Cut Width}
\label{sec:stp-by-tcw}
In this Chapter, we show that \(\stp\) is FPT by the tree-cut width of the host graph.
As \(\stp\) is a special case of \(\gstp\), where we only have a single terminal set.
This means that we can apply the results derived in the previous Chapters for \(\gstp\) to \(\stp\).
However, the previous are not strong enough to directly show that \(\stp\) is \(\fpt\) by the tree-cut width of the host-graph.
Mainly the case when the terminal set is large, but the demand is low is still missing.

Let \((G, T, d)\) be an instance of \(\stp\) and let \(\mathcal{D}=(S,\mathcal{X})\) be a friendly tree-cut decomposition of \(G\) with width \(w\).
If \(\abs{T} \leq w\), we interpret the instance as an instance of \(\gstp\).
We can find a tree-cut decomposition \(\mathcal{D}' = (S', \mathcal{X}')\) for \(G^{\{T\}}\) by adding a new root \(r'\) with the bag \(\{\aug{T}\}\) to \(S\) and making the old root a child of \(r'\).
The adhesion of \(\mathcal{D}'\) is bounded by \(2w\).
Let \(s \in V(S)\).
Any \(t\in \thinChildren[\mathcal{D}]{s}\) with \(Y_{t}^{\mathcal{D}} \cap T \neq \emptyset\) is also a thin child of \(s'\) in \(\mathcal{D}'\).
Thus, the size of the 3-center of the torso at \(s\) in \(\mathcal{D}'\) is bounded by \(1 + \abs{T} + \abs{\boldChildren{s}} + \abs{X_s} \leq 2w + 3\).
Note that the torso at \(r'\) in \(\mathcal{D}'\) consists of 2 vertices.
Thus, the width of \(\mathcal{D}'\) is bounded by \(2w+3\) and we can use \Cref{stmt:gstp-fpt-tcw} to decide whether the instance is positive in \(\fpt\)-time.

The same approach does not work for the case \(\abs{T} > w\), as in this case, the tree-cut width of the augmented graph is not necessarily bounded by a function of the tree-cut width of the host graph.
For this consider as the host graph a path of length \(n\) where we attach to each vertex of the path 3 leaves.
This graph is a tree, has tree-cut width 1, and is depicted in \Cref{fig:unlimited-tcw-stp-before}.
Let the set of terminals be the set of all leaves.
The augmented graph, which is depicted in \Cref{fig:unlimited-tcw-stp-after}, contains the graph \(S_{3,n}\) (i.e.,~the star graph with \(n\) leaves where each leaf has 3 parallel edges to the center) as an immersion.
This graph has the wall \(H_{\lfloor\sqrt{n}\rfloor}\) as an immersion~\cite{Wollan15}.
Therefore, \(H_{\lfloor\sqrt{n}\rfloor}\) has an immersion into the augmented graph.
For each \(r \in \N\), let \(G'\) be a graph for which there is an immersion from \(H_{2r^2}\) into \(G\).
Thus, \(G\) has tree-cut width at least \(r\)~\cite{Wollan15} and the augmented graph has tree-cut width at least \(\Omega(\sqrt[4]{n})\) (more careful calculations show the tree-cut width to be \(\Theta(\sqrt{n})\)), proving that the tree-cut width can grow without bound in the tree-cut width of the host graph.

\begin{remark}
There exists a family of \(\stp\) instances such that the host-graph of every instance has tree-cut width 1, but the tree-cut width of the augmented graph of the instances is not bounded.
\end{remark}

\begin{figure}
\begin{subfigure}{0.49\textwidth}
\includegraphics[width=\textwidth]{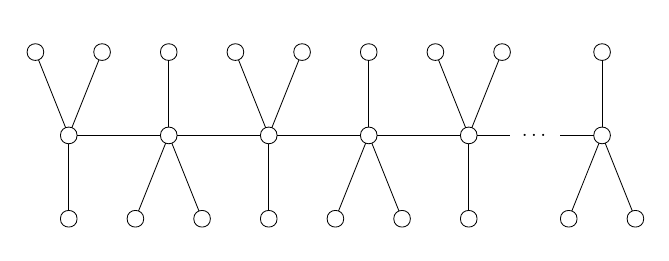}
\caption{The host graph is a path of length \(n\) with 3 leaves attached to each vertex.\newline{}}
\label{fig:unlimited-tcw-stp-before}
\end{subfigure}
\begin{subfigure}{0.49\textwidth}
\includegraphics[width=\textwidth]{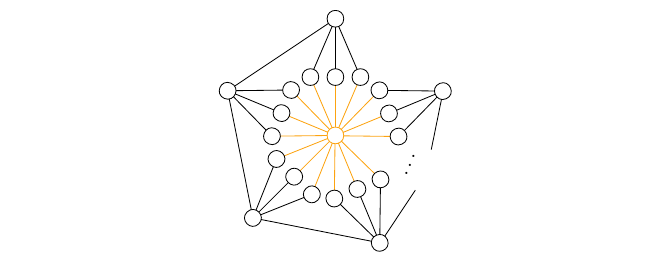}
\caption{The augmented graph with the terminal set being the set of all leaves. Augmented edges and vertices are drawn in orange.}
\label{fig:unlimited-tcw-stp-after}
\end{subfigure}
\caption{A family of host graphs with a terminal set increasing the tree-cut width of the augmented graph without bound.}
\end{figure}

In this case (i.e.,~\(\abs{T} > w\)), we observe that \(T\) is not contained in a single bag.
Let \(s,s' \in V(S)\) be such that \(X_s \cap T \neq \emptyset\), \(X_{s'} \cap T \neq \emptyset\), and \(s' \notin V(S_{s})\) that is \(s'\) is not a descendend of \(s\).
Then, \(T \in \crossLink{s}\) and applying \Cref{rr:cross-link-demand-large}, we know that if \(d(T) > w\), we get a trivial negative instance.
So, we can assume from now on that \(d(T) \leq w\).

Ganian et~al.~\cite{GanianKS22} claimed that from a tree-cut decomposition, we can find a tree decomposition of width at most \(2{w}^2 + 3w\).
Their argument relies on the claim that in a nice tree-cut decomposition, the number of bold children is bounded by \(w + 1\).
In \Cref{stmt:nice-tree-cut-unlimited-bold} we refute this claim.
However, their argument can easily be adapted to a friendly tree-cut decomposition.
Yielding a nice tree decomposition of width \(2{w}^2 + 4w\); thus, \(\tw{G} \leq 2w^2 + 4w\).
That is, in the missing case \(\tw{G} + d(T)\) is bounded by a function of \(w\).
The remainder of this Chapter is dedicated to proving the following theorem.

\begin{restatable}{theorem}{gstp-fpt-tw-D}
\label{stmt:gstp-fpt-tw-D}
Let \((G, \mathcal{T}, d)\) be an instance of the \(\gstp\) problem and set \({\Sigma_{d}} \coloneqq \sum_{T \in \mathcal{T}} d(T)\).
In time \(\abs{V(G)}2^\O{{\Sigma_{d}} \tw{G}\log\tw{G}}\), we can decide whether this instance is positive.
The \(\gstp\) problem is \(\fpt\) by the parameter \(\tw{G} + {\Sigma_{d}}\).
\end{restatable}

This is a generalization of the case missing to solve \(\stp\) parameterized by tree-cut width.
Combined with the case distinction presented above, \Cref{stmt:gstp-fpt-tcw}, and the fact that there is a \(2\)-approximation for tree-cut width running in time \(2^\O{\tcw{G}^2\log\tcw{G}}\abs{V(G)}^2\)~\cite{KimOPST18}, we get that \(\stp\) is \(\fpt\) by the tree-cut width of the host graph.

\begin{corollary}
\label{stmt:tcw-stp-fpt}
Let \((G, T, d)\) be an instance of \(\stp\) and set \(k \coloneqq \tcw{G}\).
In time \(\OstarLR{2^{2^\O{k^8}}}\), we can decide whether this instance is positive.
So, \(\stp\) is \(\fpt\) by the parameter \(\tcw{G}\).
\end{corollary}

\subsection{\texorpdfstring{$\gstp$}{GSTP} is \texorpdfstring{$\fpt$}{FPT} by Treewidth + Sum of Demands}
\label{sec:org18b7fbe}

Let \(\mathscr{P}\coloneqq (G,\mathcal{T}, d)\) be an instance of \(\gstp\) with \({\Sigma_{d}} \coloneqq \sum_{T \in \mathcal{T}} d(T)\).
Let \(\mathcal{D} \coloneqq (S, \mathcal{X})\) be a nice tree decomposition of width \(w\).
We see, that for each \({\Sigma_{d}}\) we can compute a \(\mathrm{MSO}_2\) formula for this problem.
Using Courcelles theorem~\cite{BoriePT92}, we see that this problem is FPT by \(w +{\Sigma_{d}}\).
However, the running time using this meta theorem is horrendous.
So, we additionally provide a more or less standard dynamic program by treewidth.
In this dynamic program, we essentially save for each tree all connectivity information below the current bag.

For \(\edp\) an \(\fpt\)-algorithm parameterized by \(w + {\Sigma_{d}}\) is known~\cite{ZhouTN00}.
The algorithm runs in \(\O{\abs{V(G)}(({\Sigma_{d}} + w^2){\Sigma_{d}}^{w(w+1)/2} + {\Sigma_{d}}(w + 4)^{2(w+4){\Sigma_{d}}+3})}= \abs{V(G)}(2^\O{w^2\log {\Sigma_{d}}} + 2^\O{{\Sigma_{d}} w\log w})\).
Our algorithm generalizes this result significantly, while improving the running time to \(\abs{V(G)}2^\O{{\Sigma_{d}} w\log w}\).
We conjecture that our result is a stepping stone towards applying the Cut\&Count technique introduced by Cygan et~al.~\cite{CyganNPPRW11} to obtain an algorithm running in time \(\abs{V(G)}2^\O{{\Sigma_{d}} w}\).

\looseness=-1
Consider a solution \((\mathcal{F}, \pi)\) to \(\mathscr{P}\) and a \(s \in V(S)\).
There are three types of subgraphs in \(\mathcal{F}\).
First, there are the subgraphs \(F \in \mathcal{F}\) with \(V(F) \subseteq Y_s \setminus X_s\), that is all vertices of these subgraphs have already been forgotten at \(s\).
Second, there are the subgraphs \(F \in \mathcal{F}\) with \(V(F) \cap X_s \neq \emptyset\), that is the subgraphs cross the bag \(s\).
Finally, there are the subgraphs \(F \in \mathcal{F}\) with \(V(F) \cap Y_s = \emptyset\), that is the subgraphs that are completely contained disjoint from \(Y_s\).
In the dynamic program, we store for each \(s \in V(S)\), \(T \in \mathcal{T}\), and \(i \in \natint{d(T)}\) in which type the corresponding solution subgraphs fall and for each \(F \in \mathcal{F}\) with \(V(F) \cap X_s \neq \emptyset\), we additionally store for each \(K \in \comp{F[Y_s] - E(G[X_s])}\) the set \(K \cap X_s\).

Formally, let \(\enumerateTerminalOp \colon \natint{{\Sigma_{d}}} \to \mathcal{T}\) be such that for all \(T \in \mathcal{T}\), we have \(\abs{\invEnumerateTerminal{T}} = d(T)\).
This gives indices to all solution subgraphs such that we only need to find for each \(i \in \natint{{\Sigma_{d}}}\) a solution subgraph \(F_i\) with \(\enumerateTerminal{i} \subseteq V(F_i)\).
To define the dynamic program at a node \(s \in V(S)\), partition \(\natint{{\Sigma_{d}}}\) into three parts \(I^\bot\), \(I^\times\), and \(I ^\top\) representing the indices of the three types of solutions subgraphs with respect to \(s\).
We allow any of \(I^\bot\), \(I^\times\), and \(I^\top\) to be empty.
Finally, consider for all \(i \in I^\times\) a partitioning \(\mathcal{P}_i \coloneqq \{P_{i,j}\}_{j \in \natint{\abs{\mathcal{P}_i}}}\) of a set \(U\) with \(\emptyset \neq U \subseteq X_s\).
Intuitively, these partitions describe which vertices of \(X_s\) are already connected in the subgraph below \(s\).
The crucial change to obtain the speed-up over the state-of-the-art algorithm is to, ensure the connections are realized without edges inside \(X_s\).
This allows the sub-solutions in join nodes to be treated independently.
The entries of the dynamic program \(D(s)\) are tuples of the form \(\tau \coloneqq (I^\bot, I^\times, I^\top, (\mathcal{P}_i)_{i \in I^\times})\).
We also refer to \(I^\bot\) as \(I^\bot_\tau\) and similarly for \(I^\times\) as well as \(I^\top\) and to \(\mathcal{P}_i\) as \(\mathcal{P}_{\tau,i}\).

To formally define which entries \(\tau\) should be included in the dynamic program, we define an instance of \(\gstp\) based on the corresponding tuple.
Let \(G_s \coloneqq G[Y_s] - E(G[X_s])\) be the graph induced by all edges adjacent to a node already forgotten at \(s\).
For each \(i \in I^\times\) add a vertex \(q_i\) to \(G_s\) adjacent to all \(X_s\).
Call the obtained graph \(G_{s, \tau}\).
This is the host-graph of the corresponding instance.
To define the terminal sets, let \(i \in I^\bot\).
Then \(\enumerateTerminal{i}\) is a terminal set with demand \(\abs{\invEnumerateTerminal{\enumerateTerminal{i}} \cap I^\bot}\), that is the demand is equal to the number of indices of solution subgraphs corresponding to the terminal set \(\enumerateTerminal{i}\) that are completed at \(s\).
Additionally, let \(i \in I^\times\) and set \(Q_{s,i} \coloneqq \{q_i\} \cup (\enumerateTerminal{i} \cap Y_s)\) to be a terminal set with demand 1.
So, the solution subgraph assigned to \(Q_{s,i}\) is marked with including \(q_i\) and can use \(q_i\) to simulate that connections outside \(G_s\) are made for this solution subgraph.
Call the obtained instance \(\mathscr{D}_{s, \tau}\).

\begin{definition}
\label{def:dp-tw}
The dynamic programming table \(D(s)\) contains exactly the tuples \(\tau \coloneqq (I^\bot, I^\times, I^\top, (\mathcal{P}_i)_{i \in I^\times})\) such that
\begin{enumerate}
\item \label{item:tw-dp-started} for all \(T \in \mathcal{T}\) with \(T \cap Y_s \neq \emptyset\), we have \(\invEnumerateTerminal{T} \subseteq I^\bot\cup I^\times\),
\item \label{item:tw-dp-solution} there is a solution \((\mathcal{F}, \pi)\) to the instance \(\mathscr{D}_{s, \tau}\) such that
\begin{enumerate}
\item \label{item:tw-dp-bot-are-complete} for all \(i \in I^\bot\) and \(F \in \inv\pi(\enumerateTerminal{i})\), the set \(V(F)\) is disjoint from \(X_s\) and \(\{q_i\}_{i \in I^\times}\),
\item \label{item:tw-dp-crossing} for all \(i \in I^\times\) and \(\{F\} \coloneqq \inv\pi(Q_{s,i})\), the solution subgraph \(F\) is the unique \(F \in \mathcal{F}\) with \(q_i \in V(F)\), and \(\mathcal{P}_{i} = \{V(K) \cap X_s \mid K \in \comp{F - q_i}\}\).\qedhere
\end{enumerate}
\end{enumerate}
\end{definition}

With \Cref{item:tw-dp-started}, we ensure that all indices referring to solution subgraphs starting at or below \(s\) are actually included the index sets that refer to solution subgraphs that are already started.
Using \Cref{item:tw-dp-solution}, we ensure that the solution \((\mathcal{F}, \pi)\) to \(\mathscr{D}_{s,\tau}\) witnessing \(\tau \in D(s)\) actually fulfills some consistency requirements.
Namely, we ensure with \Cref{item:tw-dp-bot-are-complete} that the vertices \(\{q_i\}_{i \in I^\times}\), which can be used to simulate connections outside \(G_s\), are not used for solution subgraphs that are marked as complete in the tuple.
Finally, we ensure two properties using \Cref{item:tw-dp-crossing}.
First, we ensure that the \(\{q_i\}_{i \in I^\times}\) are only used by the \(i\)-th solution subgraph, which means that these are used as markers.
Secondly, we make sure that the partitionings \((\mathcal{P}_i)_{i\in I^\times}\) conform to our interpretation presented above.

We now show, that this dynamic program can be used to determine whether the instance is positive.

\begin{lemma}
\label{stmt:tw-dp-output-correct}
Let \(r\) be the root of \(S\).
The instance \(\mathscr{P}\) is positive if and only if \[(\natint{{\Sigma_{d}}}, \emptyset, \emptyset, ()) \in D(r).\qedhere\]
\end{lemma}
\begin{proof}
Let \(\tau \coloneqq(\natint{{\Sigma_{d}}}, \emptyset, \emptyset, ())\).
Notice that \(\mathcal{D}\) is nice and therefore \(X_r = \emptyset\).
Thus, that \(\mathscr{D}_{t,\tau} = \mathscr{P}\).
Additionally, the conditions laid forth in \Cref{def:dp-tw} are satisfied by every solution to \(\mathscr{D}_{s,\tau}\), proving the statement.
\end{proof}

Now, all that remains is to show, how to compute this dynamic program on a nice tree decomposition.
That is, we need to show how to compute the dynamic program for leaf, introduce, join, and forget nodes given that the dynamic program of all children is already computed.

\paragraph*{Leaf node} Let \(s \in V(S)\) be a leaf node.
Since \(\mathcal{D}\) is nice, \(X_s = \emptyset\).
Observe that \(Y_s = \emptyset\).
Thus, \(I^\bot_\tau = \emptyset\).
We can verify that for all \(S \subseteq \natint{{\Sigma_{d}}}\), we have \((\emptyset, S, \natint{{\Sigma_{d}}} \setminus S,(\emptyset)_{i \in S}) \in D(s)\); so, \(D(s) = \bigcup_{S \subseteq \natint{{\Sigma_{d}}}}\{(\emptyset, S, \natint{{\Sigma_{d}}}\setminus S,(\emptyset)_{i \in S})\}\), showing that these kinds of nodes can be computed efficiently.

\paragraph*{Introduce node} Let \(s \in V(S)\) be an introduce node, that is it has exactly one child \(c\) and there is a \(v \in V(G)\setminus X_c\) with \(X_s = X_c \cup\{v\}\).
Let \(\tau \in D(c)\).
Note that \(G_c + v = G_s\).
Consider how a solution \((\mathcal{F}, \pi)\) to \(\mathscr{D}_{c, \tau}\) can be extended using this new vertex \(v\).
As \((S, \mathcal{X})\) is a tree decomposition, we know that \(N(v)\) is disjoint from all already forgotten nodes \(Y_s \setminus X_s\).
So, no solution subgraph associated with an index in \(I^\bot_\tau\) can be extended to include \(v\) without changing the state regarding \(X_c\).
Any \(F \in \inv\pi(\{Q_{s,i}\}_{i \in I^\times_\tau})\) can be extended to include \(v\) by adding the edge \(vq_i\) to \(F\).
For an \(i \in I^\top_\tau\), we can also introduce a new solution subgraph \(G[vq_i]\) only containing the edges \(vq_i\).
To formally capture this, let \(S \subseteq I^\times_\tau\cup I^\top_\tau\).
For all \(i \in I^\times \cap S\), set \(\mathcal{P}_i' \coloneqq \mathcal{P}_{\tau,i} \cup \{\{v\}\}\), for all \(i \in S \setminus I^\times\), set \(\mathcal{P}_i' \coloneqq \{\{v\}\}\) for all \(i \in I^\times \setminus S\), set \(\mathcal{P}_i' \coloneqq \mathcal{P}_{\tau,i}\).
Now, we define the function \(\extend{\tau}{S}\coloneqq (I^\bot_\tau, I^\times_\tau\cup S, I^\top_\tau \setminus S, (\mathcal{P}_i')_{i \in I^\times \cup S})\).
Intuitively, \(\extend{\tau}{S}\) is the state obtained from \(\tau\), when all solution subgraphs associated with indices in \(S\) are extended to include \(v\).

\begin{lemma}
\label{stmt:tw-dp-introduce-correct}
Let \(L \coloneqq\{i \in \natint{{\Sigma_{d}}}\mid v \in \enumerateTerminal{i}\}\).
Then, \[ D(s) = \bigcup_{\substack{\tau \in D(c),\\L\subseteq S \subseteq I^\times_\tau\cup I^\top_\tau}} \{\extend{\tau}{S} \}.\]
\end{lemma}

\begin{proof}
Let \(\tau \in D(s)\) and let \((\mathcal{F},\pi)\) be a solution to \(\mathscr{D}_{s, \tau}\) such that the additional restrictions of \Cref{item:tw-dp-bot-are-complete} in \Cref{def:dp-tw} are satisfied.
We now show, that there is are \(\gamma \in D(c)\) and \(L \subseteq S \subseteq I^\top_{\gamma} \cup I^\times_{\gamma}\) with \(\extend{\gamma}{S} = \tau\).
For this, we first define \(\gamma\).
Set \(I^\bot_\gamma \coloneqq I^\bot_\tau\).
For all \(i \in I^\times_\tau\), we set \(\mathcal{P}_i' \coloneqq \mathcal{P}_i \setminus \{\{v\}\}\), that is we remove from \(\mathcal{P}_i\) the partition only containing \(v\).
We set \(I^\times_\gamma \coloneqq \{i \in I^\times_\tau \mid \mathcal{P}_i' \neq \emptyset\}\) to be the indices that use vertices of \(X_s\) in \(\tau\) apart from \(v\), and we set for all \(i \in I^\times_\gamma\) that \(\mathcal{P}_{\gamma,i} \coloneqq \mathcal{P}_i'\).
Finally, set \(I^\top_\gamma\) to be the remaining indices \(\natint{{\Sigma_{d}}} \setminus (I^\bot_\gamma \cup I^\times_\gamma)\).
Additionally, set \(S \coloneqq \{i \in I^\times_\tau \mid v \in \bigcup \mathcal{P}_{\tau,i}\}\) to be the set of indices using \(v\) in the solution.
One can verify that \(\extend{\gamma}{S} = \tau\).

It remains to show that \(L \subseteq S\) and that \(\gamma \in D(c)\).
We first show \(L \subseteq S\).
Let \(i \in L\), then \(i \in I^\bot_\tau \cup I^\times_\tau\).
Assume \(i \in I^\bot_\tau\).
Then, the demand of \(\enumerateTerminal{i}\) is at least 1 in \(\mathscr{D}_{s,\tau}\); so, let \(F \in \inv\pi(\enumerateTerminal{i})\).
As \(v \in \enumerateTerminal{i} \subseteq V(F)\), this violates \Cref{item:tw-dp-bot-are-complete} of \Cref{def:dp-tw}.
Therefore, \(i \in I^\times_\tau\).
Let \(\{F\} \coloneqq \inv\pi(Q_{s,i})\).
As \(v \in \enumerateTerminal{i} \subseteq V(F)\), we have, by \Cref{item:tw-dp-crossing} of \Cref{def:dp-tw}, that \(i \in S\).

To show \(\gamma \in D(c)\), we verify that \(\gamma\) satisfies \Cref{item:tw-dp-started} of \Cref{def:dp-tw}.
Let \(T \in \mathcal{T}\) with \(\emptyset \neq T \cap Y_c \subseteq T \cap Y_s\).
By \Cref{item:tw-dp-started} of \Cref{def:dp-tw}, we have \(\invEnumerateTerminal{T} \subseteq I^\bot_\tau \cup I^\times_\tau\).
We know \(I^\bot_\tau = I^\bot_\gamma\) and \(\{i \in I^\times_\tau \mid \mathcal{P}_i' \neq \emptyset\}= I^\times_\gamma\).
So, it suffices to show for all \(i \in \invEnumerateTerminal{T}\) with \(i \in I^\times_\tau\) that \(\mathcal{P}_i'\neq \emptyset\).
Assume the contrary, that is there is a \(i \in \invEnumerateTerminal{T}\) with \(\mathcal{P}_i' = \emptyset\).
Then, \(\mathcal{P}_i = \{\{v\}\}\).
Let \(u \in T \cap Y_c\) and consider a \(vu\)-path \(Q\) in \(F\).
As \(N_G(v)\) is disjoint from \(Y_s \setminus X_s\) and since there are no edges among vertices in \(X_s\) in \(G_s\), there is an \(x \in X_s\setminus \{v\}\) with \(Q = vq_ix \dots u\).
Thus, \(x \in V(F) \cap X_s\).
By \Cref{item:tw-dp-crossing}, we have \(x \in \bigcup \mathcal{P}_i\) and violating \(\mathcal{P}_i = \{\{v\}\}\); so, \(\gamma\) satisfies \Cref{item:tw-dp-started} of \Cref{def:dp-tw}.

Finally, we provide a solution to \(\mathscr{D}_{c, \gamma}\) satisfying the additional requirements of \Cref{item:tw-dp-solution} in \Cref{def:dp-tw}.
For this, let \(i \in I^\bot_\gamma\).
As \(I^\bot_\gamma = I^\bot_\tau\), there is a \(F \in \inv\pi(\enumerateTerminal{i})\).
By \Cref{item:tw-dp-bot-are-complete} of \Cref{def:dp-tw}, \(F\) is fully contained in \(G[Y_s \setminus X_s] = G[Y_c \setminus X_c]\).
So, we assign all such \(F\) to \(\enumerateTerminal{i}\) in our solution to \(\mathscr{D}_{c,\gamma}\) and since \(I^\bot_\gamma = I^\bot_\tau\), this satisfies the demand of \(\enumerateTerminal{i}\) while satisfying \Cref{item:tw-dp-bot-are-complete} of \Cref{def:dp-tw}.
Now, let \(i \in I^\times_\gamma \subseteq I^\times_\tau\) and let \(\{F_i\} \coloneqq \inv\pi(Q_{s,i})\).
The graph \(F^*_i\coloneqq F - v\) is contained in \(G_{c, \rho}\) and we see \(Q_{c,i} = q_i \cup (\enumerateTerminal{i} \cap Y_c) = (q_i \cup (\enumerateTerminal{i} \cap Y_s)) \setminus \{v\} = Q_{s,i} \setminus \{v\} \subseteq V(F^*_i)\).
Consequently, we assign \(F^*_i\) to \(Q_{c,i}\) in \(\mathscr{D}_{c, \gamma}\), satisfying its demands.
We can verify that this solution satisfies \Cref{item:tw-dp-bot-are-complete} and that for all \(i \in I^\times_\gamma\), the solution subgraph \(F^*_i\) is the unique solution subgraph containing \(q_i\).
Finally, consider \(K \in \comp{F_i - q_i}\).
Assume \(v \in V(K)\).
As \(N_{G_s}(v) = \{q_i\}\), we have \(V(K) = \{v\}\).
So, if \(v \notin V(K)\), there is a \(K' \in \comp{F^*_i - q_i}\) with \(K = K'\) and vice versa, concluding the proof that \(\gamma \in D(c)\).

Now, let \(\gamma \in D(c)\) and \(L \subseteq S \subseteq I^\times_\gamma \cup I^\top_\gamma\) be given.
We now show that \(\tau \coloneqq\extend{\gamma}{S} \in D(s)\) by verifying \Cref{item:tw-dp-started,item:tw-dp-solution} of \Cref{def:dp-tw} one after the other.
Let \(T \in \mathcal{T}\) with \(T \cap Y_s \neq \emptyset\).
If \(T \cap Y_c \neq \emptyset\), by \Cref{item:tw-dp-started} of \Cref{def:dp-tw}, we have \(\invEnumerateTerminal{T} \subseteq I^\bot_\gamma \cup I^\times_\gamma\subseteq I^\bot_\tau \cup I^\times_\tau\).
Otherwise, \(v \in T\).
Thus, \(\invEnumerateTerminal{T}\subseteq L \subseteq S \subseteq I^\times_\tau\), showing that \Cref{item:tw-dp-started} of \Cref{def:dp-tw} is satisfied.

To construct a solution to \(\mathscr{D}_{s,\tau}\) satisfying the additional requirements of \Cref{item:tw-dp-solution} in \Cref{def:dp-tw}, let \((\mathcal{F},\pi)\) be a solution to \(\mathscr{D}_{c, \gamma}\) satisfying these additional requirements.
First, let \(i \in I^\bot_\tau = I^\bot_\gamma\) and consider \(F \in \inv\pi(\enumerateTerminal{i})\).
By \Cref{item:tw-dp-bot-are-complete}, we have \(\enumerateTerminal{i}\subseteq V(F) \subseteq Y_c \setminus X_c = Y_s \setminus X_s\); so, we assign each such \(F\) to \(\enumerateTerminal{i}\) satisfying its demand in \(\mathscr{D}_{s,\tau}\).
Now, let \(i \in I^\times_\gamma \subseteq I^\times_\tau\) and let \(\{F\} \coloneq \inv\pi(Q_{c,i})\).
If \(i \in S\), we consider the subgraph \(F^*_i \coloneqq F + q_iv\) of \(G_{s, \tau}\).
As \(q_i \in V(F)\), we know that \(F^*_i\) is connected.
Furthermore, \(Q_{s,i} = \{q_i\} \cup (\enumerateTerminal{i}\cap Y_s) \subseteq \{q_i, v\}\cup(\enumerateTerminal{i}\cap Y_c) = Q_{c,i} \cup \{v\} \subseteq V(F^*_i)\).
So, we assign \(F^*_i\) to \(Q_{s,i}\) in our solution, satisfying its demand.
Now, let \(i \in I^\times_\tau \setminus I^\times_\gamma\).
By \Cref{item:tw-dp-started}, we know that \(\enumerateTerminal{i} \cap Y_c = \emptyset\).
Thus, \(Q_{s,i} = \{q_i\}\cup(\enumerateTerminal{i} \cap Y_s) \subseteq \{q_i, v\}\).
We assign the subgraph \(F^*_i \coloneqq G_{s,\tau}[q_iv]\), that is the subgraph induced by the edge \(q_iv\), to the terminal set \(Q_{s,i}\), completing the description of the solution to \(\mathscr{D}_{s,\tau}\).
We can verify that this solution additionally satisfies \Cref{item:tw-dp-solution} of \Cref{def:dp-tw}.
\end{proof}

This gives us an easy way to compute the dynamic program for introduce nodes.

\begin{corollary}
\label{stmt:tw-dp-introduce-running-time}
The dynamic program for an introduce node with child \(c\) can be computed in time \(\O{1 + ({\Sigma_{d}} + w)2^{\Sigma_{d}}\abs{D(c)}}\) given that \(D(c)\) is provided.
\end{corollary}

\paragraph*{Join node} Let \(s \in V(S)\) be a join node with children \(a, b\) such that \(X_s = X_a = X_b\).
Notice that \(Y_a \setminus X_s\) and \(Y_b \setminus X_s\) are disjoint, which means that \(G_a\) and \(G_b\) are edge-disjoint, since vertices in \(X_s\) are not adjacent.
Let \(\alpha \in D(a)\) and \(\beta \in D(b)\).
Consider solutions \((\mathcal{F}_a,\pi_a)\) and \((\mathcal{F}_b, \pi_b)\) to \(\mathscr{D}_{a,\alpha}\) and \(\mathscr{D}_{b, \beta}\), respectively.
Then, every edge in \((\bigcup \mathcal{F}_a) \cap (\bigcup \mathcal{F}_b)\) is adjacent to a vertex in \(\{q_i\}_{i \in I^\times_\alpha \cup I^\times_\beta}\).
Meaning that these solutions can be easily combined.
Assume that \(I^\times_\alpha = I^\times_\beta\), and let \(i \in I^\times_\alpha\) and set \(\mathcal{E}_i \coloneqq \mathcal{P}_{\alpha,i} \cup \mathcal{P}_{\beta, i}\) to be the sets of vertices in \(X_s\) that are connected inside \(G_a\) or \(G_b\) in the solution subgraphs associated with index \(i\).
Consider the hypergraph \(H_i \coloneqq (\bigcup \mathcal{E}_i,\mathcal{E}_i)\) and set \(\mathcal{K}_i \coloneqq \{V(K) \mid K \in \comp{H_i}\}\).
Let \(I^\bot \coloneqq I^\bot_\alpha \cup I^\bot_\beta\) be all indices of subgraphs that are associated with a complete solution subgraph in either \(\alpha\) or \(\beta\).
We now define the function \(\merge\alpha\beta \coloneqq (I^\bot, I^\times_\alpha, I^\top_\alpha\setminus I^\bot_\beta, (\mathcal{K}_i)_{i \in I^\times_\alpha})\).
Intuitively, after merging \(\alpha\) and \(\beta\), all indices which were completed in either \(\alpha\) or \(\beta\) are completed after the merge and all solution subgraphs that are associated with an index crossing \(X_s\) can use connections made in either \(G_a\) and \(G_b\).

\begin{lemma}
\label{stmt:tw-dp-join-correct}
We have \[
D(s) = \bigcup_{\substack{\alpha \in D(a), \beta \in D(b)\!\colon\\I^\times_\alpha = I^\times_\beta}} \{\merge{\alpha}{\beta}\}
\]
\end{lemma}

\begin{proof}
Let \(\tau \in D(s)\) and let \((\mathcal{F},\pi)\) be a solution to \(\mathscr{D}_{s, \tau}\) satisfying the additional requirements of \Cref{item:tw-dp-solution} in \Cref{def:dp-tw}.
We now show, that there is are \(\alpha \in D(a)\) and \(\beta \in D(b)\) with \(I^\times_\alpha = I^\times_\beta\) such that \(\merge{\alpha}{\beta} = \tau\).
For this, we first define \(\alpha\) and \(\beta\).
Set \(I^\times_\alpha \coloneqq I^\times_\tau\) and \(I^\times_\beta \coloneqq I^\times_\tau\).
Additionally, set \(I^\bot_\alpha \coloneqq \{i \in I^\bot_\tau \mid \enumerateTerminal{i} \subseteq Y_a\}\) to be the indices of solution subgraphs that are fully contained in \(G_a\).
Analogously set \(I^\bot_\beta\) and choose \(I^\top_\alpha\) and \(I^\top_\beta\) accordingly.
Consider \(i \in I^\times_\alpha\) and let \(\{F\} \coloneqq \inv\pi(Q_{s,i})\).
Set \(\mathcal{P}_{\alpha,i} \coloneqq \{V(K) \cap X_a \mid K \in \comp{F - q_i - (Y_b \setminus X_s)}\}\) to be the partitions of \(X_a\) connected inside \(G_b\) and set \(\mathcal{P}_{\beta, i}\) analogously.

We now show that \(\merge\alpha\beta = \tau\).
First, we show \(I^\bot_\alpha \cup I^\bot_\beta = I^\bot_\tau\).
As \(I^\bot_\alpha\) and \(I^\bot_\beta\) are subsets of \(I^\bot_\tau\), is suffices to show \(I^\bot_\alpha \cup I^\bot_\beta \supseteq I^\bot_\tau\).
Assume there is a \(i \in I^\bot_\tau \setminus (I^\bot_\alpha\cup I^\bot_\beta)\).
As \(\enumerateTerminal{i}\) is a terminal set in \(\mathscr{D}_{s,\tau}\) with positive demand and this instance is positive, we have \(\enumerateTerminal{i} \subseteq Y_s\).
Since \(i \notin I^\bot_\alpha\), there is a \(u \in \enumerateTerminal{i} \setminus Y_a = \enumerateTerminal{i} \cap (Y_b \setminus X_s)\).
Similarly, there is a \(v \in \enumerateTerminal{i}\cap(Y_a \setminus X_s)\).
Let \(F \in \inv\pi(\enumerateTerminal{i})\) and consider a \(uv\) path \(P\) in \(F\).
As we are considering a tree decomposition, we have \(N_G(Y_a) \subseteq X_s\) and \(N_G(Y_b) \subseteq X_s\).
Thus, the path \(P\) contains a vertex of \(X_s\) violating \Cref{item:tw-dp-bot-are-complete} of \Cref{def:dp-tw}.

To show \(\merge\alpha\beta = \tau\), it remains to show that for all \(i \in I^\times_\tau\), we have \(\mathcal{P}_{\tau, i} = \mathcal{P}_{\merge\alpha\beta, i}\).
First, we note that \(\bigcup \mathcal{P}_{\tau, i}=\bigcup \mathcal{P}_{\alpha, i}=\bigcup \mathcal{P}_{\beta, i} = \mathcal{P}_{\merge\alpha\beta, i}\).
Let \(\{F\} \coloneqq \inv\pi(Q_{s,i})\), \(P \in \mathcal{P}_{\tau, i}\), and \(u,v \in \mathcal{P}_{\tau,i}\).
There is a simple \(uv\)-path \(R\) in \(F - q_i\).
Let \(x_1, x_2, \dots, x_\ell\) be the sequence of vertices in \(X_s\) on \(R\).
For all \(i \in \natint{\ell - 1}\), the sub-path in \(R\) starting at \(x_i\) and ending at \(x_{i+1}\) only contains inner vertices in either \(Y_a \setminus X_s\) or \(Y_b \setminus X_s\).
Thus, there is a \(P \in \mathcal{P}_{\alpha, i} \cup \mathcal{P}_{\beta, i} = \mathcal{E}_i\) with \(x_i,x_{i+1} \in P\).
Therefore, in the hypergraph \(H_i\) there is a \(uv\)-path and there is a \(P' \in \mathcal{P}_{\merge\alpha\beta, i}\) with \(u,v \in P'\).
Applying this logic in reverse, we get that for all \(P \in \mathcal{P}_{\merge\alpha\beta, i}\) and \(u,v \in P\), there is a \(P' \in \mathcal{P}_{\tau, i}\) with \(u,v \in P'\).
This shows that \(\mathcal{P}_{\tau, i} = \mathcal{P}_{\merge\alpha\beta, i}\).

We now show that \(\alpha \in D(a)\).
The proof that \(\beta \in D(b)\) can be done analogously.
We now show that \Cref{item:tw-dp-started} of \Cref{def:dp-tw} is satisfied for \(\alpha\).
Let \(T \in \mathcal{T}\) with \(\emptyset \neq T \cap Y_a \subseteq T\cap Y_s\).
Let \(i \in \invEnumerateTerminal{T}\).
By \Cref{item:tw-dp-started} of \Cref{def:dp-tw}, we have \(i \in I^\bot_\tau \cup I^\times_\tau\).
If \(i \in I^\bot_\tau\), we know that \(i\in I^\bot_\alpha\cup I^\bot_\beta\).
If \(i \in I^\bot_\alpha\), we are done, so assume \(i \in I^\bot_\beta\).
Then, \(\enumerateTerminal{i} \subseteq Y_b\).
As \(T \cap Y_a \neq \emptyset\), we have \(T \cap X_s \neq \emptyset\).
This is however not possible by \Cref{item:tw-dp-bot-are-complete} of \Cref{def:dp-tw} and since \(I^\times_\tau = I^\times_\alpha\), we have \(\invEnumerateTerminal{T} \subseteq I^\bot_\alpha \cup I^\times_\alpha\).
Hence, \Cref{item:tw-dp-started} of \Cref{def:dp-tw} is satisfied.

Now, we provide a solution to \(\mathscr{D}_{a, \alpha}\) satisfying the additional requirements of \Cref{item:tw-dp-solution} in \Cref{def:dp-tw}.
Let \(i \in I^\bot_\alpha \subseteq I^\bot_\tau\) and consider \(F \in \inv\pi(\enumerateTerminal{i})\).
As \(\enumerateTerminal{i} \subseteq Y_a\), since \(V(F)\) is disjoint from \(X_s\) by \Cref{item:tw-dp-bot-are-complete} of \Cref{def:dp-tw}, and since \(N_{G_s}(Y_a) \subseteq X_s\), we have \(V(F) \subseteq Y_a \setminus X_s\).
So, we assign \(F\) to \(\enumerateTerminal{i}\) in our solution to \(\mathscr{D}_{a, \alpha}\) satisfying the demand of this terminal set as well as \Cref{item:tw-dp-bot-are-complete} of \Cref{def:dp-tw}.
Now, let \(i \in I^\times_\alpha = I^\times_\tau\) and consider \(\{F\} \coloneqq \inv\pi(Q_{s,i})\).
Set \(F^*_i\) equal to \(F - (Y_b \setminus X_s) + \{q_ix\}_{x \in V(F) \cap X_s}\), that is, we remove all vertices not contained in \(Y_a \cup \{q_i\}\) and add all edges between \(q_i\) and vertices of the bag used in \(F\).
To see that \(F^*_i\) is connected, note that all \(X_s\cap V(F^*_i)\) are adjacent to \(q_i\).
For any \(v \in V(F^*_i) \setminus (X_s \cup \{q_i\})\) there is a \(vq_i\)-path \(P\) in \(F\).
As \(N(q_i) \subseteq X_s\), there is an \(x \in X_s\) that appears first in \(P\).
As \(v \in Y_a\) and \(N_{G_{a,\alpha}}(Y_a \setminus X_s) \subseteq X_s\), there is no vertex contained in \(Y_b \setminus X_s\) between \(v\) and \(x\) on \(P\).
Thus, there is a \(vx\)-path in \(F^*_i\) and \(v\) is in the same connected component as \(x\) and by extension \(q_i\), proving that \(F^*_i\) has exactly one connected component.
Furthermore, \(Q_{a, i} = Q_{s,i} \setminus (Y_b \setminus X_s) \subseteq V(F^*_i)\).
Consequently, we assign \(F^*_i\) to \(Q_{a,i}\) in our solution to \(\mathscr{D}_{a,\alpha}\), satisfying the demand of \(Q_{a,\alpha}\) and concluding our solution to \(\mathscr{D}_{a,\alpha}\).

That this solution satisfies \Cref{item:tw-dp-bot-are-complete} of \Cref{def:dp-tw} follows directly from the fact that \((\mathcal{F}, \pi)\) satisfies this condition.
Also, for each \(i \in I^\times_\alpha\), it is clear that the solution subgraph \(F^*_i\) is the unique solution subgraph containing \(q_i\).
Additionally, we have \(F^*_i - q_i = F - q_i - (Y_b \setminus X_s)\) which shows that \Cref{item:tw-dp-crossing} of \Cref{def:dp-tw} is satisfied, showing that \(\alpha \in D(a)\).

Now, let \(\alpha \in D(a)\) and \(\beta \in D(b)\) with \(I^\times_\alpha = I^\times_\beta\) be given.
We now show that \(\tau \coloneqq\merge{\alpha}{\beta} \in D(s)\).
First, we show that \Cref{item:tw-dp-bot-are-complete} of \Cref{def:dp-tw} is satisfied.
Let \(T \in \mathcal{T}\) with \(T \cap Y_s \neq \emptyset\) and consider \(i\in\invEnumerateTerminal{T}\).
If \(T \cap Y_a \neq \emptyset\), we have \(i \in I^\bot_\alpha \cup I^\times_\alpha \subseteq I^\bot_\tau \cup I^\times_\tau\) since \(\alpha \in D(a)\).
Otherwise, \(T \cap Y_a = \emptyset\).
In this case, \(T \cap Y_b \neq \emptyset\) and \(i \in I^\bot_\beta \cup I^\times_\beta \subseteq I^\bot_\tau \cup I^\times_\tau\), showing that \Cref{item:tw-dp-bot-are-complete} of \Cref{def:dp-tw} is satisfied.

Next, we construct a solution to \(\mathscr{D}_{s,\tau}\) satisfying the additional requirements of \Cref{item:tw-dp-solution} in \Cref{def:dp-tw}.
For this, let \((\mathcal{F}_a, \pi_a)\) and \((\mathcal{F}_b, \pi_b)\) be solutions to \(\mathscr{D}_{a,\alpha}\) and \(\mathscr{D}_{b,\beta}\) satisfying the additional requirements of \Cref{item:tw-dp-solution} in \Cref{def:dp-tw}, respectively.
Let \(i \in I^\bot_\tau = I^\bot_\alpha \cup I^\bot_\beta\).
Assume without loss of generality that \(i \in I^\bot_\alpha\) and let \(F \in \inv\pi_a(\enumerateTerminal{i})\).
As \(F\) is completely contained in \(G_a[Y_a \setminus X_s]\), it is also completely contained in \(G_s[Y_s\setminus X_s]\) and we assign all such \(F\) to \(\enumerateTerminal{i}\) in \(\mathscr{D}_{s, \tau}\).
To see that this already satisfies the demand of \(\enumerateTerminal{i}\), note that the demand of \(\enumerateTerminal{i}\) in \(\mathscr{D}_{s,\tau}\) is \(\abs{\invEnumerateTerminal{i} \cap I^\bot_{\tau}} \leq \abs{\invEnumerateTerminal{i} \cap I^\bot_{\alpha}} + \abs{\invEnumerateTerminal{i} \cap I^\bot_{\beta}}\) and we assign \(\abs{\invEnumerateTerminal{i} \cap I^\bot_{\alpha}}\) solution subgraphs.
As \(\enumerateTerminal{i} \subseteq V(F) \subseteq Y_a \setminus X_a\), we have \(\enumerateTerminal{i} \cap Y_b \setminus X_b = \emptyset\).
Thus, \(\invEnumerateTerminal{i} \cap I^\bot_\beta = \emptyset\) and the demand of \(\enumerateTerminal{i}\) is satisfied in our solution.
Apply the analogous argument if \(i\in I^\bot_\beta\) to see that for all \(i \in I^\bot_\tau\) the demand of \(\enumerateTerminal{i}\) is satisfied.
Note that this construction directly satisfies \Cref{item:tw-dp-bot-are-complete} of \Cref{def:dp-tw} for our solution.

Let \(i \in I^\times_\tau = I^\times_\alpha = I^\times_\beta\), \(\{F_{a,i}\} \coloneqq \inv\pi_a(Q_{a,i})\), and \(\{F_{b,i}\} \coloneqq \inv\pi_b(Q_{b,i})\).
Consider \(F^*_i \coloneqq F_a \cup F_b\).
As both \(F_a\) and \(F_b\) are connected and contain \(q_i\), \(F^*_i\) is connected as well.
Additionally, \(Q_{s,i} = \{q_i\} \cup (\enumerateTerminal{i} \cap Y_s) = \{q_i\} \cup (\enumerateTerminal{i} \cap Y_a)\cup(\enumerateTerminal{i} \cap Y_b) = Q_{a,i} \cup Q_{b,i} \subseteq V(F^*_i)\).
So, we assign \(F^*_i\) to \(Q_{s,i}\) in our solution to \(\mathscr{D}_{s,\tau}\), satisfying the demand of \(Q_{s,i}\), and concluding the description of our solution to \(\mathscr{D}_{s,\tau}\).

It remains to show that \Cref{item:tw-dp-crossing} of \Cref{def:dp-tw} is satisfied.
For this, let \(i \in I^\times_\tau\).
For each \(P \in \mathcal{P}_{\alpha,i}\) the vertices in \(P\) are in a connected component of \(F_a\) and for each \(P \in \mathcal{P}_{\beta,i}\) the vertices in \(P\) are in a connected component of \(F_b\).
Thus, if there is a path in \(H_i\) between two vertices in \(X_s\), there is a path in \(F^*_i\) between these vertices.
Now, let \(K \in \comp{F^*_i - q_i}\), \(u,v \in V(K)\cap X_s\), and \(P\) be a \(uv\)-path in \(K\).
Consider the sequence of \(x_1, x_2,\dots, x_\ell\) of vertices in \(X_s\) on \(P\).
For each \(i \in \natint{\ell - 1}\), there is a component in either \(F^*_i - q_i - (Y_b \setminus X_s) = F_{a, i} - q_i\) of \(F^*_i - q_i - (Y_a \setminus X_s) = F_{b, i} - q_i\) containing both \(x_i\) and \(x_{i+1}\).
That is, there is a \(R \in \mathcal{E}_i\) with \(x_i, x_{i+1} \in R\).
Therefore, \(x_i\) and \(x_{i+1}\) are in the same component of \(H_i\); in particular, \(u\) and \(v\) are in the same component of \(H_i\), showing that \Cref{item:tw-dp-crossing} of \Cref{def:dp-tw} is satisfied and that \(\tau \in D(s)\).
\end{proof}

This gives us an easy and efficient way to compute the dynamic program for join nodes.

\begin{corollary}
\label{stmt:tw-dp-join-running-time}
The dynamic program for a join node with children \(a\) and \(b\) can be computed in time \(\O{1 + ({\Sigma_{d}} + w)\abs{D(a)}\abs{D(b)}}\) given that \(D(a)\) and \(D(b)\) are provided.
\end{corollary}

\paragraph*{Forget node} Let \(s \in V(S)\) be a forget node, that is it has a child \(c\) and there is a \(v \in X_c\) with \(X_s = X_c\setminus \{v\}\).
Let \(E_v \coloneqq \{vx\}_{x \in N(v) \cap X_s}\) to be the set of edges incident to \(v\) and a vertex in \(X_s\).
Then, we have that \(G_s = G_c + E_v\).
Let \(\gamma \in D(c)\).
Consider a solution \((\mathcal{F}, \pi)\) to \(\mathscr{D}_{c, \gamma}\) satisfying the additional requirements of \Cref{item:tw-dp-solution} in \Cref{def:dp-tw}.
We need to distribute the additional edges among the solution subgraphs which cross \(X_c\).
Consider a function \(\lambda \colon E_v \to I^\times_\gamma \cup \{\blockedIndex\}\) where \(\blockedIndex\) is a symbol indicating that the edge is not used and for all \(i \in I^\times_\gamma\), we add the edges \(\inv\lambda(i)\) to \(\inv\pi(Q_{c,i})\).

To see how this effects the solution state, let \(i \in I^\times_\tau\) and \(\mathcal{E}_i\coloneqq \inv\lambda(i)\cup \mathcal{P}_i\).
Consider the hypergraph \(H_i \coloneqq (\{v\} \cup \mathcal{E}_i, \mathcal{E}_i)\) this is, intuitively speaking, the hypergraph on the vertices of \(X_c\) with the edges incident to \(v\) designated for this solution subgraph combined with the edges that correspond to connections already completed in \(G_c\).
Now, we set \(\mathcal{P}_i'\coloneqq \{V(K) \cap X_s \mid K \in \comp{H_i}\}\) to be the vertex sets of the components of \(H_i\) with \(v\) removed.
Set \(C_{\gamma,\lambda} \coloneqq \{i \in I^\times_\gamma \mid \mathcal{P}_i' = \{\emptyset\}\}\) to be the indices that cross \(X_c\) but not \(X_s\) under this edge distribution.
We now define the function \(\distribute\gamma\lambda \coloneqq (I^\bot_\gamma \cup C_{\gamma, \lambda}, I^\times_\gamma \setminus C_{\gamma,\lambda}, I^\top_\gamma, (\mathcal{P}_i')_{i \in I^\times_\gamma \setminus C_{\gamma,\lambda}})\).
This function describes the state of the dynamic program after forgetting \(v\) and distributing the newly available edges among the solution subgraphs according to \(\lambda\).
When distributing edges, we must watch out not to disconnect solution subgraphs that are not yet finished.

\begin{lemma}
\label{stmt:tw-dp-forget-correct}
We have \[
D(s) = \bigcup_{\substack{\gamma \in D(c),\\\lambda \colon E_v \to I^\times_\gamma \cup \{\blockedIndex\}\!\colon\\\forall i \in C_{\gamma,\lambda}\colon \enumerateTerminal{i}\subseteq Y_c}} \{\distribute{\gamma}{\lambda}\}.
\]
\end{lemma}
\begin{proof}
Let \(\tau \in D(s)\) and let \((\mathcal{F},\pi)\) be a solution to \(\mathscr{D}_{s, \tau}\) satisfying the additional requirements of \Cref{item:tw-dp-solution} in \Cref{def:dp-tw}.
We now show that there is are \(\gamma \in D(c)\) and \(\lambda \colon E_v \to I^\times_\gamma \cup \{\blockedIndex\}\) such that \(\distribute{\gamma}{\lambda} = \tau\) and that for all \(i \in C_{\gamma, \lambda}\), we have \(\enumerateTerminal{i} \subseteq Y_c\).
For this, we first define \(\gamma\) and \(\lambda\).

Denote with \(\mathcal{F}_v \coloneqq \{F \in \mathcal{F}\mid \{v\} = V(F)\cap X_c\}\) the set of solution subgraphs that contain \(v\) but none of \(X_s\).
For all \(F \in \mathcal{F}_v\), the demand of \(\pi(F)\) is exactly \(\abs{\invEnumerateTerminal{\pi(F)}\cap I^\bot_\tau}\).
So, we are able to choose a function \(\sigma \colon \mathcal{F}_v \to I^\bot_\tau\) such that for all \(F \in \mathcal{F}_v\), we have \(\pi(F) = \enumerateTerminal{\sigma(F)}\).
We now set \(I^\bot_\gamma \coloneqq I^\bot_\tau \setminus \image{\sigma}\), \(I^\times_\gamma\coloneqq I^\times_\tau \cup \image{\sigma}\), and \(I^\top_\gamma \coloneqq I^\top_\tau\).
Let \(i \in I^\times_\tau\).
We set \(\mathcal{P}_{\gamma,i} \coloneqq \{V(K) \cap X_c \mid K \in \comp{\inv\pi(Q_{s,i}) - E_v - q_i}\}\) to be the connected components of the solution subgraph assigned to \(Q_{s,i}\) with \(q_i\) and all edges adjacent to \(v\) removed.
For \(i \in \image{\sigma}\), we set \(\mathcal{P}_{\gamma, i} \coloneqq \{\{v\}\}\) to be the partitioning of \(X_c\) only containing the partition \(\{v\}\).
To define \(\alpha\), consider any \(e \in E_v\).
If there is a \(F \in \mathcal{F}_v\) with \(e \in E(F)\), we set \(\alpha(e) \coloneqq \sigma(e)\).
If there is a \(F \in \mathcal{F} \setminus \mathcal{F}_v\) with \(e \in E(F)\), let \(i \in I^\times_\tau\) be the unique \(i\) with \(q_i \in V(F)\).
We set \(\alpha(e) \coloneqq i\).
Otherwise, we set \(\alpha(e) \coloneqq \blockedIndex\).

To show \(\distribute{\gamma}{\lambda} = \tau\), we first note that \(C_{\gamma,\lambda} = \image{\sigma}\).
With this knowledge, we see that \(I^\bot_{\distribute{\gamma}{\lambda}} = I^\bot_\tau\), \(I^\times_{\distribute{\gamma}{\lambda}} = I^\times_\tau\), and \(I^\top_{\distribute{\gamma}{\lambda}} = I^\top_\tau\).
Now, let \(i \in I^\times_\tau\).
By definition, we know that \(\bigcup \mathcal{P}_{\distribute{\gamma}{\lambda}, i}= V(\inv\pi(Q_{s,i})) \cap X_s\); by \Cref{item:tw-dp-crossing} of \Cref{def:dp-tw}, we know \(V(\inv\pi(Q_{s,i})) \cap X_s = \bigcup \mathcal{P}_{\tau, i}\), yielding \(\bigcup \mathcal{P}_{\distribute{\gamma}{\lambda}, i} = \bigcup \mathcal{P}_{\tau, i}\).
It remains to show for \(x,y \in \bigcup \mathcal{P}_{\tau, i}\) that there is a \(P \in \mathcal{P}_{\tau,i}\) with \(x,y \in P\) if and only if there is a \(P' \in \mathcal{P}_{\distribute{\gamma}\lambda,i}\) with \(x,y \in P'\).
Assume there is a \(P \in \mathcal{P}_{\tau, i}\) with \(x,y \in P\).
By \Cref{item:tw-dp-crossing} of \Cref{def:dp-tw}, this is exactly the case if there is a \(K \in \comp{\inv\pi(Q_{s,i}) - q_i}\) with \(x,y \in V(K) \cap X_s\).
So, there is a path \(xy\)-path \(P\) in \(\inv\pi(Q_{s,i})\) not containing the vertex \(q_i\).
Consider the vertices \(q_1, q_2,\dots, q_\ell\) of \(P\) in \(X_c\) in order.
For all \(k \in \natint{\ell - 1}\), if \(q_k\) and \(q_{k+1}\) are directly connected by an edge, then, as \(X_s\) is an independent set in \(G_s\), either \(q_k = v\) or \(q_{k+1} = v\).
By definition of \(\alpha\), we have \(\alpha(q_kq_{k+1}) = i\); meaning \(\{q_k, q_{k+1}\} \in \mathcal{E}_i\).
Otherwise, \(q_k\) and \(q_{k+1}\) are connected by a path in \(\inv\pi(Q_{s,i}) - q_i - E_v\); so, there is a \(K \in \comp{\inv\pi(Q_{s,i}) - q_i - E_v}\) with \(q_k, q_{k+1} \in V(K) \cap X_c\).
As \(V(K)\cap X_c \in \mathcal{E}_i\), in either case \(q_k\) and \(q_{k+1}\) are adjacent by a hyperedge in \(\mathcal{E}_i\).
Thus, \(x\) and \(y\) are in the same connected component of \(H_i\) and there is a \(P' \in \mathcal{P}_{\distribute{\gamma}{\alpha},i}\) with \(x,y \in P'\).
Applying these arguments in reverse shows the claimed equivalence and, by extension, \(\distribute{\gamma}{\alpha} = \tau\).

Now, let \(i \in C_{\gamma, \lambda} = \image{\sigma}\).
If \(i \in I^\times_\tau\), let \(\{F\} \coloneqq \inv\pi(Q_{s,i})\).
We have that \(v \in V(F)\) and that \(q_i \in V(F)\).
Additionally, \(V(F)\) contains no vertex in \(X_s\).
As \(N_{G_{s,\tau}}(q_i) = X_s\), there is no \(vq_i\)-path contained in \(F\) and so \(F\) is disconnected, which is a contradiction.
Thus, \(i \in I^\bot_\tau\) and \(\enumerateTerminal{i} \subseteq V(G_s) = Y_s = Y_c\).

It remains to show that \(\gamma \in D(c)\).
We now show \Cref{item:tw-dp-started} and \Cref{item:tw-dp-solution} of \Cref{def:dp-tw} separately for \(\gamma\).
First, consider any \(T \in \mathcal{T}\) with \(T \cap Y_c \neq \emptyset\).
We know that \(Y_c = Y_s\); so, by \Cref{item:tw-dp-started} of \Cref{def:dp-tw}, we know that \(\invEnumerateTerminal{T} \subseteq I^\bot_{\tau} \cup I^\times_\tau\).
Noting that \(I^\bot_\gamma \cup I^\times_\gamma = I^\bot_\tau \cup I^\times_\tau\), we get \(\invEnumerateTerminal{T} \subseteq I^\bot_{\gamma} \cup I^\times_\gamma\). Hence \Cref{item:tw-dp-started} of \Cref{def:dp-tw} is satisfied for \(\gamma\).

To show \Cref{item:tw-dp-solution} of \Cref{def:dp-tw} for \(\gamma\), we construct a solution to \(\mathscr{D}_{c, \gamma}\) satisfying \Cref{item:tw-dp-bot-are-complete,item:tw-dp-crossing} of \Cref{def:dp-tw}.
Let \(i \in I^\bot_\gamma\).
Consider \(F \in \inv\pi(\enumerateTerminal{i}) \setminus \mathcal{F}_v\).
We notice that \(F\) is completely contained in \(G_c\).
We assign all such \(F\) to \(\enumerateTerminal{i}\).
This amounts to \(\abs{\inv\pi(\enumerateTerminal{i})} - \abs{\mathcal{F}_v \cap \inv\pi(F)} = \abs{\inv\pi(\enumerateTerminal{i})} - \abs{\inv\sigma(\enumerateTerminal{i})} = d_{\mathscr{D}_{s, \tau}}(\enumerateTerminal{i}) - \abs{\inv\sigma(\enumerateTerminal{i})}\) solution subgraphs that get assigned to \(\enumerateTerminal{i}\).
The demand of \(\enumerateTerminal{i}\) is \(\abs{\invEnumerateTerminal{\enumerateTerminal{i}} \cap I^\bot_\gamma} = \abs{\invEnumerateTerminal{\enumerateTerminal{i}} \cap (I^\bot_\tau \setminus \image{\sigma})}\).
As \(\image{\sigma} \subseteq I^\bot_\tau\), we have \(d_{\mathscr{D}_{c,\gamma}}(\enumerateTerminal{i}) = \abs{\invEnumerateTerminal{\enumerateTerminal{i}} \cap I^\bot_\tau} - \abs{\invEnumerateTerminal{\enumerateTerminal{i}} \cap \image{\sigma}} = d_{\mathscr{D}_{s, \tau}}(\enumerateTerminal{i}) - \abs{\invEnumerateTerminal{\enumerateTerminal{i}} \cap \image{\sigma}}\).
Additionally, \(\inv\sigma(\enumerateTerminal{i}) \subseteq \invEnumerateTerminal{\enumerateTerminal{i}}\) implies that \(\inv\sigma(\enumerateTerminal{i}) = \invEnumerateTerminal{\enumerateTerminal{i}} \cap \image{\sigma}\).
Thus, the demand of \(\enumerateTerminal{i}\) is satisfied by assigning all such \(F\).

Now, let \(i \in I^\times_\gamma\).
If \(i \in \image{\sigma}\).
Let \(\{F\} \coloneqq \inv\sigma(i)\) and set \(F^*_i \coloneqq F + vq_i\), which is connected as \(v \in V(F)\).
By choice of \(\sigma\), we have \(\pi(F) = \enumerateTerminal{i}\).
Thus, \(Q_{c,i} = \{q_i\} \cup (\enumerateTerminal{i} \cap Y_c) \subseteq \{q_i\} \cup \pi(F) \subseteq V(F^*_i)\).
Consequently, we assign \(F^*_i\) to \(Q_{c,i}\) in our solution.
If \(i \notin \image{\sigma}\), we have \(i \in I^\times_\tau\).
Let \(\{F\} \coloneqq \inv\pi(Q_{s,i})\).
We set \(F^*_i\coloneqq F - E_v + \{xq_i \mid x \in X_c\}\).
If \(F - E_v\) is disconnected, for every \(K \in \comp{F - E_v}\) there is an \(x \in V(K)\) that incident to an edge \(e \in E_v\).
As \(E_v\) only contains edges inside \(X_c\), we have \(x \in X_c\).
Thus, \(V(K) \cap X_c \neq \emptyset\) and in \(F^*_i\), every connected component contains \(q_i\) showing that \(F^*_i\) is connected.
By \Cref{item:tw-dp-solution} of \Cref{def:dp-tw}, this \(F^*_i\) is edge-disjoint from all \(\{F^*_j\}_{j \in I^\times_\gamma}\) and all other already assigned solution subgraphs.
Additionally, notice that \(Q_{c,i} \subseteq Q_{s,i} \subseteq V(F) \subseteq V(F^*_i)\); thus, we assign \(F^*_i\) to \(Q_{c,i}\) in our solution to \(\mathscr{D}_{c,\gamma}\) completing the description of our solution description.

We finally show that \Cref{item:tw-dp-bot-are-complete,item:tw-dp-crossing} of \Cref{def:dp-tw} are satisfied.
That \Cref{item:tw-dp-bot-are-complete} of \Cref{def:dp-tw} is satisfied directly follows from the fact that \((\mathcal{F}, \pi)\) satisfies this condition with respect to \(\tau\) and \(s\).
Now, let \(i \in I^\times_\gamma\).
It follows by construction and by the fact that \(\tau\) satisfies \Cref{item:tw-dp-crossing} of \Cref{def:dp-tw}, that \(F^*_i\) is the unique solution subgraph containing the vertex \(q_i\).
It remains to show that \(\mathcal{P}_{\gamma,i} = \{V(K) \cap X_c \mid K \in \comp{F^*_i - q_i}\}\).
If \(i \in \image{\sigma}\), let \(\{F\} \coloneqq \inv\sigma(i)\).
As \(F \in \mathcal{F}_v\), we have \(V(F) \cap X_c = \{v\}\).
So, \(\{V(K) \cap X_c \mid K \in \comp{F^*_i - q_i}\} = \{\{v\}\} = \mathcal{P}_{\gamma,i}\).
If \(i \notin \image{\sigma}\), we have that \(i \in I^\times_\tau\).
By definition, we have \(\mathcal{P}_{\gamma, i} = \{V(K) \cap X_c \mid K \in \comp{F^*_i - E_v - q_i}\}\).
Since, \(F^*_i = F - E_v + \{xq_i\mid x \in X_c\}\) there is no edge of \(E_v\) contained in \(F^*_i\); meaning that \(F^*_i - E_v -q_i = F^*_i - q_i\), which shows the desired equality and that \(\gamma \in D(c)\).

Now, let \(\gamma \in D(c)\) and \(\alpha \colon E_v \to I^\times_\gamma \cup \{\blockedIndex\}\) be given, such that for all \(i \in C_{\gamma, \lambda}\) we have that \(\enumerateTerminal{i} \subseteq Y_c\).
To show \(\tau \coloneqq\distribute{\gamma}{\alpha} \in D(s)\), we show \Cref{item:tw-dp-started,item:tw-dp-solution} of \Cref{def:dp-tw} separately.
First, let \(T \in \mathcal{T}\) with \(T \cap Y_s \neq \emptyset\) be given.
By \Cref{item:tw-dp-started} of \Cref{def:dp-tw} applied to \(\gamma\), we know that \(\invEnumerateTerminal{T} \subseteq I^\bot_\gamma \cup I^\times_\gamma\).
As \(I^\bot_\gamma \cup I^\times_\gamma = I^\bot_\tau \cup I^\times_\tau\), this shows \(\invEnumerateTerminal{T} \subseteq I^\bot_\tau \cup I^\times_\tau\).

Now, we provide a solution to \(\mathscr{D}_{s,\tau}\) satisfying the additional requirements of \Cref{item:tw-dp-solution} in \Cref{def:dp-tw}.
For this, let \((\mathcal{F}, \pi)\) be a solution to \(\mathscr{D}_{c, \gamma}\) satisfying the additional requirements of \Cref{item:tw-dp-solution} in \Cref{def:dp-tw}.
First, let \(i \in I^\bot_\gamma\) and consider \(F \in \inv\pi(\enumerateTerminal{i})\).
As \(F\) is completely contained in \(G_c\), it is also completely contained in \(G_s\) and we assign \(F\) to \(\enumerateTerminal{i}\) in our solution for \(\mathscr{D}_{s,\tau}\).
Now, let \(i \in C_{\gamma,\lambda}\) and set \(\{F\} \coloneqq \inv\pi(Q_{c,i})\).
By \Cref{item:tw-dp-crossing} of \Cref{def:dp-tw} and choice of \(C_{\gamma, \lambda}\), we know that \(V(F)\) is disjoint from \(X_s\).
Thus, \(F - q_i\) is contained in \(G_s\).
Since \(\enumerateTerminal{i} \subseteq Y_c\), we have \(\enumerateTerminal{i} = \enumerateTerminal{i} \cap Y_c = Q_{c,i} \setminus \{q_i\} \subseteq V(F - q_i)\).
Consequently, we assign \(F - q_i\) to \(\enumerateTerminal{i}\).

Now, let \(i \in I^\bot_\tau\).
We assign \(d_{\mathscr{D}_{c, \gamma}}(\enumerateTerminal{i}) + \abs{\invEnumerateTerminal{\enumerateTerminal{i}} \cap C_{\gamma, \lambda}}\) solution subgraphs to \(\enumerateTerminal{i}\) in our solution to \(\mathscr{D}_{s,\tau}\).
As \(d_{\mathscr{D}_{c, \gamma}}(\enumerateTerminal{i}) = \abs{\invEnumerateTerminal{\enumerateTerminal{i}} \cap I^\bot_\gamma}\) and since \(I^\bot_\gamma\) is disjoint from \(C_{\gamma, \lambda}\), we assign \(\abs{\invEnumerateTerminal{\enumerateTerminal{i}} \cap (I^\bot_\gamma \cup C_{\gamma,\lambda})} = \abs{\invEnumerateTerminal{\enumerateTerminal{i}} \cap I^\bot_\tau} = d_{\mathscr{D}_{s, \tau}}(\enumerateTerminal{i})\) solution subgraphs to \(\enumerateTerminal{i}\) satisfying its demand.

Finally, consider \(i \in I^\times_\tau = I^\times_\gamma \setminus C_{\gamma,\lambda}\).
Let \(\{F\} \coloneqq \inv\pi(Q_{c,i})\) and set \(F^*_i \coloneqq F + \inv\lambda(i) + \{xq_i\mid x \in X_s\}\).
Since \(i \notin C_{\gamma, \lambda}\), there is an \(x \in X_s \cap V(F)\) and so \(F^*_i\) is connected.
By \Cref{item:tw-dp-crossing} of \Cref{def:dp-tw}, we know that all \(\{F^*_j\}_{j \in I^\times_\tau}\) are edge-disjoint from each other and all previously assigned solution subgraphs.
Additionally, \(Q_{c,i} = Q_{s,i} \subseteq V(F) \subseteq V(F^*_i)\); so, we assign \(F^*_i\) to the terminal set \(Q_{s,i}\) concluding the description of our solution to \(\mathscr{D}_{s,\tau}\).
By construction this solution already satisfies \Cref{item:tw-dp-bot-are-complete} of \Cref{def:dp-tw}.

We also see that \(F^*_i\) is the unique solution subgraph containing \(q_i\).
It remains to show that \(\{V(K) \cap X_s \mid K \in \comp{F^*_i - q_i}\}= \mathcal{P}_{\distribute{\gamma}{\alpha}, i}\).
First, we see that \(\bigcup_{K \in \comp{F^*_i - q_i}} V(K) \cap X_s = V(F^*_i) \cap X_s\) and that \(\bigcup \mathcal{P}_{\distribute{\gamma}{\alpha},i} = \{V(K) \cap X_s \mid K \in \comp{H_i}\} = X_s \cap (\{v\} \cup \{x \mid vx \in \inv\lambda(i)\} \cup \bigcup \mathcal{P}_{\gamma,i})\).
The vertices \(V(F^*_i) \cap X_s\) are exactly the endpoints of all \(\inv\lambda(i)\) which are not \(v\) and the vertices \(V(F) \cap X_s\); which, by \Cref{item:tw-dp-crossing} of \Cref{def:dp-tw}, are exactly \(\bigcup\mathcal{P}_{\gamma,i}\).
Therefore, \(V(F^*_i) \cap X_s = \bigcup \mathcal{P}_{\distribute{\gamma}{\alpha},i}\).
We now show that for all \(x,y \in V(F^*_i) \cap X_s\) there is a \(K \in \comp{F^*_i - q_i}\) with \(x,y \in V(K)\) if and only if there is a \(P \in \mathcal{P}_{\distribute{\gamma}\alpha, i}\) with \(x,y \in P\), completing this proof.
Consider \(x,y \in V(F^*_i) \cap X_s\) such that there is a \(K \in \comp{F^*_i - q_i}\) with \(x,y \in V(K)\).
This is exactly the case when there is a simple \(xy\)-path \(P\) in \(F^*_i\) not containing \(q_i\).
Denote with \(x_1, x_2, \dots, x_\ell\) the vertices of \(X_c\) on \(P\) in order.
For all \(k \in \natint{\ell - 1}\), either \(x_k\) and \(x_{k+1}\) are directly connected by an edge in \(\inv\lambda(i)\), or there is a path connecting \(x_k\) and \(x_{k+1}\) in \(F\).
Therefore, \(x_k\) and \(x_{k+1}\) are adjacent in \(H_i\), meaning that \(x,y\) are in the same connected component of \(H_i\).
Applying, this reasoning in reverse shows that \(\{V(K) \cap X_s \mid K \in \comp{F^*_i - q_i}\}= \mathcal{P}_{\distribute{\gamma}{\alpha}, i}\).
\end{proof}

The last lemma allows us to easily and efficiently compute \(D(s)\) based on \(D(c)\).

\begin{corollary}
\label{stmt:tw-dp-forget-running-time}
The dynamic program for a forget node with child \(c\) can be computed in time \(\O{1 + ({\Sigma_{d}} + w)({\Sigma_{d}} + 1)^w\abs{D(c)}}\) given that \(D(c)\) is provided.
\end{corollary}

Based on \Cref{stmt:tw-dp-introduce-running-time,stmt:tw-dp-join-running-time,stmt:tw-dp-forget-running-time} we are able to compute the dynamic program for the whole tree decomposition.
To be able to easily compute leaf nodes, we modify the given tree decomposition by extending leaves that are not empty to actually introduce one of its vertices.
We do this until all leaf nodes are empty.
Finally, we extend the root node to forget its nodes until the root node is empty.
Using \Cref{stmt:tw-dp-output-correct}, this allows to decide whether the instance is positive.
We note that for all \(s \in V(S)\), we have \(\abs{D(s)} \leq 2^\O{{\Sigma_{d}} w\log w}\).
This allows us to bound the running time.

\begin{corollary}
\label{stmt:tw-dp-running-time}
Given a nice tree decomposition of width \(w\), we can compute whether the instance is positive in time \(\abs{V(S)}2^\O{{\Sigma_{d}} w \log w}\).
\end{corollary}

For each \(k \in \N\), we can either find in time \(\abs{V(G)}2^\O{k}\) a tree decomposition of width at most \(5k + 4\), or conclude that \(\tw{G} > k\)~\cite{BodlaenderDDFLP16}.
Additionally, any tree decomposition \((S', \mathcal{X}')\) of width \(w\) can be made nice in time \(\O{w^2\max(\abs{V(G)}, \abs{V(S'})}\) while using at most \(w\abs{V(G)}\) nodes~\cite{Kloks94}.
Combined with \Cref{stmt:tw-dp-running-time}, this concludes the proof of \Cref{stmt:gstp-fpt-tw-D}.


\section{Conclusion and Outlook}
\label{sec:conclusion}
In this paper, we provide the first fixed-parameter tractable algorithm for \textsc{Steiner Tree Packing} (\(\stp\)) parameterized by a structural parameter.
Concretely, we show that \(\stp\) is \(\fpt\) when parameterized by fracture number as well as tree-cut width.
This significantly extends the number of instances for which we know an exact polynomial time algorithm.
Previously known polynomial time algorithms are typically based on heuristics or approximations.
In case of the result that \(\stp\) is \(\fpt\) by \(\abs{T} + d\), we do not even know a concrete algorithm, but only that one exists~\cite{RobertsonS95b}.

To achieve this goal, we generalize the notion of the augmented graph from \textsc{Edge-Disjoint Paths} (\(\edp\)) to \textsc{Generalized Steiner Tree Packing} (\(\gstp\)) and \(\stp\).
This is the first result utilizing this tool on a problem where the terminals are arbitrary sets and not pairs of vertices.
The notion of augmentation has been used extensively for \(\edp\), but was originally introduced for \textsc{Multicut}~\cite{GottlobL07}.
Despite the fact that many parameterized complexity results for the generalized version of this problem (\textsc{Steiner Multicut}) are known~\cite{BringmannHML15}, the augmented graph has not yet been considered in this setting.
We think that augmentation will also prove to be a valuable tool for \textsc{Steiner Multicut} and other similar problems in future research.

Further, we extend all known \(\fpt\) algorithms for \(\edp\) parameterized by a structural parameter to \(\gstp\).
In addition, we provide a novel \(\fpt\) algorithm for \(\gstp\) parameterized by the tree-cut width of the augmented graph.
This settles whether \(\gstp\) is \(\fpt\) or \(\wonehard\) parameterized by all eight commonly used structural parameters described in \Cref{sec:prelims} with respect to the augmented graph as well as the host graph.
As all these results coincide between \(\edp\) and \(\gstp\), this also completes such a complexity classification for \(\edp\), where previously the result that \(\edp\) is \(\fpt\) by the tree-cut width of the augmented graph was not known.

For \(\stp\) the established results are almost as complete.
We prove for six of these eight parameters that \(\stp\) is \(\fpt\).
It is known, that \(\stp\) is \(\wonehard\) parameterized by treewidth, even if \(\abs{T} = 3\)~\cite{BodlaenderMOPL23,AazamiCJ12}.
So, the only question remaining here is whether \(\stp\) is \(\fpt\) parameterized by feedback vertex set number.
As \(\edp\) is \(\wonehard\) parameterized by the feedback vertex set number of the augmented graph~\cite{GanianOR21}, the approach employed in this paper—generalizing results from \(\edp\) to \(\gstp\) and applying them to \(\stp\)—is not suited to decide this questions.
Also, the techniques used by Bodlaender et~al.~\cite{BodlaenderMOPL23} to obtain the \(\wonehard\)ness result for \textsc{Integer 2-Commodity Flow} parameterized by treewidth, which generalizes to \(\stp\) with \(\abs{T} = 3\)~\cite{AazamiCJ12}, do not easily apply with respect to the feedback vertex set number.
We leave answering this question to future research.

Additionally, we improve upon the known running times for \(\edp\) parameterized by the fracture number of the augmented graph and the sum of treewidth and number of terminal pairs.
For the former algorithm, we do not see a clear path to improving the running time to sub-doubly-exponential and we would not be surprised, if this result is conditionally optimal up to improvements in the polynomial.
For the later algorithm, we think that our algorithm is a stepping stone towards further improvement.
In fact, we conjecture that an algorithm building on our idea running in time \(\Ostar{2^{\O{\tw{G}\sum_{T\in \mathcal{T}} d(T)}}}\) can be obtained using the Cut\&Count technique introduced by Cygan et~al.~\cite{CyganNPPRW11}.

\bibliography{papers}
\newpage
\appendix
\section{Finding a Fracture Modulator is in FPT}
\label{sec:find-frac-mod-in-fpt}
In this Chapter, we first examine why the known algorithm to find a fracture modulator in FPT-time in the size of the modulator is inaccurate.
Then, we present a fix to this algorithm.
This algorithm first appeared in the context of the incidence graph of ILP-instances~\cite{DvorakEGKO21} and was later republished for general graphs~\cite{GanianOR21}.
The algorithm is a given a graph \(G\) and an integer \(k\) and should now, decide whether is a fracture modulator in \(G\) of size \(k\).
It is based on two main ideas, which together yield a branching algorithm with a bounded search tree.

\begin{claim}[name=\cite{DvorakEGKO21,GanianOR21}]
\label{claim:fracture-modulator-observations}
\leavevmode
\begin{enumerate}
\item For every \(U \subseteq V(G)\) with \(\abs{U} = k+1\), if \(G[U]\) is connected then at least one vertex of \(U\) is contained in any fracture modulator of size \(k\).
\item Assume \(G\) is disconnected and let \(S\) be a fracture modulator of \(G\). Then, for every \(C \in \comp{G}\), the set \(S \cap V(C)\) is a fracture modulator for \(C\).
\end{enumerate}
\end{claim}

To see the inaccuracy, consider the graph \(P_5\).
We can verify that no set of size \(1\) is a fracture modulator in \(P_5\).
Now, let \(H\) be a graph that consists of two disconnected \(P_5\) graphs and denote with \(S^*\) the set of the two central vertices of each of the \(P_5\).
In \(H - S^*\) every connected component has at most two vertices.
Thus, \(S^*\) is a fracture modulator for \(H\), which uses exactly one vertex of each of the disconnected \(P_5\) graphs.
This refutes the claim by Dvorak et~al.~\cite{DvorakEGKO21}.

To fix this problem, we solve a slightly more general problem.
We call it \textsc{\((k,d)\)-Fracture Deletion}.
In it we are given a graph \(G\) and two natural number \(k, d \in \N\).
The task is now to decide, whether there is a set \(S \subseteq V(G)\) with \(\abs{S} = d\) such that every connected component of \(G - S\) has at most \(k\) vertices.
We call \(S\) a \emph{\(k\)-fracture deletion set} of size \(d\).
We see that checking whether a fracture modulator of size \(k' \in \N\) exists is equivalent to deciding \textsc{\((k', k')\)-Fracture Deletion}.
We can adopt \Cref{claim:fracture-modulator-observations} to this modified definition.

\begin{lemma}
\label{stmt:k-d-fracture-deletion-helper}
Let \(k,d \in \N\), \(G\) be a graph, and \(S\) be a \(k\)-fracture deletion with \(\abs{S} = d\).
\looseness=-1
Then,
\begin{enumerate}
\item for all \(U \subseteq V(G)\) with \(k < \abs{U}\) such that \(G[U]\) is connected, we have \(U \cap S \neq \emptyset\)
\item for all \(C \in \comp{G}\), the set \(S \cap V(C)\) is a \(k\)-fracture deletion set of \(C\).
\end{enumerate}
\end{lemma}
\begin{proof}
\leavevmode
\begin{enumerate}
\item Let \(U \subseteq V(G)\) with \(k < \abs{U}\) be given such that \(G[U]\) is connected and that \(U \cap S = \emptyset\).
Then, \(U\) is contained in a single connected component \(C\) of \(G - S\).
As \(U \subseteq V(C)\), we have \(k < \abs{V(C)}\) which violates that \(S\) is a \(k\)-fracture deletion set.
So, such an \(U\) can not exist.
\item Now, let \(C \in \comp{G}\).
For all \(D \in \comp{C - (S \cap V(C))}\), we also have \(D \in \comp{G - S}\).
Thus, \(V(D) \leq k\) and \(S \cap V(C)\) is a \(k\)-fracture deletion set of \(C\). \qedhere
\end{enumerate}
\end{proof}

Based on this, we can give a simple branching algorithm with a bounded search tree.
We distinguish three cases.
First, if \(\abs{V(G)} < d\), we report that no \(k\)-fracture deletion set of size \(d\) exists.
Second, if \(G\) is connected; find any \(U \subseteq V(G)\) with \(\abs{U} = k+1\) and \(G[U]\) connected.
We can find \(U\) using a breadth- or depth-first-search.
Now, we branch on every \(u \in U\), whether to include it in the solution.
We recursively decide whether \(G - u\) is a positive instance of \textsc{\((k, d-1)\)-Fracture Deletion} and output that \(G\) is a positive instance of \textsc{\((k, d)\)-Fracture Deletion} if and only if at least one subinstance is positive.
Finally, assume \(G\) is disconnected.
Let \(\mathcal{C} \coloneqq \{C \in \comp{G}\mid k < \abs{V(C)} \}\) be the components of \(G\) that need to be broken up.
If \(\abs{\mathcal{C}}=0\), we report that \(G\) is a positive instance of \textsc{\((k, d)\)-Fracture Deletion}.
If \(\abs{\mathcal{C}}=1\), let \(\{C\} \coloneqq \mathcal{C}\) and we report that \(G\) is a positive instance of \textsc{\((k, d)\)-Fracture Deletion} if and only if \(C\) is a positive instance of \textsc{\((k, d)\)-Fracture Deletion}.
If \(\abs{\mathcal{C}} > d\), we report that there is no \(k\)-fracture deletion set of size \(d\).
Otherwise, for each \(C \in \mathcal{C}\), we search recursively for the smallest \(d_C \in \natint{d - 1}\) such that a \(k\)-fracture deletion set of size \(d_C\) exists in \(C\).
If there is no such \(d_C\), report that there is no \(k\)-fracture deletion set in \(G\).
Finally, we output that \(G\) is a positive instance of \textsc{\((k, d)\)-Fracture Deletion} if and only if \(\sum_{C \in \mathcal{C}} d_C \leq d\).

\begin{theorem}
\label{stmt:k-d-fracture-deletion-algo}
Let \(k,d \in \N^+\), and \(G\) be a graph.
The presented algorithm correctly determines whether \(G\) is a positive instance of \textsc{\((k, d)\)-Fracture Deletion} and can be implemented in time \(\O{\max(k+1, 2d - 1)^d\abs{G}}\).
The algorithm can be modified to report a \(k\)-fracture deletion set of size \(d\) in the same time, if one exists.
\end{theorem}
\begin{proof}
We first focus on the correctness in case of termination, and after that bound the running time which proves termination.

Using \Cref{stmt:k-d-fracture-deletion-helper}, we see that the first two cases are correct.
Now, we focus on the third case.
When \(\abs{\mathcal{C}} = 0\), any subset of size \(d\) of \(V(G)\) is a \(k\)-fracture deletion set.
As \(\abs{V(G)} \geq d\), such a set exists and returning that the instance is positive in this case is correct.
If \(\abs{\mathcal{C}} = 1\), let \(\{C\} \coloneqq \mathcal{C}\).
Consider any \(k\)-fracture deletion set \(S\) of size \(d\) in \(G\).
We know that \(S \cap V(C)\) is a \(k\)-fracture deletion set in \(C\) and as \(\abs{S \cap V(C)} \leq d \leq \abs{\mathcal{C}}\), we know that we can extend \(S \cap V(C)\) to a \(k\)-fracture deletion set of size \(d\) in \(C\).
Additionally, any \(k\)-fracture deletion set of size \(d\) in \(C\) is a \(k\)-fracture deletion set of size \(d\) in \(G\).
So, the algorithm is correct in this case as well.
Regarding the case \(\abs{\mathcal{C}} > d\); we know from \Cref{stmt:k-d-fracture-deletion-helper}, that for each \(k\)-fracture deletion set \(S\) and \(C \in \mathcal{C}\), we have \(S \cap V(C) \neq \emptyset\).
Thus, any \(k\)-fracture deletion set has size at least \(\abs{\mathcal{C}}\); so, returning that the instance is negative when \(\abs{\mathcal{C}} > d\) is correct.

Now, for each \(C \in \mathcal{C}\) denote with \(S_C\) a \(k\)-fracture deletion set of size \(d_C\) and set \(S \coloneqq \bigcup_{C \in \mathcal{C}} S_C\).
We can verify that \(S\) is a \(k\)-fracture deletion set in \(G\) with \(\abs{S} = \sum_{C \in \mathcal{C}} d_C \leq d\).
To obtain a \(k\)-fracture deletion set of size \(d\), we add arbitrary vertices of \(G\) to \(S\) until \(\abs{S} = d\).
This again works as \(\abs{V(G)} \geq d\).
Assume that there is a \(k\)-fracture deletion set \(S\) of size \(d\).
By \Cref{stmt:k-d-fracture-deletion-helper}, we have \(\abs{\mathcal{C}} \leq d\).
For each \(C \in \mathcal{C}\), we have \(d_C \leq \abs{S \cap V(C)} = \abs{S} - \abs{S \setminus V(C)} \leq d - 1\).
Additionally, \(\sum_{C \in \mathcal{C}} d_C \leq \sum_{C \in \mathcal{C}} \abs{S\cap V(C)} \leq \abs{S} = d\) and the algorithm correctly reports that there is no \(k\)-fracture deletion set of size \(d\) if \(\sum_{C \in \mathcal{C}} d_C > d\).

To bound the running time, we bound the number of nodes in the recursion tree.
The first case does not make any recursive calls, while the second case makes at most \(k + 1\) recursive calls each decreasing \(d\) by one.
Now, consider the third case.
If \(\abs{\mathcal{C}} = 1\), we make one recursive call that does not decrease \(d\).
However, the instance \(C\), on which we recurse, is connected.
Thus, we make at most \(k+1\) recursive calls directly on \(C\), which all decrease \(d\) by at least one.
If \(\abs{\mathcal{C}} \geq 2\), we make recursive calls which all decrease \(d\) by at least one.
To bound, how many such recursive calls we make, we first need to specify how the search for \(d_C\) is carried out.
We choose to select each \(C \in \mathcal{C}\) after another and linearly search from 1 to \(d_C\).
Additionally, we abort when it is clear that the sum of \(d_C\) will exceed \(d\).
Each call to a negative instance increases the minimum sum by \(1\).
As this sum starts at \(\abs{\mathcal{C}}\), we do at most \(d - \abs{\mathcal{C}} + 1\) calls to instances which turn out to be negative.
Additionally, since we stop the search for each \(C \in \mathcal{C}\) once we recursively call a positive instance, we do at most \(d\) such calls.
Therefore, we do at most \(2d - \abs{\mathcal{C}} + 1 \leq 2d - 1\) such recursive calls.
This means, the branching factor is bounded by \(\max(k + 1, 2d - 1)\) and the branching depth is bounded by \(d\), which combined with the fact that at each recursion step we do at most a linear amount of additional work, yields the desired result.
\end{proof}

We can combine \Cref{stmt:k-d-fracture-deletion-algo} with the fact that for each \(k \in \N\) a fracture modulator of size \(k\) is a \(k\)-fracture deletion set of size \(k\) to obtain an \(\fpt\)-algorithm for finding minimum fracture modulator.

\begin{corollary}
\label{stmt:fracture-modulator-algo}
Let \(G\) be a non-empty graph and \(k\coloneqq\fracNum{G}\) be the fracture number of \(G\).
There is an algorithm that finds a fracture modulator of size \(k\) for \(G\) in time \(\O{(2k - 1)^k\abs{G}}\) if one exists.
\end{corollary}


\end{document}